\documentclass[a4paper,11pt]{article}

\usepackage[
  a4paper,
  bindingoffset=0cm,
  left=2.0cm,
  right=2.0cm,
  top=2.5cm,
  bottom=2.5cm,
  footskip=1.0cm
]{geometry}

\usepackage{amsthm}
\usepackage{amssymb}

\usepackage{graphicx}
\usepackage{dcolumn}
\usepackage{bm}
\usepackage[colorlinks=true,allcolors=blue,backref=page]{hyperref}
\hypersetup{colorlinks=true,breaklinks=true}
\renewcommand*{\backref}[1]{}

\renewcommand*{\backrefalt}[4]{%
  \ifcase #1%
  \or [p.~#2]%
  \else [pp.~#2]%
  \fi}
\usepackage{booktabs}

\usepackage{mathtools}
\usepackage{amsthm}
\usepackage{xcolor}
\usepackage{comment}
\usepackage{fancybox}
\usepackage{stmaryrd}
\usepackage{ytableau}
\usepackage{dsfont}
\usepackage{physics}
\usepackage{multirow}

\usepackage[normalem]{ulem}
\usepackage{cancel}
\usepackage{bm}

\newtheorem{theorem}{Theorem}

\newtheorem{lemma}[theorem]{Lemma}
\newtheorem{proposition}[theorem]{Proposition}
\newtheorem{corollary}[theorem]{Corollary}

\newcommand{\mcA}{\mathcal{A}}
\newcommand{\mcB}{\mathcal{B}}
\newcommand{\mcC}{\mathcal{C}}
\newcommand{\mcD}{\mathcal{D}}
\newcommand{\mcE}{\mathcal{E}}
\newcommand{\mcF}{\mathcal{F}}
\newcommand{\mcG}{\mathcal{G}}
\newcommand{\mcH}{\mathcal{H}}
\newcommand{\mcI}{\mathcal{I}}

\newcommand{\mcK}{\mathcal{K}}
\newcommand{\mcL}{\mathcal{L}}
\newcommand{\mcM}{\mathcal{M}}
\newcommand{\mcN}{\mathcal{N}}
\newcommand{\mcO}{\mathcal{O}}
\newcommand{\mcP}{\mathcal{P}}
\newcommand{\mcQ}{\mathcal{Q}}
\newcommand{\mcR}{\mathcal{R}}
\newcommand{\mcS}{\mathcal{S}}
\newcommand{\mcT}{\mathcal{T}}
\newcommand{\mcU}{\mathcal{U}}
\newcommand{\mcV}{\mathcal{V}}
\newcommand{\mcW}{\mathcal{W}}
\newcommand{\mcX}{\mathcal{X}}

\newcommand{\mcZ}{\mathcal{Z}}

\newcommand{\Wr}{\mathsf{Wr}}
\newcommand{\sfT}{\mathsf{T}}
\newcommand{\young}[2]{\mathsf{Y}_{#2}^{#1}}

\makeatletter
\newcommand{\doublewidetilde}[1]{{%
  \mathpalette\double@widetilde{#1}%
}}
\newcommand{\double@widetilde}[2]{%
  \sbox\tw@{$\m@th#1#2$}%
  \sbox\z@{$\m@th#1\widetilde{\copy\tw@}$}%
  \ht\z@=.9\ht\z@
  \widetilde{\copy\z@}%
}
\makeatother

\newcommand{\mfS}{\mathfrak{S}}

\newcommand{\U}{\mathsf{U}}
\newcommand{\irrep}[1]{\widehat{#1}}

\newcommand{\CC}{\mathbb{C}}
\newcommand{\ZZ}{\mathbb{Z}}
\newcommand{\RR}{\mathbb{R}}

\newcommand{\dbra}[1]{\langle\!\langle {#1} \vert}
\newcommand{\dket}[1]{\vert {#1} \rangle\!\rangle}

\NewDocumentCommand{\dketbra}{m g}{\vert {#1} \rangle\!\rangle\!\langle\!\langle {\IfNoValueTF{#2}{#1}{#2}} \vert}

\newcommand{\1}{\mathds{1}}

\newcommand{\Comm}{\mathrm{Comm}}

\DeclareMathOperator*{\argmin}{arg\,min}
\newcommand{\spec}{\operatorname{spec}}
\newcommand{\Prob}{\operatorname{Prob}}
\newcommand{\vol}{\operatorname{vol}}
\newcommand{\pol}{\operatorname{pol}}
\DeclareMathOperator{\poly}{poly}
\DeclareMathOperator{\polylog}{polylog}
\newcommand{\St}{\mathsf{St}}
\newcommand{\su}{\mathfrak{su}}
\newcommand{\id}{\mathrm{id}}

\newcommand{\W}{\mathsf{W}}
\newcommand{\Dil}{\mathsf{Dil}}
\newcommand{\Pur}{\mathsf{Pur}}
\newcommand{\sgn}{\operatorname{sgn}}
\newcommand{\supp}{\operatorname{supp}}

\newcommand{\PCovChan}{\mathsf{PCovChan}}
\newcommand{\CovChan}{\mathsf{CovChan}}

\begin{document}

\title{Optimal learning of covariant quantum states and channels}

\author{Satoshi Yoshida\thanks{Department of Physics, Graduate School of Science, The University of Tokyo, Hongo 7-3-1, Bunkyo-ku, Tokyo 113-0033, Japan. \texttt{satoshiyoshida.phys@gmail.com}}, Kazuki Okigami\thanks{OptQC Corp., 3-28-13 Nishi-Ikebukuro, Toshima-ku, Tokyo 171-0021, Japan. \texttt{kazuki.okigami@optqc.com}}, Pietro M. Posta\thanks{Department of Mathematical Sciences, University of Copenhagen, Denmark. \texttt{pmp@math.ku.dk}}, Dmitry Grinko\thanks{QuSoft \& Institute for Logic, Language and Computation \& Korteweg-de Vries Institute for Mathematics, University of Amsterdam, The Netherlands. \texttt{d.grinko@uva.nl}}}

\date{\today}

\maketitle

\begin{abstract}
  We give collective tomography protocols for quantum states and channels with known symmetries, using random purification and dilation to reduce learning to pure-state estimation.
  For states commuting with a compact-group representation with multiplicities $m_\lambda$, the optimal copy complexity is $\Theta((\sum_\lambda m_\lambda^2+\log\eta^{-1})/\varepsilon^2)$ for sufficiently small trace-distance error $\varepsilon$ and failure probability $\eta$ for the nontrivial case $\sum_\lambda m_\lambda^2>1$.
  For $G$-covariant channels with finite-dimensional unitary representations of a compact group $G$, parallel $\widetilde{O}(D_G/\varepsilon^2)$ queries achieve a diamond-distance error $\varepsilon $ at fixed success probability $2/3$, where $D_G$ counts the real parameters of covariant Choi operators before imposing trace preservation.
  For permutation-covariant channels on $k$ qudits of fixed local dimension $d$, at sufficiently small fixed error and fixed success probability $2/3$, we obtain the optimal query scalings $\Theta_d(k^{d^4-1})$ for diamond distance and $\Theta_d(k^{d^4-d^2})$ for Choi trace distance.
  Furthermore, we resolve an open problem in quantum tomography of constructing efficient quantum circuits that approximate the Hayashi measurement for optimal pure-state estimation.
  We encode states in the symmetric subspace into bosonic occupation modes, realize the measurement by heterodyne detection and normalization, and approximate this procedure on qubits using the quantum Hermite transform.
  Combined with our symmetry-compatible purification and dilation circuits, this yields query-optimal and gate-efficient learners of permutation-covariant states and channels with gate complexity $O_d(\mathrm{poly}(k,\varepsilon^{-1},\log\eta^{-1}))$.
\end{abstract}

\clearpage

\tableofcontents

\clearpage

\section{Introduction}
\label{sec:introduction}

Quantum tomography studies how we can learn a description of a quantum system from experiments in a mathematically optimal way.
Tomography protocols usually reconstruct a state from independently prepared copies, or a channel from its action on chosen inputs.
Early work established optimal collective measurements for pure states~\cite{massar1995optimal,hayashi1998asymptotic} and protocols for reconstructing quantum processes~\cite{chuang1997prescription,poyatos1997complete}.
Subsequent results determined the sample complexity of mixed-state tomography~\cite{haah2017sample,odonnell2016efficient,pelecanos2025mixed} and the query bounds for general channels~\cite{mele2025optimal,chen2026quantum}.
Unitary channels form an important special case, linking estimation of group transformations~\cite{chiribella2005optimal,chiribella2011group} to query-optimal reconstruction in diamond distance~\cite{haah2023query}.
For many-body systems, however, unrestricted tomography remains expensive because the Hilbert-space dimension grows exponentially with the number of sites.

However, knowing a symmetry can change this scaling.
Degrees of freedom of a state commuting with a group representation are determined by its multiplicity blocks, rather than an arbitrary matrix on the physical Hilbert space.
Permutation-invariant tomography already illustrates how this structure reduces experimental and computational requirements~\cite{toth2010permutationally,moroder2012permutationally}.
The corresponding channel condition is covariance: transforming the input by a known representation is equivalent to transforming the output by another.
For state tomography, the full sample has a wreath-product symmetry that combines independent symmetry transformations of each copy with permutations of the copies.
This brings us to the central questions which we aim to answer in this paper:
\begin{center}
\textit{How does symmetry determine the number of queries needed for learning, and do the resulting measurements admit efficient quantum circuits?}
\end{center}

It turns out that the random purification channel (RPC) \cite{tang2025conjugate} provides a route to resolve these two questions.
RPC converts $n$ copies of a mixed state into an ensemble of $n$ copies of one shared random purification~\cite{tang2025conjugate,girardi2025random}.
Estimating the pure representative via the Hayashi measurement~\cite{hayashi1998asymptotic} and discarding its reference system then yields a mixed-state estimate~\cite{pelecanos2025mixed}.
A general construction extends this transformation to states with prescribed symmetries~\cite{chen2026gaussian}.
For channels, random dilation similarly converts parallel channel queries into parallel uses of one shared random Stinespring isometry~\cite{girardi2025random2,yoshida2025random}.
Together, these reductions connect state and channel tomography to collective estimation of pure states and isometries.
However, these works left the question of efficient implementability of the Hayashi measurement unanswered.
In order to describe efficient protocols for covariant state and channel learning, we consider the following question:
\begin{center}
\textit{Does the Hayashi measurement admit efficient implementation?}
\end{center}

\subsection{Summary of results}

Our results determine how known symmetries reduce the cost of tomography and provide quantum circuits for the required measurements.
The state protocols process copies collectively, and the channel protocols use parallel queries.

\paragraph{Optimal learning of symmetric states.}\leavevmode
For states commuting with a known compact-group representation, the symmetry fixes part of the density operator and leaves smaller matrix blocks to be learned, with block sizes $m_\lambda$ and dimension parameter $D_\mathrm{st} \coloneqq \sum_\lambda m_\lambda^2$, where $m_\lambda$ are the multiplicities of the irreducible representations $\lambda$.
For any nontrivial such family, we provide a learning protocol achieving the optimal sample complexity $\Theta((D_\mathrm{st} + \log \eta^{-1})/\varepsilon^2)$ (Thm.~\ref{thm:state-learning-optimal}) for sufficiently small trace-distance error $\varepsilon$ and failure probability $0<\eta\leq1/3$.
For instance, permutation-invariant $k$-qudit states at fixed local dimension $d$ require a polynomial number of copies in $k$ at fixed error and success probability.
This protocol utilizes the symmetry-compatible random purification and the Hayashi measurement, both of which can be implemented efficiently for permutation symmetry (see also the paragraphs below).

\paragraph{Learning covariant channels.}
We provide a general reduction from learning covariant channels to learning covariant isometries.
We give a protocol for learning $G$-covariant channels with parallel queries that achieves the scaling $\widetilde{O}(D_G/\varepsilon^2)$ with diamond distance error $\varepsilon$ at fixed success probability (Cor.~\ref{cor:covariant-channel-tomography-high-probability}), where $D_G$ counts the real parameters of covariant Choi operators before imposing trace preservation.
When the real affine dimension of the set of $G$-covariant channels, denoted by $D''_G$, satisfies $D''_G = \Theta(D_G)$, we provide a matching lower bound $\Omega(D_G)$ up to polylogarithmic factors at sufficiently small fixed error and fixed success probability (Thm.~\ref{thm:permutation-channel-strong-diamond-lower}).
We also provide an upper bound on the query complexity for learning covariant channels with respect to the Choi trace-distance error, but the matching lower bound for general symmetry is still open.
For permutation-covariant channels on $k$ qudits of fixed local dimension $d$, the optimal query complexity is given by $\Theta_d(k^{d^4-1})$ for diamond-distance error and $\Theta_d(k^{d^4-d^2})$ for Choi trace-distance error at sufficiently small fixed error and fixed success probability (Thms.~\ref{thm:permutation-covariant-channel-sql-scaling}, \ref{thm:learning-permutation-covchannel-compression}, Cor.~\ref{cor:permutation-channel-strong-diamond-lower} and Thm.~\ref{thm:permutation-channel-trace-matching-lower}).
Similarly to the permutation-covariant state tomography, our construction is efficient in terms of gate complexity for permutation-covariant channel tomography.

\paragraph{Symmetry-compatible random purification and dilation.}
For finite groups $G$, we construct quantum circuits that convert copies of $G$-covariant states $\rho$ or channels $\Lambda$ into copies of one shared random purification or Stinespring isometry, respectively (Thms.~\ref{thm:random-covariant-purification} and~\ref{thm:random-covariant-dilation}), a construction that we further extend to compact groups in Sec.~\ref{sec:compact-oracle-implementation}.
The random purification and dilation utilize the wreath-product symmetry of $\rho^{\otimes n}$ and $\Lambda^{\otimes n}$, which combines independent symmetry transformations of each copy with permutations of the copies.
The circuit is implemented by using the quantum Fourier transform and the generalized Schur transform for the wreath product.
By using the recent construction of the quantum Fourier transform for the wreath product~\cite{bruinsma2026toolbox}, we show that the random purification and dilation can be implemented efficiently for permutation symmetry at fixed local dimension (Sec.~\ref{sec:random-purification-efficiency}).

\paragraph{Efficient implementation of the Hayashi measurement.}
We resolve the open problem of efficiently implementing the Hayashi measurement for optimal pure-state estimation (Thm.~\ref{thm:efficient-hayashi-measurement}).
The key observation is that, by encoding the symmetric subspace into a multi-mode bosonic system, the Hayashi measurement can be implemented by the heterodyne detection (Thm.~\ref{thm:hayashi-exact-heterodyne}).
We approximate this procedure on qubits using the quantum Hermite transform~\cite{jain2025hermite}.
The gate complexity for implementing the Hayashi measurement on $t$ copies of a $q$-dimensional system is given by
\begin{align}
    O((t+q)\polylog(t,q,1/\tau,1/\zeta)),
\end{align}
where $\tau$ and $\zeta$ are the implementation accuracy and failure probability, respectively.
For permutation-covariant states and channels, we combine the Hayashi measurement with our symmetry-compatible purification and dilation circuits to give learners whose gate complexity is polynomial in the number of sites, the inverse target error, and the logarithm of the inverse failure probability at fixed local dimension.

\paragraph{Optimal estimation of covariant isometries.}
We also characterize the minimum mean squared Choi trace-distance error of learning covariant isometries as $1-\lambda_{\max}(E_n)$, where $\lambda_{\max}(E_n)$ is the maximum eigenvalue of a certain matrix $E_n$ which only depends on the input- and output-space representations (Thm.~\ref{thm:learning-colored-estimation}).
Using this characterization, we determine the leading behavior of the minimum mean squared error (Thm.~\ref{thm:learning-leading-estimation}).

\subsection{Related work}

Previous work on permutation-invariant tomography uses symmetry to obtain polynomial block descriptions of states and protocols for their reconstruction~\cite{toth2010permutationally,moroder2012permutationally,klimov2013optimal,schwemmer2014experimental,marconi2026symmetric}.
Our state theorem gives a symmetry-based counterpart to the usual unrestricted tomography setting~\cite{haah2017sample,odonnell2016efficient,odonnell2017efficient,scharnhorst2025optimal}.

For unrestricted $d$-dimensional unitary channels, Haah et al.~\cite{haah2023query} prove the optimal $\Theta(d^2/\varepsilon)$ diamond-distance query complexity at constant success probability.
Unitary estimation is also connected to optimal parallel transposition under average fidelity~\cite{quintino2022deterministic}, and an $n$-query estimation protocol corresponds to deterministic port-based teleportation with $n+1$ ports~\cite{yoshida2026correspondence}.
The optimal fidelity of port-based teleportation can itself be expressed through the largest eigenvalue of a matrix indexed by Young diagrams~\cite{studzinski2017port,mozrzymas2018optimal}.
Our spectral characterization in Thm.~\ref{thm:learning-colored-estimation} extends the results for unrestricted-isometry estimation~\cite{yoshida2025quantum} to coupled multiplicity sectors, while random dilation transfers learning guarantees to general covariant channels.

The random purification channel~\cite{tang2025conjugate} converts state copies into copies of one shared random purification, and the random dilation superchannel~\cite{girardi2025random2,yoshida2025random} gives the analogous transformation for parallel channel queries.
Allowing a random purification avoids the obstruction to deterministic universal purification~\cite{liu2025no}.
Mele and Bittel~\cite{mele2025optimal} combine Choi-state preparation, random purification, and pure-state estimation for unrestricted channel learning in diamond distance.
Purification compatible with a prescribed symmetry is given in Ref.~\cite{chen2026gaussian}, but its application to our case requires the quantum Schur transform for the wreath-product representations, whose explicit construction remained open.
We provide an alternative construction of a finite-group circuit for symmetry-compatible purification for wreath-product representations, which we extend to random dilation by combining the state-channel duality~\cite{choi1975completely,jamiolkowski1972linear} with the transpose-trick circuit construction.
Their efficiency requires efficient Fourier transforms, controlled representation operations, and generalized Schur transforms, which we implement for permutation symmetry using the recent construction of the quantum Fourier transform for the wreath product~\cite{bruinsma2026toolbox}, which generalizes the corresponding construction for the symmetric group~\cite{beals1997quantum,kawano2016quantum,bruinsma2026symmetric}.

Turning these reductions into efficient learning algorithms also requires an implementation of the collective pure-state measurement.
Optimal collective measurements, including finite positive operator-valued measures (POVMs), are known~\cite{bruss1999optimal,hayashi1998asymptotic,hayashi2005pure,gill2000state,bagan2006optimal,pelecanos2025mixed}.
Among these, the Hayashi measurement~\cite{hayashi1998asymptotic} provides a covariant measurement that is crucial for our channel learning protocols in diamond distance since covariance enables us to bound the diamond distance error without incurring additional dimension-dependent factors~\cite{mele2025optimal}.
However, an efficient finite-qubit approximation of the Hayashi measurement with explicit gate-complexity bounds in the number of copies, state dimension, accuracy, and failure probability remained open~\cite{pelecanos2025mixed}.
Although Ref.~\cite{chiribella2007continuous} guarantees that a continuous POVM can be realized as a probabilistic mixture of finite-outcome POVMs, it does not provide an explicit construction of these finite POVMs.
References~\cite{shojaee2018optimal,jackson2019implement} provide an implementation of the Hayashi measurement using continuous-time weak measurements, but do not give gate-complexity bounds for an approximation over a fixed finite universal gate set with prescribed accuracy and failure probability.
Reference~\cite{mele2025optimal} provides an alternative implementation of the covariant measurement using the Haar twirling, but its efficient implementation is not known.
Reference~\cite{pelecanos2025mixed} utilizes unitary $t$-designs to approximate the twirling, but it only guarantees a high-probability bound on the trace-distance error, which is not sufficient for the diamond distance error guarantee in channel learning.
We address this question by constructing a finite-qubit approximation of the Hayashi measurement using the quantum Hermite transform~\cite{jain2025hermite}, resulting in a gate-efficient implementation of the Hayashi measurement for any finite number of copies.

\subsection{Discussion and outlook}

We discuss some open problems and directions for future work.

\paragraph{Matching bounds for channel tomography.}
While we provide matching upper and lower bounds for learning permutation-covariant channels at fixed error and success probability, the optimal query complexity for general symmetries remains open.
For diamond distance, our protocol achieves the optimal scaling $\widetilde{\Theta}(D_G)$ when there exists $\kappa=\Theta(1)$ satisfying $D''_G\geq \kappa D_G$, but it remains to determine the optimal query complexity when $\kappa=o(1)$.
A similar problem is investigated for general quantum channels with low Kraus rank~\cite{chen2026optimal}, and it would be interesting to understand how the scaling of $\kappa$ affects the transition between the Heisenberg and standard quantum limits for covariant channels.
For Choi trace distance, applying the learner for general symmetries to permutation-covariant channels gives the query complexity $O_d(k^{d^4-(d^2+1)/2})$ (see Thm.~\ref{thm:permutation-covariant-channel-sql-scaling}), which is not optimal since the optimal scaling is given by $\Theta_d(k^{d^4-d^2})$.
Even for the permutation-covariant channel, the matching lower bound for diamond-distance error is only established for sufficiently small fixed error, and it remains open to determine the optimal query complexity for general target error.

\paragraph{Efficient measurements for covariant unitary and isometry tomography.}
This work constructs an efficient circuit for the Hayashi measurement by embedding the symmetric subspace into multi-mode bosonic systems to implement the continuous-outcome measurement.
This is possible due to the high symmetry of the Hayashi measurement.
Similarly to pure-state tomography, unitary and isometry tomography also admits the optimal covariant measurements with continuous outcomes~\cite{chiribella2005optimal,chiribella2011group,yang2020optimal, yoshida2025quantum}.
We leave it an open problem to construct efficient circuits for the optimal covariant POVMs used in unitary and isometry tomography.

\subsection{Organization of the paper}
After the preliminaries, Sec.~\ref{sec:random-purification-dilation} constructs random purification for states and random dilation for channels.
Section~\ref{sec:learning-covariant-states} uses the wreath-product symmetry of $\rho^{\otimes n}$ to establish optimal state-learning bounds and compare them with standard tomography.
Section~\ref{sec:learning-covariant-channels} develops the channel-learning upper and lower bounds.
Section~\ref{sec:hayashi-qubit-implementation} implements the Hayashi measurement efficiently.
Appendix~\ref{appendix:efficient-wreath-schur} gives the generalized Schur transform circuits for wreath-product representations.
Appendix~\ref{appendix:learning-covariant-isometries-choi-trace-norm} derives the optimal covariant-isometry estimation in mean squared Choi trace distance error.
Appendices~\ref{app:state-tomography-lower} and~\ref{appendix:lower-bound-channel} prove the state- and channel-tomography converses.
Appendix~\ref{sec:uniform-fixed-height-young-estimates} provides the asymptotic estimates of quantities related to Young diagrams used in the upper and lower bounds.

\section{Preliminaries}
\label{sec:preliminaries}

\subsection{Notation}
We denote $[n]\coloneqq \{1, \ldots, n\}$, and $\log(\cdot)$ denotes the natural logarithm.
We denote the set of linear operators on a Hilbert space $\mcH$ by $\mcL(\mcH)$, and the identity operator on $\mcH$ as $\1_\mcH$.
We write the probability of an event $E$ as $\Prob[E]$ and the expectation of a random variable $X$ as $\mathbb{E}[X]$.
We utilize Landau's big-$O$ notation $O(\cdot)$, the big-$\Omega$ notation $\Omega(\cdot)$, the big-$\Theta$ notation $\Theta(\cdot)$, and the soft big-$O$ and big-$\Theta$ notations $\widetilde{O}(\cdot)$ and $\widetilde{\Theta}(\cdot)$, which suppress polylogarithmic factors.
We put a variable $x$ on the subscripts of the big-$O$ notation, e.g., $O_x(\cdot)$, to indicate that the hidden constant may depend on $x$, and do the same for the big-$\Omega$ and big-$\Theta$ notations.

\subsection{Distance measures for quantum states and channels}
\label{sec:distance-measures}

We use the trace distance between two quantum states $\rho$ and $\sigma$ to quantify the estimation error in quantum state tomography, which is defined by
\begin{align}
    \label{eq:prelim-state-distance-measures}
    d_\mathrm{tr}(\rho,\sigma) \coloneqq {1\over 2}\|\rho-\sigma\|_1,
\end{align}
where $\|\cdot\|_1$ is the trace norm defined by
\begin{align}
    \|X\|_1 \coloneqq \Tr\sqrt{X^\dagger X}.
\end{align}
The trace distance quantifies the distinguishability of two quantum states~\cite{helstrom1967detection}.

For process tomography, we use the diamond distance and Choi trace distance\footnote{We adopt this terminology from Ref.~\cite{chen2026quantum}, with a small modification; Reference~\cite{chen2026quantum} defined the Choi trace norm as $\|\rho_\Lambda-\rho_\Gamma\|_1$ without the $1/2$ factor.} to quantify the estimation error between two quantum channels $\Lambda$ and $\Gamma$, which are defined by
\begin{align}
    \label{eq:permutation-learning-losses}
    d_\diamond(\Lambda,\Gamma) \coloneqq {1\over 2}\|\Lambda-\Gamma\|_\diamond, \qquad
    d_\mathrm{tr}(\rho_\Lambda,\rho_\Gamma),
\end{align}
respectively, where $\|\cdot\|_\diamond$ is the diamond norm defined by~\cite{kitaev1997quantum}
\begin{align}
    \|\Lambda\|_\diamond \coloneqq \sup_{r\geq 1} \sup_{\substack{X\in\mcL(\CC^r\otimes\mcI)\\X\neq 0}} {\|(\1_{\mcL(\CC^r)}\otimes\Lambda)(X)\|_1 \over \|X\|_1},
\end{align}
for $\Lambda: \mcL(\mcI) \to \mcL(\mcO)$, and $\rho_\Lambda$ and $\rho_\Gamma$ are the normalized Choi states of the channels $\Lambda$ and $\Gamma$, respectively, defined by~\cite{choi1975completely,jamiolkowski1972linear}
\begin{align}
    \label{eq:def-choi-state}
    \rho_\Lambda \coloneqq {J_\Lambda \over d_\mcI}, \qquad 
    J_\Lambda \coloneqq \sum_{i,j=1}^{d_\mcI} \ketbra{i}{j}_{\overline{\mcI}} \otimes \Lambda(\ketbra{i}{j})_{\mcO},
\end{align}
using the Choi operators $J_\Lambda$ and the computational basis $\{\ket{i}\}_{i=1}^{d_\mcI}$ of $\mcI$ with $d_\mcI\coloneqq \dim \mcI$.
The diamond distance quantifies the worst-case error of the quantum channels, while the Choi trace distance quantifies the average-case error~\cite{rosenthal2024quantum}.
For isometry channels $\mcW(\cdot) \coloneqq W(\cdot) W^\dagger$ and $\mcW'(\cdot) \coloneqq W'(\cdot)W'^\dagger$ for $W, W': \mcI \to \mcO$, we denote the diamond distance as
\begin{align}
    d_\diamond(W,W')\coloneqq d_\diamond(\mcW,\mcW').
\end{align}
The normalized Choi states are denoted as
\begin{align}
    \rho_W\coloneqq\dketbra{W}/d_\mcI, \qquad
    \rho_{W'}\coloneqq\dketbra{W'}/d_\mcI,
\end{align}
where $\dket{W}$ is the Choi vector of the isometry $W$ defined by
\begin{align}
    \dket{W} \coloneqq \sum_{i=1}^{d_\mcI} \ket{i}_{\overline{\mcI}} \otimes W\ket{i}_\mcO.
\end{align}
Then, the Choi trace distance between the isometry channels is denoted by
\begin{align}
    d_\mathrm{tr}(\rho_W,\rho_{W'}).
\end{align}

In the analysis of the learning errors, we utilize the $q$-norm of an operator $X\in \mcL(\mcH)$ given by
\begin{align}
    \|X\|_q \coloneqq \qty[\Tr (X^\dagger X)^{q/2}]^{1/q},
\end{align}
and the operator norm given by
\begin{align}
    \|X\|_\infty \coloneqq \sup_{\ket{\psi}\in\mcH, \|\ket{\psi}\|=1} \|X\ket{\psi}\|.
\end{align}
We also utilize the total variation distance between two probability distributions $p$ and $q$ over a finite set $\mathsf{X}$, which is defined by
\begin{align}
    d_\mathrm{TV}(p,q) \coloneqq {1 \over 2} \sum_{x\in\mathsf{X}} \abs{p(x) - q(x)}.
\end{align}

\subsection{Schur--Weyl duality and Young diagrams}\label{sec:schur-weyl-basis-conventions}

We use the standard notation for the Schur--Weyl duality and Young diagrams (see the standard textbooks, e.g., Refs.~\cite{fulton1997young,georgi2000lie,ceccherini2010representation} for more details).
We consider the representation of the symmetric group $\mfS_n$ on the $n$-fold tensor product $(\CC^d)^{\otimes n}$ of the $d$-dimensional Hilbert space $\CC^d$ given by
\begin{align}
    \label{eq:def-permutation-operator}
    \pi_d(\sigma) \ket{i_1\cdots i_n} \coloneqq \ket{i_{\sigma^{-1}(1)} \cdots i_{\sigma^{-1}(n)}},
\end{align}
using the computational basis $\{\ket{i}\}_{i=1}^d$ of $\CC^d$.
We also consider the representation of the unitary group $\U(d)$ on $(\CC^d)^{\otimes n}$ given by
\begin{align}
    U^{\otimes n}.
\end{align}
The representations $\pi_d: \mfS_n \to \mcL((\CC^d)^{\otimes n})$ and $U\mapsto U^{\otimes n}$ admit the following simultaneous decomposition into tensor products of irreducible representations (irreps) of $\mfS_n$ and $\U(d)$:
\begin{align}
    \label{eq:schur-weyl-decomposition}
    (\CC^d)^{\otimes n} \cong \bigoplus_{\lambda \vdash_d n} \mcS_\lambda \otimes \mcU_\lambda^{(d)}, \qquad \pi_d(\sigma) \cong \bigoplus_{\lambda \vdash_d n} g_\lambda(\sigma) \otimes \1_{\mcU_\lambda^{(d)}}, \qquad U^{\otimes n} \cong \bigoplus_{\lambda \vdash_d n} \1_{\mcS_\lambda} \otimes f_\lambda^{(d)}(U),
\end{align}
where $\lambda\vdash_d n$ denotes a Young diagram of size $n$ with at most $d$ rows, represented by a tuple $\lambda = (\lambda_1, \ldots, \lambda_d)\in\ZZ^d$ satisfying $\lambda_1\geq \cdots \geq \lambda_d \geq 0$ and $\sum_{i=1}^d \lambda_i = n$.
The space $\mcS_\lambda$ is the representation space of the irrep $g_\lambda$ corresponding to the Young diagram $\lambda$, and $\mcU_\lambda^{(d)}$ is the multiplicity space of $g_\lambda$ in $(\CC^d)^{\otimes n}$, which is also the representation space of the corresponding irrep $f_\lambda^{(d)}$ of $\U(d)$.
Throughout, we equip each space $\mcS_\lambda$ with the orthonormal Young--Yamanouchi basis, indexed by standard Young tableaux. We fix the Young orthogonal matrix conventions used for the symmetric-group irreps in Ref.~\cite{bruinsma2026symmetric}. The matrices $g_\lambda(\sigma)$, the symmetric-group Fourier transforms, and the symmetric-group irrep registers of Schur transforms all use this convention.
We also define the set of Young diagrams of size $n$ with at most $d$ rows as
\begin{align}
    \young{d}{n} \coloneqq \{\lambda \mid \lambda\vdash_d n\}.
\end{align}
We denote the dimension of the irrep $g_\lambda$ by $d_\lambda \coloneqq \dim \mcS_\lambda$ and the dimension of the multiplicity space $\mcU_\lambda^{(d)}$ by $m_\lambda^{(d)} \coloneqq \dim \mcU_\lambda^{(d)}$.
The dimensions $d_\lambda$ and $m_\lambda^{(d)}$ can be computed by the Frobenius formula and the Weyl dimension formula given by
\begin{align}
    \label{eq:hook-length}
    d_\lambda&={n! \prod_{1\leq i<j\leq d}(\widetilde{\lambda}_i - \widetilde{\lambda}_j) \over \prod_{i=1}^{d} \widetilde{\lambda}_i!},\\
    \label{eq:weyl-dimension}
    m_\lambda^{(d)}&=\prod_{1\leq i<j\leq d}{\widetilde{\lambda}_i - \widetilde{\lambda}_j \over j-i},
\end{align}
where $\widetilde{\lambda}_i \coloneqq \lambda_i + d - i$ for $i \in [d]$.
We define the quantum Schur transform $U_\mathrm{Sch}^{(n,d)}$ by the unitary corresponding to the isomorphism in Eq.~\eqref{eq:schur-weyl-decomposition}.
Then, we define the Young projector $P_\lambda$ by
\begin{align}
    \label{eq:def-young-projector}
    P_\lambda \coloneqq U_\mathrm{Sch}^{(n,d)\dagger} \qty(\1_{\mcS_\lambda} \otimes \1_{\mcU_\lambda^{(d)}}) U_\mathrm{Sch}^{(n,d)}.
\end{align}

Similarly to the decomposition~\eqref{eq:schur-weyl-decomposition} of $U^{\otimes n}$ for $U\in \U(d)$, we can also consider the decomposition of $W^{\otimes n}$ for an isometry $W:\CC^d\to\CC^D$ with $D\geq d$.
Since $\pi_D(\sigma)^\dagger W^{\otimes n} \pi_d(\sigma) = W^{\otimes n}$ holds for all $\sigma\in\mfS_n$, the operator $W^{\otimes n}$ is decomposed as~\cite{yoshida2023universal}
\begin{align}
    \label{eq:prelim-isometry-schur-decomposition}
    W^{\otimes n} \cong \bigoplus_{\lambda \vdash_d n} \1_{\mcS_\lambda} \otimes f_\lambda^{(d\to D)}(W),
\end{align}
where $f_\lambda^{(d\to D)}(W): \mcU_\lambda^{(d)} \to \mcU_\lambda^{(D)}$ is an isometry.

\subsection{Hayashi measurement and optimal pure-state estimation}
\label{subsec:prelim-hayashi}

We consider the task of learning an unknown pure state $\ket{\psi}\in \CC^q$ for $q\geq 2$ using $t$ copies of the state, i.e., we apply a measurement on $\ket{\psi}^{\otimes t}$ and output an estimate $\ket{u}$ of the state $\ket{\psi}$.
Since the input state $\ket{\psi}^{\otimes t}$ is supported on the symmetric subspace of $(\CC^q)^{\otimes t}$, denoted by $\operatorname{Sym}^t(\CC^q)$, we can restrict the measurement to this subspace without loss of generality.
The \emph{Hayashi measurement}~\cite{hayashi1998asymptotic} is a continuous-outcome measurement represented by a positive operator-valued measure (POVM) on $\operatorname{Sym}^t(\CC^q)$ defined by
\begin{align}
    H_{t,q}(\mathrm{d}u)
    \coloneqq D_{t,q} \ketbra{u}^{\otimes t} \dd u, \qquad D_{t,q} \coloneqq \dim\operatorname{Sym}^t(\CC^q) = \binom{t+q-1}{t},
    \label{eq:prelim-hayashi-povm}
\end{align}
where $\dd u$ is the normalized Haar measure on the unit sphere of $\CC^q$.
It provides the optimal expected fidelity for pure-state estimation, which is given by~\cite{hayashi1998asymptotic,bruss1999optimal}
\begin{align}
    \mathbb{E}\qty[\abs{\braket{\psi}{u}}^2] = {t+1 \over t+q}.
    \label{eq:prelim-hayashi-mean-fidelity}
\end{align}
In addition, due to the covariance of the measurement, the probability distribution of the measurement outcome $\ket{u}$ is given by~\cite{pelecanos2025mixed,mele2025optimal}
\begin{align}
\ket{u}
=e^{i\phi}\sqrt{1-\delta}\ket{\psi}
+\sqrt\delta\ket{r},
\label{eq:covariant-additive-confidence-beta-output-prelim}
\end{align}
where $1-\delta\sim\mathrm{Beta}(t+1,q-1)$,
$\phi$ is uniformly distributed on $[0,2\pi)$, and $\ket{r}$ is
Haar distributed on the unit sphere of $\ket{\psi}^{\perp}$.
The random variables $\delta$, $\phi$, and $\ket{r}$
are mutually independent.
This property is crucial in the diamond-distance error analysis of the channel tomography~\cite{mele2025optimal}.
We can also derive a high-probability guarantee for the Hayashi measurement, which is useful in the analysis of the sample complexity of quantum state tomography\footnote{A similar result is shown in Ref.~\cite{mele2025optimal}, but the constants are taken differently.}.

\begin{proposition}[High-probability guarantee for the Hayashi measurement]
\label{prop:prelim-hayashi-confidence}
Let $q\geq2$ and $t\geq1$ be integers, and let $0<\varepsilon,\eta<1$.
Assume that the number of copies satisfies
\begin{align}
    t\geq\left\lceil\frac{2(q-1)+8\log(1/\eta)}{\varepsilon^2}\right\rceil.
    \label{eq:prelim-hayashi-confidence}
\end{align}
Then, for every unit vector $\ket\psi\in\CC^q$, the outcome $\ket u$ of the Hayashi measurement on $\ket\psi^{\otimes t}$ satisfies
\begin{align}
    \Prob_\psi\!\left[1-|\braket{\psi}{u}|^2>\varepsilon^2\right]\leq\eta.
    \label{eq:prelim-hayashi-high-probability}
\end{align}
\end{proposition}

\begin{proof}
By Eq.~\eqref{eq:covariant-additive-confidence-beta-output-prelim}, the infidelity $\delta=1-|\braket{\psi}{u}|^2$ has distribution $\operatorname{Beta}(q-1,t+1)$ for every $\psi$.
The beta--binomial tail identity~\cite[Eq.~(6.6.4)]{abramowitz1964handbook} therefore gives
\begin{align}
    \Prob_\psi[\delta>\varepsilon^2]=\Prob[B\leq q-2],\qquad B\sim\operatorname{Bin}(t+q-1,\varepsilon^2).
    \label{eq:prelim-hayashi-exact-tail}
\end{align}
Under Eq.~\eqref{eq:prelim-hayashi-confidence}, the mean $\mathbb{E}[B]=(t+q-1)\varepsilon^2$ satisfies $\mathbb{E}[B]\geq2(q-1)+8\log(1/\eta)$, so $q-2<\mathbb{E}[B]/2$.
Thus, the Chernoff bound yields
\begin{align}
    \Prob[B\leq q-2]\leq\Prob[B\leq \mathbb{E}[B]/2]\leq e^{-\mathbb{E}[B]/8}\leq\eta,
\end{align}
which concludes the proof.
\end{proof}

\section{Random purification of $G$-invariant states and $G$-covariant channels}
\label{sec:random-purification-dilation}
\subsection{$G$-invariant quantum states and their purifications}
Suppose $G$ is a compact group and $U:G\to\U(\mcH)$ is a unitary representation on $\mcH\simeq\CC^d$.
Then, due to the Peter--Weyl theorem, the representation $U$ decomposes into irreducible representations (irreps) as
\begin{align}
    \label{eq:schur-G}
    \mcH &\simeq\bigoplus_{\lambda\in\irrep{G}}\mcR_\lambda\otimes\mcM_\lambda,\\
    U(g) &\simeq\bigoplus_{\lambda\in\irrep{G}}U_\lambda(g)\otimes\1_{\mcM_\lambda},
    \label{eq:single-copy-representation-decomposition}
\end{align}
where $\irrep{G}$ is the set of labels of the irreducible representations of $G$, $U_\lambda: G \to \U(\mcR_\lambda)$ is an irrep labeled by $\lambda\in\irrep{G}$, and $\mcM_\lambda$ is the multiplicity space of $U_\lambda$.
For each $\lambda$, we fix an orthonormal basis of $\mcR_\lambda$ and regard $U_\lambda(g)$ as a fixed matrix in this basis. All local isotypic decompositions use these same matrices; when $G$ is finite, the $G$-Fourier transforms used below also use these representatives and bases. For $G=\mfS_k$, we use the Young--Yamanouchi convention fixed in Section~\ref{sec:schur-weyl-basis-conventions}.
We call the homomorphism~\eqref{eq:schur-G} the generalized Schur transform for the group $G$, denoted by $U_\mathrm{Sch}^{(G \curvearrowright \mcH)}$.
We abbreviate it as $U_\mathrm{Sch}^{(G)}$ when the Hilbert space $\mcH$ is clear from the context.
We define the dimension of the irrep $U_\lambda$ as $d_\lambda\coloneqq \dim\mcR_\lambda$ and the multiplicity of $U_\lambda$ as $m_\lambda\coloneqq \dim\mcM_\lambda$.
We consider a set of $G$-invariant quantum states defined by
\begin{align}
    \mcS_G(\mcH)\coloneqq \left\{\rho\in\mcL(\mcH) \;\middle | \; \rho\succeq 0,\ \Tr\rho=1,\ U(g)\rho U(g)^\dagger=\rho \quad \forall g\in G \right\}.
\end{align}
Then, due to Schur's lemma, any $G$-invariant quantum state $\rho\in\mcS_G(\mcH)$ can be expressed as
\begin{align}
    \label{eq:invariant-state-block-decomposition}
    \rho = \sum_{\lambda\in\irrep{G}} {\1_{\mcR_\lambda} \over d_\lambda} \otimes \rho_\lambda,
\end{align}
where $\rho_\lambda\in \mcL(\mcM_\lambda)$ satisfies $\rho_\lambda\succeq 0$ and $\sum_\lambda\Tr\rho_\lambda=1$.

For a $G$-invariant quantum state $\rho\in \mcS_G(\mcH)$, we define the canonical purification $\dket{\sqrt{\rho}}$ in the doubled Hilbert space $\mcE\otimes\mcH$ for $\mcE\simeq \overline{\mcH}$ by
\begin{align}
\dket{\1_d}&\coloneqq \sum_{i=1}^d\ket{\bar{i}}_{\mcE} \otimes \ket{i}_{\mcH},\\
\dket{\sqrt{\rho}}&\coloneqq (\1_{\mcE}\otimes\sqrt{\rho})\dket{\1_d},
\label{eq:vectorization-convention}
\end{align}
where $\{\ket{i}\}_{i=1}^d$ is the computational basis of $\mcH$ and $\{\ket{\bar{i}}\}_{i=1}^d$ is the corresponding basis of $\mcE$.
Then, $G$-invariance of $\rho$ is equivalent to
\begin{align}
    (\overline{U}(g)\otimes U(g))\dket{\sqrt{\rho}}=\dket{\sqrt{\rho}}\quad \forall g\in G.
\end{align}
We define the set of $G$-invariant purifications of $\rho$ as
\begin{align}
    \Pur_G(\rho)\coloneqq \left\{\ket{\psi}\in\mcE\otimes\mcH \; \middle | \; \Tr_{\mcE}\ketbra{\psi}=\rho,\ (\overline{U}(g)\otimes U(g))\ket{\psi}=\ket{\psi}\quad \forall g\in G\right\}.
\end{align}
Then, this set is characterized by the following lemma.

\begin{lemma}[$G$-invariant purification]
\label{lem:invariant-purifications}
For every $G$-invariant density operator $\rho \in \mcS_G(\mcH)$, the set of $G$-invariant purifications $\Pur_G(\rho)$ is given by
\begin{align}
\Pur_G(\rho)=\left\{(\overline{V}\otimes\1_{\mcH})\dket{\sqrt{\rho}}\; \middle | \; V\in \U_G(\mcH)\right\},
\label{eq:purification-orbit}
\end{align}
where $\U_G(\mcH)$ is the group of unitaries commuting with $U(G)$, defined by
\begin{align}
    \U_G(\mcH)\coloneqq \left\{V\in \U(\mcH)\; \middle | \; [V, U(g)] = 0 \quad \forall g\in G\right\}.
\end{align}
\end{lemma}

\begin{proof}
    Every purification of $\rho$ has the form $(\overline{V}\otimes\1_{\mcH})\dket{\sqrt{\rho}}$ for some $V\in \U(\mcH)$.
    The $G$-invariance is equivalent to
    \begin{align}
        (U(g)VU(g)^\dagger-V)\sqrt{\rho}=0
    \quad \forall g\in G.
    \end{align}
    Defining $\mcS\coloneqq \supp \rho$ and $\mcT\coloneqq V(\mcS)$, the above condition implies $V\vert_{\mcS}: \mcS \to \mcT$ is a unitary $G$-intertwiner.
    We take a unitary $G$-intertwiner $V_\perp: \mcS^\perp \to \mcT^\perp$ and define $\widetilde{V} \coloneqq V\vert_{\mcS} \oplus V_\perp$.
    Then, $\widetilde{V} \in \U_G(\mcH)$ and $(\overline{\widetilde{V}}\otimes\1_{\mcH})\dket{\sqrt{\rho}} = (\overline{V}\otimes\1_{\mcH})\dket{\sqrt{\rho}}$, which proves Eq.~\eqref{eq:purification-orbit}.
\end{proof}

Due to Schur's lemma, the group $\U_G(\mcH)$ is isomorphic to a product of unitary groups as
\begin{align}
    \U_G(\mcH)\simeq \prod_{\lambda\in\irrep{G}}\U(m_\lambda).
\end{align}
Using this isomorphism, we can define the Haar measure on $\U_G(\mcH)$ as the product of normalized Haar measures on $\U(m_\lambda)$.
We also define the Haar measure on $\Pur_G(\rho)$ as the pushforward of the Haar measure on $\U_G(\mcH)$.
For a nonzero $G$-invariant positive operator $A$, we extend this notation by $\Pur_G(A)\coloneqq\sqrt{\Tr A}\,\Pur_G(A/\Tr A)$ and use the correspondingly rescaled pushforward measure; these purification vectors have norm $\sqrt{\Tr A}$.

\subsection{$\U_G(\mcH) - \Wr_n$ duality}

We define the wreath product $\Wr_n\coloneqq G\wr\mfS_n$ as the semidirect product $G^n\rtimes \mfS_n$, where $\mfS_n$ is the symmetric group on $n$ elements.
It is defined by
\begin{align}
    \Wr_n = \left\{(\bm{g}, \sigma) \; \middle | \; \bm{g} = (g_1,\ldots,g_n)\in G^n, \sigma \in \mfS_n \right\}
\end{align}
with the product
\begin{align}
 (\bm{g},\sigma)(\bm{h},\tau)=\qty((g_i h_{\sigma^{-1}(i)})_{i=1}^n,\sigma\tau).
\end{align}
We define the representation $R:\Wr_n\to \U(\mcH^{\otimes n})$ by
\begin{align}
    R(\bm{g},\sigma) \coloneqq U(\bm{g})\pi_d(\sigma),
\label{eq:wreath-representation}
\end{align}
where $\pi_d(\sigma)$ is the permutation operator defined in Eq.~\eqref{eq:def-permutation-operator}, and $U(\bm{g})$ is defined by
\begin{align}
    U(\bm{g}) \coloneqq U(g_1)\otimes \cdots \otimes U(g_n).
\end{align}
Then, the representation $R$ commutes with $V^{\otimes n}$ for every $V\in \U_G(\mcH)$:
\begin{align}
    \label{eq:commutation-relation}
    [V^{\otimes n}, R(\bm{g}, \sigma)] = 0 \quad \forall V\in \U_G(\mcH), (\bm{g},\sigma)\in \Wr_n.
\end{align}
In fact, we have the following duality between the actions of $\U_G(\mcH)$ and $\Wr_n$ on $\mcH^{\otimes n}$.

\begin{proposition}[$\U_G(\mcH) - \Wr_n$ duality]
\label{prop:double-centralizer}
The commutant of $\{V^{\otimes n} \; | \; V\in \U_G(\mcH)\}$ is given by the span of the representation $R$ of the wreath product $\Wr_n$:
\begin{align}
    \Comm\left(\left\{V^{\otimes n}\; \middle | \; V\in \U_G(\mcH)\right\}\right)=\mathrm{span}_{\CC}\left\{R(\bm{g},\sigma)\; \middle | \; (\bm{g},\sigma)\in\Wr_n\right\}.
\label{eq:generalized-schur-weyl-duality}
\end{align}
\end{proposition}
\begin{proof}
Since $\mcH$ is finite dimensional, the representation support $\Lambda_U\coloneqq\left\{\lambda\in\irrep{G}\;\middle|\;m_\lambda>0\right\}$ is finite, even when $G$ is compact.
All irrep labels in this proof range over $\Lambda_U$.
By Schur's lemma, $V\in\U_G(\mcH)$ can be written as
\begin{align}
    V\cong\bigoplus_{\lambda\in\Lambda_U}\1_{\mcR_\lambda}\otimes V_\lambda,
\end{align}
using $V_\lambda\in\U(\mcM_\lambda)$.
By taking the tensor product, we have
\begin{align}
    \label{eq:basis-change-1}
    V^{\otimes n} \cong \bigoplus_{\lambda_1, \ldots, \lambda_n \in \Lambda_U} \1_{\mcR_{\lambda_1}} \otimes \cdots \otimes \1_{\mcR_{\lambda_n}} \otimes V_{\lambda_1} \otimes \cdots \otimes V_{\lambda_n}.
\end{align}
We fix a total order $\mu_1\prec\cdots\prec\mu_{|\Lambda_U|}$ on $\Lambda_U=\left\{\mu_1,\ldots,\mu_{|\Lambda_U|}\right\}$.
For a word $\bm{\lambda}=(\lambda_1,\ldots,\lambda_n)$, define the stable sorting permutation $p_{\bm{\lambda}}\in\mfS_n$ to send each original position to its position after sorting in increasing order, and define the action of $p\in \mathfrak{S}_n$ on $\bm{\lambda}$ by $p\cdot\bm{\lambda}=(\lambda_{p^{-1}(1)},\ldots,\lambda_{p^{-1}(n)})$.
It is uniquely determined by
\begin{align}
    \lambda_{p_{\bm{\lambda}}^{-1}(i)}\preceq\lambda_{p_{\bm{\lambda}}^{-1}(j)}, \qquad
    \lambda_i=\lambda_j\implies p_{\bm{\lambda}}(i)<p_{\bm{\lambda}}(j)\qquad\forall i<j.
\end{align}
The second condition preserves the original order of equal labels.
We define the occupation vector $\bm{n}=(n_\lambda)_{\lambda\in\Lambda_U}$ by $n_\lambda\coloneqq\#\left\{i\in[n]\;\middle|\;\lambda_i=\lambda\right\}$.
For each occupation vector $\bm{n}$, define its canonical sorted tuple by
\begin{align}
    \label{eq:def-sorted-tuple}
    \bm{\lambda}^{\uparrow}(\bm{n})\coloneqq(\underbrace{\mu_1,\ldots,\mu_1}_{n_{\mu_1}},\underbrace{\mu_2,\ldots,\mu_2}_{n_{\mu_2}},\ldots,\underbrace{\mu_{|\Lambda_U|},\ldots,\mu_{|\Lambda_U|}}_{n_{\mu_{|\Lambda_U|}}}).
\end{align}
Write $\lambda_i^{\uparrow}(\bm{n})$ for its $i$th component.
For every word $\bm{\lambda}$ with occupation vector $\bm{n}$, stable sorting gives
\begin{align}
    p_{\bm{\lambda}}\cdot\bm{\lambda}=\bm{\lambda}^{\uparrow}(\bm{n}),\qquad \lambda_i^{\uparrow}(\bm{n})=\lambda_{p_{\bm{\lambda}}^{-1}(i)}\quad(1\leq i\leq n).
\end{align}
Thus the sorted tuple depends only on $\bm{n}$, independently of the original order of the word.
For each occupation vector $\bm{n}$, we define the set of permutations that preserve the order within each equal-label block of $\bm{\lambda}^{\uparrow}(\bm{n})$ by
\begin{align}
    \mathsf{P}_{\bm{n}}\coloneqq\left\{p\in\mfS_n\;\middle|\;\lambda_i^{\uparrow}(\bm{n})=\lambda_j^{\uparrow}(\bm{n})\implies p^{-1}(i)<p^{-1}(j)\quad\forall\,1\leq i<j\leq n\right\}.
\end{align}
Then, the tuple $\bm{\lambda}$ is uniquely determined by the pair $(\bm{n},p_{\bm{\lambda}})$.
We define $\CC[\mathsf{P}_{\bm{n}}]\coloneqq\operatorname{span}_{\CC}\left\{\ket{q}\;\middle|\;q\in\mathsf{P}_{\bm{n}}\right\}$ with its orthonormal basis $\ket{q}$, and $\mcR_{\bm{n}}\coloneqq\bigotimes_{\lambda\in\Lambda_U}\mcR_\lambda^{\otimes n_\lambda}$.
Reordering $\mcR_{\bm{\lambda}} \coloneqq \mcR_{\lambda_1} \otimes \cdots \otimes \mcR_{\lambda_n}$ and $\mcM_{\lambda_1} \otimes \cdots \otimes \mcM_{\lambda_n}$ by $p_{\bm{\lambda}}$ in Eq.~\eqref{eq:basis-change-1}, we have
\begin{align}
    \label{eq:basis-change-2}
    V^{\otimes n} \cong \bigoplus_{\bm{n}} \1_{\CC[\mathsf{P}_{\bm{n}}]} \otimes \1_{\mcR_{\bm{n}}} \otimes \bigotimes_{\lambda\in \Lambda_U} V_\lambda^{\otimes n_\lambda}.
\end{align}
We apply the Schur--Weyl duality to $V_\lambda^{\otimes n_\lambda}$ to obtain the following isotypic decomposition:
\begin{align}
    \mcM_\lambda^{\otimes n_\lambda} \cong \bigoplus_{\alpha^\lambda \vdash_{m_\lambda} n_\lambda} \mcS_{\alpha^\lambda} \otimes \mcU_{\alpha^\lambda}^{(m_\lambda)}, \qquad V_\lambda^{\otimes n_\lambda} \cong \bigoplus_{\alpha^\lambda \vdash_{m_\lambda} n_\lambda} \1_{\mcS_{\alpha^\lambda}} \otimes f_{\alpha^\lambda}^{(m_\lambda)}(V_\lambda),
\end{align}
and for $\bm{\alpha} = (\alpha^\lambda)_{\lambda\in \Lambda_U}$ and $\bm{m} = (m_\lambda)_{\lambda\in \Lambda_U}$, we denote
\begin{align}
    \mcS_{\bm{\alpha}} \coloneqq \bigotimes_{\lambda\in \Lambda_U} \mcS_{\alpha^\lambda}, \qquad \mcU_{\bm{\alpha}} \coloneqq \bigotimes_{\lambda\in \Lambda_U} \mcU_{\alpha^\lambda}^{(m_\lambda)}, \qquad \mcG_{\bm{\alpha}} \coloneqq \CC[\mathsf{P}_{\bm{n}}] \otimes \mcR_{\bm{n}} \otimes \mcS_{\bm{\alpha}}, \qquad f_{\bm{\alpha}}^{(\bm{m})}(V) \coloneqq \bigotimes_{\lambda\in \Lambda_U} f_{\alpha^\lambda}^{(m_\lambda)}(V_\lambda).
\end{align}
Then, we have the following isotypic decomposition of $V^{\otimes n}$ for $V\in \U_G(\mcH)$:
\begin{align}
    \label{eq:generalized-schur-decomposition}
    \mcH^{\otimes n} \cong \bigoplus_{\bm{\alpha}} \mcG_{\bm{\alpha}} \otimes \mcU_{\bm{\alpha}}, \qquad V^{\otimes n} \cong \bigoplus_{\bm{\alpha}} \1_{\mcG_{\bm{\alpha}}} \otimes f_{\bm{\alpha}}^{(\bm{m})}(V),
\end{align}
where the summation is taken over all $\bm{\alpha} = (\alpha^\lambda)_{\lambda\in \Lambda_U}$ and $\bm{n} = (n_\lambda)_{\lambda\in \Lambda_U}$ such that $\alpha^\lambda \vdash_{m_\lambda} n_\lambda$ for each $\lambda\in \Lambda_U$ and $\sum_\lambda n_\lambda = n$.
Since $\bm{n}$ is uniquely determined by $\bm{\alpha}$, we omit the summation over $\bm{n}$.

We then determine the representation $R$ of the wreath product $\Wr_n$ on the multiplicity space $\mcG_{\bm{\alpha}}$.
Since the commutation relation~\eqref{eq:commutation-relation} holds, the representation $R$ is block-diagonalized in the decomposition~\eqref{eq:generalized-schur-decomposition} as
\begin{align}
    \label{eq:generalized-wreath-decomposition}
    R(\bm{g}, \sigma) \cong \bigoplus_{\bm{\alpha}} \gamma_{\bm{\alpha}}(\bm{g}, \sigma) \otimes \1_{\mcU_{\bm{\alpha}}} \quad \forall (\bm{g}, \sigma)\in\Wr_n,
\end{align}
where $\gamma_{\bm{\alpha}}:\Wr_n \to \U(\mcG_{\bm{\alpha}})$ is a representation.
For $\bm{n}$, we define the stabilizer of the sorted tuple and the corresponding subgroup of $\Wr_n$ by
\begin{align}
    H_{\bm{n}} \coloneqq \qty{\eta\in\mfS_n \;\middle|\; \eta\cdot\bm{\lambda}^{\uparrow}(\bm{n}) = \bm{\lambda}^{\uparrow}(\bm{n})}, \qquad I_{\bm{n}} \coloneqq G^n \rtimes H_{\bm{n}}.
\end{align}
Since $\eta\in H_{\bm{n}}$ preserves $\bm{\lambda}^{\uparrow}(\bm{n})$, it decomposes as
\begin{align}
    \eta = \prod_{\lambda\in \Lambda_U} \eta_\lambda,
\end{align}
where $\eta_\lambda$ permutes the $\lambda$-block in Eq.~\eqref{eq:def-sorted-tuple}.
This defines the homomorphism
\begin{align}
    H_{\bm{n}} \cong \prod_{\lambda\in \Lambda_U} \mfS_{n_\lambda}.
\end{align}
We also define the subspace of $\mcH^{\otimes n}$ corresponding to the sorted tuple by
\begin{align}
    \mcK_{\bm{n}} \coloneqq \bigotimes_{i=1}^n \qty(\mcR_{\lambda_i^{\uparrow}(\bm{n})} \otimes \mcM_{\lambda_i^{\uparrow}(\bm{n})}) \subset \mcH^{\otimes n}, \qquad \mcK_{\bm{n}} \cong \mcR_{\bm{n}} \otimes \bigotimes_{\lambda\in\Lambda_U} \mcM_\lambda^{\otimes n_\lambda}.
\end{align}
For $\bm{g}\in G^n$ and $\eta\in H_{\bm{n}}$, define
\begin{align}
    U_{\bm{n}}(\bm{g})\coloneqq \bigotimes_{i=1}^n U_{\lambda_i^{\uparrow}(\bm{n})}(g_i), \qquad g_{\bm{\alpha}}(\eta)\coloneqq \bigotimes_{\lambda\in\Lambda_U} g_{\alpha^\lambda}(\eta_\lambda).
\end{align}
Let $P_{\mcR_{\bm{n}}}(\eta)$ permute the irrep tensor factors according to the same convention as Eq.~\eqref{eq:def-permutation-operator}:
\begin{align}
    P_{\mcR_{\bm{n}}}(\eta)\bigotimes_{i=1}^n\ket{v_i}
    =\bigotimes_{i=1}^n\ket{v_{\eta^{-1}(i)}},
    \qquad \ket{v_i}\in\mcR_{\lambda_i^{\uparrow}(\bm{n})}.
\end{align}
This is an operator on $\mcR_{\bm{n}}$ since $\eta$ preserves each equal-label block.
The subspace $\mcK_{\bm{n}}$ is $I_{\bm{n}}$-invariant, and its restriction decomposes as
\begin{align}
    R(\bm{g},\eta)\vert_{\mcK_{\bm{n}}}
    \cong\bigoplus_{\bm{\alpha}:\,|\alpha^\lambda|=n_\lambda}
    \bigl[U_{\bm{n}}(\bm{g})P_{\mcR_{\bm{n}}}(\eta)\bigr]
    \otimes g_{\bm{\alpha}}(\eta)\otimes\1_{\mcU_{\bm{\alpha}}},
    \label{eq:induced-representation}
\end{align}
where the sum retains the constraints $\alpha^\lambda\vdash_{m_\lambda}n_\lambda$ from Eq.~\eqref{eq:generalized-schur-decomposition}.
For a fixed $\bm{\alpha}$, let $e\in\mathsf{P}_{\bm{n}}$ be the identity permutation and define its sorted-word subspace by
\begin{align}
    \mcL_{\bm{n},\bm{\alpha}}\coloneqq
    \operatorname{span}_{\CC}\left\{\ket{e}\right\}\otimes\mcR_{\bm{n}}\otimes\mcS_{\bm{\alpha}}
    \subset\mcG_{\bm{\alpha}},
    \qquad \mcL_{\bm{n},\bm{\alpha}}\cong\mcR_{\bm{n}}\otimes\mcS_{\bm{\alpha}}.
\end{align}
Thus $\mcK_{\bm{n}}\cong\bigoplus_{\bm{\alpha}:\,|\alpha^\lambda|=n_\lambda}\mcL_{\bm{n},\bm{\alpha}}\otimes\mcU_{\bm{\alpha}}$, and on this fixed $\bm{\alpha}$ block Eq.~\eqref{eq:induced-representation} provides
\begin{align}
    \gamma_{\bm{\alpha}}(\bm{g},\eta)\vert_{\mcL_{\bm{n},\bm{\alpha}}}
    \cong\bigl[U_{\bm{n}}(\bm{g})P_{\mcR_{\bm{n}}}(\eta)\bigr]\otimes g_{\bm{\alpha}}(\eta).
\end{align}
Define the action of $p\in\mfS_n$ on $\bm{g}\in G^n$ by $p\cdot\bm{g}\coloneqq(g_{p^{-1}(1)},\ldots,g_{p^{-1}(n)})$.
The covariance identity
\begin{align}
    P_{\mcR_{\bm{n}}}(\eta)U_{\bm{n}}(\bm{h})P_{\mcR_{\bm{n}}}(\eta)^\dagger
    =U_{\bm{n}}(\eta\cdot\bm{h})
\end{align}
shows that $\widetilde{U}_{\bm{n}}(\bm{g},\eta)\coloneqq U_{\bm{n}}(\bm{g})P_{\mcR_{\bm{n}}}(\eta)$ is a representation of $I_{\bm{n}}$, since
\begin{align}
    \widetilde{U}_{\bm{n}}(\bm{g},\eta)\widetilde{U}_{\bm{n}}(\bm{h},\theta)
    =\widetilde{U}_{\bm{n}}\bigl(\bm{g}(\eta\cdot\bm{h}),\eta\theta\bigr)
\end{align}
holds.
Together with the representation $\widetilde{g}_{\bm{\alpha}}(\bm{g},\eta)\coloneqq g_{\bm{\alpha}}(\eta)$ of $I_{\bm{n}}$, this yields
\begin{align}
    \operatorname{Res}_{I_{\bm{n}}}^{\Wr_n}\gamma_{\bm{\alpha}}\vert_{\mcL_{\bm{n},\bm{\alpha}}}
    \cong\widetilde{U}_{\bm{n}}\otimes\widetilde{g}_{\bm{\alpha}}.
\end{align}
To recover the representation $\gamma_{\bm{\alpha}}$ on the entire space $\mcG_{\bm{\alpha}}$, we induce this representation from $I_{\bm{n}}$ to $\Wr_n$.
We define
\begin{align}
    t_p \coloneqq (\bm{e}, p^{-1}) \in \Wr_n \quad (p\in\mathsf{P}_{\bm{n}}),
\end{align}
where $\bm{e} = (e, \ldots, e) \in G^n$ is the tuple of identity elements in $G$.
Then, $\left\{t_p\right\}_{p\in\mathsf{P}_{\bm{n}}}$ forms a left transversal for $I_{\bm{n}}$ in $\Wr_n$, and the spaces $\gamma_{\bm{\alpha}}(t_p)\mcL_{\bm{n},\bm{\alpha}}$ correspond to the distinct words with occupation vector $\bm{n}$ and are mutually orthogonal.
Thus, we have the following decomposition:
\begin{align}
    \mcG_{\bm{\alpha}} \cong \bigoplus_{p\in \mathsf{P}_{\bm{n}}} \gamma_{\bm{\alpha}}(t_p)\mcL_{\bm{n},\bm{\alpha}}.
\end{align}
For $w = (\bm{g}, \sigma)\in\Wr_n$, we define
\begin{align}
    q\coloneqq p_{\sigma\cdot(p^{-1}\cdot\bm{\lambda}^{\uparrow}(\bm{n}))}, \qquad \eta\coloneqq q\sigma p^{-1}\in H_{\bm{n}}.
\end{align}
Then, for $\bm{\lambda} = p^{-1}\cdot \bm{\lambda}^{\uparrow}(\bm{n})$, we have $q = p_{\sigma\cdot \bm{\lambda}}$.
Thus, we have
\begin{align}
    t_q^{-1} w t_p = (q\cdot \bm{g}, q\sigma p^{-1}) = (q\cdot \bm{g}, \eta) \in I_{\bm{n}}.
\end{align}
Thus, for $\ket{\kappa}\in\mcL_{\bm{n},\bm{\alpha}}$, we have
\begin{align}
    \gamma_{\bm{\alpha}}(w)\bigl[\gamma_{\bm{\alpha}}(t_p)\ket{\kappa}\bigr]
    =\gamma_{\bm{\alpha}}(t_q)\gamma_{\bm{\alpha}}(q\cdot\bm{g},\eta)\ket{\kappa}.
\end{align}
Equivalently, under the identification $\gamma_{\bm{\alpha}}(t_p)(\ket{e}\otimes\ket{v})\leftrightarrow\ket{p}\otimes\ket{v}$, the action on $\ket{v}\in\mcR_{\bm{n}}\otimes\mcS_{\bm{\alpha}}$ is given by
\begin{align}
    \gamma_{\bm{\alpha}}(\bm{g},\sigma)(\ket{p}\otimes\ket{v})
    =\ket{q}\otimes\Bigl(\bigl[U_{\bm{n}}(q\cdot\bm{g})P_{\mcR_{\bm{n}}}(\eta)\bigr]
    \otimes g_{\bm{\alpha}}(\eta)\Bigr)\ket{v}.
    \label{eq:wreath-irrep-matrices}
\end{align}
Therefore, the representation $\gamma_{\bm{\alpha}}$ of $\Wr_n$ on $\mcG_{\bm{\alpha}}$ is equivalent to
\begin{align}
    \gamma_{\bm{\alpha}} \cong \operatorname{Ind}_{I_{\bm{n}}}^{\Wr_n} \qty(\widetilde{U}_{\bm{n}} \otimes \widetilde{g}_{\bm{\alpha}}),
\end{align}
which is shown to be irreducible and pairwise inequivalent for different $\bm{\alpha}$ in Ref.~\cite[Thm.~8.7]{bruinsma2026toolbox}\footnote{Reference~\cite{bruinsma2026toolbox} shows the theorem for finite groups, but the same argument can be extended to compact groups.}.
Thus, from Eqs.~\eqref{eq:generalized-schur-decomposition} and \eqref{eq:generalized-wreath-decomposition}, we obtain Eq.~\eqref{eq:generalized-schur-weyl-duality}.
\end{proof}

We write $d_{\bm{\alpha}}\coloneqq\dim\mcG_{\bm{\alpha}}=\dim\gamma_{\bm{\alpha}}$ for this common irrep dimension.
We call a unitary realizing Eqs.~\eqref{eq:generalized-schur-decomposition} and \eqref{eq:generalized-wreath-decomposition}, with the fixed induced matrices $\gamma_{\bm\alpha}$ of Eq.~\eqref{eq:wreath-irrep-matrices}, a generalized Schur transform and denote it by $U_{\mathrm{Sch}}^{(\Wr_n)}$.

\subsection{Random purification channel of $G$-invariant quantum states}

The polynomial representations $f_{\alpha^\lambda}^{(m_\lambda)}$ extend from unitaries to arbitrary operators.
For a state $\rho=\bigoplus_\lambda(\1_{\mcR_\lambda}/d_\lambda)\otimes\rho_\lambda$, we use the shorthand $f_{\bm{\alpha}}(\rho)\coloneqq\bigotimes_\lambda f_{\alpha^\lambda}^{(m_\lambda)}(\rho_\lambda)$, where the $\rho_\lambda$ are the subnormalized blocks defined above.
Then, similarly to $V^{\otimes n}$ for $V\in \U_G(\mcH)$, the tensor product $\rho^{\otimes n}$ for $\rho\in\mcS_G(\mcH)$ is decomposed as
\begin{align}
    \label{eq:tensor-rho-decomposition}
    \rho^{\otimes n} \cong \bigoplus_{\bm{\alpha}} c_{\bm{\alpha}} \1_{\mcG_{\bm{\alpha}}} \otimes f_{\bm{\alpha}}(\rho),
\end{align}
where $c_{\bm{\alpha}}$ is defined by $c_{\bm{\alpha}}\coloneqq \prod_\lambda d_\lambda^{-n_\lambda}$.
We define the randomized $n$-fold purification of $\rho\in \mcS_G(\mcH)$ as
\begin{align}
    \Omega_{G,n}(\rho)
    &\coloneqq \int_{\Pur_G(\rho)} \dd \ket{\psi} \; \ketbra{\psi}^{\otimes n}\\
    &=\int_{\U_G(\mcH)} \dd V \; (\overline{V}^{\otimes n}\otimes \1_{\mcH^{\otimes n}})(\dketbra{\sqrt{\rho}})^{\otimes n}(\overline{V}^{\otimes n}\otimes \1_{\mcH^{\otimes n}})^\dagger.
\end{align}
The isotypic decomposition of $V^{\otimes n}$ for $V\in \U_G(\mcH)$ in Eq.~\eqref{eq:generalized-schur-decomposition} gives the following decomposition of $\Omega_{G,n}(\rho)$:
\begin{align}
    \Omega_{G,n}(\rho)\cong \bigoplus_{\bm{\alpha}}c_{\bm{\alpha}}\,
\dketbra{\1_{\mcG_{\bm{\alpha}}}}\otimes
\tau_{\overline{\mcU}_{\bm{\alpha}}}\otimes f_{\bm{\alpha}}(\rho),
\label{eq:twirled-purification}
\end{align}
where $\tau_{\overline{\mcU}_{\bm{\alpha}}}$ is the maximally mixed state on $\overline{\mcU}_{\bm{\alpha}}$ defined by $\tau_{\overline{\mcU}_{\bm{\alpha}}}\coloneqq \1_{\overline{\mcU}_{\bm{\alpha}}}/\dim\mcU_{\bm{\alpha}}$.
Similarly to the random purification channel given in Refs.~\cite{girardi2025random2,yoshida2025random}, when $G$ is a finite group, the quantum state $\Omega_{G,n}(\rho)$ can be prepared by using the plus state $\ket{+_{\Wr_n}}$, the controlled representation $\mathrm{ctrl}-R$, and the quantum Fourier transform $\mathrm{QFT}_{\Wr_n}$ over $\Wr_n$ defined by\footnote{This definition of $\mathrm{QFT}_{\Wr_n}$ is slightly different from the one used in Refs.~\cite{girardi2025random2,yoshida2025random}, where $w$ on the right-hand side is replaced by $w^{-1}$.}
\begin{align}
    \ket{+_{\Wr_n}}&\coloneqq {1 \over \sqrt{|\Wr_n|}}\sum_{w\in\Wr_n}\ket{w} =\ket{+_G}^{\otimes n}\otimes\ket{+_{\mfS_n}}\in \CC[\Wr_n],\\
    \mathrm{ctrl}-R&\coloneqq \sum_{w\in\Wr_n} \ketbra{w} \otimes R(w),\\
    \mathrm{QFT}_{\Wr_n}\ket{w}&\coloneqq
    \sum_{\eta\in\irrep{\Wr_n}}\sum_{i,j=1}^{d_\eta}
    \sqrt{{d_\eta \over |\Wr_n|}}\,
    \overline{[\eta(w)]_{ij}}\ket{\eta,i,j} \quad \forall w\in \Wr_n,
    \label{eq:def-qft-wreath}
\end{align}
where $\ket{+_G}$ and $\ket{+_{\mfS_n}}$ are the plus states on $\CC[G]$ and $\CC[\mfS_n]$ defined similarly, respectively, $\irrep{\Wr_n}$ is the set of irreducible representations of $\Wr_n$ and $d_\eta$ is the dimension of $\eta\in\irrep{\Wr_n}$.
Throughout, the wreath-product QFT uses the induced little-group basis of Ref.~\cite[Section~7]{bruinsma2026toolbox}, with the local $G$-irrep bases fixed above and the Young--Yamanouchi bases fixed in Section~\ref{sec:schur-weyl-basis-conventions} for the little-group factors. We take the underlying total order in that reference to be the reverse of $\prec$, so that its decreasingly sorted orbit representatives coincide with our increasingly sorted tuples. Its stable orbit transversal is then $t_p=(\bm e,p^{-1})$, as above. In particular, for every irrep occurring in $\mcH^{\otimes n}$, the matrices $\eta(w)$ in Eq.~\eqref{eq:def-qft-wreath} are exactly $\gamma_{\bm\alpha}(w)$ from Eq.~\eqref{eq:wreath-irrep-matrices}, and the Fourier indices are identified with the corresponding induced-basis indices. The same conventions apply to all representations used in the purification and dilation circuits.

\begin{figure}
\centering
\includegraphics[width=.75\linewidth]{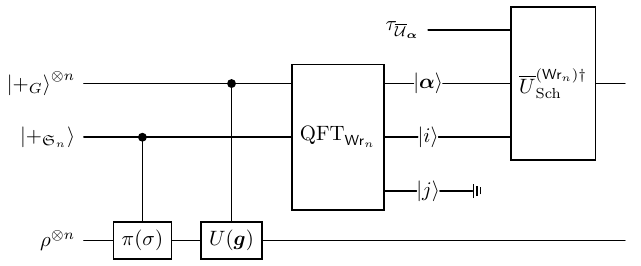}
\caption{Quantum circuit for the random purification channel of $G$-invariant quantum states.
The two controlled gates jointly implement $\mathrm{ctrl}-R$, and after the quantum Fourier transform $\mathrm{QFT}_{\Wr_n}$, the Fourier index $j$ is discarded.
The maximally mixed state $\tau_{\overline{\mcU}_{\bm{\alpha}}}$ is prepared conditionally on the irrep label $\bm{\alpha}$, and $\overline{U}_\mathrm{Sch}^{(\Wr_n)\dagger}$ is applied.}
\label{fig:covariant_random_purification}
\end{figure}

\begin{theorem}[Random purification channel of $G$-invariant quantum states]
\label{thm:random-covariant-purification}
The quantum circuit shown in Fig.~\ref{fig:covariant_random_purification} implements $\Omega_{G,n}(\rho)$ using $n$ copies of a $G$-invariant quantum state $\rho\in\mcS_G(\mcH)$, when $G$ is a finite group.
\end{theorem}

\begin{proof}
    The quantum state just after $\mathrm{QFT}_{\Wr_n}$ is given by
    \begin{align}
        &{1\over \abs{\Wr_n}}\sum_{v, w\in \Wr_n} \mathrm{QFT}_{\Wr_n}\ketbra{v}{w}\mathrm{QFT}_{\Wr_n}^\dagger \otimes R(v)\rho^{\otimes n}R(w)^\dagger\nonumber\\
        &\cong {1\over \abs{\Wr_n}}\sum_{v, w\in \Wr_n} \sum_{\bm{\alpha}, \bm{\alpha}'\in \irrep{\Wr_n}} \sum_{i,j=1}^{d_{\bm{\alpha}}} \sum_{i',j'=1}^{d_{\bm{\alpha}'}} {\sqrt{d_{\bm{\alpha}} d_{\bm{\alpha}'}} \over \abs{\Wr_n}} \overline{[\gamma_{\bm{\alpha}}(v)]_{ij}}[\gamma_{\bm{\alpha}'}(w)]_{i'j'} \ketbra{\bm{\alpha},i,j}{\bm{\alpha}',i',j'}\notag\\
        &\hspace{240pt} \otimes \bigoplus_{\bm{\beta}} c_{\bm{\beta}}\gamma_{\bm{\beta}}(v)\gamma_{\bm{\beta}}(w)^\dagger \otimes f_{\bm{\beta}}(\rho)\\
        &= \bigoplus_{\bm{\alpha}} {c_{\bm{\alpha}} \over d_{\bm{\alpha}}} \sum_{i,j, i', j'=1}^{d_{\bm{\alpha}}} \ketbra{\bm{\alpha},i,j}{\bm{\alpha},i',j'} \otimes  \ketbra{i}{j} \cdot \ketbra{j'}{i'}\otimes f_{\bm{\alpha}}(\rho),
    \end{align}
    where we use the definition~\eqref{eq:def-qft-wreath} of $\mathrm{QFT}_{\Wr_n}$ and the decomposition~\eqref{eq:tensor-rho-decomposition} of $\rho^{\otimes n}$ in the first equality, and the Schur orthogonality relation
    \begin{align}
        \sum_{v\in \Wr_n} \overline{[\gamma_{\bm{\alpha}}(v)]_{ji}} \gamma_{\bm{\alpha}'}(v) = \delta_{\bm{\alpha}\bm{\alpha}'} {\abs{\Wr_n} \over d_{\bm{\alpha}}} \ketbra{j}{i}
    \end{align}
    in the last equality.
    Thus, after discarding the Fourier index $j$ and preparing the maximally mixed state $\tau_{\overline{\mcU}_{\bm{\alpha}}}$, we obtain the output state
    \begin{align}
        \bigoplus_{\bm{\alpha}} c_{\bm{\alpha}}\dketbra{\1_{\mcG_{\bm{\alpha}}}} \otimes \tau_{\overline{\mcU}_{\bm{\alpha}}} \otimes f_{\bm{\alpha}}(\rho) \cong \Omega_{G,n}(\rho),
    \end{align}
    which concludes the proof.
\end{proof}

Our construction of the random purification channel of $G$-invariant quantum states is different from that in Ref.~\cite{chen2026gaussian}.
Theorem~\ref{thm:random-covariant-purification} gives a finite-group circuit-level realization of this general construction in terms of the wreath-product Fourier transform and the generalized Schur transform.
Our construction enables the random-dilation construction in the next subsection.

\subsection{Random dilation superchannel of \texorpdfstring{$G$}{G}-covariant quantum channels}

We consider a compact group $G$ and two finite-dimensional representations $U_1: G\to \U(\mcI)$ and $U_2: G\to \U(\mcO)$ on Hilbert spaces $\mcI$ and $\mcO$.
For $i\in\{1,2\}$ and $(\bm g,\sigma)\in\Wr_n$, define $R_i(\bm g,\sigma)\coloneqq U_i(\bm g)\pi_{d_i}(\sigma)$, where $U_i(\bm g)\coloneqq\bigotimes_{a=1}^n U_i(g_a)$, $d_1=\dim\mcI$, $d_2=\dim\mcO$, and $R_1$ and $R_2$ act on $\mcI^{\otimes n}$ and $\mcO^{\otimes n}$, respectively.
We consider the set of $G$-covariant quantum channels given by
\begin{align}
    \mathrm{CovChan}_G(\mcI, \mcO)\coloneqq \left\{\Lambda:\mcL(\mcI)\to\mcL(\mcO) \;\middle|\; \Lambda \text{ is CPTP}, \Lambda\circ\mcU_1(g)=\mcU_2(g)\circ\Lambda\quad \forall g\in G\right\},
\end{align}
where $\mcU_i(g)$ is the unitary channel defined by $\mcU_i(g)(\cdot)\coloneqq U_i(g)(\cdot)U_i(g)^\dagger$ for $i\in\{1,2\}$.
For the Choi operator $J_{\Lambda} \in \mcL(\overline{\mcI} \otimes \mcO)$, the covariance condition is equivalent to
\begin{align}
    [\overline{U}_1(g)\otimes U_2(g)]J_\Lambda [\overline{U}_1(g)\otimes U_2(g)]^\dagger=J_\Lambda \quad \forall g\in G.
\label{eq:choi-covariance}
\end{align}
The canonical purification $\dket{\sqrt{J_\Lambda}}$ of $J_\Lambda$ is defined in $\mcE \otimes \overline{\mcI} \otimes \mcO$ with $\mcE\simeq \mcI\otimes \overline{\mcO}$ satisfying
\begin{align}
[[U_1(g)\otimes \overline{U}_2(g)]_{\mcE}\otimes \overline{U_1}(g) \otimes U_2(g)]
\dket{\sqrt{J_{\Lambda}}}=\dket{\sqrt{J_{\Lambda}}}
\quad \forall g\in G.
\label{eq:choi-canonical-purification-invariance}
\end{align}
Then the operator $W:\mcI \to \mcE\otimes \mcO$ constructed from the canonical purification $\dket{\sqrt{J_{\Lambda}}}$ is defined by
\begin{align}
    W\ket{i} \coloneqq (\1_{\mcE}\otimes\bra{i}_{\overline{\mcI}} \otimes \1_{\mcO})\dket{\sqrt{J_{\Lambda}}},
\end{align}
which satisfies $\Tr_{\mcE}[W(\cdot)W^\dagger]=\Lambda(\cdot)$, i.e., $W$ is a Stinespring dilation of $\Lambda$.
The $G$-invariance~\eqref{eq:choi-canonical-purification-invariance} provides
\begin{align}
[[U_1(g)\otimes \overline{U}_2(g)]_{\mcE}\otimes U_2(g)]W=WU_1(g)
\quad \forall g\in G.
\end{align}
Suppose $U_1$ and $U_1\otimes\overline U_2\otimes U_2$ are decomposed as
\begin{align}
    \label{eq:decomposition-multiplicity-1}
    U_1(g) &\cong \bigoplus_{\lambda\in\irrep{G}} U_\lambda(g)\otimes \1_{m_{\lambda}},\\
    \label{eq:decomposition-multiplicity-2}
    U_1(g)\otimes\overline U_2(g)\otimes U_2(g)
    &\cong \bigoplus_{\lambda\in\irrep{G}} U_\lambda(g)\otimes \1_{M_{\lambda}},
\end{align}
using the multiplicities $m_\lambda$ and $M_\lambda$. Then, the $G$-covariant Stinespring dilation $W$ decomposes as
\begin{align}
    \label{eq:decomposition_W}
    W \cong \bigoplus_{\lambda\in\irrep{G}} \1_{\mcR_\lambda}\otimes W_\lambda,
\end{align}
using isometry operators $W_\lambda:\CC^{m_\lambda}\longrightarrow\CC^{M_\lambda}$.
We define the set of $G$-covariant isometries as
\begin{align}
    \W_G\coloneqq \left\{W \cong \bigoplus_{\lambda\in\irrep{G}} \1_{\mcR_\lambda}\otimes W_\lambda \;\middle|\; W_\lambda\in \W_{M_\lambda, m_\lambda}\right\} \cong \prod_{\lambda\in\irrep{G}} \W_{M_\lambda, m_\lambda},
\end{align}
where $\W_{M_\lambda, m_\lambda}$ is the set of isometries from $\CC^{m_\lambda}$ to $\CC^{M_\lambda}$.
Similarly to the $G$-invariant quantum state, we define the set of $G$-covariant Stinespring dilations of $\Lambda$ as
\begin{align}
\Dil_G(\Lambda)&\coloneqq \left\{W:\mcI\to\mcE\otimes\mcO \;\middle|\; \Tr_{\mcE}[W(\cdot)W^\dagger]=\Lambda(\cdot), W\in \W_G\right\}\\
    &=\left\{W:\mcI\to\mcE\otimes\mcO\;\middle|\; \dket{W} \in \Pur_G(J_\Lambda)\right\}.
\label{eq:dilation-orbit}
\end{align}
Here, when identifying Choi vectors with purifications of $J_\Lambda$, we implicitly swap the registers $\overline{\mcI}$ and $\mcE$, mapping $\overline{\mcI}\otimes\mcE\otimes\mcO$ to $\mcE\otimes\overline{\mcI}\otimes\mcO$.
We equip $\Dil_G(\Lambda)$ with the pushforward of Haar measure on $\Pur_G(J_\Lambda)$.
We define the randomized $n$-fold Stinespring dilation channel of $\Lambda$ by
\begin{align}
    \Omega_{G, n}(\Lambda) \coloneqq \mathbb{E}_{W \sim \Dil_G(\Lambda)} \qty[W^{\otimes n}(\cdot)W^{\dagger \otimes n}].
\label{eq:random-dilation-superchannel}
\end{align}
Similarly to the construction of random dilation superchannels in Refs.~\cite{girardi2025random2,yoshida2025random}, for a finite group $G$, we can implement a random dilation superchannel of $G$-covariant quantum channels as follows.

\begin{figure}[t]
\centering
\includegraphics[width=\linewidth]{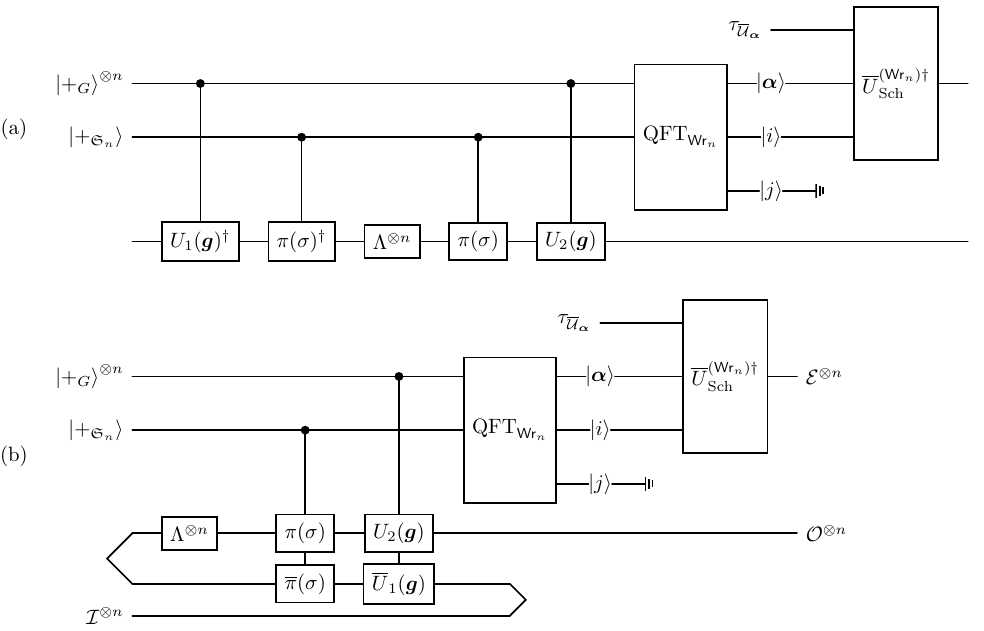}
\caption{(a) Quantum circuit for the random dilation superchannel on $G$-covariant quantum channels.
(b) Its derivation from random purification of $J_\Lambda^{\otimes n}$.
The left cup denotes the unnormalized maximally entangled vector $\dket{\1_{d_{\mcI}^{\,n}}}$.
The right cap represents a postselected closed timelike curve (P-CTC), used here as a formal contraction with the corresponding unnormalized maximally entangled bra.
Moving the reference-side controlled gates through this contraction replaces $\overline{U}_1(\bm g)\overline{\pi}(\sigma)$ by $\pi(\sigma)^\dagger U_1(\bm g)^\dagger$ and yields the deterministic circuit in (a).}
\label{fig:covariant_random_dilation}
\end{figure}

\begin{theorem}[Random covariant dilation circuit]
\label{thm:random-covariant-dilation}
The quantum circuit shown in Fig.~\ref{fig:covariant_random_dilation}~(a) implements $\Omega_{G,n}(\Lambda)$ using $n$ parallel calls to a $G$-covariant quantum channel $\Lambda$, when $G$ is a finite group.
\end{theorem}
\begin{proof}
We follow the Choi-purification and transpose-trick argument of Ref.~\cite{yoshida2025random}.
Due to the correspondence between $G$-covariant channels and $G$-invariant states, we can construct a random dilation superchannel of $G$-covariant channels from the random purification channel of $G$-invariant states.
In particular, define $\mcP_{G,n}:\mcL(\mcH^{\otimes n})\to\mcL(\mcE^{\otimes n}\otimes\mcH^{\otimes n})$ to be the random purification channel constructed in Thm.~\ref{thm:random-covariant-purification} for $U\coloneqq \overline{U}_1 \otimes U_2$.
Then, we have
\begin{align}
    \mcP_{G,n}(J_\Lambda^{\otimes n})=J_{\Omega_{G,n}(\Lambda)}.
\end{align}
Figure~\ref{fig:covariant_random_dilation}(b) expresses this identity as the following formal circuit construction using the P-CTC~\cite{lloyd2011quantum}:
\begin{enumerate}
\item Apply $\Lambda^{\otimes n}$ to one half of an unnormalized maximally entangled state to obtain $J_\Lambda^{\otimes n}$.
\item Apply $\mcP_{G,n}$ to obtain the Choi operator $J_{\Omega_{G,n}(\Lambda)}$.
\item Retrieve the action of the quantum channel $\Omega_{G,n}(\Lambda)$ from its Choi state by applying the P-CTC.
\end{enumerate}
The controlled operation $\mathrm{ctrl}-\overline{R_1}\coloneqq \sum_{\bm{g}\in G^n, \sigma\in \mfS_n} \ketbra{\bm{g}}\otimes \ketbra{\sigma} \otimes \overline{U}_1(\bm{g}) \overline{\pi}(\sigma)$ can be rewritten as $\mathrm{ctrl}-R_1^\dagger$ by using the transpose trick.
This transformation provides the deterministic circuit in Fig.~\ref{fig:covariant_random_dilation}~(a) from the formal P-CTC construction in Fig.~\ref{fig:covariant_random_dilation}~(b).
\end{proof}

\subsection{Efficient implementation}
\label{sec:random-purification-efficiency}

As shown in Appendix~\ref{appendix:efficient-wreath-schur}, for the case of the permutation group $G=\mfS_k$ with the representations $U_1(\sigma)=U_2(\sigma)=\pi_d(\sigma)$ on $\mcI\simeq\mcO\simeq(\CC^d)^{\otimes k}$, the generalized Schur transform $U_\mathrm{Sch}^{(\Wr_n)}$ and its inverse can be implemented with a quantum circuit of size $O_d\qty(\poly(n,k, \log \delta^{-1}))$, where $\delta$ is the compilation error.
The quantum Fourier transform $\mathrm{QFT}_{\Wr_n}$ over the wreath product $\Wr_n$, in the little-group basis specified after Eq.~\eqref{eq:def-qft-wreath}, can also be implemented with a quantum circuit of size $O\qty(\poly(n,k, \log \delta^{-1}))$~\cite[Thm.~8.1]{bruinsma2026toolbox}. Appendix~\ref{appendix:efficient-wreath-schur} verifies that the Schur circuit produces the same irrep basis.
The plus state $\ket{+_G}^{\otimes n} \otimes \ket{+_{\mfS_n}}$ can also be prepared efficiently by using the following equation:
\begin{align}
    \label{eq:plus-state-wreath}
    \ket{+_G}^{\otimes n} \otimes \ket{+_{\mfS_n}} = \ket{+_{\Wr_n}} = \mathrm{QFT}_{\Wr_n}^\dagger (\ket{\varnothing} \otimes \ket{0} \otimes \ket{0}),
\end{align}
where $\varnothing$ denotes the trivial representation of $\Wr_n$ given by $\gamma_\varnothing: \Wr_n \ni w \mapsto 1 \in \Gamma_{\varnothing} \cong \CC$, and $\ket{0}$ denotes the basis vector of the $1$-dimensional space $\Gamma_\varnothing$.
Since each space $\mcU_{\alpha^\lambda}$ is spanned by the Gelfand--Tsetlin basis, which is labelled by semistandard Young tableaux, the maximally mixed state $\tau_{\overline{\mcU}_{\bm{\alpha}}}$ can be prepared efficiently by uniformly sampling a semistandard Young tableau of shape $\alpha^\lambda$ with entries in $[m_\lambda]$ for each $\lambda$ using Algorithm HC of Ref.~\cite{krattenthaler1999another}, independently for each $\lambda$.
The controlled permutation can be implemented efficiently by using the decomposition of a permutation into transpositions.
Thus, the random purification for permutation-invariant states and random dilation circuit for permutation-covariant channels can be implemented with a quantum circuit of size $O_d\qty(\poly(n,k, \log \delta^{-1}))$.

\subsection{Extension to compact groups}
\label{sec:compact-oracle-implementation}

\begin{figure}
\centering
\includegraphics[width=\linewidth]{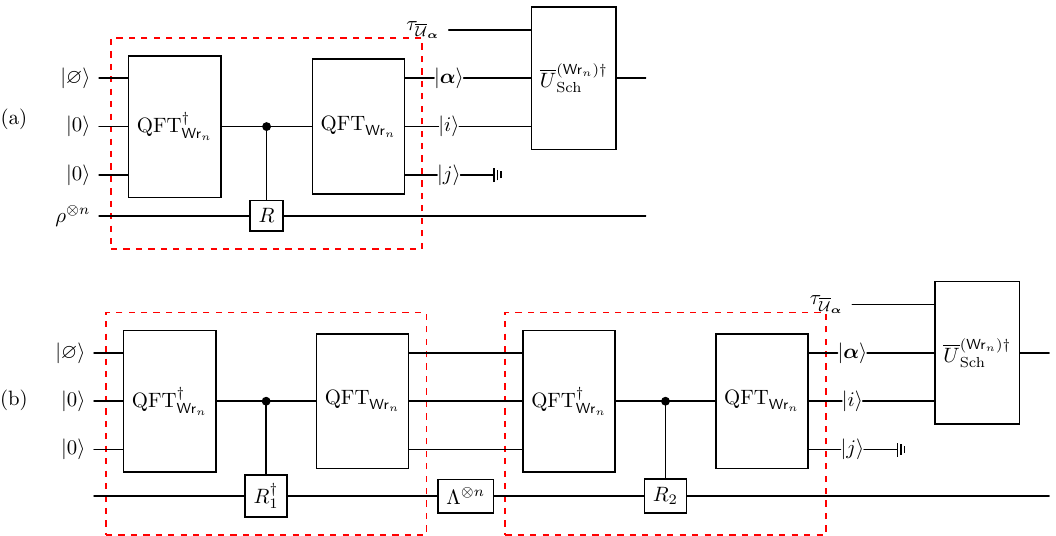}
\caption{(a) The random purification circuit and (b) the random dilation circuit using the generalized phase estimation shown in the red boxes.
By replacing the generalized phase estimation (red boxes) with the Clebsch--Gordan transforms, we can implement these circuits for compact groups using a finite number of qubits.}
\label{fig:expanded-oracle-circuits}
\end{figure}

The construction of random purification and dilation circuits in the preceding subsections is based on the finite-group Fourier transform.
For an infinite compact group, the full Fourier transform acts on an infinite-dimensional space, but only finitely many representation sectors are reached by the circuits considered there, and we can implement the random purification and dilation circuits as below.

For a finite group $G$, the plus state $\ket{+_{\Wr_n}}$ can be written as Eq.~\eqref{eq:plus-state-wreath}.
This provides the random purification circuit using the generalized phase estimation~\cite{harrow2005applications}, which is the controlled-$R$ gate sandwiched between the quantum Fourier transforms, in Fig.~\ref{fig:expanded-oracle-circuits}.
This generalized phase estimation can be implemented using the Clebsch--Gordan transforms associated to the representation $R$~\cite{harrow2005applications}, which only acts on finite-dimensional spaces (see also Refs.~\cite{grinko2025quantum,foxman2026quantum} for its explicit construction).
Thus, the random purification circuit can be implemented by replacing the generalized phase estimation in Fig.~\ref{fig:expanded-oracle-circuits} with the Clebsch--Gordan transforms, and similarly for the random dilation circuit.

\section{\texorpdfstring{Learning covariant quantum states}{Learning covariant quantum states}}
\label{sec:learning-covariant-states}
We develop the state-tomography tools used for covariant channel learning in Sec.~\ref{sec:learning-covariant-channels}.
We consider the task of learning an unknown state using independent copies, allowing collective measurements and measuring error in trace distance.
The compact group $G$, its finite-dimensional unitary representation $U$ on $\mcH$, and a representation basis are known.

\subsection{Covariant quantum states}

A covariant, or $G$-invariant, quantum state satisfies
\begin{align}\label{eq:state-covariance-promise}
\rho U(g)=U(g)\rho \quad\forall g\in G.
\end{align}
A state is supported on invariant vectors if the following stricter promise is satisfied.
\begin{align}
\rho U(g)=U(g)\rho=\rho \quad\forall g\in G. \label{eq:state-invariant-support-promise}
\end{align}
Using the irrep and multiplicity spaces of Eq.~\eqref{eq:single-copy-representation-decomposition}, the standard block parameterization~\cite{corte2023parameterizing,sauvage2024classical} is
\begin{align}
\rho=\bigoplus_\lambda\frac{\1_{\mcR_\lambda}}{d_\lambda}\otimes\rho_\lambda,\qquad \rho_\lambda\succeq0,\qquad\sum_\lambda\Tr\rho_\lambda=1,\qquad D_{\mathrm{st}}\coloneqq\sum_\lambda m_\lambda^2,
\label{eq:state-statistical-dimension}
\end{align}
where $\Tr\rho_\lambda$ is the sector probability, all sums run over $\lambda$ satisfying $m_\lambda>0$, and the family $\mcS_G(\mcH)$ has $D_{\mathrm{st}}-1$ real parameters.
For the stronger promise~\eqref{eq:state-invariant-support-promise}, let $\varnothing$ denote the trivial irrep and set
\begin{align}
\mcH_{\varnothing}\coloneqq\left\{\ket{v}\in\mcH \; \middle| \; U(g)\ket{v}=\ket{v}\ \forall g\in G\right\},\qquad P_{\varnothing}\coloneqq\int_G U(g)\, \dd g,
\label{eq:state-fixed-subspace}
\end{align}
where $\dd g$ is the Haar measure over $G$ and $P_{\varnothing}$ projects onto $\mcH_{\varnothing}\simeq\mcM_{\varnothing}$, with $\mcR_{\varnothing} \otimes \mcM_{\varnothing}= \mcM_{\varnothing}$.
We normalize Haar measure so that $\int_G\dd g=1$, and set $\mcM_{\varnothing}=\{0\}$ and $m_{\varnothing}=0$ if the trivial irrep is absent.
Integrating Eq.~\eqref{eq:state-invariant-support-promise} gives $\rho P_{\varnothing}=P_{\varnothing}\rho=\rho$, and the converse follows from $U(g)P_{\varnothing}=P_{\varnothing}$.
Thus this promise is equivalent to $\rho=P_{\varnothing}\rho P_{\varnothing}$: it allows every state on $\mcM_{\varnothing}$ and is empty when $m_{\varnothing}=0$.

\begin{proposition}[Compression and purification dimension]
\label{prop:state-compression}
The family $\mcS_G(\mcH)$ is equivalent to the block-diagonal states $\sigma=\bigoplus_\lambda\rho_\lambda$ on $\bigoplus_\lambda\mcM_\lambda$ through known quantum channels that preserve trace distance on these families.
The invariant purifications of Sec.~\ref{sec:random-purification-dilation} lie in
\begin{align}
\mcK_G\coloneqq(\overline{\mcH}\otimes\mcH)^G\simeq\bigoplus_\lambda\overline{\mcM_\lambda}\otimes\mcM_\lambda,\qquad \dim\mcK_G=D_{\mathrm{st}}.
\label{eq:state-purification-space}
\end{align}
Here $G$ acts by $\overline U\otimes U$.
\end{proposition}
\begin{proof}
Measure the sector and discard the irrep register, as in symmetry-adapted tomography~\cite{sauvage2024classical}; the inverse channel restores $\1_{\mcR_\lambda}/d_\lambda$, and
\begin{align}
d_{\mathrm{tr}}(\rho,\rho')=\frac12\sum_\lambda\|\rho_\lambda-\rho'_\lambda\|_1=d_{\mathrm{tr}}(\sigma,\sigma').
\label{eq:state-compression-isometry}
\end{align}
Schur's lemma gives $(\overline{\mcR_\lambda}\otimes\mcR_\mu)^G=\left\{0\right\}$ for $\lambda\neq\mu$ and the span of $\dket{\1_{d_\lambda}}/\sqrt{d_\lambda}$ otherwise, proving Eq.~\eqref{eq:state-purification-space}.
With $\mcE\simeq\overline{\mcH}$ and the vectorization convention~\eqref{eq:vectorization-convention}, the support isometry $J:\bigoplus_\lambda\overline{\mcM_\lambda}\otimes\mcM_\lambda\to\mcE\otimes\mcH$ and reconstruction map are
\begin{align}
J\!\left(\bigoplus_\lambda\dket{A_\lambda}\right)=\bigoplus_\lambda\frac{\dket{\1_{d_\lambda}}}{\sqrt{d_\lambda}}\otimes\dket{A_\lambda},\qquad \Tr_{\mcE}\!\left[J\ketbra{u}J^\dagger\right]=\bigoplus_\lambda\frac{\1_{\mcR_\lambda}}{d_\lambda}\otimes A_\lambda A_\lambda^\dagger,
\label{eq:state-support-isometry}
\end{align}
where $\ket{u}=\bigoplus_\lambda\dket{A_\lambda}$.
The normalized irrep vectors give $\|J\ket u\|_2^2=\sum_\lambda\|A_\lambda\|_2^2=\|u\|_2^2$, and the range of $J$ is $\mcK_G$.
\end{proof}

\subsection{Wreath-product symmetry}
\label{sec:state-symmetries}

Let $\mcQ=\bigoplus_\lambda\mcM_\lambda$, $Q=\sum_\lambda m_\lambda$, and $d_\mcH=\sum_\lambda d_\lambda m_\lambda$.
Under the identification in Eq.~\eqref{eq:state-purification-space},
\begin{align}
\mcK_G\subset\overline{\mcQ}\otimes\mcQ,\qquad D_{\mathrm{st}}=\sum_\lambda m_\lambda^2\leq\Big(\sum_\lambda m_\lambda\Big)^2=Q^2\leq\Big(\sum_\lambda d_\lambda m_\lambda\Big)^2=d_\mcH^2.
\label{eq:state-dimension-chain}
\end{align}

\paragraph{The $n$-copy input.}
The sample $\rho^{\otimes n}$ commutes with independent group actions and permutations of the copies, hence with the wreath-product representation
\begin{align}
R(\bm g,\sigma)=\left(\bigotimes_{i=1}^nU(g_i)\right)\pi(\sigma),
\label{eq:state-learning-wreath-symmetry}
\end{align}
where $(\bm g,\sigma)\in G\wr\mfS_n$ and $\pi(\sigma)$ permutes the copies.
The generalized Schur decomposition~\eqref{eq:tensor-rho-decomposition} uses this symmetry to separate the irrep registers from the multiplicity registers that contain the unknown state.
Random purification then produces a mixture of tensor powers of invariant purifications; encoding each factor by $J^\dagger$ places the output in $\operatorname{Sym}^n(\CC^{D_{\mathrm{st}}})$, where the Hayashi measurement acts.

\paragraph{Three levels of symmetry.}
An unrestricted learner on $\mcH$ uses the dimension parameter $d_\mcH^2$; compressing each copy but ignoring the block constraints gives $Q^2$; retaining the full promise gives $D_{\mathrm{st}}$.
For each family, Theorem~\ref{thm:state-learning-optimal} gives copy complexity $\Theta((D-1+\log(1/\eta))/\varepsilon^2)$ when $D\geq 2$ holds.
The unrestricted results~\cite{haah2017sample,odonnell2016efficient,scharnhorst2025optimal} follow by taking the trivial group on $\mcH$ or $\mcQ$.
Treating the state as unrestricted allows $Q^2-D_{\mathrm{st}}=\sum_{\lambda\neq\mu}m_\lambda m_\mu$ off-diagonal parameters that the symmetry fixes to zero.
If $L$ irreps occur in $U$, Cauchy--Schwarz gives $Q^2/D_{\mathrm{st}}\leq L$.
These are worst-case comparisons over different state families; they do not give a lower bound for an unrestricted learner at each covariant input.
Twirling an unrestricted estimate cannot increase its trace error on covariant inputs, but leaves the measurement and copy count unchanged, so it alone gives no improved copy bound.

\subsection{Optimal learning}
\label{sec:state-optimal-learning}

The upper bounds below specialize symmetry-compatible random purification~\cite{chen2026gaussian} and the purification-based tomography analysis of Ref.~\cite{pelecanos2025mixed} to the $D_{\mathrm{st}}$-dimensional invariant purification space.
We combine these guarantees with a matching lower bound for the covariant state family.
For a nonempty family $\mcF$, let $N_{\mcF}(\varepsilon,\eta)$ denote the least number of copies needed by a collective protocol to estimate every $\rho\in\mcF$ within trace distance $\varepsilon$, with failure probability at most $\eta$.
We use the squared fidelity
\begin{equation}
F(\rho,\sigma)\coloneqq\left\|\sqrt{\rho}\sqrt{\sigma}\right\|_1^2.
\label{eq:state-fidelity}
\end{equation}

\begin{theorem}[Optimal learning of covariant states]
\label{thm:state-learning-optimal}
For every known finite-dimensional unitary representation $U$ of a compact group $G$ and every $n\geq1$, there is a collective protocol that takes $n$ copies of any $\rho\in\mcS_G(\mcH)$ and returns a quantum state $\widehat\rho\in\mcS_G(\mcH)$ satisfying
\begin{align}
\sup_{\rho\in\mcS_G(\mcH)}\mathbb E_\rho\!\left[d_{\mathrm{tr}}(\rho,\widehat\rho)^2\right]\leq\sup_{\rho\in\mcS_G(\mcH)}\mathbb E_\rho\!\left[1-F(\rho,\widehat\rho)\right]\leq\frac{D_{\mathrm{st}}-1}{n+D_{\mathrm{st}}}.
\label{eq:state-learning-mean}
\end{align}
For $D_{\mathrm{st}}\geq2$, $0<\varepsilon<1$, and $0<\eta<1$, it satisfies $\Prob_\rho[d_{\mathrm{tr}}(\rho,\widehat\rho)>\varepsilon]\leq\eta$ uniformly over $\rho$ whenever
\begin{align}
n\geq\left\lceil\frac{2(D_{\mathrm{st}}-1)+8\log(1/\eta)}{\varepsilon^2}\right\rceil.
\label{eq:state-learning-confidence}
\end{align}
There is a universal $\varepsilon_0>0$ such that, uniformly over the representation, $0<\varepsilon\leq\varepsilon_0$, and $0<\eta\leq1/3$,
\begin{align}
N_{\mcS_G(\mcH)}(\varepsilon,\eta)=\Theta\!\left(\frac{D_{\mathrm{st}}-1+\log(1/\eta)}{\varepsilon^2}\right)\qquad(D_{\mathrm{st}}\geq2).
\label{eq:state-learning-minimax}
\end{align}
The implicit constants are universal, and the lower bound allows arbitrary collective measurements and quantum-state estimates outside $\mcS_G(\mcH)$.
If $D_{\mathrm{st}}=1$, the family contains a single known state and requires no copies.
\end{theorem}

\begin{proof}
For $D_{\mathrm{st}}\geq2$, random purification reduces the task to pure-state estimation in dimension $D_{\mathrm{st}}$, followed by reconstruction through $J$ defined in Eq.~\eqref{eq:state-support-isometry}.
\paragraph{Protocol.}
\begin{enumerate}
\item \emph{Random purification and support encoding.}
Apply Theorem~\ref{thm:random-covariant-purification} (Sec.~\ref{sec:compact-oracle-implementation} for compact groups), then encode the output using $J^\dagger$:
\begin{align}
\int_{\Pur_G(\rho)}\left(J^\dagger\ketbra{\Psi}J\right)^{\otimes n}\,d\Psi=\int\ketbra{\psi}^{\otimes n}\,d\nu_\rho(\psi),\qquad\ket{\psi}=J^\dagger\ket{\Psi}\in\CC^{D_{\mathrm{st}}},
\label{eq:state-learning-shared-purifier}
\end{align}
where $\nu_\rho$ is the distribution obtained by mapping a Haar-random invariant purification $\Psi$ to $J^\dagger\ket\Psi$.
\item \emph{Covariant pure-state measurement.}
Apply the Hayashi POVM~\eqref{eq:prelim-hayashi-povm} with $t=n$ and $q=D_{\mathrm{st}}$~\cite{hayashi1998asymptotic,bruss1999optimal,hayashi2005pure}.
\item \emph{Reconstruction of the quantum state.}
For the unit-vector outcome $\ket{u}=\bigoplus_\lambda\dket{A_\lambda}$, return
\begin{align}
\widehat\rho_\lambda\coloneqq A_\lambda A_\lambda^\dagger,\qquad\widehat\rho\coloneqq\bigoplus_\lambda\frac{\1_{\mcR_\lambda}}{d_\lambda}\otimes\widehat\rho_\lambda=\Tr_{\mcE}\!\left[J\ketbra{u}J^\dagger\right].
\label{eq:state-learning-estimator}
\end{align}
\end{enumerate}

\paragraph{Error analysis.}
The Fuchs--van de Graaf inequality combined with contractivity of fidelity under the partial trace in Eq.~\eqref{eq:state-learning-estimator} provides $d_\mathrm{tr}(\rho,\widehat\rho)^2\leq 1-F(\rho, \widehat{\rho}) \leq 1-|\braket{\psi}{u}|^2$.
We show the mean squared error~\eqref{eq:state-learning-mean} and the high-probability guarantee~\eqref{eq:state-learning-confidence} by applying Eq.~\eqref{eq:prelim-hayashi-mean-fidelity} and Prop.~\ref{prop:prelim-hayashi-confidence} with $t=n$ and $q=D_\mathrm{st}$.

\paragraph{Optimality.}
Appendix~\ref{app:state-tomography-lower} extends the Holevo--Fano method~\cite{haah2017sample} to multiplicity blocks and uses classical distribution bounds~\cite{canonne2020learning}.
This gives $\Omega((D_{\mathrm{st}}-1)/\varepsilon^2)$ copies, while binary testing gives $\Omega(\log(1/\eta)/\varepsilon^2)$.
Their maximum is at least half their sum, proving Eq.~\eqref{eq:state-learning-minimax} for a sufficiently small universal $\varepsilon_0>0$.
\end{proof}

\paragraph{Stronger support promises.}
Encoding $\mcH_{\varnothing}$ and applying the theorem to the trivial representation provides $D_{\mathrm{st}}=m_{\varnothing}^2$.
More generally, suppose each block is supported on a known subspace $F_\lambda\subseteq\mcM_\lambda$ of dimension $s_\lambda$, with $\sum_\lambda s_\lambda>0$.
Restrict $U$ to $\mcV=\bigoplus_\lambda\mcR_\lambda\otimes F_\lambda$ and apply the theorem with $D_{\mathrm{st}}$ replaced by $\sum_\lambda s_\lambda^2$.
The lower bound still allows estimates on all of $\mcH$: if $\Pi$ projects onto $\mcV$ and $\omega$ is a fixed state on $\mcV$, the channel $X\mapsto\Pi X\Pi+\Tr[(\1-\Pi)X]\omega$ fixes the family and maps every estimate into $\mcV$ without increasing its error.

\paragraph{Permutation symmetry.}
For site permutations on $k$ qudits, conjugation permutes the factors of $(\mcL(\CC^d))^{\otimes k}$, so the invariant operator space is $\operatorname{Sym}^k(\mcL(\CC^d))$ and
\begin{align}
D_{\mathrm{st}}=\sum_{\lambda\vdash_d k}(m_\lambda^{(d)})^2=\binom{k+d^2-1}{d^2-1} = \Theta_d(k^{d^2-1})\qquad m_{\varnothing}=\dim\operatorname{Sym}^k(\CC^d)=\binom{k+d-1}{d-1} = \Theta_d(k^{d-1}).
\label{eq:state-permutation-dimension}
\end{align}
As $k\to\infty$ at fixed $d\geq2$, Proposition~\ref{prop:symmetric-group-parameters-asymptotics} provides $Q=\sum_{\lambda} m_\lambda^{(d)}=\sum_{\lambda\vdash_d k}m_\lambda^{(d)}=\Theta_d(k^{(d-1)(d+2)/2})$.
Table~\ref{tab:state-permutation-comparison} compares the resulting copy bounds.
\begin{table}[tb]
\centering
\caption{Dimension parameter $D$ in the optimal collective copy complexity $\Theta((D-1+\log(1/\eta))/\varepsilon^2)$ for fixed local dimension $d\geq2$, $0<\eta\leq1/3$, and sufficiently small $\varepsilon$.
We consider the unrestricted case, and the three levels of symmetry in Sec.~\ref{sec:state-symmetries} for $k$ qudits.}
\label{tab:state-permutation-comparison}
\begingroup
\renewcommand{\arraystretch}{1.4}
\begin{tabular*}{\linewidth}{@{\extracolsep{\fill}}llcc@{}}
\toprule
State family used by learner & $D$ & Qudits (local dimension $d$) & Qubits ($d=2$) \\
\midrule
Unrestricted on $\mcH$ & $d_\mcH^2$ & $d^{2k}$ & $4^k$ \\
Unrestricted on $\mcQ$ & $Q^2$ & $\Theta_d\!\left(k^{(d-1)(d+2)}\right)$ & $\lfloor(k+2)^2/4\rfloor^2$ \\
Permutation-invariant & $D_{\mathrm{st}}$ & $\tbinom{k+d^2-1}{d^2-1}=\Theta_d\!\left(k^{d^2-1}\right)$ & $\tbinom{k+3}{3}$ \\
Supported on $\operatorname{Sym}^k(\CC^d)$ & $m_{\varnothing}^2$ & $\tbinom{k+d-1}{d-1}^{2}=\Theta_d\!\left(k^{2d-2}\right)$ & $(k+1)^2$ \\
\bottomrule
\end{tabular*}
\endgroup
\end{table}

\paragraph{Implementation.}
Theorem~\ref{thm:efficient-hayashi-measurement} implements the Hayashi measurement with $t=n$ and $q=D_{\mathrm{st}}$, using the same $n$ input copies.
Suppose purification and encoding prepare the full ensemble within trace distance $\delta_{\mathrm p}$ of the ideal one.
On the ideal ensemble, the theorem couples the normalized circuit output $\widehat u$ to the ideal outcome $u$ so that $\|u-\widehat u\|_2\leq\tau$ except with probability $\zeta$.
Return $\widehat\rho_{\mathrm{impl}}=\Tr_{\mcE}[J\ketbra{\widehat u}J^\dagger]$.
For normalized vectors,
\begin{align}
d_{\mathrm{tr}}(\ketbra{u},\ketbra{\widehat u})=\sqrt{1-|\braket{u}{\widehat u}|^2}\leq\|u-\widehat u\|_2.
\label{eq:state-readout-stability}
\end{align}
Contractivity under reconstruction, the triangle inequality, and the preparation error give
\begin{align}
\Prob[d_{\mathrm{tr}}(\rho,\widehat\rho_{\mathrm{impl}})>\varepsilon+\tau]\leq\eta+\delta_{\mathrm p}+\zeta,
\label{eq:state-learning-finite-implementation}
\end{align}
whenever $n$ satisfies Eq.~\eqref{eq:state-learning-confidence}.
For total targets $0<\varepsilon,\eta<1$, use statistical targets $(\varepsilon/2,\eta/2)$ and choose $\delta_{\mathrm p}=\eta/4$, $\tau=\varepsilon/2$, and $\zeta=\eta/4$.
The resulting copy count and Hayashi gate cost are
\begin{align}
n=\left\lceil\frac{8(D_{\mathrm{st}}-1)+32\log(2/\eta)}{\varepsilon^2}\right\rceil,\qquad G_H=O\!\left((n+D_{\mathrm{st}})\polylog(n,D_{\mathrm{st}},\varepsilon^{-1},\eta^{-1})\right).
\label{eq:state-implementation-resources}
\end{align}

The total gate cost also includes random purification and support encoding, and a small $D_{\mathrm{st}}$ alone does not make these steps efficient for an arbitrary representation.
For permutations at fixed $d$, Sec.~\ref{sec:random-purification-efficiency} supplies random purification, and the Schur transform and its conjugate implement the support encoding by removing the known irrep vector in Eq.~\eqref{eq:state-support-isometry}.
That vector is a uniform superposition of matching standard tableaux, which can be prepared recursively using the branching probabilities $d_{\lambda\setminus c}/d_\lambda$ from the Frobenius formula~\eqref{eq:hook-length}.
Reversing this preparation and indexing the sector and multiplicity coordinates gives the $D_{\mathrm{st}}$-dimensional encoding.
Allocating error $\eta/8$ to purification and $\eta/(8n)$ to each copy's encoding ensures $\delta_{\mathrm p}\leq\eta/4$.
These steps have gate cost polynomial in $k$, $n$, and the logarithms of inverse implementation errors; with the copy count above, the total is $O_d(\poly(k,\varepsilon^{-1},\log\eta^{-1}))$.
The classical estimate is reported as multiplicity blocks, whose size is polynomial in $k$ at fixed $d$; writing the full matrix in the physical basis has a different output cost.

\section{Learning covariant quantum channels}
\label{sec:learning-covariant-channels}

We consider the problem of learning an unknown $G$-covariant quantum channel $\Lambda\in \CovChan_G(\mcI, \mcO)$ for a compact group $G$.
As a concrete example, we consider permutation-covariant channels, which are covariant with respect to the symmetric group $G=\mfS_k$ with representations $U_1(\sigma) = U_2(\sigma) = \pi_d(\sigma)$ on $\mcI\simeq \mcO\simeq (\CC^d)^{\otimes k}$.
We define the set of all permutation-covariant quantum channels as
\begin{align}
    \PCovChan_{k,d}\coloneqq \CovChan_{\mfS_k}((\CC^d)^{\otimes k}, (\CC^d)^{\otimes k}).
\end{align}
Then, an $\mfS_k$-covariant dilation isometry $W\in \Dil_{\mfS_k}(\Lambda)$ satisfies
\begin{align}
    W\pi_d(\sigma)=\pi_d(\sigma)^{\otimes 3} W\quad \forall \sigma\in \mfS_k.
\end{align}
Since $\pi_d(\sigma)^{\otimes 3} \cong \pi_{d^3}(\sigma)$ holds, the multiplicities $m_\lambda$ and $M_\lambda$ in Eqs.~\eqref{eq:decomposition-multiplicity-1} and \eqref{eq:decomposition-multiplicity-2} are given by
\begin{align}
    m_\lambda = m_{\lambda}^{(d)}, \qquad M_\lambda = m_{\lambda}^{(d^3)}.
\end{align}
We consider the problem of learning an unknown permutation-covariant quantum channel $\Lambda$ by using $n$ calls to the channel and evaluating the diamond distance or the Choi trace distance between the true channel $\Lambda$ and the estimated channel $\widehat{\Lambda}$.
We summarize the main results of this section in Tab.~\ref{tab:learning-protocol-comparison}.

\begin{table}
\centering
\caption{The upper and lower bounds on the query complexity of learning covariant channels at constant success probability.
For $G$-covariant channels, we define $D_G$, $D''_G$, and $C'_G$ as detailed in Cor.~\ref{cor:covariant-channel-tomography-high-probability} and Thm.~\ref{thm:permutation-channel-strong-diamond-lower}, and define $\kappa$ such that $D''_G \geq \kappa D_G$.
The diamond distance upper bound hides a polylogarithmic factor in the number of irreps $L$, which is denoted by $\widetilde{O}(\cdot)$.
For the diamond distance at sufficiently small constant error, when $\kappa = \Theta(1)$ holds, the upper bound and lower bound coincide up to a polylogarithmic factor, while the matching lower bound for the Choi trace distance remains open.
For permutation-covariant channels in $\PCovChan_{k,d}$ at fixed local dimension $d$, the upper and lower bounds on the query complexity match up to a constant factor for the Choi trace distance in the stated small-error regime, while the diamond distance lower bound matches the upper bound for sufficiently small constant error.}
\label{tab:learning-protocol-comparison}
\begingroup
\newcommand{\tablerowspace}{\rule[-0.9\baselineskip]{0pt}{3\baselineskip}}
\begin{tabular*}{\linewidth}{@{\extracolsep{\fill}}llcc@{}}
\toprule
Covariance & Figure of merit & Upper bound & Lower bound \\
\midrule
\tablerowspace\multirow{2}{*}[-2mm]{Any compact group $G$} & Diamond distance
& $\displaystyle \widetilde{O}\left(\frac{D_G}{\varepsilon_\diamond^2}\right)$ (Cor.~\ref{cor:covariant-channel-tomography-high-probability})
& $\displaystyle\Omega_\kappa\left(\frac{D_G}{\varepsilon_\diamond^{\kappa/2}}\right)$ (Thm.~\ref{thm:permutation-channel-strong-diamond-lower}) \\[4mm]
\tablerowspace & Choi trace distance
& $\displaystyle O\left(\frac{C'_G}{\varepsilon_\mathrm{tr}^2}\right)$ (Cor.~\ref{cor:covariant-channel-tomography-high-probability})
& Open \\[4mm]
\midrule
\tablerowspace\multirow{2}{*}[-4mm]{Symmetric group} & Diamond distance
& $\displaystyle O_d\left(\frac{k^{d^4-1}}{\varepsilon_\diamond^2}\right)$ (Thm.~\ref{thm:permutation-covariant-channel-sql-scaling})
& $\displaystyle\Omega_d\left(k^{d^4-1}\over \varepsilon_\diamond^{1/4}\right)$ (Cor.~\ref{cor:permutation-channel-strong-diamond-lower}) \\[4mm]
\tablerowspace & Choi trace distance
& $\displaystyle O_d\left(\frac{k^{d^4-d^2}}{\varepsilon_\mathrm{tr}^2}\right)$ (Thm.~\ref{thm:learning-permutation-covchannel-compression})
& $\displaystyle\Omega_d\left(\frac{k^{d^4-d^2}}{\varepsilon_\mathrm{tr}^2}\right)$ (Thm.~\ref{thm:permutation-channel-trace-matching-lower}) \\[4mm]
\bottomrule
\end{tabular*}
\endgroup
\end{table}

\subsection{Learning covariant isometries and channels}

We show the upper bound on the query complexity of learning covariant isometries in diamond distance and Choi trace distance with high probability in Thm.~\ref{thm:covariant-isometry-tomography-high-probability}, which also provides upper bounds for learning covariant channels using the random dilation superchannel (Cor.~\ref{cor:covariant-channel-tomography-high-probability}).
We also show the optimal mean square Choi trace distance error of learning covariant isometries in Appendix~\ref{appendix:learning-covariant-isometries-choi-trace-norm}.

\begin{theorem}[Learning covariant isometries in diamond distance and Choi trace distance with high probability]
\label{thm:covariant-isometry-tomography-high-probability}
Suppose $W: \mcI\to \mcO$ is an unknown isometry with the block structure $W = \bigoplus_{\lambda} \1_{d_\lambda} \otimes W_\lambda$ where $W_\lambda: \CC^{m_\lambda} \to \CC^{M_\lambda}$, and define the parameters
\begin{align}
    \label{eq:parameters-for-diamond-norm}
    L\coloneqq \#\left\{\lambda\;\middle|\;m_\lambda>0\right\}, \quad D_G\coloneqq\sum_{\lambda}m_\lambda M_\lambda,\quad K_G\coloneqq \sum_{\lambda}m_\lambda, \quad
    D'_G\coloneqq D_G + K_G \log(4L),
\end{align}
\begin{align}
    \label{eq:parameters-for-choi-trace-norm}
    d_\mcI\coloneqq \sum_{\lambda} d_\lambda m_\lambda, \quad
    C'_G\coloneqq {1\over d_\mcI}\!\left(\sum_{\lambda}m_\lambda\sqrt{d_\lambda M_\lambda}\right)^2.
\end{align}
For every $0<\varepsilon_\diamond\leq1$ and $0<\eta\leq1/2$, there exists a parallel estimation protocol using
\begin{align}
    n = O\qty({D'_G + K_G\log(\eta^{-1})\over\varepsilon_\diamond^2})
\end{align}
queries of $W$ and outputting an estimate $\widehat{W}$ such that
\begin{align}
    \label{eq:guarantee-diamond-norm}
    \Prob\qty[d_\diamond(W,\widehat{W})>\varepsilon_\diamond]
    \leq\eta.
\end{align}
For every $0<\varepsilon_\mathrm{tr}\leq1$ and $0<\eta\leq1/2$, there exists a parallel protocol using
\begin{align}
    n = O\qty({C'_G + \log(\eta^{-1})\over\varepsilon_\mathrm{tr}^2})
\end{align}
queries of $W$ and outputting an estimate $\widehat{W}$ such that
\begin{align}
    \label{eq:guarantee-choi-trace-norm}
    \Prob\qty[d_\mathrm{tr}(\rho_W,\rho_{\widehat{W}})>\varepsilon_\mathrm{tr}]
    \leq\eta,
\end{align}
where $\rho_W$ and $\rho_{\widehat{W}}$ are the normalized Choi states of $W$ and $\widehat{W}$, respectively.
The gate complexity of the protocol is given by
\begin{align}
    O\qty(n N_\mathrm{Sch}^{(G \curvearrowright \mcI)}\qty({\eta\over 16n}) + n N_\mathrm{Sch}^{(G \curvearrowright \mcO)}\qty({\eta\over 16n})+ \poly\qty(D_G, n, \log \eta^{-1}, \log \varepsilon^{-1})),
\end{align}
where $\varepsilon = \varepsilon_\diamond$ for diamond distance and $\varepsilon = \varepsilon_\mathrm{tr}$ for Choi trace distance, and $N_\mathrm{Sch}^{(G\curvearrowright X)}(\xi)$ denotes the gate count for implementing the Schur transform on $X\in\left\{\mcI,\mcO\right\}$ with operator-norm error at most $\xi$.
\end{theorem}

\begin{proof}
    We first construct a protocol, which is a symmetry-adapted version of the isometry tomography via its Choi state in Ref.~\cite{mele2025optimal}.
    Then, we evaluate its diamond-distance error and Choi trace distance error to conclude the proof.
    If $D_G=1$, the isometry channel is unique and both claims are trivial; hence below we assume $D_G\geq2$.

\paragraph{Protocol.}
We fix $\bm{v} = (v_\lambda)_{\lambda \in \irrep{G}, m_\lambda>0}$ satisfying $v_\lambda > 0$ and $\sum_\lambda v_\lambda = 1$, which will be optimized for the diamond distance error and Choi trace distance error separately.
\begin{enumerate}
\item
\emph{Preparation of the probe state.}
Using $n$ parallel queries of $W$, prepare $n$ copies of
\begin{align}
    \label{eq:psi_W_v}
\ket{\psi_{W,\bm{v}}}
\coloneqq\bigoplus_\lambda
\sqrt{v_\lambda\over m_\lambda}\dket{W_\lambda}
\in\bigoplus_\lambda\CC^{m_\lambda}\otimes\CC^{M_\lambda}.
\end{align}
\item
\emph{Covariant measurement.}
Apply the covariant pure-state POVM~\cite{hayashi1998asymptotic} to $\ket{\psi_{W,\bm{v}}}^{\otimes n}$ and denote its pure-state estimate by $\ket{\widehat{\psi}}$.
\item
\emph{Construction of the estimate.}
For the projector $\Pi_\lambda$ onto the $\lambda$-th vectorized block, we define
\begin{align}
\dket{\widetilde{W}_\lambda}
\coloneqq\sqrt{m_\lambda \over v_\lambda}\,
\Pi_\lambda\ket{\widehat{\psi}},
\qquad
\widehat{W}_\lambda\coloneqq\pol(\widetilde{W}_\lambda), \qquad
\widehat{W}\coloneqq
\bigoplus_\lambda\1_{d_\lambda}\otimes\widehat{W}_\lambda,
\end{align}
where $\pol(X)$ is the polar isometry of $X$ defined by the polar decomposition
\begin{align}
    X = \pol(X) Z
\end{align}
using an isometry $\pol(X)$ and a positive semidefinite matrix $Z=(X^\dagger X)^{1/2}$.
It is characterized as~\cite{fan1955metric,audenaert2014generalisation}
\begin{align}
    \label{eq:pol-argmin}
    \pol(X) \in \argmin_{Y: Y^\dagger Y = \1} \|X-Y\|_q
\end{align}
for any $q\in [1,\infty]$.
When $X^\dagger X$ is full-rank, $\pol(X)$ is unique and given by $\pol(X) = X(X^\dagger X)^{-1/2}$.
Otherwise, we can take the Moore-Penrose pseudoinverse in the above expression and pad outside the support of $X^\dagger X$ to make $\pol(X)$ an isometry.
The operator $\widehat{W}$ is the estimate of the $G$-covariant isometry $W$.
\end{enumerate}

The covariant-POVM output has the form
\begin{align}
\ket{\widehat{\psi}}
=e^{i\phi}\sqrt{1-\delta}\ket{\psi_{W,\bm{v}}}
+\sqrt\delta\ket{r},
\label{eq:covariant-additive-confidence-beta-output}
\end{align}
where $1-\delta\sim{\rm Beta}(n+1,D_G-1)$, and $\ket{r}$ is an independent Haar-random unit vector in $\ket{\psi_{W,\bm{v}}}^\perp$ [see Eq.~\eqref{eq:covariant-additive-confidence-beta-output-prelim} in Sec.~\ref{subsec:prelim-hayashi}].
Let $\ket G$ be a standard complex Gaussian vector on the probe space, i.e.,
\begin{align}
    \ket{G} = \sum_j {x_j + i y_j \over \sqrt{2}} \ket{j},
\end{align}
using the computational basis $\{\ket{j}\}_j$ and independent standard normal random variables $x_j, y_j$.
For each $\lambda$, we define $G_\lambda\in\CC^{M_\lambda\times m_\lambda}$ by $\dket{G_\lambda}\coloneqq\Pi_\lambda\ket{G}$.
Then, entries in $G_\lambda$ are independent standard complex Gaussian random variables.
Then, $\abs{\braket{\psi_{W,\bm{v}}}{G}}^2$ obeys the same distribution as ${x^2+y^2 \over 2}$ for $x, y\sim \mcN(0,1)$, i.e., $\abs{\braket{\psi_{W,\bm{v}}}{G}}^2 \sim \mathrm{Exp}(1)$.
We define $\ket{H} \coloneqq (\1-\ketbra{\psi_{W,\bm{v}}})\ket{G}$, and independently take $Y\sim{\rm Gamma}(n+1,1)$.
Since $\|H\|_2^2\sim{\rm Gamma}(D_G-1,1)$ and its direction is independent of its norm, Eq.~\eqref{eq:covariant-additive-confidence-beta-output} admits the exact coupling
\begin{align}
\ket{r}={\ket{H}\over\|H\|_2},
\qquad
{\delta\over1-\delta}={\|H\|_2^2\over Y},
\end{align}
and thus we have
\begin{align}
    \label{eq:delta-r-h-y}
\sqrt{\delta\over1-\delta}\ket{r}={\ket{H}\over\sqrt{Y}}.
\end{align}
With $\dket{R_\lambda}\coloneqq\Pi_\lambda\ket{r}$, the operator $\widetilde{W}_\lambda$ is given by $\widetilde{W}_\lambda=e^{i\phi}\sqrt{1-\delta}\,W_\lambda+\sqrt\delta\sqrt{m_\lambda/v_\lambda}\,R_\lambda$.
Since $\pol(aX)=\pol(X)$ for every $a>0$, we have
\begin{align}
\widehat{W}_\lambda
=\pol\left(e^{i\phi}W_\lambda
+\sqrt{\delta\over 1-\delta} \sqrt{m_\lambda \over v_\lambda} R_\lambda\right).
\end{align}
Since the polar isometry is characterized by Eq.~\eqref{eq:pol-argmin}, we have
\begin{align}
\|e^{i\phi}W_\lambda-\widehat{W}_\lambda\|_q
\leq 2\sqrt{\delta\over1-\delta}\sqrt{m_\lambda \over v_\lambda}\|R_\lambda\|_q
={2\over\sqrt{Y}}\sqrt{m_\lambda\over v_\lambda}\,
\|H_\lambda\|_q,
\label{eq:covariant-additive-confidence-polar-bound}
\end{align}
where $H_\lambda$ is defined by $\dket{H_\lambda}\coloneqq \Pi_\lambda\ket{H}$.

\paragraph{Error analysis in diamond distance.}
We use the following inequality\footnote{Reference~\cite{haah2023query} proves this inequality for unitaries, but its proof extends directly to isometries.}~\cite[Prop.~1.6]{haah2023query}
\begin{align}
    d_\diamond(W, \widehat{W}) \leq \|e^{i\phi} W - \widehat{W}\|_\infty
\end{align}
to bound the diamond distance by the operator norm.
Using Eq.~\eqref{eq:covariant-additive-confidence-polar-bound} with $q=\infty$, we have
\begin{align}
    \label{eq:diamond-norm-bound}
    d_\diamond(W, \widehat{W})\leq {2\over \sqrt{Y}}\max_\lambda\sqrt{m_\lambda\over v_\lambda}\,\|H_\lambda\|_\infty.
\end{align}
By definition of $H_\lambda$, we have
\begin{align}
H_\lambda
=G_\lambda-\braket{\psi_{W,\bm{v}}}{G} \sqrt{v_\lambda\over m_\lambda}W_\lambda.
\end{align}
The Gaussian concentration inequality~\cite[Thm.~II.13]{davidson2001local} (see also Ref.~\cite[Prop.~6.33]{aubrun2017alice}) yields
\begin{align}
    \Prob\qty[\|G_\lambda\|_\infty \geq \sqrt{m_\lambda}+\sqrt{M_\lambda} + t] \leq e^{-t^2}
\end{align}
for $t\geq 0$.
Since $|\braket{\psi_{W,\bm{v}}}{G}|^2\sim{\rm Exp}(1)$, we have $\Prob[|\braket{\psi_{W,\bm{v}}}{G}|\geq t]\leq e^{-t^2}$ for $t\geq0$.
The triangle inequality provides
\begin{align}
    \|H_\lambda\|_\infty
    &\leq \|G_\lambda\|_\infty + |\braket{\psi_{W,\bm{v}}}{G}| \sqrt{v_\lambda\over m_\lambda}\\
    &\leq \|G_\lambda\|_\infty + |\braket{\psi_{W,\bm{v}}}{G}|.
\end{align}
By using the union bound, we have
\begin{align}
    \Prob\qty[\|H_\lambda\|_\infty \geq \sqrt{m_\lambda}+\sqrt{M_\lambda} + 2t]\leq 2e^{-t^2}.
\end{align}
On the other hand, for $Y\sim \mathrm{Gamma}(n+1, 1)$, we have
\begin{align}
    \label{eq:Y-tail-bound}
    \Prob\qty[Y<{n+1\over 2}]\leq \exp\left[-\left(\log2-{1\over2}\right)(n+1)\right],
\end{align}
where we use the Chernoff bound
\begin{align}
    \Prob\{X\leq x\}
    \leq e^{\theta x}\mathbb E[e^{-\theta X}]
    =e^{\theta x}(1+\theta)^{-q},
    \qquad X\sim\operatorname{Gamma}(q,1),\quad\theta>0,
\end{align}
with $q = n+1$, $x={n+1\over 2}$, and $\theta=1$.
Using the union bound with Eq.~\eqref{eq:diamond-norm-bound}, we have
\begin{align}
    &\Prob\qty[d_\diamond(W, \widehat{W}) \geq \sqrt{8\over n+1} \max_\lambda \sqrt{m_\lambda \over v_\lambda}(\sqrt{m_\lambda}+\sqrt{M_\lambda}+2t)]\notag\\
    &\leq \Prob\qty[Y<{n+1\over 2}] + \sum_\lambda \Prob\qty[\|H_\lambda\|_\infty \geq \sqrt{m_\lambda}+\sqrt{M_\lambda} + 2t]\\
    &\leq \exp\left[-\left(\log2-{1\over2}\right)(n+1)\right] + 2L e^{-t^2}.
\end{align}
By choosing $t = \sqrt{\log(4L/\eta)}$ and $n+1\geq \log(2/\eta)/(\log2-1/2)$ so that the first and second terms are at most $\eta/2$, we have
\begin{align}
    \Prob\qty[d_\diamond(W, \widehat{W}) \geq \sqrt{8\over n+1} \max_\lambda \sqrt{m_\lambda \over v_\lambda}(\sqrt{m_\lambda}+\sqrt{M_\lambda}+2\sqrt{\log(4L/\eta)})] \leq \eta.
\end{align}
We choose $v_\lambda$ to minimize $\max_\lambda \sqrt{m_\lambda \over v_\lambda}(\sqrt{m_\lambda}+\sqrt{M_\lambda}+2\sqrt{\log(4L/\eta)})$, given by
\begin{align}
    v_\lambda \coloneqq {m_\lambda (\sqrt{m_\lambda}+\sqrt{M_\lambda}+2\sqrt{\log(4L/\eta)})^2 \over \sum_\lambda m_\lambda (\sqrt{m_\lambda}+\sqrt{M_\lambda}+2\sqrt{\log(4L/\eta)})^2},
\end{align}
and then we have
\begin{align}
    \max_\lambda \sqrt{m_\lambda \over v_\lambda}(\sqrt{m_\lambda}+\sqrt{M_\lambda}+2\sqrt{\log(4L/\eta)})
    &= \sqrt{\sum_\lambda m_\lambda (\sqrt{m_\lambda}+\sqrt{M_\lambda}+2\sqrt{\log(4L/\eta)})^2}\\
    &\leq \sqrt{3 \sum_\lambda m_\lambda (m_\lambda+M_\lambda+4\log(4L/\eta))}.
\end{align}
Thus, we have
\begin{align}
    \Prob\qty[d_\diamond(W, \widehat{W}) \geq \varepsilon_\diamond] \leq \eta,
\end{align}
with
\begin{align}
    n
    &= O\qty(\sum_{\lambda} m_\lambda (m_\lambda+M_\lambda+\log(4L/\eta))\over \varepsilon_\diamond^2)\\
    &=O\qty(D'_G + K_G \log(\eta^{-1}) \over \varepsilon_\diamond^2),
\end{align}
where we use $m_\lambda\leq M_\lambda$ in the last line.

\paragraph{Error analysis in Choi trace distance.}
The squared Choi trace distance error is bounded by
\begin{align}
    d_\mathrm{tr}(\rho_W, \rho_{\widehat{W}})^2 \leq {1\over d_\mcI} \|e^{i\phi} W - \widehat{W}\|_2^2 = {1\over d_\mcI} \sum_\lambda d_\lambda \|e^{i\phi} W_\lambda - \widehat{W}_\lambda\|_2^2.
\end{align}
Using Eq.~\eqref{eq:covariant-additive-confidence-polar-bound} with $q=2$, we have
\begin{align}
    \label{eq:choi-trace-norm-bound-H_lambda}
    d_\mathrm{tr}(\rho_W, \rho_{\widehat{W}})^2\leq {4\over Y} \sum_\lambda{d_\lambda m_\lambda \over d_\mcI v_\lambda} \|H_\lambda\|_2^2.
\end{align}
By definition of $H_\lambda$, we have
\begin{align}
    \sum_\lambda{d_\lambda m_\lambda \over d_\mcI v_\lambda} \|H_\lambda\|_2^2 = \bra{G} A \ket{G},
\end{align}
where $A$ is defined by
\begin{align}
    A\coloneqq \sum_\lambda {d_\lambda m_\lambda \over d_\mcI v_\lambda} (\1-\ketbra{\psi_{W,\bm{v}}}) \Pi_\lambda (\1-\ketbra{\psi_{W,\bm{v}}}),
\end{align}
satisfying
\begin{align}
    \label{eq:A-infty-norm}
    \|A\|_\infty &\leq \max_\lambda {d_\lambda m_\lambda \over d_\mcI v_\lambda},\\
    \Tr A &\leq \sum_{\lambda}  {d_\lambda m_\lambda \over d_\mcI v_\lambda} \Tr \Pi_\lambda = \sum_{\lambda} {d_\lambda m_\lambda^2 M_\lambda \over d_\mcI v_\lambda},\\
    \label{eq:A-trace-square}
    \Tr A^2 &\leq \|A\|_\infty \Tr A.
\end{align}
The Gaussian concentration inequality~\cite[Prop.~1.1]{hsu2012tail} provides
\begin{align}
    \label{eq:gaussian-tail-bound}
    \Prob[\bra{G}A\ket{G} \geq \Tr A + 2\sqrt{\Tr A^2 t}+2\|A\|_\infty t] \leq e^{-t}
\end{align}
for $t\geq 0$.
Since
\begin{align}
    \Tr A + 2\sqrt{\Tr A^2 t}+2\|A\|_\infty t
    &\leq \Tr A + 2\sqrt{\|A\|_\infty \Tr A t}+2\|A\|_\infty t\\
    &\leq 2\Tr A + 3\|A\|_\infty t\\
    &\leq 2\sum_\lambda {d_\lambda m_\lambda^2 M_\lambda \over d_\mcI v_\lambda} + 3\max_\lambda {d_\lambda m_\lambda \over d_\mcI v_\lambda} t
\end{align}
holds, where we use the AM--GM inequality and Eqs.~\eqref{eq:A-infty-norm}--\eqref{eq:A-trace-square}, we have
\begin{align}
    Q_t
    &\coloneqq 2\sum_\lambda {d_\lambda m_\lambda^2 M_\lambda \over d_\mcI v_\lambda}
    + 3\max_\lambda {d_\lambda m_\lambda \over d_\mcI v_\lambda} t,\\
    \Prob\qty[\bra{G}A\ket{G} \geq Q_t] &\leq e^{-t}.
\end{align}
By using the union bound with Eqs.~\eqref{eq:Y-tail-bound} and \eqref{eq:gaussian-tail-bound}, we have
\begin{align}
    \Prob\qty[d_\mathrm{tr}(\rho_W, \rho_{\widehat{W}}) \geq \sqrt{{8Q_t\over n+1}}] \leq e^{-t} + \exp\left[-\left(\log2-{1\over2}\right)(n+1)\right].
\end{align}
By choosing $t = \log(2/\eta)$ and $n+1\geq \log(2/\eta)/(\log2-1/2)$ so that the first and second terms are at most $\eta/2$, we have
\begin{align}
    \Prob\qty[d_\mathrm{tr}(\rho_W, \rho_{\widehat{W}}) \geq \sqrt{{8Q_{\log(2/\eta)}\over n+1}}] \leq \eta.
\end{align}
By choosing $v_\lambda$ by
\begin{align}
    v_\lambda\coloneqq {1\over 2} \qty({\sqrt{d_\lambda m_\lambda^2 M_\lambda} \over \sum_\lambda \sqrt{d_\lambda m_\lambda^2 M_\lambda}}+{d_\lambda m_\lambda \over d_\mcI}),
\end{align}
we have
\begin{align}
    \sum_\lambda {d_\lambda m_\lambda^2 M_\lambda \over d_\mcI v_\lambda}
    &\leq 2 C'_G, \qquad \max_\lambda {d_\lambda m_\lambda \over d_\mcI v_\lambda} \leq 2,
\end{align}
i.e.,
\begin{align}
     \Prob\qty[d_\mathrm{tr}(\rho_W, \rho_{\widehat{W}}) \geq \sqrt{{8\over n+1}\qty(4 C'_G + 6 \log (2\over \eta))}] \leq \eta.
\end{align}
Thus, we have
\begin{align}
    \Prob\qty[d_\mathrm{tr}(\rho_W, \rho_{\widehat{W}}) \geq \varepsilon_\mathrm{tr}] \leq \eta,
\end{align}
with
\begin{align}
    n=O\qty(C'_G + \log(\eta^{-1}) \over \varepsilon_\mathrm{tr}^2).
\end{align}

\paragraph{Gate complexity.}\leavevmode
We fix the final targets $\varepsilon,\eta$ and apply the preceding statistical analysis with accuracy $\varepsilon/2$ and failure probability $\eta/4$, so that we can leave additional error budgets for the implementation errors.
Choose $n$ to satisfy these statistical bounds and $n+1\geq\log(8/\eta)/(\log2-1/2)$, and set the per-operation compilation precision to $\xi\coloneqq\eta/(16n)$.
Given a classical description of $v_\lambda$, the probe state $\ket{\psi_{W, \bm{v}}}$ in Eq.~\eqref{eq:psi_W_v} can be implemented in a quantum circuit as follows.
\begin{enumerate}
    \item Prepare the superposition state $\sum_\lambda \sqrt{v_\lambda} \ket{\lambda}$.
    This step can be implemented by $O(L)$ gates with arbitrary-angle single-qubit rotations and CNOT gates using the state preparation shown in Ref.~\cite{mottonen2005transformation}.
    By compiling the arbitrary-angle single-qubit rotations into a universal gate set using the Solovay--Kitaev theorem~\cite{kitaev1997quantum,dawson2005solovay}, this preparation unitary can be implemented with operator-norm error at most $\xi$ using $O(L\polylog(L/\xi))$ gates.
    \item Controlled on $\ket{\lambda}$, prepare the maximally entangled state $m_\lambda^{-1/2}\sum_{i=1}^{m_\lambda} \ket{i}\otimes \ket{i}$.
    The preparation of the maximally entangled state for each $\lambda$ can be implemented by $O(\log m_\lambda)$ gates allowing arbitrary-angle single-qubit rotations~\cite{shukla2024efficient}.
    Similarly to Step~1, each preparation unitary is compiled to operator-norm error at most $\xi$ using $O(\log m_\lambda\polylog(\log m_\lambda/\xi))$ gates over the universal gate set.
    Using unary iteration~\cite{babbush2018encoding}, the controlled preparation can be implemented with operator-norm error at most $\xi$ for the entire controlled operation using $O(\sum_\lambda\log m_\lambda\polylog(\log m_\lambda/\xi))$ gates.
    We can append the representation-space register in any efficiently implementable state, e.g., $\ket{0}$, which only has a negligible contribution to the overall gate complexity.
    \item Apply $U_\mathrm{Sch}^{(G\curvearrowright \mcO)} W U_\mathrm{Sch}^{(G \curvearrowright \mcI)\dagger}$.
    Compile each Schur transform to operator-norm error at most $\xi$, using $O(N_\mathrm{Sch}^{(G\curvearrowright\mcI)}(\xi)+N_\mathrm{Sch}^{(G\curvearrowright\mcO)}(\xi))$ gates in total.
    \item The resulting state is in the space $\mathrm{span}\{\ket{\lambda, 0, i} \otimes \ket{\lambda, 0, j}\}$, where the first register encodes the irrep label $\lambda$, the second register is the representation-space register, and the third register is the multiplicity-space register.
    We compress it to a $D_G$-dimensional space by discarding the representation-space register and relabeling the multiplicity-space register as
    \begin{align}
        (\lambda, i, j) \mapsto \mathrm{offset}(\lambda) + (i-1) M_\lambda + j,
    \end{align}
    where $\mathrm{offset}(\lambda) \coloneqq \sum_{\lambda'\prec \lambda} m_{\lambda'} M_{\lambda'}$.
    This relabeling can be implemented with $O(L\polylog(D_G))$ gates.
\end{enumerate}
Each of the two preparation operations and the two Schur transforms is used $n$ times and contributes at most $n\xi=\eta/16$ to the trace-distance error of the prepared state.
By contractivity of errors under the subsequent measurement and classical processing, the final output distribution differs in total variation distance by at most $\eta/4$ from that of the same measurement and processing applied to the ideally prepared state.
Then, the Hayashi measurement can be implemented by $O((n+D_G)\polylog(n, D_G, \zeta^{-1}, \tau^{-1}))$ gates with approximation error $\tau$ and failure probability $\zeta$ as shown in Sec.~\ref{sec:hayashi-qubit-implementation}.
We set $\zeta=\eta/4$ and choose $\tau$ explicitly below.
For the classical post-processing, we use the polar estimator defined above and count its singular value decomposition cost as $O(m_\lambda^2 M_\lambda)$ arithmetic operations~\cite{dongarra2018singular} for each $\lambda$.

We first analyze the implemented Hayashi measurement on the ideally prepared state, and then add the preparation error budget $\eta/4$.
Let $\ket{u}$ be the ideal Hayashi measurement outcome and $\ket{\widehat{u}}$ be the implemented Hayashi measurement outcome, satisfying
\begin{align}
    \Prob\qty[\|\ket{\widehat{u}}-\ket{u}\|_2 > \tau] \leq \zeta.
\end{align}
The ideal Hayashi measurement outcome $\ket{u}$ has the form [see Eqs.~\eqref{eq:covariant-additive-confidence-beta-output}, \eqref{eq:delta-r-h-y}]
\begin{align}
    \ket{u} = \sqrt{Y\over Y+\|H\|_2^2}\left(e^{i\phi}\ket{\psi_{W,\bm{v}}} + {1\over \sqrt{Y}}\ket{H}\right).
\end{align}
Thus, the implemented Hayashi measurement outcome $\ket{\widehat{u}}$ can be written as
\begin{align}
    \ket{\widehat{u}} = \sqrt{Y\over Y+\|H\|_2^2}\left(e^{i\phi}\ket{\psi_{W,\bm{v}}} + {1\over \sqrt{Y}}\ket{H}\right) + \ket{E},
\end{align}
where $\ket{E}$ is the error vector satisfying $\|\ket{E}\|_2 \leq \tau$ with probability at least $1-\zeta$.
With $E_\lambda$ defined by $\dket{E_\lambda} = \Pi_\lambda \ket{E}$, $\|\ket{E}\|_2 \leq \tau$ implies $\sum_\lambda \|E_\lambda\|_2^2 \leq \tau^2$.
Defining $\widehat{W}_\lambda^{\mathrm{fin}}$ by the same post-processing procedure as $\widehat{W}_\lambda$ but using $\ket{\widehat{u}}$ instead of $\ket{u}$, we have
\begin{align}
    \|\widehat{W}_\lambda^{\mathrm{fin}} - e^{i\phi} W_\lambda\|_q
    &\leq 2\sqrt{m_\lambda \over v_\lambda}\left({\|H_\lambda\|_q \over \sqrt{Y}}  + \sqrt{1+{\|H\|_2^2 \over Y}}\|E_\lambda\|_q\right).
\end{align}
For the diamond distance case, similarly to Eq.~\eqref{eq:diamond-norm-bound}, we have
\begin{align}
    d_\diamond(W, \widehat{W}^{\mathrm{fin}}) \leq 2\max_\lambda \sqrt{m_\lambda \over v_\lambda}\left({\|H_\lambda\|_\infty \over \sqrt{Y}} + \sqrt{1+{\|H\|_2^2 \over Y}}\|E_\lambda\|_\infty \right).
\end{align}
By the chosen statistical budget, the first term is at most $\varepsilon_\diamond/2$ except with probability $\eta/4$.
By Eq.~\eqref{eq:Y-tail-bound} and the choice of $n$, $Y\geq(n+1)/2$ holds except with probability $\eta/8$.
Since $\|H\|_2^2\sim{\rm Gamma}(D_G-1,1)$, the Chernoff bound gives $\|H\|_2^2\leq2(D_G-1+\log(8/\eta))$ except with probability $\eta/8$.
Consequently, with probability at least $1-\eta/4$,
\begin{align}
    \label{eq:s-bound}
    \sqrt{1+{\|H\|_2^2\over Y}}\leq A_{n,\eta}\coloneqq\sqrt{1+{4(D_G-1+\log(8/\eta))\over n+1}}.
\end{align}
We also have $\|E_\lambda\|_\infty \leq \|E_\lambda\|_2 = \|\Pi_\lambda \ket{E}\|_2 \leq \|\ket{E}\|_2 \leq \tau$.
For diamond distance, we choose
\begin{align}
    \tau\coloneqq\frac{\varepsilon_\diamond}{4A_{n,\eta}\sqrt{\max_\lambda(m_\lambda/v_\lambda)}},
\end{align}
such that
\begin{align}
    d_\diamond(W,\widehat{W}^{\mathrm{fin}})\leq\frac{\varepsilon_\diamond}{2}+2A_{n,\eta}\sqrt{\max_\lambda(m_\lambda/v_\lambda)}\,\tau=\varepsilon_\diamond.
\end{align}
For the Choi trace distance case, similarly to Eq.~\eqref{eq:choi-trace-norm-bound-H_lambda}, we have
\begin{align}
    d_\mathrm{tr}(\rho_W,\rho_{\widehat{W}^{\mathrm{fin}}})
    &\leq {2\over\sqrt{Y}}\left(\sum_\lambda {d_\lambda m_\lambda\over d_\mcI v_\lambda}\|H_\lambda\|_2^2\right)^{1/2}+2\sqrt{1+{\|H\|_2^2\over Y}}\left(\sum_\lambda {d_\lambda m_\lambda\over d_\mcI v_\lambda}\|E_\lambda\|_2^2\right)^{1/2}.
    \label{eq:finite-choi-bound}
\end{align}
The first term is at most $\varepsilon_\mathrm{tr}/2$ except with probability $\eta/4$, and Eq.~\eqref{eq:s-bound} bounds the amplification factor by $A_{n,\eta}$ except with probability $\eta/4$.
Since $\max_\lambda {d_\lambda m_\lambda \over d_\mcI v_\lambda} \leq 2$ and $\sum_\lambda \|E_\lambda\|_2^2 \leq \tau^2$, we have
\begin{align}
    d_\mathrm{tr}(\rho_W,\rho_{\widehat{W}^{\mathrm{fin}}})\leq\frac{\varepsilon_\mathrm{tr}}{2}+2\sqrt2 A_{n,\eta}\tau=\varepsilon_\mathrm{tr},
\end{align}
on the corresponding three success events, when we choose
\begin{align}
    \tau\coloneqq\frac{\varepsilon_\mathrm{tr}}{4\sqrt2 A_{n,\eta}}.
\end{align}
For either distance, a union bound gives failure probability at most $\eta/4+\eta/4+\zeta=3\eta/4$ on the ideally prepared input.
Adding the preparation contribution $\eta/4$ proves the stated guarantee with failure probability at most $\eta$.
These choices change only constant factors in the stated query bounds.
Substituting the above choices of $\xi$, $\zeta$, and $\tau$ into the gate counts, and including the classical post-processing, provides
\begin{align}
    O\qty(n N_\mathrm{Sch}^{(G\curvearrowright\mcI)}\qty({\eta \over 16n})+n N_\mathrm{Sch}^{(G\curvearrowright\mcO)}\qty({\eta \over 16n})+\poly\qty(D_G,n,\log\eta^{-1},\log\varepsilon^{-1})).
\end{align}
\end{proof}

By combining the above theorem with the random dilation superchannel, we obtain the following corollary for learning covariant channels with high probability.

\begin{corollary}[Learning covariant channels in diamond distance and Choi trace distance with high probability]
\label{cor:covariant-channel-tomography-high-probability}
Suppose $\Lambda\in \CovChan_G(\mcI, \mcO)$ is an unknown covariant channel with respect to a compact group $G$, and define the parameters $D'_G, K_G, C'_G$ as in Eqs.~\eqref{eq:parameters-for-diamond-norm} and \eqref{eq:parameters-for-choi-trace-norm} with the multiplicities $m_\lambda, M_\lambda$ given in Eqs.~\eqref{eq:decomposition-multiplicity-1} and \eqref{eq:decomposition-multiplicity-2}.
For every $0<\varepsilon_\diamond \leq 1$ and $0<\eta\leq1/2$, there exists a parallel estimation protocol using
\begin{align}
    n = O\qty({D'_G + K_G\log(\eta^{-1})\over\varepsilon_\diamond^2})
\end{align}
queries of $\Lambda$ and outputting an estimate $\widehat{\Lambda}$ such that
\begin{align}
    \Prob\qty[d_\diamond(\Lambda,\widehat{\Lambda})>\varepsilon_\diamond]
    \leq\eta.
\end{align}
For every $0<\varepsilon_\mathrm{tr}\leq1$ and $0<\eta\leq1/2$, there exists a parallel protocol using
\begin{align}
    n = O\qty({C'_G + \log(\eta^{-1})\over\varepsilon_\mathrm{tr}^2})
\end{align}
queries of $\Lambda$ and outputting an estimate $\widehat{\Lambda}$ such that
\begin{align}
    \Prob\qty[d_\mathrm{tr}(\rho_\Lambda,\rho_{\widehat{\Lambda}})>\varepsilon_\mathrm{tr}]
    \leq\eta,
\end{align}
where $\rho_\Lambda$ and $\rho_{\widehat{\Lambda}}$ are the normalized Choi states of $\Lambda$ and $\widehat{\Lambda}$, respectively.
\end{corollary}
\begin{proof}
    The random dilation superchannel combined with the covariant isometry tomography protocol in Thm.~\ref{thm:covariant-isometry-tomography-high-probability} provides the high probability bounds in diamond distance or Choi trace distance with respect to the dilation isometries.
    The data-processing inequality transfers the error guarantees to the original covariant channel, which completes the proof.
\end{proof}

\subsection{Learning permutation-covariant channels}

We can combine the random dilation superchannel with the parallel estimation protocol of permutation-covariant isometry channels to obtain a learning protocol for permutation-covariant channels (Thm.~\ref{thm:permutation-covariant-channel-sql-scaling}).
This protocol is further improved for the Choi trace distance error by using the Choi state compression technique (Thm.~\ref{thm:learning-permutation-covchannel-compression}).

\begin{theorem}[Learning permutation-covariant channels in diamond distance and Choi trace distance via random dilation]
\label{thm:permutation-covariant-channel-sql-scaling}
Suppose $d\geq2$, $k\geq1$, $0<\varepsilon_\diamond,\varepsilon_\mathrm{tr}\leq1$, $0<\eta\leq1/2$, and $\Lambda$ is an unknown permutation-covariant channel $\Lambda\in \PCovChan_{k,d}$.
Then, there exists a parallel protocol that, given $n$ parallel calls to $\Lambda$, outputs a classical description of a channel $\widehat{\Lambda}$ such that
\begin{align}
    \Prob\qty[d_\diamond(\Lambda, \widehat{\Lambda})\geq \varepsilon_\diamond]\leq \eta
\end{align}
with
\begin{align}
    \label{eq:permutation-covariant-channel-sql-scaling}
    n = O_d\qty({k^{d^4-1} + k^{(d-1)(d+2)/2}\log \eta^{-1} \over \varepsilon_\diamond^2}).
\end{align}
Similarly, there exists a parallel protocol outputting a classical description of $\widehat\Lambda$ such that
\begin{align}
    \Prob\qty[d_\mathrm{tr}(\rho_\Lambda,\rho_{\widehat\Lambda})\geq \varepsilon_\mathrm{tr}]\leq \eta
\end{align}
with
\begin{align}
    \label{eq:permutation-covariant-channel-choi-trace-norm-scaling}
    n = O_d\qty({k^{d^4-(d^2+1)/2} + \log \eta^{-1} \over \varepsilon_\mathrm{tr}^2}).
\end{align}
The gate complexity of the protocol is given by
\begin{align}
    O_d\qty(\poly(k, \varepsilon^{-1}, \log \eta^{-1})),
\end{align}
where $\varepsilon = \varepsilon_\diamond$ for the diamond distance case and $\varepsilon = \varepsilon_\mathrm{tr}$ for the Choi trace distance case.
\end{theorem}
\begin{proof}
    The protocol is given as follows.
    \begin{enumerate}
        \item \emph{Random dilation.} Apply the $n$-slot random dilation superchannel to obtain the correlated channel $\int_{\Dil_{\mfS_k}(\Lambda)}\dd W\,W^{\otimes n}(\mathord\cdot)W^{\dagger\otimes n}$.
        \item \emph{Learning the isometry channel.} Apply the parallel estimation measurement in Thm.~\ref{thm:covariant-isometry-tomography-high-probability} to obtain $\widehat{W}$.
        \item \emph{Output the estimate.} Output $\widehat{\Lambda}\coloneqq \Tr_\mcE[\widehat{W}(\cdot)\widehat{W}^\dagger]$.
    \end{enumerate}
    The data-processing inequality transfers the diamond-distance and Choi trace-distance error guarantees from $W$ to $\Lambda$.
    For the finite-gate implementation, we apply Thm.~\ref{thm:covariant-isometry-tomography-high-probability} with failure probability $\eta/2$ and compile the entire $n$-slot random-dilation circuit so that its output channel, including the environment, is within diamond distance $\eta/2$ of the ideal channel.
    This does not change the scaling of the query complexity.
    By contractivity under the subsequent measurement and classical processing, the final output distribution changes by at most $\eta/2$ in total variation distance, giving total failure probability at most $\eta/2+\eta/2=\eta$ without changing the asymptotic query bounds.
    The query bounds in Eqs.~\eqref{eq:permutation-covariant-channel-sql-scaling} and \eqref{eq:permutation-covariant-channel-choi-trace-norm-scaling} follow from the query bounds in Cor.~\ref{cor:covariant-channel-tomography-high-probability} and the asymptotic estimates of $D'_G$, $K_G$, and $C'_G$ shown in Prop.~\ref{prop:symmetric-group-parameters-asymptotics}.
    The Schur transforms for $\mfS_k$ on $(\CC^d)^{\otimes k}$ and $(\CC^{d^3})^{\otimes k}$ can both be implemented with operator-norm error $\xi$ using $O_d(\poly(k,\log\xi^{-1}))$ gates~\cite{bacon2005quantum,bacon2006efficient}.
    Thus, the total gate complexity follows from the gate complexity in Thm.~\ref{thm:covariant-isometry-tomography-high-probability} and the gate complexity of the random dilation superchannel shown in Sec.~\ref{sec:random-purification-efficiency}, which is given by $O_d(\poly(k,\log\eta^{-1}))$.
\end{proof}

For the Choi trace distance error, we can further improve the sample complexity by compressing the Choi state of the unknown channel to a smaller support.
The Choi space decomposes as follows:
\begin{align}
    (\overline{\CC^d})^{\otimes k} \otimes (\CC^d)^{\otimes k}
    &\cong \bigoplus_{\mu, \alpha\vdash_d k} \overline{\mcS_\mu} \otimes \CC^{m_\mu^{(d)}} \otimes \mcS_\alpha \otimes \CC^{m_\alpha^{(d)}}\\
    \label{eq:choi-space-decomposition-2}
    &\cong \bigoplus_{\mu, \alpha\vdash_d k} \bigoplus_{\nu\vdash k} \mcS_\nu \otimes \CC^{g_{\mu\alpha\nu}} \otimes \CC^{m_\mu^{(d)}} \otimes \CC^{m_\alpha^{(d)}}\\
    &= \bigoplus_{\nu\vdash k} \mcS_\nu \otimes \bigoplus_{\mu\vdash_d k} \CC^{w_{\mu\nu}},
    \label{eq:choi-space-decomposition}
\end{align}
where $g_{\mu\alpha\nu}$ is the Kronecker coefficient, which is  the multiplicity of the irrep $\nu$ in the decomposition of the tensor product $\mcS_\mu \otimes \mcS_\alpha$, and $w_{\mu\nu}$ is defined by
\begin{align}
    w_{\mu\nu} \coloneqq \sum_{\alpha\vdash_d k} m_\mu^{(d)} m_\alpha^{(d)} g_{\mu\alpha\nu},
\end{align}
which quantifies the multiplicity of the irrep $\nu$ in the decomposition of the Choi space coming from the input Young sector $\mu$.
We then take $\Delta>0$ and define a centered Young window by
\begin{align}
    \label{eq:young-window}
    \young{d}{k}(\Delta)\coloneqq \qty{\lambda\vdash_d k \; \middle|\; \|\bm{x}\|_1\leq \sqrt{d(d+1)\Delta} \text{ with } x_i \coloneqq {\lambda_i - k/d \over \sqrt{k}}}.
\end{align}
Then, we consider the compressed Choi space
\begin{align}
    \label{eq:compressed-choi-space-decomposition}
    \bigoplus_{\mu\in \young{d}{k}(\Delta)} \overline{\mcS_\mu} \otimes \CC^{m_\mu^{(d)}} \otimes (\CC^d)^{\otimes k}
    &\cong \bigoplus_{\nu\vdash k} \mcS_\nu \otimes \CC^{c_\nu(\Delta)},
\end{align}
where $c_\nu(\Delta)$ is defined by
\begin{align}
    c_\nu(\Delta)\coloneqq \sum_{\mu\in \young{d}{k}(\Delta)} w_{\mu\nu}.
\end{align}
Let $P_{k,d}(\Delta)$ be the orthogonal projector onto the Schur sectors corresponding to $\mu\in\young{d}{k}(\Delta)$, given by
\begin{align}
    P_{k,d}(\Delta)\coloneqq \sum_{\mu\in\young{d}{k}(\Delta)}P_\mu,
\end{align}
using the Young projector $P_\mu$ defined in Eq.~\eqref{eq:def-young-projector}.
Let $\mcH_\Delta\coloneqq\operatorname{Im}(\overline{P_{k,d}(\Delta)}\otimes\1_{\mcO})$ denote the compressed Choi space in Eq.~\eqref{eq:compressed-choi-space-decomposition}, where $\operatorname{Im}$ denotes the image of the operator.
We restrict $\overline\pi_d\otimes\pi_d$ to $\mcH_\Delta$ and apply random purification with environment $\overline{\mcH_\Delta}$ and the commutant of this restricted representation.
The canonical purification, and every purification randomized by this restricted commutant, of a permutation-invariant state on $\mcH_\Delta$ then lies in
\begin{align}
    \bigoplus_{\nu\vdash k} \CC\dket{\1_{\mcS_\nu}} \otimes \overline{\CC^{c_\nu(\Delta)}} \otimes \CC^{c_\nu(\Delta)}.
    \label{eq:permutation-random-purification-support}
\end{align}
The dimension of the support of the purification is given by
\begin{align}
    D_\Delta\coloneqq \sum_{\nu\vdash k} c_\nu(\Delta)^2.
    \label{eq:permutation-fidelity-compressed-pure-estimation}
\end{align}
Then, we can show the following lemma:

\begin{lemma}[Compression of Schur sectors]
\label{lem:prefactor-free-schur-window}
Suppose $\lambda$ is a random Young diagram sampled by applying the weak Schur sampling to $k$ copies of a state $\rho$ with $\bm{p} = \spec(\rho)$, i.e.,
\begin{align}
\Prob[\lambda = \lambda'] = \Tr[P_{\lambda'}\rho^{\otimes k}],
\end{align}
where $\spec(\rho)$ denotes the vector of eigenvalues of $\rho$, arranged in nonincreasing order and padded with zeros to length $d$.
Then, we have
\begin{align}
\Prob\left[
\left\|{\lambda \over k}-\bm{p}\right\|_1\geq{d\over\sqrt k}+t
\right]
\leq\exp\!\left[-{kt^2\over2d^2}\right]
\qquad(t\geq0).
\label{eq:prefactor-free-schur-concentration}
\end{align}
Suppose $a_d\coloneqq \sqrt{d(d+1)}$ and define the constant
\begin{align}
c_d^{\mathrm{SW}}\coloneqq{(a_d-d)^2\over2d^2}>0.
\end{align}
Then, for any $\Delta\geq1$, we have
\begin{align}
1-{\Tr P_{k,d}(\Delta)\over d^k}
&\leq e^{-c_d^{\mathrm{SW}}\Delta},
\label{eq:prefactor-free-uniform-window-tail}\\
\Tr[P_{k,d}(\Delta)\rho^{\otimes k}]
&\leq\exp\left\{-{k\over2d^2}
\left[\left\|\rho-\frac{\1}{d}\right\|_1
-(a_d+d)\sqrt{\Delta\over k}\right]_+^2\right\},
\label{eq:prefactor-free-offcenter-kernel}
\end{align}
for any quantum state $\rho\in \mcL(\CC^d)$, where $[\cdot]_+ \coloneqq \max\{0, \cdot\}$.
The compressed dimension defined in Eq.~\eqref{eq:permutation-fidelity-compressed-pure-estimation} is bounded by
\begin{align}
D_\Delta = O_d(k^{d^4-d^2}\Delta^{d^2-1})
\qquad(k\geq1,\ \Delta\geq1).
\label{eq:prefactor-free-compressed-dimension}
\end{align}
\end{lemma}

\begin{proof}
    The average of $\|\lambda/k-\bm{p}\|_1$ is bounded as
    \begin{align}
        \mathbb{E}\left\|{\lambda\over k}-\bm{p}\right\|_1 \leq \sqrt{d} \sqrt{\mathbb{E}\left\|{\lambda\over k}-\bm{p}\right\|_2^2} \leq {d\over \sqrt{k}},
    \end{align}
    where the first inequality uses the Cauchy--Schwarz inequality and the second inequality is from Ref.~\cite[Thm.~1.1]{odonnell2016efficient}.
    The probability distribution of $\lambda$ is the same as the distribution of the shape of the Young tableau produced by applying the RSK correspondence to a word $w$ of length $k$ with letters drawn i.i.d. from $\bm{p}$, i.e., $w_i\sim \bm{p}$ for $i\in \{1, \ldots, k\}$~\cite{kuperberg2002random,its2001random,alicki1988symmetry,feray2012asymptotics,odonnell2016efficient}.
    It is given by $\lambda$ such that $\lambda_1+\cdots+ \lambda_i$ is the length of the longest disjoint union of $i$ weakly increasing subsequences in the word $w$~\cite{greene1974extension}.
    Suppose $w'$ is obtained from $w$ by changing one letter at position $i$, and $\lambda'$ is the corresponding Young diagram.
    Then, we have
    \begin{align}
        \abs{\left\|{\lambda \over k}-\bm{p}\right\|_1 - \left\|{\lambda' \over k}-\bm{p}\right\|_1} \leq c_i
    \end{align}
    with $c_i = 2d/k$.
    Thus, the bounded-differences inequality~\cite[Thm.~6.7]{mcdiarmid1989method} provides
    \begin{align}
        \Prob\left[
        \left\|{\lambda \over k}-\bm{p}\right\|_1\geq \mathbb{E}\left\|{\lambda \over k}-\bm{p}\right\|_1+t
        \right]
        \leq \exp[-{2t^2\over \sum_{i=1}^k c_i^2}] = \exp[-{kt^2\over2d^2}],
    \end{align}
    which proves Eq.~\eqref{eq:prefactor-free-schur-concentration}.
    Since $\Prob[\lambda = \lambda'] = {1\over d^k} \Tr[P_{\lambda'}]$ holds for $\bm{p} = \bm{u}_d \coloneqq (1/d, \ldots, 1/d)$, we have
    \begin{align}
        1-{\Tr P_{k,d}(\Delta) \over d^k}
        &= \Prob\qty[\left\|{\lambda \over k}-\bm{u}_d\right\|_1>a_d \sqrt{\Delta \over k}]\\
        &\leq \exp[-{k\over 2d^2}\qty(a_d\sqrt{\Delta \over k}-{d\over \sqrt{k}})^2]\\
        &\leq \exp[-{(a_d-d)^2\over 2d^2}\Delta],
    \end{align}
    where we use Eq.~\eqref{eq:prefactor-free-schur-concentration} in the second line and $\Delta\geq 1$ in the third line, which proves Eq.~\eqref{eq:prefactor-free-uniform-window-tail}.
    For a general $\bm{p} = \spec(\rho)$, the triangle inequality implies
    \begin{align}
        \left\|{\lambda \over k}-\bm{p}\right\|_1 \geq \|\bm{p}-\bm{u}_d\|_1 - \left\|{\lambda \over k} - \bm{u}_d\right\|_1 = \left\|\rho-{\1\over d}\right\|_1 - \left\|{\lambda \over k} - \bm{u}_d\right\|_1.
    \end{align}
    Thus, we have
    \begin{align}
        \Tr[P_{k,d}(\Delta)\rho^{\otimes k}]
        &= \Prob\qty[\left\|{\lambda\over k} - \bm{u}_d\right\|_1\leq a_d \sqrt{\Delta \over k}]\\
        &\leq \Prob\qty[\left\|{\lambda \over k}-\bm{p}\right\|_1 \geq \left\|\rho-{\1\over d}\right\|_1-a_d\sqrt{\Delta \over k}]\\
        &\leq \exp[-{k\over2d^2}\qty(\|\bm p-\bm u_d\|_1-a_d\sqrt{\Delta\over k}-{d\over\sqrt{k}})_+^2],
    \end{align}
    which proves Eq.~\eqref{eq:prefactor-free-offcenter-kernel}.

    We conclude the proof by showing Eq.~\eqref{eq:prefactor-free-compressed-dimension}.
    Suppose $\chi_\Delta$ is the character of the representation $\overline{\pi_d}\otimes \pi_d$ restricted to the subspace projected by $P_{k,d}(\Delta)$.
    Then, for $\sigma\in \mfS_k$, we have
    \begin{align}
        \chi_\Delta(\sigma) = \sum_{\nu\vdash k} c_\nu(\Delta) \chi_\nu(\sigma).
    \end{align}
    Due to the orthogonality of characters, we have
    \begin{align}
        D_\Delta = \sum_{\nu\vdash k} c_\nu(\Delta)^2 = {1\over k!}\sum_{\sigma\in\mfS_k} \chi_\Delta(\sigma)^2.
    \end{align}
    The character $\chi_\Delta$ can be expressed as
    \begin{align}
        \chi_\Delta(\sigma)
        &= \Tr[P_{k,d}(\Delta)\overline\pi_d(\sigma)] \cdot \Tr[\pi_d(\sigma)].
    \end{align}
    Then, we have
    \begin{align}
        D_\Delta = {1\over k!}\sum_{\sigma\in\mfS_k} \Tr[P_{k,d}(\Delta)\overline\pi_d(\sigma)]^2 \Tr[\pi_d(\sigma)]^2.
    \end{align}
    On the other hand, by taking the Haar-random unit vector $\ket{\psi}$ on $\mcA\otimes\mcB\otimes\mcC\otimes\mcD$ with $\mcA\simeq \mcB\simeq \mcC\simeq \mcD\simeq \CC^d$, we have
    \begin{align}
        \int \dd \psi \ketbra{\psi}^{\otimes k} = {1\over k! \binom{k+d^4-1}{d^4-1}} \sum_{\sigma\in \mfS_k} \overline{\pi}_d(\sigma) \otimes \overline{\pi}_d(\sigma) \otimes \pi_d(\sigma) \otimes \pi_d(\sigma),
    \end{align}
    and thus
    \begin{align}
        D_\Delta = \binom{k+d^4-1}{d^4-1}\int\dd\psi\,
        \Tr[P_{k,d}(\Delta)^{\otimes2}\rho_{\mcA\mcB}^{\otimes k}].
        \label{eq:permutation-compressed-dimension-haar-average}
    \end{align}
    The Cauchy--Schwarz inequality applied to Eq.~\eqref{eq:permutation-compressed-dimension-haar-average} provides
    \begin{align}
        {D_\Delta\over\binom{k+d^4-1}{d^4-1}}
        &= \int\dd\psi \Tr[P_{k,d}(\Delta)^{\otimes2}\rho_{\mcA\mcB}^{\otimes k}]\\
        &= \int\dd\psi \Tr{\qty[(P_{k,d}(\Delta) \otimes \1)\sqrt{\rho_{\mcA\mcB}^{\otimes k}}]^\dagger \qty[(\1 \otimes P_{k,d}(\Delta))\sqrt{\rho_{\mcA\mcB}^{\otimes k}}]}\\
        &\leq \int\dd\psi \sqrt{\Tr[P_{k,d}(\Delta)\rho_\mcA^{\otimes k}]\Tr[P_{k,d}(\Delta)\rho_\mcB^{\otimes k}]}.
        \label{eq:prefactor-free-kernel-haar-integral}
    \end{align}
    For a Hilbert space $X$, we define $\St(X)\coloneqq\{\rho\in\mcL(X):\rho\succeq0,\ \Tr\rho=1\}$.
    As shown in Ref.~\cite[Eq.~(3.5)]{zyczkowski2001induced}, if $\ket{\psi}$ obeys the Haar measure on $\CC^{d^4}$, the reduced density operator $\rho_{\mcA\mcB}$ obeys the uniform distribution over $\St(\mcA\otimes \mcB)$ with respect to the Hilbert--Schmidt measure induced by the Hilbert--Schmidt norm $\|A\|_2\coloneqq\sqrt{\Tr(A^\dagger A)}$.
    Since $\St(\mcA\otimes \mcB)$ is compact, the total volume is a constant depending only on $d$.
    We define $\su(X)\coloneqq \{H\in \mcL(X)\mid H =H^\dagger, \Tr H = 0\}$ for a Hilbert space $X$.
    Then, $\St(X)$ can be embedded in $\su(X)$, by $\iota: \St(X) \ni \rho \mapsto \rho-\1_X/\dim X\in \su(X)$.
    We consider a surjective linear map $L: \su(\mcA\otimes \mcB) \to \su(\mcA) \times \su(\mcB)$ defined by
    \begin{align}
        L(H)\coloneqq (\Tr_{\mcB} H, \Tr_{\mcA} H).
    \end{align}
    For every Borel set $E\subseteq\su(\mcA)\times\su(\mcB)$, the coarea formula~\cite{federer1959curvature} provides
    \begin{align}
    \Prob\left[L(\iota(\rho_{\mcA\mcB}))\in E\right]
    =\frac{1}{J_L\operatorname{vol}(\St(\mcA\otimes\mcB))}
    \int_E \dd H_\mcA\dd H_\mcB\,
    \operatorname{vol}_{q_d}\!\left(
    L^{-1}(H_\mcA,H_\mcB)\cap\Im\iota\right),
    \end{align}
    where $\vol_{q_d}$ is the $q_d$-dimensional Hausdorff measure for $q_d\coloneqq \dim \ker L = (d^2-1)^2$, and $J_L$ is the coarea Jacobian given by
    \begin{align}
        J_L\coloneqq \sqrt{\det(LL^\dagger)} = d^{d^2-1},
    \end{align}
    using the adjoint map $L^\dagger: \su(\mcA)\times \su(\mcB) \to \su(\mcA\otimes \mcB)$ given by $L^\dagger(H_\mcA, H_\mcB) = H_\mcA\otimes \1_\mcB + \1_\mcA\otimes H_\mcB$.
    Since $\|\iota(\rho)\|_2^2 = \Tr \rho^2 - d^{-2}<1$ holds, the image $\Im \iota$ is in the unit ball of $\su(\mcA\otimes \mcB)$.
    Thus, the intersection $L^{-1}(H_\mcA, H_\mcB) \cap \Im \iota$ is contained in the $q_d$-dimensional unit ball, and its Hausdorff measure is bounded by the volume of the $q_d$-dimensional unit ball, which is a constant depending only on $d$.
    Thus, we have
    \begin{align}
        &\int\dd\psi \sqrt{\Tr[P_{k,d}(\Delta)\rho_\mcA^{\otimes k}]\Tr[P_{k,d}(\Delta)\rho_\mcB^{\otimes k}]}\notag\\
        &\leq O_d\qty(\int_{\iota(\St(\mcA))\times\iota(\St(\mcB))}\dd H_\mcA\dd H_\mcB
        \sqrt{\Tr[P_{k,d}(\Delta)(H_\mcA+\1/d)^{\otimes k}]\Tr[P_{k,d}(\Delta)(H_\mcB+\1/d)^{\otimes k}]})\\
        &\leq O_d\qty(\int_{\su(d)}\dd H
       \exp\left\{-{k\over 4d^2}
        \left[\left\|H\right\|_1
        -(a_d+d)\sqrt{\Delta\over k}\right]_+^2\right\})^2\\
        &\leq O_d\qty[\qty[(a_d+d)\sqrt{\Delta\over k}]^{d^2-1} + \int_0^\infty \dd r \exp[-kr^2 \over 4d^2] \qty[r+(a_d+d)\sqrt{\Delta\over k}]^{d^2-2}]^2\\
        &\leq O_d\qty({\Delta \over k})^{d^2-1},
        \label{eq:prefactor-free-offcenter-kernel-integral}
    \end{align}
    where we use Eq.~\eqref{eq:prefactor-free-offcenter-kernel} in the second line, and $\vol(\{H\in \su(d) \mid \|H\|_1\leq r\}) = \Theta_d(r^{d^2-1})$ in the third line.
    From Eqs.~\eqref{eq:prefactor-free-kernel-haar-integral} and \eqref{eq:prefactor-free-offcenter-kernel-integral}, we obtain
    \begin{align}
        D_\Delta\leq \binom{k+d^4-1}{d^4-1} \cdot O_d\qty(\qty({\Delta \over k})^{d^2-1}) = O_d(k^{d^4-d^2}\Delta^{d^2-1}),
    \end{align}
    which concludes the proof.
\end{proof}

\begin{theorem}[Learning permutation-covariant channels in Choi trace distance via Choi state compression]
\label{thm:learning-permutation-covchannel-compression}
Suppose $d\geq2$, $k\geq1$, $0<\varepsilon_\mathrm{tr}\leq1$, $0<\eta\leq1/2$, and $\Lambda$ is an unknown permutation-covariant channel $\Lambda\in \PCovChan_{k,d}$.
Then there exists a parallel protocol that, given $n$ calls to $\Lambda$, outputs a classical description of a channel $\widehat\Lambda$ such that
\begin{align}
    \label{eq:choi-trace-norm-compression-error-bound}
    \Prob\qty[d_\mathrm{tr}(\rho_\Lambda,\rho_{\widehat\Lambda})\geq \varepsilon_\mathrm{tr}]\leq \eta
\end{align}
with
\begin{align}
    n = O_d\qty({k^{d^4-d^2} + \log \eta^{-1} \over \varepsilon_\mathrm{tr}^2}).
\end{align}
The gate complexity of the protocol is given by $O_d(\poly(k, \varepsilon^{-1}_\mathrm{tr}, \log \eta^{-1}))$.
\end{theorem}

\begin{proof}
We first provide the protocol, then analyze its error, query complexity, and gate complexity.
\paragraph{Protocol.}
We define a projector $E_j$ for $j\geq 1$ by
\begin{align}
    E_1 \coloneqq \overline{P_{k,d}}(1), \qquad E_j \coloneqq \overline{P_{k,d}}(j)-\overline{P_{k,d}}(j-1) \quad (j\geq 2),
\end{align}
and define $r_j\coloneqq\Tr E_j$.
Then $\sum_{j}E_j=\1$, and $r_j = 0$ for $j\geq 2+\lceil{4k\over d(d+1)}\rceil$.
For each $j$ satisfying $r_j>0$, define the normalized state
\begin{align}
    \ket{\Phi_j}
    \coloneqq {(E_j \otimes \1)\ket{\Phi_{d^k}^+} \over \|(E_j \otimes \1)\ket{\Phi_{d^k}^+}\|_2} = \sqrt{d^k \over r_j} (E_j \otimes \1)\ket{\Phi_{d^k}^+},
\end{align}
where $\ket{\Phi_{d^k}^+} \coloneqq {1\over \sqrt{d^k}}\sum_{i=1}^{d^k}\ket{i}\otimes\ket{i}$ is the maximally entangled state on $(\CC^d)^{\otimes k}\otimes(\CC^d)^{\otimes k}$.
The maximally entangled state then satisfies
\begin{align}
    \ket{\Phi_{d^k}^+} = \sum_{j} \sqrt{r_j \over d^k} \ket{\Phi_j}.
\end{align}
We also define the superposition states for $i<j$ by
\begin{align}
    \ket{\Phi_{ij}^{\pm}} \coloneqq {\ket{\Phi_i} \pm \ket{\Phi_j} \over \sqrt{2}}.
\end{align}
We define $\mathsf{A}_k \coloneqq \{i\mid r_i>0\}\cup \{(i,j,\pm), i<j, r_i>0, r_j>0\}$, and we make a slight abuse of notation to write $\ket{\Phi_a}$ for $a = (i,j,\pm)$ as $\ket{\Phi_{ij}^{\pm}}$, similarly for $\sigma_a$, $\widehat{\sigma}_a$, and $D_a$ below.
The protocol consists of the following four steps:
\begin{enumerate}
    \item \emph{Preparation of the probe states.}
    We define the probe states
    \begin{align}
        \sigma_j \coloneqq \qty(\id \otimes \Lambda)\qty(\ketbra{\Phi_j}),\qquad
        \sigma_{ij}^{\pm} \coloneqq \qty(\id \otimes \Lambda)\qty(\ketbra{\Phi_{ij}^{\pm}}).
    \end{align}
    For each $a\in\mathsf{A}_k$, prepare $n_a$ copies of $\sigma_a$, where $n_a$ is specified below; all preparations use parallel calls to $\Lambda$.
    \item \emph{Estimation of the probe states.} For every $a$, perform quantum state tomography on the $n_a$ copies of $\sigma_a$ to obtain an estimate $\widehat{\sigma}_a$ satisfying
    \begin{align}
        \label{eq:permutation-multiscale-probe-state-estimation}
        \Prob\left[d_\mathrm{tr}(\sigma_a, \widehat{\sigma}_a) \leq \varepsilon_a\right] \geq 1 - \eta_a,
    \end{align}
    where $\varepsilon_a$ and $\eta_a$ are specified below.
    \item \emph{Reconstruction of the Choi state.} By definition of $\sigma_a$, we have
    \begin{align}
        \rho_\Lambda = \sum_{a\in \mathsf{A}_k} p_a \sigma_a,
    \end{align}
    with $p_j\coloneqq {r_j \over d^k}$ and $p_{ij}^{\pm}\coloneqq \pm {\sqrt{r_i r_j} \over d^k}$.
    We obtain an estimate $\widehat{\rho}_\Lambda'$ of the Choi state $\rho_\Lambda$ by replacing each $\sigma_a$ with its estimate $\widehat{\sigma}_a$:
    \begin{align}
        \widehat{\rho}_\Lambda' = \sum_{a\in \mathsf{A}_k} p_a \widehat{\sigma}_a.
    \end{align}
    \item \emph{CPTP completion.} We obtain the final estimate $\widehat{\Lambda}$ by
    \begin{align}
        \label{eq:permutation-multiscale-cptp-completion}
        \widehat{\Lambda} \in \argmin_{\widehat{\Lambda}\in\PCovChan_{k,d}} d_\mathrm{tr}(\rho_{\widehat{\Lambda}}, \widehat{\rho}_\Lambda').
    \end{align}
\end{enumerate}

\paragraph{Error analysis.}
For each $j$, the state $\sigma_j$ is permutation invariant and its canonical purification lies in the support~\eqref{eq:permutation-random-purification-support} with $\Delta=j$.
For $i<j$, the same holds for $\sigma_{ij}^{\pm}$ since $E_i+E_j\preceq P_{k,d}(j)$ holds.
Applying Thm.~\ref{thm:state-learning-optimal} to the representation restricted to $\mcH_j$, as in the known-support reduction in Sec.~\ref{sec:state-optimal-learning}, therefore provides
\begin{align}
    \label{eq:permutation-multiscale-probe-state-sample-complexity}
    n_a = O_d\qty[{D_a + \log \eta_a^{-1} \over \varepsilon_a^2}]
\end{align}
to achieve Eq.~\eqref{eq:permutation-multiscale-probe-state-estimation}, where $D_a$ is defined using Eq.~\eqref{eq:permutation-fidelity-compressed-pure-estimation} by
\begin{align}
    D_a\coloneqq D_{\Delta =j}
\end{align}
for $a=j$ and $a=(i,j,\pm)$ with $i<j$.
By the triangle inequality, we have
\begin{align}
    d_\mathrm{tr}(\rho_\Lambda, \widehat{\rho}_\Lambda')
    \leq \sum_{a\in\mathsf{A}_k} \abs{p_a} d_\mathrm{tr}(\sigma_a, \widehat{\sigma}_a),
\end{align}
where the number of summands is $\abs{\mathsf{A}_k} \leq (1+\lfloor 4k/a_d^2 \rfloor)^2 = O_d(k^2)$.
We define
\begin{align}
    t_a\coloneqq{\varepsilon_\mathrm{tr}\over 4}
    {\qty[D_a+\log\qty(\abs{\mathsf{A}_k}\over\eta)]^{1/3}\abs{p_a}^{-1/3}\over \sum_{b\in\mathsf{A}_k} \qty[D_b+\log\qty(\abs{\mathsf{A}_k}\over\eta)]^{1/3}\abs{p_b}^{2/3}},
\end{align}
and we take $\varepsilon_a = t_a$, $\eta_a = \eta/\abs{\mathsf{A}_k}$, and choose $n_a$ by Eq.~\eqref{eq:permutation-multiscale-probe-state-sample-complexity} when $t_a\leq 1/2$.
Otherwise, we take $n_a = 0$ and $\widehat{\sigma}_a$ is taken to be any fixed $G$-covariant quantum state.
When $d_\mathrm{tr}(\sigma_a, \widehat{\sigma}_a)\leq t_a$ for all $a\in \mathsf{A}_k$ satisfying $t_a\leq 1/2$, we have
\begin{align}
    d_\mathrm{tr}(\rho_\Lambda, \widehat{\rho}_\Lambda')
    \leq \sum_{t_a\leq 1/2}\abs{p_a}t_a
    +\sum_{t_a>1/2}\abs{p_a}
    \leq 2\sum_{a\in\mathsf{A}_k}\abs{p_a} t_a
    ={\varepsilon_\mathrm{tr}\over 2}.
\end{align}
Thus, due to the union bound, we have
\begin{align}
    \Prob\qty[d_\mathrm{tr}(\rho_\Lambda, \widehat{\rho}_\Lambda')>{\varepsilon_\mathrm{tr}\over 2}]
    \leq\sum_{t_a\leq 1/2}\Prob\qty[d_\mathrm{tr}(\sigma_a,\widehat\sigma_a)>t_a]
    \leq \sum_{t_a\leq 1/2} \eta_a \leq \eta.
\end{align}
Since the true channel $\Lambda$ satisfies $\Lambda\in\PCovChan_{k,d}$, the triangle inequality and Eq.~\eqref{eq:permutation-multiscale-cptp-completion} yield
\begin{align}
    d_\mathrm{tr}(\rho_\Lambda, \rho_{\widehat\Lambda})
    \leq d_\mathrm{tr}(\rho_\Lambda, \widehat{\rho}_\Lambda')
    + d_\mathrm{tr}(\widehat{\rho}_\Lambda', \rho_{\widehat\Lambda})
    \leq 2 d_\mathrm{tr}(\rho_\Lambda, \widehat{\rho}_\Lambda').
\end{align}
Thus, we obtain the guarantee of the error probability shown in Eq.~\eqref{eq:choi-trace-norm-compression-error-bound}.
The query complexity is given by
\begin{align}
    n
    &= \sum_{t_a\leq 1/2} n_a \\
    &\leq O_d\qty[{\qty[\sum_{a\in\mathsf{A}_k}
    \qty[D_a+\log(\abs{\mathsf{A}_k}/\eta)]^{1/3}\abs{p_a}^{2/3}]^3
    \over\varepsilon_\mathrm{tr}^2}]\\
    &\leq O_d\qty[{\qty[\sum_{a\in\mathsf{A}_k}
    \qty[D_a^{1/3}+\log^{1/3}(\abs{\mathsf{A}_k}/\eta)]\abs{p_a}^{2/3}]^3
    \over\varepsilon_\mathrm{tr}^2}],
\end{align}
where we use $(x+y)^{1/3}\leq x^{1/3}+y^{1/3}$ for $x,y\geq 0$ in the last line.
Lemma~\ref{lem:prefactor-free-schur-window} shows
\begin{align}
    \sum_{a\in \mathsf{A}_k}D_a^{1/3}\abs{p_a}^{2/3}
    &= \sum_{j} D_j^{1/3} p_j^{2/3} + 2 \sum_{i<j} D_j^{1/3} p_i^{1/3} p_j^{1/3}\\
    &\leq O_d\left(k^{(d^4-d^2)/3}\right)
    \left[\sum_{j\geq1}j^{(d^2-1)/3}e^{-2c_d^{\mathrm{SW}}j/3}
    +2\sum_{1\leq i<j}j^{(d^2-1)/3}e^{-c_d^{\mathrm{SW}}(i+j)/3}\right]\\
    &=O_d\!\left(k^{(d^4-d^2)/3}\right),\\
    \sum_{a\in\mathsf{A}_k}\abs{p_a}^{2/3}
    &= \sum_{j} p_j^{2/3} + 2 \sum_{i<j} p_i^{1/3} p_j^{1/3}\\
    &\leq \sum_j O_d(e^{-2c_d^{\mathrm{SW}}j/3}) + 2 \sum_{i<j} O_d(e^{-c_d^{\mathrm{SW}}(i+j)/3})\\
    &= O_d(1).
\end{align}
Thus, we have
\begin{align}
    n
    =O_d\qty[{k^{d^4-d^2}+\log\eta^{-1}\over\varepsilon_\mathrm{tr}^2}].
\end{align}

\paragraph{Gate complexity.}
Let $\delta$ denote the total compilation error, with preparation errors measured in trace distance and unitary-approximation errors in operator norm.
Similarly to the probe state preparation in Thm.~\ref{thm:covariant-isometry-tomography-high-probability}, preparing each copy of $\sigma_a$ to error $\delta/(2n)$ requires $O_d(\poly(k,\log(n/\delta)))$ gates uniformly over $a\in\mathsf{A}_k$, since $r_j\leq d^k$, the number of sectors $L=\abs{\young{d}{k}}$ is polynomial in $k$ at fixed $d$, and the Schur transform has gate complexity $N_\mathrm{Sch}=O_d(\poly(k))$.
In the Schur basis for paired input--output sites of local dimension $d^2$, the projector $\overline{P_{k,d}(\Delta)}\otimes\1_{\mcO}$ acts only on multiplicity spaces of polynomial dimension at fixed $d$.
Its multiplicity-block entries can be computed by contracting the sequential Clebsch--Gordan formulas~\cite{bacon2005quantum,bacon2006efficient}, fixing a paired-site tableau path and summing matching input tableau paths in the ket and bra.
Each step retains only a constant number of intermediate Young labels and $\mathrm{U}(d)$ or $\mathrm{U}(d^2)$ multiplicity indices, all with polynomial ranges, so these entries can be computed to precision $\delta$ in $O_d(\poly(k,\log\delta^{-1}))$ time.
Using the spectral gap between $0$ and $1$ to identify each projector block's range stably from sufficiently accurate entries, we can compute an orthonormal range basis, complete it to a unitary, and synthesize the corresponding restricted-support basis change to operator-norm error $\delta$ using gate complexity $O_d(\poly(k,\log\delta^{-1}))$ including classical computation.
Together with Prop.~\ref{prop:coherent-wreath-schur-complexity} and the implementation following Thm.~\ref{thm:state-learning-optimal} in Sec.~\ref{sec:state-optimal-learning}, this gives gate complexity $O_d(\poly(k,n_a,\log(|\mathsf{A}_k|/\delta)))$ for estimation of each $\sigma_a$, with total unitary-approximation error at most $\delta/(2|\mathsf{A}_k|)$ in its subsequent circuit.
Summing these errors and the $n$ preparation errors bounds the accumulated compilation error by $\delta$.
For each $a$, use statistical accuracy $\varepsilon_a/2$ and failure probability $\eta_a/4$, and finite-readout vector error $\varepsilon_a/2$ with failure probability $\eta_a/4$, as in Eq.~\eqref{eq:state-learning-finite-implementation}.
Taking $\delta\leq\eta/4$, the union bound gives total failure probability at most $\eta/2+\delta\leq\eta$, while these constant-factor changes in the statistical targets and logarithmic compilation overheads preserve the stated sample and gate complexities and the accuracy in Eq.~\eqref{eq:choi-trace-norm-compression-error-bound}.
The CPTP completion in Eq.~\eqref{eq:permutation-multiscale-cptp-completion} can be expressed as the following semidefinite program (SDP):
\begin{align}
\begin{split}
    \min \; & {1\over 2} \Tr(X_+ +X_-)\\
    \text{s.t.} \; & 0 \preceq \rho_{\widehat{\Lambda}}, X_+, X_- \in \Comm(\overline{\pi}_d \otimes \pi_d) \cong \Comm(\pi_{d^2}),\\
    & \rho_{\widehat{\Lambda}}-\widehat{\rho}_\Lambda' = X_+ - X_-, \qquad \Tr_\mcO \rho_{\widehat{\Lambda}} = {\1_{\overline{\mcI}} \over d^k}.
\end{split}
\end{align}
This SDP can be efficiently constructed as shown in Ref.~\cite{bergh2026permutation}.
Since
\begin{align}
    \dim \Comm(\overline{\pi}_d \otimes \pi_d) = \sum_{\lambda\vdash_{d^2} k} \qty[m_\lambda^{(d^2)}]^2 = O_d(k^{d^4-1})
\end{align}
holds, the SDP has $O_d(\poly(k))$ real scalar variables.
Thus, the SDP can be solved in $O_d(\poly(k)\log \varepsilon^{-1})$ arithmetic operations with additive accuracy $\varepsilon$ using the interior-point method~\cite{vandenberghe1996semidefinite}.
Running the preceding probe-estimation and reconstruction steps with target accuracy $\varepsilon_\mathrm{tr}/2$ and taking the SDP objective error to be at most $\varepsilon=\varepsilon_\mathrm{tr}/2$ gives $d_\mathrm{tr}(\rho_\Lambda,\rho_{\widehat\Lambda})\leq 2d_\mathrm{tr}(\rho_\Lambda,\widehat{\rho}_\Lambda')+\varepsilon\leq\varepsilon_\mathrm{tr}$, preserving the guarantee in Eq.~\eqref{eq:choi-trace-norm-compression-error-bound} and the stated asymptotic complexities.
Thus, the total gate complexity including the classical post-processing is given by $O_d(\poly(k, \varepsilon^{-1}_\mathrm{tr}, \log \eta^{-1}))$.
\end{proof}

\subsection{Lower bounds}
\label{sec:lower-bound}

We show the following theorem on the lower bound for learning $G$-covariant channels (see Appendix~\ref{appendix:lower-bound-channel} for the proof).

\begin{theorem}[Lower bound for learning covariant quantum channels in diamond distance]
\label{thm:permutation-channel-strong-diamond-lower}
For $m_\lambda$ and $M_\lambda$ defined in Eqs.~\eqref{eq:decomposition-multiplicity-1} and \eqref{eq:decomposition-multiplicity-2}, we define
\begin{align}
    \label{eq:parameters-for-diamond-lower-bound}
    D_G \coloneqq \sum_\lambda m_\lambda M_\lambda,
    \qquad
    D''_G \coloneqq \sum_\lambda m_\lambda(M_\lambda-m_\lambda),
\end{align}
which are the dimension of the commutant $\Comm(\overline{U}_1 \otimes U_2)$ and the real affine dimension of the set of $G$-covariant channels, respectively.
Suppose $D_G>1$ and $D''_G>0$, and we take $\kappa\in (0,1)$ such that $D''_G\geq \kappa D_G$.
Then, there exists a constant $\varepsilon_0(\kappa)>0$, depending only on $\kappa$, with the following property.
For any $0<\varepsilon_\diamond\leq\varepsilon_0(\kappa)$, if there exists a quantum circuit using $n$ queries to an unknown $\Lambda\in\CovChan_G(\mcI, \mcO)$ that outputs an estimate $\widehat\Lambda$ satisfying
\begin{align}
    \Prob\qty[d_\diamond(\Lambda,\widehat{\Lambda})\leq\varepsilon_\diamond]\geq {2\over3}
\end{align}
for all $\Lambda\in\CovChan_G(\mcI, \mcO)$, then 
\begin{align}
    n = \Omega_\kappa\qty(D_G \over \varepsilon_\diamond^{\kappa/2})
\end{align}
holds.
\end{theorem}

As a corollary, we obtain the following lower bound for learning permutation-covariant channels in diamond distance.

\begin{corollary}[Lower bound for learning permutation-covariant quantum channels in diamond distance]
\label{cor:permutation-channel-strong-diamond-lower}
Fix $d\geq2$.
There exist $k_0(d)\in\mathbb{N}$ and $\varepsilon_0>0$ such that, for any $k\geq k_0(d)$ and $0<\varepsilon_\diamond\leq\varepsilon_0$, if a quantum circuit using $n$ queries to an unknown $\Lambda\in\PCovChan_{k,d}$ outputs an estimate $\widehat\Lambda$ satisfying
\begin{align}
    \Prob\qty[d_\diamond(\Lambda,\widehat{\Lambda})\leq\varepsilon_\diamond]\geq {2\over3}
\end{align}
for all $\Lambda\in\PCovChan_{k,d}$, then
\begin{align}
    n = \Omega_d\qty({k^{d^4-1 }\over\varepsilon_\diamond^{1/4}}).
\end{align}
\end{corollary}
\begin{proof}
    We apply Thm.~\ref{thm:permutation-channel-strong-diamond-lower} to the case of $G=\mfS_k$.
    Since
    \begin{align}
        D_{\mfS_k} &= \sum_{\lambda\in \irrep{\mfS_k}} m_\lambda^{(d)} m_\lambda^{(d^3)} = \binom{k+d^4-1}{d^4-1} = \Theta_d(k^{d^4-1}),\\
        D''_{\mfS_k} &= D_{\mfS_k} - \sum_{\lambda\in\irrep{\mfS_k}} [m_\lambda^{(d)}]^2 = D_{\mfS_k} - \binom{k+d^2-1}{d^2-1} = D_{\mfS_k} - \Theta_d(k^{d^2-1})
    \end{align}
    hold (see also Prop.~\ref{prop:symmetric-group-parameters-asymptotics}), we can take $\kappa = 1/2$ for sufficiently large $k\geq k_0(d)$ to conclude the proof.
\end{proof}

We then show the following theorem on the lower bound for learning permutation-covariant channels in the Choi trace distance (see Appendix~\ref{appendix:lower-bound-channel} for the proof).

\begin{theorem}[Matching lower bound for learning permutation-covariant channels in Choi trace distance]
\label{thm:permutation-channel-trace-matching-lower}
We fix $d\geq 2$.
There exist $k_0(d)\in\mathbb{N}$ and $\varepsilon_{0}(d)>0$ such that, for any $k\geq k_0(d)$ and $0<\varepsilon_\mathrm{tr} \leq \varepsilon_0(d)$, if a quantum circuit using $n$ queries to an unknown $\Lambda\in\PCovChan_{k,d}$ outputs an estimate $\widehat\Lambda$ satisfying
\begin{align}
    \Prob\qty[d_\mathrm{tr}(\rho_\Lambda,\rho_{\widehat\Lambda})\leq\varepsilon_\mathrm{tr}]\geq {2\over3}
\end{align}
for all $\Lambda\in\PCovChan_{k,d}$, then
\begin{align}
    n = \Omega_d\qty({k^{d^4-d^2}\over\varepsilon_\mathrm{tr}^2})
\end{align}
holds.
\end{theorem}

\subsection{Comparison with direct Choi-state tomography}
\label{sec:channel-choi-comparison}

For the permutation-covariant channel learning, we can also consider a direct approach that prepares the normalized Choi state $\rho_\Lambda$ of the unknown channel $\Lambda \in \PCovChan_{k,d}$, applies the random purification, and then performs the optimal pure-state tomography on the purified Choi state.
This strategy is utilized in the optimal channel tomography in Ref.~\cite{mele2025optimal}.
Since the purified Choi state is in the symmetric subspace of $(\CC^{d^4})^{\otimes k}$, the dimension of the canonical invariant-purification space is given by $D_{\mathrm{Choi}}(k,d) = \binom{k+d^4-1}{d^4-1} = \Theta_d(k^{d^4-1})$.
Thus, by applying the optimal pure-state tomography, we can achieve the Choi trace distance error $\varepsilon_\mathrm{tr}$ with
\begin{align}
    O_d\qty({k^{d^4-1}+\log(1/\eta)\over \varepsilon_\mathrm{tr}^2})
\end{align}
queries at failure probability $\eta$ (see also Thm.~\ref{thm:state-learning-optimal}).
On the other hand, our protocol in Thm.~\ref{thm:learning-permutation-covchannel-compression} achieves the same error and failure probability with
\begin{align}
    O_d\qty({k^{d^4-d^2}+\log(1/\eta)\over \varepsilon_\mathrm{tr}^2})
\end{align}
queries, yielding a polynomial improvement in $k$ over the direct Choi-state tomography.

\section{Efficient implementation of the Hayashi measurement}
\label{sec:hayashi-qubit-implementation}

We present an efficient finite-qubit implementation of the Hayashi measurement defined in Section~\ref{subsec:prelim-hayashi}.
First, we show that encoding a symmetric qudit state in bosonic occupation modes, performing multimode heterodyne measurement, and normalizing the outcome realizes the Hayashi POVM exactly.
We then approximate this continuous-variable measurement by a qubit circuit using occupation registers and the quantum Hermite transform (QHT)~\cite{jain2025hermite}.
Figure~\ref{fig:hayashi-heterodyne-correspondence} summarizes how these two steps compose into a finite-qubit circuit.

\begin{figure}
\centering
\includegraphics[width=\linewidth]{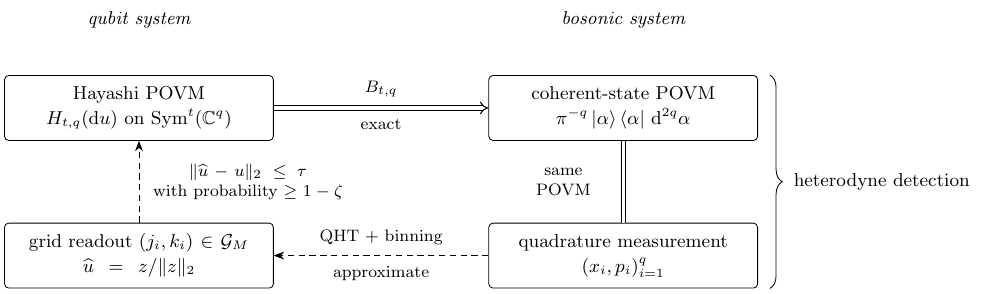}
\caption{Schematic picture of the relation between the qubit and bosonic descriptions of the Hayashi measurement.
}
\label{fig:hayashi-heterodyne-correspondence}
\end{figure}

\subsection{Hayashi measurement via heterodyne measurement}
\label{subsec:hayashi-exact-realization}

\begin{theorem}[Exact realization of the Hayashi measurement via heterodyne measurement]
    \label{thm:hayashi-exact-heterodyne}
        Let $B_{t,q}$ be the isometry mapping the normalized Dicke basis of $\operatorname{Sym}^t(\CC^q)$ to the occupation basis of the $q$-mode bosonic system with total occupation $t$, as defined in Eq.~\eqref{eq:hayashi-occupation-encoding}.
        Applying $B_{t,q}$ and then performing heterodyne measurement on each mode induces a POVM $E_{t,q}(\mathrm{d}\alpha)$ on $\operatorname{Sym}^t(\CC^q)$ with outcomes $\alpha\in\CC^q$.
        In the polar coordinates $s=\norm{\alpha}_2^2$ and $u=\alpha/\sqrt{s}$, defined for $\alpha\ne0$, $E_{t,q}(\mathrm{d}\alpha)$ is given by
    \begin{align}
    \label{eq:hayashi-exact-povm-factorization}
        E_{t,q}(\mathrm{d}s,\mathrm{d}u)=g_{t+q}(s)\,\mathrm{d}s\,H_{t,q}(\mathrm{d}u),\qquad
        g_k(s)\coloneqq\frac{e^{-s}s^{k-1}}{\Gamma(k)}\quad(s>0).
    \end{align}
    Thus, for every density operator supported on $\operatorname{Sym}^t(\CC^q)$, computing $u$ realizes the Hayashi POVM exactly, while the squared radius $s$ is independent of $u$ and follows the Gamma distribution with rate $1$ regardless of the input state.
    The outcome $\alpha=0$ has probability zero, so $u$ may be assigned arbitrarily there.
\end{theorem}

\paragraph{Heterodyne measurement and its coherent-state POVM.}
For a bosonic mode with annihilation operator $a$ and creation operator $a^\dagger$, the quadratures
\begin{align}
    X=\frac{a+a^\dagger}{\sqrt2},\qquad
    P=\frac{a-a^\dagger}{\mathrm{i}\sqrt2}
\end{align}
are the analogues of position and momentum and satisfy $[X,P]=\mathrm{i}$, with $\hbar=1$.
Since they do not commute, they do not allow for simultaneous measurements.
Heterodyne measurement, which is a standard measurement technique in quantum optics, instead mixes the signal with an auxiliary vacuum mode $b$ on a half beam splitter $U_{\mathrm{BS}}$ and measures one quadrature in each output mode, at the cost of the vacuum noise contributed by $b$.
We use the convention $A=(a+b)/\sqrt2$ and $B=(a-b)/\sqrt2$ for the output annihilation operators, so that the observables $X_A$ and $P_B$ commute, and their outcomes $x,p\in\RR$ are combined into the complex outcome $\alpha=x+\mathrm{i}p$.
For $\alpha\in\CC$, we define the normalized coherent state by
\begin{align}
    \ket{\alpha}\coloneqq
    e^{-|\alpha|^2/2}\sum_{n=0}^{\infty}
    \frac{\alpha^n}{\sqrt{n!}}\ket n,
\end{align}
where $\ket n$ denotes the occupation-number state.

\begin{lemma}[Heterodyne measurement amplitude~\cite{lvovsky2009continuous}]
\label{lem:hayashi-heterodyne-amplitude}
    For $\alpha=x+\mathrm{i}p$, with the input state $\ket{\chi}$, the measurement functional $K_{x,p}\ket{\chi}\coloneqq (\bra{x}_A\bra{p}_B) U_{\mathrm{BS}}(\ket{\chi}\ket0)$ satisfies
    $K_{x,p}=\pi^{-1/2}\bra{\alpha}$.
    Hence the single-mode heterodyne POVM is $\pi^{-1}\ketbra{\alpha}\,\mathrm{d}^2\alpha$, where $\mathrm{d}^2\alpha=\mathrm{d}x\,\mathrm{d}p$.
\end{lemma}

\begin{proof}
    In the input representation, $X_A-\mathrm{i}P_B=a^\dagger+b$.
    Since $\ket{x}_A\ket{p}_B$ is an eigenvector of $X_A-\mathrm{i}P_B$ with eigenvalue $x-\mathrm{i}p=\overline{\alpha}$ and $b\ket0=0$, we obtain $K_{x,p}a^\dagger=\overline{\alpha}K_{x,p}$ and therefore
    \begin{align}
        K_{x,p}\ket{n+1}=\frac{\overline{\alpha}}{\sqrt{n+1}}K_{x,p}\ket n.
    \end{align}
    Writing the vacuum quadrature wavefunctions as $\braket{x}{0}=\pi^{-1/4}e^{-x^2/2}$ and $\braket{p}{0}=\pi^{-1/4}e^{-p^2/2}$ gives $K_{x,p}\ket0=\pi^{-1/2}e^{-|\alpha|^2/2}$.
    Solving the recursion yields
    \begin{align}
        K_{x,p}\ket n
        =\frac{e^{-|\alpha|^2/2}(\overline{\alpha})^n}{\sqrt{\pi n!}}
        =\pi^{-1/2}\braket{\alpha}{n}.
    \end{align}
    Consequently, $K_{x,p}^\dagger K_{x,p}=\pi^{-1}\ketbra{\alpha}$, which gives the claimed POVM.
\end{proof}

\paragraph{Coherent state occupying a common mode.}
Consider $q$ bosonic modes with annihilation operators $a_1,\ldots,a_q$ satisfying $[a_i,a_j^\dagger]=\delta_{ij}$ and $[a_i,a_j]=0$, and let $\ket{\mathrm{vac}}$ be a state with vacuum on all modes.
A unit vector $u=(u_1,\ldots,u_q)\in\CC^q$ defines a collective mode by
\begin{align}
\label{eq:hayashi-common-mode}
    b_u^\dagger\coloneqq\sum_{i=1}^q u_i a_i^\dagger,\qquad
    b_u\coloneqq\sum_{i=1}^q \overline{u_i}a_i,\qquad
    [b_u,b_u^\dagger]=1.
\end{align}
The normalized state of $k\geq0$ bosons occupying this mode is
\begin{align}
\label{eq:hayashi-common-mode-number-state}
    \ket{k;u}\coloneqq\frac{(b_u^\dagger)^k}{\sqrt{k!}}\ket{\mathrm{vac}},\qquad
    \braket{k;u}{\ell;u}=\delta_{k\ell}.
\end{align}
All $k$ bosons occupy the same one-particle state $u$, which is a superposition of the original modes. 

For $\alpha\in\CC^q$, the multimode coherent state is
\begin{align}
\label{eq:hayashi-coherent-state}
    \ket{\alpha}\coloneqq\bigotimes_{i=1}^q\ket{\alpha_i}
    =e^{-\norm{\alpha}_2^2/2}\exp\!\left(\sum_{i=1}^q\alpha_i a_i^\dagger\right)\ket{\mathrm{vac}}.
\end{align}
Applying Lemma~\ref{lem:hayashi-heterodyne-amplitude} to each mode gives the $q$-mode heterodyne POVM $\pi^{-q}\ketbra{\alpha}\,\mathrm{d}^{2q}\alpha$.
For $\alpha\ne0$, write $s=\norm{\alpha}_2^2$ and $u=\alpha/\sqrt{s}$ as in Theorem~\ref{thm:hayashi-exact-heterodyne}.
Then $\sum_i\alpha_i a_i^\dagger=\sqrt{s}\,b_u^\dagger$, and expanding the exponential gives
\begin{align}
\label{eq:hayashi-common-mode-coherent-expansion}
    \ket{\alpha}
    =e^{-s/2}\sum_{k=0}^{\infty}\frac{s^{k/2}}{\sqrt{k!}}\ket{k;u}.
\end{align}
Let $\Pi_k$ denote the projector onto the sector of total occupation $k$.
Then
\begin{align}
    \label{eq:hayashi-number-projection}
    \Pi_k\ket{\alpha}=e^{-s/2}\frac{s^{k/2}}{\sqrt{k!}}\ket{k;u},\qquad
    \norm{\Pi_k\ket{\alpha}}_2^2=e^{-s}\frac{s^k}{k!}.
\end{align}
This separation of the radius-dependent amplitude from the mode direction is what yields Eq.~\eqref{eq:hayashi-exact-povm-factorization} after restriction to $k=t$.

\paragraph{Mapping symmetric qudit states to bosonic occupation states.}
\label{paragraph:hayashi-occupation-encoding}

For $\bm n=(n_1,\ldots,n_q)\in\ZZ_{\geq0}^q$ with $|\bm n|\coloneqq\sum_i n_i=t$, write $\bm n!\coloneqq\prod_i n_i!$ and $u^{\bm n}\coloneqq\prod_i u_i^{n_i}$.
Define
\begin{align}
    \mathcal W_{\bm n}\coloneqq
    \left\{(i_1,\ldots,i_t)\in\{1,\ldots,q\}^t:
    \#\{r:i_r=j\}=n_j\ \text{for every }j\right\}.
\end{align}
The set $\mathcal W_{\bm n}$ consists of all distinct permutations of a string containing $n_j$ copies of $j$.
Thus $|\mathcal W_{\bm n}|=t!/\bm n!$, and the normalized Dicke state with occupation vector $\bm n$ is
\begin{align}
    \ket{D^{(t)}_{\bm n}}
    =\sqrt{\frac{\bm{n}!}{t!}}
    \sum_{(i_1,\ldots,i_t)\in\mathcal W_{\bm n}}
    \ket{i_1\cdots i_t}.
\end{align}
These Dicke states form an orthonormal basis of $\operatorname{Sym}^t(\CC^q)$.
Let $\mcF_q^{(t)}\coloneqq\operatorname{span}\{\ket{n_1,\ldots,n_q}:|\bm n|=t\}$ be the total-occupation-$t$ sector of the $q$-mode Fock space, whose occupation basis is also orthonormal.
The occupation encoding $B_{t,q}$ is the map
\begin{align}
    \label{eq:hayashi-occupation-encoding}
    B_{t,q}:\operatorname{Sym}^t(\CC^q)\longrightarrow\mcF_q^{(t)},\qquad
    B_{t,q}\ket{D^{(t)}_{\bm n}}=\ket{n_1,\ldots,n_q},
\end{align}
which maps an orthonormal basis onto another orthonormal basis.

For a unit vector $u\in\CC^q$, expand $\ket{u}^{\otimes t}$ in the computational basis and group the terms with the same occupation vector $\bm n$. Each of the $t!/\bm n!$ terms in such a group has amplitude $u^{\bm n}$, so the definition of the normalized Dicke states gives
\begin{align}
    \ket{u}^{\otimes t}
    =\sum_{|\bm n|=t}\sqrt{\frac{t!}{\bm n!}}u^{\bm n}\ket{D^{(t)}_{\bm n}}.
\end{align}
Applying $B_{t,q}$ term by term using Eq.~\eqref{eq:hayashi-occupation-encoding}, and then using the multinomial expansion of the creation operators, yields
\begin{align}
    \label{eq:hayashi-bosonic-product-state}
    B_{t,q}\ket{u}^{\otimes t}
    =\sum_{|\bm n|=t}\sqrt{\frac{t!}{\bm n!}}u^{\bm n}\ket{\bm n}
    =\frac{(\sum_i u_i a_i^\dagger)^t}{\sqrt{t!}}\ket{\mathrm{vac}}
    =\ket{t;u},
\end{align}
which links to the fixed-occupation components of coherent states in Eq.~\eqref{eq:hayashi-number-projection}.

\paragraph{Proof of the exact realization.}
\begin{proof}[Proof of Theorem~\ref{thm:hayashi-exact-heterodyne}]
    Equations~\eqref{eq:hayashi-number-projection} and~\eqref{eq:hayashi-bosonic-product-state} give the first identity below, and polar coordinates in $\CC^q\simeq\RR^{2q}$ give the second:
    \begin{align}
        \label{eq:hayashi-coherent-restriction}
        B_{t,q}^{\dagger}\Pi_t\ket{\alpha}
        =e^{-s/2}\frac{s^{t/2}}{\sqrt{t!}}\ket{u}^{\otimes t},\qquad
        \mathrm{d}^{2q}\alpha=\frac{\pi^q}{\Gamma(q)}s^{q-1}\mathrm{d}s\,\dd u,
    \end{align}
    where $\dd u$ is the normalized Haar measure on the unit sphere in $\CC^q \simeq \RR^{2q}$.
    The unit sphere in $\RR^{2q}$ has surface area $2\pi^q/\Gamma(q)$, and the substitution $s=r^2$ for $r=\norm{\alpha}_2$ gives $r^{2q-1}\mathrm{d}r=\tfrac12s^{q-1}\mathrm{d}s$.
    Implementing the multimode heterodyne POVM through $B_{t,q}$ therefore yields
    \begin{align}
        \label{eq:hayashi-radial-factorization}
        E_{t,q}(\mathrm{d}\alpha)
        &=\pi^{-q}B_{t,q}^{\dagger}\Pi_t\ketbra{\alpha}\Pi_tB_{t,q}\,\mathrm{d}^{2q}\alpha\notag\\
        &=\frac{e^{-s}s^{t+q-1}}{t!\,\Gamma(q)} (\ketbra{u}{u})^{\otimes t}\,\mathrm{d}s\,\dd u\notag\\
        &=\frac{e^{-s}s^{t+q-1}}{\Gamma(t+q)}\,\mathrm{d}s\,H_{t,q}(\mathrm{d}u),
    \end{align}
    where the last step uses $D_{t,q}=\binom{t+q-1}{t}=\Gamma(t+q)/(t!\,\Gamma(q))$ and the definition of $H_{t,q}$ in Eq.~\eqref{eq:prelim-hayashi-povm}.
    The prefactor is the Gamma density $g_{t+q}(s)$, which integrates to one, so discarding $s$ from the result leaves the Hayashi POVM.
\end{proof}

It is noted that the projector $\Pi_t$ in Eq.~\eqref{eq:hayashi-radial-factorization} simply expresses that the encoded input lies in $\mcF_q^{(t)}$, which involves no postselection.
Because Eq.~\eqref{eq:hayashi-radial-factorization} is an operator identity, the equivalence also holds for the joint distribution of the classical outcome and the retained reference system when the input is entangled with a reference.

\subsection{Efficient implementation on a finite-qubit system}
\label{subsec:hayashi-finite-guarantee}

Although we have shown that the heterodyne measurement can implement the Hayashi measurement exactly, realizing this on hardware is not easy since we have to combine a qubit system and an optical system to perform heterodyne measurement on the optical system.
Therefore, we now approximate the continuous-variable measurement presented in the previous section solely by a qubit circuit with finite-dimensional classical output.
The discretization is the mapping given by
\begin{align}
    \ket{\chi} = \sum_{n=0}^{t} c_n \ket{n}
    \longmapsto
    \ket{\chi}_{\mathrm{encoded}} = \sum_{n=0}^{t} c_n \ket{\mathrm{bin}(n)},
\end{align}
and the quadrature coordinates in a discretized grid representation are
\begin{align}
    x_j = jh,
    \qquad
    p_k = kh,
\end{align}
where $h$ is a grid spacing.
Since the occupation number of each mode is at most $t$, the sector $\mcF_q^{(t)}$ can also be represented as a subspace of $(\CC^{t+1})^{\otimes q}$ by identifying each Fock state $\ket{n_i}$ with an occupation basis state.
Although the exact optical realization uses the Fock-space interpretation, the finite-qubit approximate implementation uses these factors as occupation registers, each of dimension $t+1$ and realized by $\lceil\log_2(t+1)\rceil$ qubits.

\begin{theorem}[Efficient approximate Hayashi measurement on a qubit system]
    \label{thm:efficient-hayashi-measurement}
    For every $0<\tau,\zeta<1/2$, there exists a qubit circuit which, on any density operator $\rho$ supported on $\operatorname{Sym}^t(\CC^q)$, outputs a discretized unit vector $\widehat u\in\CC^q$ approximating the outcome of the Hayashi measurement.

    More precisely, the ideal Hayashi outcome $u$ and the circuit output $\widehat u$ admit a coupling, with their respective output distributions as marginals, satisfying the following bound.
    \begin{align}
        \Prob\!\left[
        \norm{\widehat u-u}_2>\tau
        \right]\leq\zeta.
    \end{align}
    The gate complexity over a fixed finite universal gate set is given by
    \begin{align}
        G_H(t,q;\tau,\zeta) = O((t+q) \polylog (t, q, \tau^{-1}, \zeta^{-1})).
    \end{align}
\end{theorem}

\paragraph{Protocol.}
\label{subsec:hayashi-finite-circuit}
The finite-qubit implementation of the Hayashi measurement is summarized in Fig.~\ref{fig:hayashi-finite-qubit-protocol}, which follows the three steps below.

\begin{figure}
\centering
\includegraphics{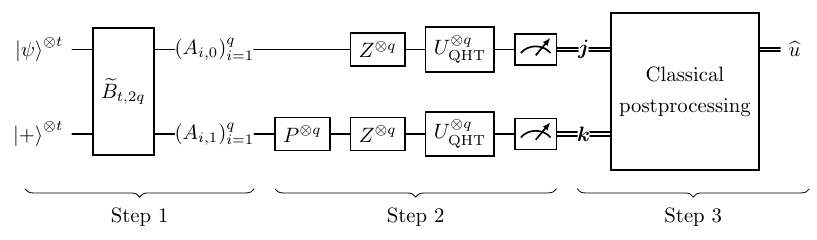}
\caption{
    The efficient implementation of the Hayashi measurement on an input state $\ket{\psi}^{\otimes t} \in \operatorname{Sym}^t(\CC^q)$.
    Step 1 implements the occupation encoding, which realizes the beam-splitter operation.
    Step 2 performs quadrature measurements on each occupation register to obtain the heterodyne measurement outcomes.
    Each $Z$ gate acts on the least significant occupation bit to match the Hermite-function convention in Eq.~\eqref{eq:hayashi-hermite-function}.
    By postprocessing the measurement outcomes, Step 3 outputs a unit vector $\widehat u$ that approximates the Hayashi measurement outcome $u$.
}
\label{fig:hayashi-finite-qubit-protocol}
\end{figure}

\begin{enumerate}
    \item \emph{Occupation encoding realizing the beam-splitter operation}.
    Append a qubit in the state $\ket{+} \coloneqq \frac{1}{\sqrt{2}}(\ket{0} + \ket{1})$ to each input qudit.
    For each resulting system, use the product basis $\ket{i,b}\coloneqq\ket i\otimes\ket b$ of $\CC^q\otimes\CC^2$, where $i=1,\ldots,q$ and $b\in\{0,1\}$, and apply the approximate occupation encoding $\widetilde B_{t,2q}$ of Prop.~\ref{prop:hayashi-occupation-encoder}.
    This produces $2q$ occupation registers $A_{i,b}$, grouped into the $q$ ordered pairs $(A_{i,0},A_{i,1})$, with total error $\varepsilon_B$.

    \item \emph{Quadrature measurements}.
    For every $i=1,\ldots,q$, apply $P\coloneqq \sum_n (-\mathrm{i})^n \ketbra{n}$ to the second register $A_{i,1}$ as preparation for the momentum measurement. Apply a position-type QHT to every register $A_{i,b}$ using the Hermite functions in Eq.~\eqref{eq:hayashi-hermite-function}, and measure the resulting grid registers in the computational basis. Denote the outcomes from $A_{i,0}$ and $A_{i,1}$ by $j_i,k_i$, respectively.
    The parameters of the QHT are chosen later.

    \item \emph{Classical postprocessing}. Obtain $z_i=j_i+\mathrm{i}k_i$ from the measurement results and compute
    \begin{align}
        \label{eq:hayashi-finite-output}
        \widehat u=\frac{z}{\sqrt{\sum_{i=1}^{q}(j_i^2+k_i^2)}}.
    \end{align}
    When $z=0$, we assign $\widehat u$ to an arbitrary unit vector.
\end{enumerate}

\paragraph{Proof strategy.}
The ideal versions implement occupation-basis encoding and heterodyne detection; the final normalization then returns the Hayashi measurement result by Theorem~\ref{thm:hayashi-exact-heterodyne}.
We compare the qubit-circuit implementation with the ideal measurement after binning its outcomes by evaluating the distance between their probability distributions.
After establishing error bounds for these three processes, we choose the circuit parameters and evaluate the gate complexity.
For each process, we describe its approximate implementation together with the corresponding error analysis.

\paragraph{Step 1: Occupation encoding and beam-splitter operation.}
The first step converts the symmetric qudit input into occupation registers and realizes the beam-splitter operation required for the heterodyne measurement.
The following lemma identifies this combined digital operation with the corresponding ideal optical operation.

\begin{lemma}[Combined occupation encoding and beam splitters]
    \label{lem:hayashi-combined-encoding}
    Let $V_{\mathrm{vac}}$ add one vacuum to each of the $q$ modes, and let $U_{\mathrm{BS}}^{\otimes q}$ act on the resulting pairs.
    On $\operatorname{Sym}^t(\CC^q)$,
    $B_{t,2q}V_+^{\otimes t}=U_{\mathrm{BS}}^{\otimes q}V_{\mathrm{vac}}B_{t,q}$ holds, where $V_+(\cdot)\coloneqq (\cdot) \otimes \ket{+}$.
    In particular,
    \begin{align}
        B_{t,2q}V_+^{\otimes t}\ket{D^{(t)}_{\bm n}}
        =\sum_{\substack{0\leq m_i\leq n_i\\1\leq i\leq q}}
        \sqrt{2^{-t}\prod_{i=1}^{q}\binom{n_i}{m_i}}
        \bigotimes_{i=1}^{q}\ket{m_i,n_i-m_i}.
        \label{eq:hayashi-combined-encoding}
    \end{align}
\end{lemma}

\begin{proof}
    Each basis state after $V_+^{\otimes t}$ has amplitude
    $\sqrt{\bm n!/(2^t t!)}$.
    There are $t!/\prod_i[m_i!(n_i-m_i)!]$ basis states with occupation $(m_i,n_i-m_i)_i$.
    Their normalized Dicke component therefore has the coefficient in Eq.~\eqref{eq:hayashi-combined-encoding}.
    Indeed, multiplying the common basis amplitude by the square root of the number of basis states in the corresponding refined type class gives
    \begin{align}
        \sqrt{\frac{\bm n!}{2^t t!}}
        \sqrt{\frac{t!}{\prod_i m_i!(n_i-m_i)!}}
        =
        \sqrt{2^{-t}\prod_{i=1}^q \binom{n_i}{m_i}}.
    \end{align}
    Applying \(B_{t,2q}\) maps each normalized Dicke state of occupation \((m_i,n_i-m_i)_i\) to the corresponding occupation-number basis state \(\bigotimes_i \ket{m_i,n_i-m_i}\), which proves Eq.~\eqref{eq:hayashi-combined-encoding}.
    For the optical side, place the vacuum in the second input port and use the convention $U_{\mathrm{BS}}a^\dagger U_{\mathrm{BS}}^\dagger=(A^\dagger+B^\dagger)/\sqrt2$ and $U_{\mathrm{BS}}\ket{0,0}=\ket{0,0}$ to obtain
    \begin{align}
        U_{\mathrm{BS}}\ket{n,0}
        =\frac{(A^\dagger+B^\dagger)^n}{\sqrt{2^n n!}}\ket{0,0}
        =\sum_{m=0}^{n}\sqrt{2^{-n}\binom{n}{m}}\ket{m,n-m}.
    \end{align}
    Taking the tensor product over the $q$ pairs gives the right-hand side of Eq.~\eqref{eq:hayashi-combined-encoding}, proving the claimed identity.
\end{proof}

We provide the gate complexity of this operation below.

\begin{proposition}[Coherent occupation encoder]
    \label{prop:hayashi-occupation-encoder}
    For $0<\epsilon<1/2$ there is a circuit $\widetilde B_{t,q}$ over a fixed universal gate set with
    \begin{align}
        \norm{\widetilde B_{t,q}-B_{t,q}}_{\mathrm{op},\operatorname{Sym}^t(\CC^q)}\leq\epsilon,\qquad
        G_B(t,q;\epsilon)=O\!\left((t+q)\polylog\frac{t+q+2}{\epsilon}\right),
        \label{eq:hayashi-occupation-cost}
    \end{align}
    where both maps act on the input, the occupation registers, and clean work registers, and the ideal map returns every register other than the occupation output to a fixed zero state.
\end{proposition}
\begin{proof}
    Let $\ell(\bm n)=1^{n_1}\cdots q^{n_q}$ be the sorted word with occupation $\bm n$.
    The coherent symmetrization circuit of Ref.~\cite{liu2025symmetrization} implements $S:\ket{\ell(\bm n),0}\mapsto\ket{D^{(t)}_{\bm n},0}$ on superpositions of sorted words; it consists of reversible sorting and merging steps on $O(\log(t+2))$-bit records and of the preparation of $O(t)$ interval-uniform registers, and synthesizing each preparation to error $O(\epsilon/t)$ gives an approximation $\widetilde S$ with $\norm{(\widetilde S-S)P_{\mathrm{sort}}}_{\mathrm{op}}\leq\epsilon$ on the sorted-word subspace and $O(t\polylog((tq+2)/\epsilon))$ gates.
    A sorted word and its occupation vector determine each other, and the conversion $R:\ket{\ell(\bm n),0}\mapsto\ket{0,\bm n}$ is implemented exactly by a fixed sorting network on $t+q$ records of $O(\log(t+q+2))$ bits~\cite{liu2025symmetrization}.
    By using Batcher's odd-even mergesort network~\cite{batcher1968sorting}, the conversion uses $O((t+q)\log^3(t+q+2))$ gates.
    Thus, the total gate complexity of $\widetilde B_{t,q} = R \widetilde{S}^{\dagger}$ is given by $O((t+q)\polylog((t+q+2)/\epsilon))$.
\end{proof}

\paragraph{Step 2: Quadrature measurements via the quantum Hermite transform.}
To implement the quadrature measurements with QHTs, we use $\braket{x}{n}=\varphi_n(x)$ and $\braket{p}{n}=(-\mathrm{i})^n\varphi_n(p)$~\cite{lvovsky2009continuous}, where $\varphi_n$ is defined by
\begin{align}
    \varphi_n(x)\coloneqq\frac{H_n(x)e^{-x^2/2}}{\pi^{1/4}\sqrt{2^n n!}},
    \label{eq:hayashi-hermite-function}
\end{align}
using the Hermite polynomials $H_n$.
A position-type QHT approximates the map from occupation-basis amplitudes to sampled position wavefunctions $\varphi_n$ in the grid basis, and $P$ followed by a position-type QHT approximates the map to sampled momentum wavefunctions $(-\mathrm{i})^n\varphi_n$.
Thus, the QHTs followed by computational-basis measurements implement the heterodyne measurement on the occupation registers.

To describe the above procedure more precisely, we introduce the occupation register $\mcE_t\coloneqq\operatorname{span}\{\ket0,\ldots,\ket t\}$, and define the ideal position representation $V_t\ket n\coloneqq\varphi_n$.
For a power-of-two grid size $M\geq2$, let $\mcG_M\coloneqq\{-M/2,\ldots,M/2-1\}$ denote the centered signed grid indices.
We take a grid of size $M$ with spacing $h\coloneqq\sqrt{2\pi/M}$ and cells $I_j\coloneqq[jh-h/2,jh+h/2)$, $j\in\mcG_M$, and define the sampling map
\begin{align}
    S_Mf\coloneqq\sqrt h\sum_{j\in\mcG_M}f(jh)\ket j
\end{align}
and the isometric embedding of grid vectors into the space of piecewise constant functions
\begin{align}
    J_M\ket j\coloneqq h^{-1/2}\mathbf 1_{I_j}
\end{align}
so that $S_M J_M = \1$.
For $0<\eta<1/2$, choose an integer $N\geq t+1$ with $N>\log(1/\eta)$ as the cutoff of the QHT and pad the occupation registers with zeros to accommodate the levels $0,\ldots,N-1$.
To match Eq.~\eqref{eq:hayashi-hermite-function}, we precede the QHT of Ref.~\cite{jain2025hermite}, whose Hermite functions include a factor $(-1)^n$, by the parity phase $\sum_{n=0}^{N-1}(-1)^n\ketbra{n}$, implemented by a $Z$ gate on the least significant occupation bit.
Before fixed-gate synthesis, the QHT~\cite{jain2025hermite} provides an isometry $V_{\mathrm{QHT}}$ satisfying
\begin{align}
    \norm{V_{\mathrm{QHT}}-(S_MV_t)\otimes\ket0_{\mathrm{work}}}_{\mathrm{op},\mcE_t}\leq\eta.
\label{eq:hayashi-qht-sampling-guarantee}
\end{align}
The gate complexity is $G_Q(N;\eta)=O((\log N+\log(1/\eta))^3\log(1/\eta))$, with a grid size $M\geq N$ that can be chosen to be polynomially bounded as $M=\poly(N,1/\eta)\geq N$.

\begin{lemma}[Grid comparison]
\label{lem:hayashi-grid-comparison}
For $f\in\operatorname{span}\{\varphi_0,\ldots,\varphi_t\}$, possibly with Hilbert-space-valued coefficients,
$\norm{J_MS_Mf-f}_2\leq h\sqrt{2t+1}\,\norm{f}_2$.
\end{lemma}

\begin{proof}
    The harmonic oscillator wave function $\varphi_n(x)$ satisfies the eigenvalue equation
    \begin{align}
        \left(-\frac{d^2}{dx^2}+x^2\right) \varphi_n(x)
        = (2n + 1) \varphi_n(x).
    \end{align}
    Thus, write
    \begin{align}
        f(x)=\sum_{n=0}^{t}c_n\varphi_n(x)\ket{r_n},
    \end{align}
    where $c_n\in\CC$ and each $\ket{r_n}$ is a unit vector in an auxiliary Hilbert space $\mcR$; when $c_n=0$, the choice of $\ket{r_n}$ is arbitrary.  The harmonic-oscillator operator acts trivially on $\mcR$, and hence
    \begin{align}
        \norm{f'}_2^2 + \norm{x f}_2^2
        &= \left\langle f, \qty(-\frac{d^2}{dx^2}+x^2)\otimes\1_{\mcR} f \right\rangle\\
        &= \sum_{n=0}^{t} (2n + 1) |c_n|^2\\
        &\leq (2t+1) \sum_{n=0}^{t} |c_n|^2\\
        &= (2t+1) \norm{f}_2^2,
    \end{align}
    where $\langle f,g\rangle=\int\langle f(x),g(x)\rangle\,\mathrm{d}x$.  We used integration by parts in the first line and the orthogonality of the Hermite functions in the second and final lines.
    For $x\in I_j$, since $f(x)-f(jh)=\int_{jh}^{x}f'(y)\,\mathrm{d}y$ holds, the Cauchy--Schwarz inequality gives
    \begin{align}
    \norm{f(x)-f(jh)}^2 \leq \abs{x-jh}\abs{\int_{jh}^{x}\norm{f'(y)}^2\,\mathrm{d}y} \leq h \int_{I_j}\norm{f'(y)}^2\,\mathrm{d}y,
    \end{align}
    i.e.,
    \begin{align}
    \int_{I_j}\norm{f(x)-f(jh)}^2\,\mathrm{d}x \leq h^2 \int_{I_j}\norm{f'(x)}^2\,\mathrm{d}x.
    \end{align}
    Since $\bigcup_j I_j$ contains $[-R,R]$ with $R=(M-1)h/2$, we have
    \begin{align}
        \norm{J_MS_Mf-f}_2^2
        &\leq \sum_{j} \int_{I_j} \norm{f(x)-f(jh)}^2 \dd x + \int_{|x|>R} \norm{f(x)}^2 \dd x\\
        &\leq h^2 \sum_j \int_{I_j} \norm{f'(x)}^2 \dd x + R^{-2} \|xf\|_2^2\\
        &\leq h^2 (\norm{f'}^2_2 + \norm{xf}^2_2)\\
        &\leq h^2 (2t+1) \norm{f}_2^2,
    \end{align}
    where we use
    \begin{align}
        R^{-1} = {M \over M-1} {h \over \pi} \leq h
    \end{align}
    using $h = \sqrt{2\pi/M}$ and $M\geq 2$.
\end{proof}

Let $\eta<\delta<1/2$ and let $\widetilde V$ be a fixed-gate implementation of $V_{\mathrm{QHT}}$ with $\norm{\widetilde V-V_{\mathrm{QHT}}}_{\mathrm{op},\mcE_t}\leq\delta-\eta$, so that $\delta$ bounds the total QHT implementation error.
Combining Eq.~\eqref{eq:hayashi-qht-sampling-guarantee} with Lemma~\ref{lem:hayashi-grid-comparison} and the triangle inequality, we find that the implemented and ideal position representations of one register satisfy
\begin{align}
    \label{eq:hayashi-isometry-error}
    \norm{(J_M\otimes\1_{\mathrm{work}})\widetilde V-V_t\otimes\ket0_{\mathrm{work}}}_{\mathrm{op},\mcE_t}\leq\delta+h\sqrt{2t+1}.
\end{align}
Thus, we can bound the diamond distance error of the implemented measurement from the ideal binned heterodyne measurement as follows.

\begin{lemma}[Error before normalization]
    \label{lem:hayashi-binned-channel}
    Let $\mcM_{\mathrm{het,bin}}$ be the heterodyne measurement of Theorem~\ref{thm:hayashi-exact-heterodyne} with each of the $2q$ real quadrature outcomes binned into the cells $I_j$ and all outcomes with an overflowing coordinate merged into one label $\perp$, and let $\widetilde{\mcM}$ be the measurement implemented by steps~1--2.
    Then
    \begin{align}
        \label{eq:hayashi-total-error-budget}
        \frac12\norm{\widetilde{\mcM}-\mcM_{\mathrm{het,bin}}}_\diamond\leq\varepsilon_B+2q\bigl(\delta+h\sqrt{2t+1}\bigr)\eqqcolon E.
    \end{align}
\end{lemma}

\begin{proof}
    Replacing the implemented preprocessing by the ideal $B_{t,2q}V_+^{\otimes t}$ costs at most $\varepsilon_B$.
    Afterwards the state is supported on $\mcE_t^{\otimes2q}$, because every occupation is at most $t$.
    The phase gates $P$ applied to the registers $A_{i,1}$ are diagonal in the occupation basis and hence preserve this subspace.
    On this subspace, the $2q$ local QHTs can be exchanged one at a time for $V_t$ using Eq.~\eqref{eq:hayashi-isometry-error}, since for purified inputs the output vectors change in norm by at most $\delta+h\sqrt{2t+1}$ per register and the trace distance of pure states is at most the norm distance of their vectors.
    Measuring the grid register after $\widetilde V$ is equivalent to measuring position and binning after its embedding $J_M\widetilde V$.
    Applying this same binned measurement to the ideal representation $V_t$, and discarding the work registers, cannot increase the trace distance.
    Finally, Lemma~\ref{lem:hayashi-heterodyne-amplitude} identifies the ideal binned measurement with $\mcM_{\mathrm{het,bin}}$.
\end{proof}

\paragraph{Step 3: Classical normalization of finite output.}
The measured grid indices indicate the complex vector $z=(j_i+\mathrm{i}k_i)_{i=1}^q$.
Since computing $\hat{u}$ in Eq.~\eqref{eq:hayashi-finite-output} involves normalizing $z$, a small value of $\lVert z\rVert$ can amplify even a small error in $z$ into a considerable error in its direction.
We bound the probability of a small radius to control the error introduced by normalization.

\begin{lemma}[From bins to directions]
    \label{lem:hayashi-radial-coupling}
    For every input on $\operatorname{Sym}^t(\CC^q)$ and every $0<s_0\leq1$ with $h\sqrt{2q}<2\sqrt{s_0}$, the ideal Hayashi outcome $u$ and the protocol output $\widehat u$ admit a coupling satisfying the following bound.
    \begin{align}
        \label{eq:hayashi-strassen-neighborhood}
        \Prob\!\left[\norm{\widehat u-u}_2>\frac{h\sqrt{2q}}{\sqrt{s_0}}\right]\leq E+s_0.
    \end{align}
\end{lemma}

\begin{proof}
    Due to Lem.~\ref{lem:hayashi-binned-channel}, we can choose a joint distribution of the measurement outcomes of $\widetilde{\mcM}$ and $\mcM_{\mathrm{het,bin}}$ so that they differ with probability at most $E$.
    When the labels agree and $\norm{\alpha}_2^2\geq s_0$ holds, the label is not $\perp$ and the bin center $c(\alpha)$ satisfies $\norm{c(\alpha)-\alpha}_2\leq h\sqrt{2q}/2<\sqrt{s_0}\leq\norm{\alpha}_2$.
    Thus, we have $c(\alpha)\neq0$ and
    \begin{align}
        \label{eq:hayashi-normalization-stability}
        \left\|\frac{c(\alpha)}{\norm{c(\alpha)}_2}-\frac{\alpha}{\norm{\alpha}_2}\right\|_2\leq\frac{2\norm{c(\alpha)-\alpha}_2}{\norm{\alpha}_2}\leq\frac{h\sqrt{2q}}{\sqrt{s_0}},
    \end{align}
    where $u=\alpha/\norm{\alpha}_2$ by Thm.~\ref{thm:hayashi-exact-heterodyne} and $\widehat u=c(\alpha)/\norm{c(\alpha)}_2$ by step~3.
    By Thm.~\ref{thm:hayashi-exact-heterodyne}, $\|\alpha\|_2^2 \sim\operatorname{Gamma}(t+q,1)$ for every input, and thus we have
    \begin{align}
        \Prob\!\left[\|\alpha\|_2^2<s_0\right]\leq\int_0^{s_0}\frac{s^{t+q-1}}{\Gamma(t+q)}\,\mathrm ds=\frac{s_0^{t+q}}{\Gamma(t+q+1)}\leq s_0.
        \label{eq:hayashi-gamma-cdf}
    \end{align}
    A union bound yields
    \begin{align}
        \Prob\!\left[\norm{\widehat u-u}_2>\frac{h\sqrt{2q}}{\sqrt{s_0}}\right]
        &\leq\Prob[\|\alpha\|_2^2\geq s_0 \land \text{labels differ}]+\Prob[\|\alpha\|_2^2<s_0]\notag
        \leq E+s_0,
    \end{align}
    which completes the proof.
\end{proof}

We combine the above lemmas to prove Theorem~\ref{thm:efficient-hayashi-measurement}.

\begin{proof}[Proof of Theorem~\ref{thm:efficient-hayashi-measurement}]
    As shown in Thm.~\ref{thm:hayashi-exact-heterodyne}, the Hayashi measurement is realized by the heterodyne measurement on the $q$-mode bosonic system, and the implemented circuit approximates the heterodyne measurement by the three-step protocol.
    We conclude the proof by choosing the parameters of the protocol to satisfy the error bound in the theorem statement.
    We choose the parameters as
    \begin{align}
        \varepsilon_B=\frac{\zeta}{4},\qquad
        \delta=\frac{\zeta}{8q},\qquad s_0=\frac{\zeta}{4},\qquad
        h\leq\min\left\{\frac{\zeta}{8q\sqrt{2t+1}},
        \frac{\tau\sqrt\zeta}{2\sqrt{2q}}\right\}.
    \end{align}
    Set $\eta=\delta/2$, reserve the remaining $\delta/2$ for fixed-gate synthesis, and choose
    \begin{align}
        N=\left\lceil\max\left\{t+1,\ 1+\log(2/\delta),
        \frac{128\pi q^2(2t+1)}{\zeta^2},
        \frac{16\pi q}{\tau^2\zeta}\right\}\right\rceil.
    \end{align}
    The QHT supplies a power-of-two $M=\poly(N,1/\eta)\geq N$, which ensures the required bounds on $h=\sqrt{2\pi/M}$.
    These choices give
    \begin{align}
        E+s_0
        &\leq\frac{\zeta}{4}+2q\frac{\zeta}{8q}
        +2q\frac{\zeta}{8q}+\frac{\zeta}{4}=\zeta,\\
        \frac{h\sqrt{2q}}{\sqrt{s_0}}&\leq\tau.
    \end{align}
    In particular, $h\sqrt{2q}\leq\tau\sqrt{s_0}<2\sqrt{s_0}$, so Lemma~\ref{lem:hayashi-radial-coupling} applies and yields
    \begin{align}
        \Prob\qty[\|\widehat{u}-u\|_2 > \tau] \leq \zeta.
    \end{align}
    The chosen $N$ and the resulting $M$ are polynomially bounded in $t,q,\tau^{-1},\zeta^{-1}$. The occupation encoding and the $2q$ QHTs therefore still require $O((t+q)\polylog(t,q,\tau^{-1},\zeta^{-1}))$ gates.
\end{proof}

\section*{AI disclosure}

Part of this work is supported by GPT 5.6-sol Ultra and GPT-6 Astra Ultra as follows.
The main idea of constructing random purification and dilation of covariant quantum channels is given by human authors.
The initial ideas for learning protocols and construction of the Hayashi measurement are due to human authors, and the details of the proofs are supported by AI.
The proof of lower bounds is mainly generated by AI under instructions from human authors.
The main text is written by human authors with the assistance of AI to improve the clarity and correctness of the text.
Human authors have verified the correctness of the proofs.

\section*{Acknowledgments}
We thank Ryotaro Niwa, Takeru Utsumi, Ryuji Takagi, and Mio Murao for the collaboration on related work.
This work was supported by JSPS KAKENHI Grant No.~26K25550, NWO grant NGF.1623.23.025 (“Qudits in theory and experiment”), and EPSRC grant EP/Z000580/1.
D.G. and P.M.P. would like to thank the Isaac Newton Institute for Mathematical Sciences, Cambridge, for support and hospitality during the programme Mathematics of many-body entanglement, where work on this paper was undertaken.

\bibliographystyle{alphaurl}
\bibliography{main}

\clearpage

\appendix

\section{Quantum circuit of the Schur transform for the wreath product group}
\label{appendix:efficient-wreath-schur}

We present and analyze a quantum circuit for the generalized Schur transform
$U_{\mathrm{Sch}}^{(\Wr_n)}$ introduced in
Section~\ref{sec:random-purification-dilation}, using the single-copy
decomposition~\eqref{eq:schur-G} and the notation from the proof of
Proposition~\ref{prop:double-centralizer}.
The circuit concatenates the basis changes in
Eqs.~\eqref{eq:basis-change-1}, \eqref{eq:basis-change-2}, and
\eqref{eq:generalized-schur-decomposition}.

For the circuit description, we restrict the irrep labels to the finite set
\[
    \Lambda:=\{\lambda\in\irrep G:m_\lambda>0\}.
\]
Accordingly, the occupation vectors $\bm n$ and multipartitions $\bm\alpha$
are indexed by $\Lambda$, with all other components omitted.
Using the spaces $\mcR_{\bm n}$, $\mcS_{\bm\alpha}$,
$\mcG_{\bm\alpha}$, and $\mcU_{\bm\alpha}$ defined in the proof of Proposition~\ref{prop:double-centralizer},
the circuit implements the isomorphism
\begin{align}
    \mcH^{\otimes n}
    \overset{U_{\mathrm{Sch}}^{(\Wr_n)}}{\cong}
    \bigoplus_{\bm\alpha}
    \underbrace{
        \left(
            \CC[\mathsf P_{\bm n}]
            \otimes
            \mcR_{\bm n}
            \otimes
            \mcS_{\bm\alpha}
        \right)
    }_{\mcG_{\bm\alpha}}
    \otimes
    \underbrace{
        \left(
            \bigotimes_{\lambda\in\Lambda}
            \mcU_{\alpha^\lambda}^{(m_\lambda)}
        \right)
    }_{\mcU_{\bm\alpha}}.
    \label{eq:wreath-schur-abstract-isometry}
\end{align}

The circuit has three steps. First, decompose every copy of $\mcH$ into its
$G$-irrep and multiplicity spaces. Second, stably sort the $G$-irrep labels
while carrying the corresponding representation and multiplicity registers
through the same sort; the stable sorting permutation is retained as part
of the output. Third, apply a Schur transform recursively to the
multiplicity registers inside each block of equal labels.
Figure~\ref{fig:wreath-schur-transform} shows the quantum circuit implementing this procedure, whose gate complexity is analyzed in Proposition~\ref{prop:coherent-wreath-schur-complexity}.

\begin{figure}[htbp]
    \centering
    \includegraphics[width=0.9\linewidth]{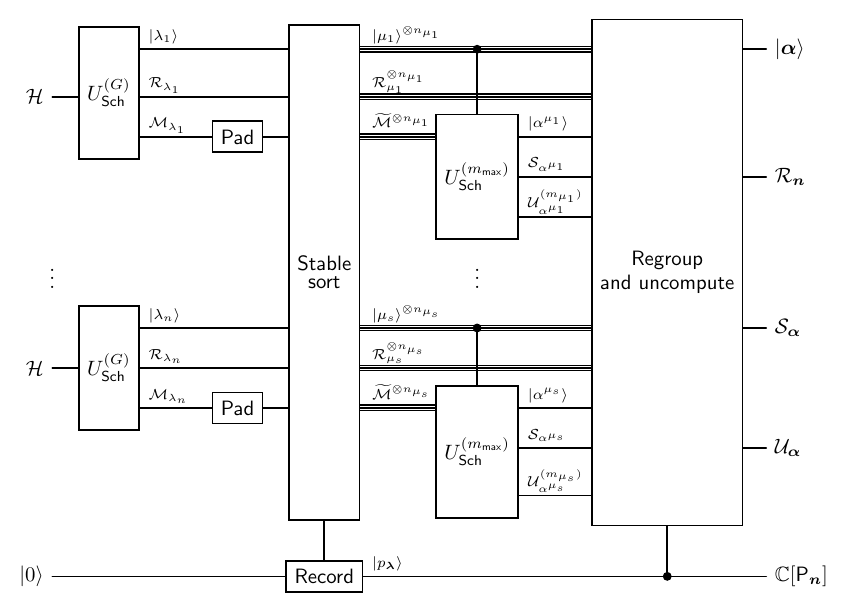}
    \caption{Quantum circuit for the generalized Schur transform $U_{\mathrm{Sch}}^{(\Wr_n)}$.
  The displayed block structure represents one branch of the coherent computation, with block boundaries controlled by the irrep labels.
    Here, $\mu_1,\ldots,\mu_s$ are the distinct labels appearing in $(\lambda_1,\ldots,\lambda_n)$ in sorted order, and $s=\#\left\{\lambda\in\Lambda\;\middle|\;n_\lambda>0\right\}$.
    $\mathsf{Pad}$ implements the inclusion $\iota_\lambda$ into $\widetilde{\mcM}$.
    $\mathsf{Stable\ sort}$ coherently sorts the labels together with their representation and multiplicity registers, recording $p_{\bm\lambda}$.
    The label registers control the recursive Schur transforms: comparison of adjacent sorted labels determines whether to apply the next Clebsch--Gordan step or initialize a new block.
    The final block uncomputes redundant labels and temporary sorting data and regroups the remaining registers into the decomposition in Eq.~\eqref{eq:wreath-schur-abstract-isometry}.}
    \label{fig:wreath-schur-transform}
\end{figure}

\paragraph{Parameters and complexity convention.}

Define
\begin{align}
    m_{\max}
    :=
    \max_{\lambda\in\Lambda}m_\lambda,
    \qquad
    \widetilde{\mcM}
    :=
    \CC^{m_{\max}}.
\end{align}
For every $\lambda\in\Lambda$, fix an isometric inclusion
\begin{align}\label{eq: iota definition}
    \iota_\lambda:
    \mcM_\lambda
    \hookrightarrow
    \widetilde{\mcM}
\end{align}
that maps a chosen basis of $\mcM_\lambda$ to the first $m_\lambda$
basis vectors of $\widetilde{\mcM}$. We write $N_{\mathrm{Sch}}^{(G)}(\delta)$ for the one- and two-qubit gate count of the local transformation
\begin{align}\label{eq: H embedded}
    \mcH
    \longrightarrow
    \bigoplus_{\lambda\in\Lambda}
    \mcR_\lambda\otimes\widetilde{\mcM},
\end{align}
obtained by composing the local $G$-irrep decomposition with the inclusions $\iota_\lambda$, with operator-norm error at most $\delta$.

Let $b_\Lambda$ and $b_R$ be the maximum number of qubits used to encode an irrep label in $\Lambda$ and a basis vector of any $\mcR_\lambda$ respectively.

\begin{proposition}[Coherent wreath-product Schur transform]
\label{prop:coherent-wreath-schur-complexity}
For $n\geq1$ and $0<\varepsilon<1$, the generalized Schur transform
$U_{\mathrm{Sch}}^{(\Wr_n)}$ admits a circuit with operator-norm error
at most $\varepsilon$ and one- and two-qubit gate count
\begin{align}
    N_{\mathrm{Sch}}^{(\Wr_n)}(\varepsilon)
    \leq
    nN_{\mathrm{Sch}}^{(G)}
    \!\left(\frac{\varepsilon}{2n}\right)
    +
    O\!\left(
        n\bigl(b_\Lambda+b_R+m_{\max}^4\bigr)
        \polylog(n,m_{\max},1/\varepsilon)
    \right).
    \label{eq:wreath-schur-gate-bound}
\end{align}
\end{proposition}

\begin{proof}
We describe the circuit as a three-step procedure, and analyze the cost of each step separately.

\begin{enumerate}
   \item \textbf{Local isotypic decomposition.}

Apply 
\begin{align}
    \mcH
    \xrightarrow{\,U_{\mathrm{Sch}}^{(G)}\,}
    \bigoplus_{\lambda\in\Lambda}
    \mcR_\lambda\otimes\mcM_\lambda
    \xhookrightarrow{\;\bigoplus_{\lambda}
        (\1_{\mcR_\lambda}\otimes\iota_\lambda)\;}
    \bigoplus_{\lambda\in\Lambda}
    \mcR_\lambda\otimes\widetilde{\mcM}
\end{align}
independently to all $n$ tensor factors.
After applying this map, the $i$th factor is
represented by an irrep label $\lambda_i$, a basis register for
$\mcR_{\lambda_i}$, and a multiplicity register with
state lying in the embedded copy $\iota_{\lambda_i}(\mcM_{\lambda_i})$.
The $\mcR_{\lambda_i}$ registers use the fixed local irrep bases of Section~\ref{sec:random-purification-dilation}, shared with the $G$-QFT subroutine.

  \item \textbf{Stable sorting of the $G$-irrep labels.}
  
Reversibly and coherently apply a single stable sort, ordered by $\lambda_i$, to all $n$ register tuples of the following form.
\begin{align}
    \bigl(
        \lambda_i,\,\mcR_{\lambda_i},\,\widetilde{\mcM}
    \bigr)
\end{align}
Let $p_{\bm\lambda}\in\mfS_n$ be the resulting stable sorting permutation, with the convention of Eq.~\eqref{eq:def-sorted-tuple}, and retain its chosen encoding as part of the output, as an element in $\CC[\mathsf P_{\bm n}]$. After the sort, all tensor factors with the same label $\lambda$ are contiguous.

  \item \textbf{Schur transforms on the multiplicity spaces.}

For every label $\lambda$ that appears, we apply the Schur transform to
the corresponding embedded multiplicity space,
\begin{align}
    \iota_\lambda(\mcM_\lambda)^{\otimes n_\lambda}
    \longrightarrow
    \bigoplus_{\alpha^\lambda\vdash_{m_\lambda}n_\lambda}
    \mcS_{\alpha^\lambda}
    \otimes
    \mcU_{\alpha^\lambda}^{(m_\lambda)} .
    \label{eq:wreath-schur-block-output}
\end{align}
We use the recursive construction by Bacon, Chuang, and Harrow (BCH)~\cite{bacon2005quantum,bacon2006efficient}, in which the tensor factors are
added one at a time by Clebsch--Gordan transforms; see
Ref.~\cite[Sec.~8]{burchardt2025high}.
The BCH construction gives the Young--Yamanouchi basis on each symmetric-group irrep space $\mcS_{\alpha^\lambda}$. We use the matrix and phase convention fixed in Section~\ref{sec:schur-weyl-basis-conventions}; in particular, the permutation action on this register is the same $g_{\alpha^\lambda}$ used by the little-group QFT.

Proceed through the sorted tensor factors from left to right.  If the current label agrees with the preceding one, apply the next BCH Clebsch--Gordan step, by adding the current multiplicity register to the partial Schur transform of that block. If the labels differ, initialize the Schur transform for a new block instead. The same Clebsch--Gordan circuit can be used for every block. 

For a block with label $\lambda$, every input lies
in the subspace
\begin{align}
    \iota_\lambda(\mcM_\lambda)
    \subseteq
    \widetilde{\mcM}.
\end{align}
With the compatible Gelfand--Tsetlin basis used by the BCH construction,
the Schur transform on $\widetilde{\mcM}$, restricted to this subspace,
agrees with the $m_\lambda$-dimensional Schur transform.  Hence the
output of the block lies in the canonical embedded copy of
\begin{align}
    \bigoplus_{\alpha^\lambda\vdash_{m_\lambda}n_\lambda}
    \mcS_{\alpha^\lambda}
    \otimes
    \mcU_{\alpha^\lambda}^{(m_\lambda)} .
\end{align}

These choices also establish compatibility with the little-group Fourier basis. Step~1 fixes the matrices on $\mcR_{\bm n}$, Step~2 identifies each word with the stable transversal element $t_p=(\bm e,p^{-1})$, and Step~3 fixes the matrices on $\mcS_{\bm\alpha}$. Consequently, the action of $R(\bm g,\sigma)$ on the output registers $\CC[\mathsf P_{\bm n}]\otimes\mcR_{\bm n}\otimes\mcS_{\bm\alpha}$ is given exactly by Eq.~\eqref{eq:wreath-irrep-matrices}.

Across all labels, these Schur transforms produce the ordered multipartition
\[
    \bm\alpha=(\alpha^\lambda)_{\lambda\in\Lambda},
\]
which is the irrep label of the wreath product.  Since
$n_\lambda=|\alpha^\lambda|$, the multipartition $\bm\alpha$ also
determines the occupation vector $\bm n$.  
The remaining registers are
\begin{align}
    \CC[\mathsf P_{\bm n}]
    \otimes
    \mcR_{\bm n}
    \otimes
    \bigotimes_{\lambda\in\Lambda}\mcS_{\alpha^\lambda}
    =
    \mcG_{\bm\alpha},
    \qquad
    \bigotimes_{\lambda\in\Lambda}
    \mcU_{\alpha^\lambda}^{(m_\lambda)}
    =
    \mcU_{\bm\alpha},
\end{align}
which gives Eq.~\eqref{eq:wreath-schur-abstract-isometry}.
\end{enumerate}
For the gate count, Step~1 uses
\begin{align}
    nN_{\mathrm{Sch}}^{(G)}
    \!\left(\frac{\varepsilon}{2n}\right)
\end{align}
gates and contributes operator-norm error at most $\varepsilon/2$.
In Step~2, use a reversible stable sorting network to sort the tuples
consisting of the irrep label, the representation
register, and the multiplicity register.  The network uses
$O(n\polylog(n))$ comparisons.  Each comparison acts on the irrep
label, and each conditional swap also acts on the
corresponding representation and multiplicity registers.  Since a
multiplicity register contains $\lceil\log_2 m_{\max}\rceil$ qubits,
the total cost of the stable sort, the comparison of neighboring labels,
and the subsequent uncomputation of the temporary sorting data is
\begin{align}
    O\!\left(
        n(b_\Lambda+b_R)
        \polylog(n,m_{\max})
    \right).
\end{align}
For Step~3, consider a block containing $n_\lambda$ copies of the
embedded multiplicity space $\iota_\lambda(\mcM_\lambda)$.  The recursive
BCH Schur transform uses $n_\lambda-1$ Clebsch--Gordan transforms on this
block.  Hence the total number of Clebsch--Gordan transforms over all
blocks is less than
\begin{align}
    \sum_{\lambda\in\Lambda} n_\lambda=n.
\end{align}
By Ref.~\cite[Sec.~8]{burchardt2025high}, one such transform can be
implemented with
\begin{align}
    O\!\left(
        m_{\max}^4
        \polylog(n,m_{\max},1/\delta)
    \right)
\end{align}
gates when its operator-norm error is at most $\delta$.  Taking
$\delta=O(\varepsilon/n)$ gives
\begin{align}
    N_{\mathrm{BCH}}
    =
    O\!\left(
        nm_{\max}^4
        \polylog(n,m_{\max},1/\varepsilon)
    \right).
    \label{eq:wreath-controlled-schur-cost}
\end{align}
The total error contributed by these transforms is then at most
$\varepsilon/2$.
Combining the three steps gives
\begin{align}
    N_{\mathrm{Sch}}^{(\Wr_n)}(\varepsilon)
    \leq
    nN_{\mathrm{Sch}}^{(G)}
    \!\left(\frac{\varepsilon}{2n}\right)
    +
    O\!\left(
        n\bigl(b_\Lambda+b_R+m_{\max}^4\bigr)
        \polylog(n,m_{\max},1/\varepsilon)
    \right),
\end{align}
which proves Eq.~\eqref{eq:wreath-schur-gate-bound}.
\end{proof}

\paragraph{Permutation representation on $k$ qudits.}
Let $G=\mfS_k$ act by the permutation representation
$\pi_d$ on $(\CC^d)^{\otimes k}$, with $d$ fixed.  By Schur--Weyl
duality, only partitions $\lambda\vdash k$ with at most $d$ rows occur.
Moreover,
\begin{align}
    m_\lambda
    =
    \dim\mcM_\lambda
    =
    \prod_{1\leq i<j\leq d}
    \frac{\lambda_i-\lambda_j+j-i}{j-i},
\end{align}
and therefore
\begin{align}
    m_{\max}
    =
    O_d\!\left(k^{d(d-1)/2}\right).
\end{align}
The number of such partitions is polynomial in $k$, so
\[
    b_\Lambda=O_d(\log k).
\]
Furthermore, since $\dim\mcS_\lambda\leq d^k$, a compact encoding gives
\[
    b_R=O_d(k).
\]
Substituting these bounds into
Eq.~\eqref{eq:wreath-schur-gate-bound} gives
\begin{align}
    \begin{aligned}
    N_{\mathrm{Sch}}^{(\mfS_k\wr\mfS_n)}(\varepsilon)
    &\leq
    nN_{\mathrm{Sch}}^{(\pi_d)}
    \!\left(\frac{\varepsilon}{2n}\right)\\
    &\quad+
    O_d\!\left(
        nk^{2d(d-1)}
        \polylog(n,k,1/\varepsilon)
    \right).
    \end{aligned}
    \label{eq:symmetric-wreath-schur-gate-bound}
\end{align}
Since for fixed $d$, the embedding $\iota_\lambda$ requires $O_d(\polylog k)$ gates, using the BCH implementation also for the one-copy
Schur transform gives
\[
    N_{\mathrm{Sch}}^{(\pi_d)}
    \!\left(\frac{\varepsilon}{2n}\right)
    =
    O_d\!\left(
        k\,\polylog(n,k,1/\varepsilon)
    \right),
\]
so that
\begin{align}
    N_{\mathrm{Sch}}^{(\mfS_k\wr\mfS_n)}(\varepsilon)
    =
    O_d\!\left(
        nk^{2d(d-1)}
        \polylog(n,k,1/\varepsilon)
    \right).
\end{align}

After regrouping tensor factors, the channel applications use $\overline{\pi_d}\otimes\pi_d\cong\pi_{d^2}$ on Choi operators and $\pi_d\otimes\overline{\pi_d}\otimes\pi_d\cong\pi_{d^3}$ on dilation outputs, so the corresponding bounds follow by replacing $d$ with $d^2$ and $d^3$, respectively, and remain polynomial for fixed $d$.

This construction generalizes, in reverse, the embedding steps of the plethysm algorithm in Ref.~\cite[Fig.~1 and Sec.~3]{christandl_et_al:LIPIcs.CCC.2026.28}: those steps use inverse Schur transforms for a fixed local irrep $\nu$, whereas here we coherently allow several local irreps $\lambda$ and apply the corresponding Schur transforms blockwise.

\section{Optimal learning of covariant isometries in mean squared Choi trace distance}
\label{appendix:learning-covariant-isometries-choi-trace-norm}

We consider the estimation of covariant isometries
\begin{align}
    W\in \W_G\coloneqq \left\{W: \mcI \to \mcO \; \middle| \; W^\dagger W = \1_{\mcI}, \quad  W U_1(g) = U_2(g) W \quad \forall g\in G\right\},
\end{align}
which can be decomposed as
\begin{align}
    W = \bigoplus_{\lambda\in \irrep{G}} \1_{d_\lambda} \otimes W_\lambda, \quad
    W_\lambda: \CC^{m_\lambda} \to \CC^{M_\lambda},
\end{align}
where $m_\lambda$ and $M_\lambda$ are the multiplicities of the irreducible representation $\lambda$ in $U_1$ and $U_2$, respectively.
We can define the Haar measure on $\W_G$ by the product of the Haar measures on each multiplicity isometry $W_\lambda$.
We consider the quantum tester~\cite{chiribella2009theoretical} $\sfT = (\sfT_{\widehat{W}})$, which takes $n$ queries of $W$ as inputs and outputs an estimate $\widehat{W}$ of $W$ as a measurement outcome\footnote{Here we assume that the measurement outcome satisfies $\widehat{W}\in \W_G$.}.
The probability distribution is given by
\begin{align}
    \Prob_\sfT(\widehat{W}\mid W) = \Tr[T_{\widehat{W}}^{\top} \dketbra{W}^{\otimes n}],
\end{align}
using the Choi operator $T_{\widehat{W}}$ of the tester $\sfT_{\widehat{W}}$.
We consider the mean squared Choi trace distance error of estimating covariant isometries defined as
\begin{align}
    \int_{\W_G} \dd W \mathbb{E}_\sfT \qty[d_\mathrm{tr}(\rho_W, \rho_{\widehat{W}})^2].
\end{align}
We consider a collection of Young diagrams $\bm{\alpha} = (\alpha^\lambda)_{\lambda\in \irrep{G}}$ with $\alpha^\lambda\vdash_{m_\lambda} n_\lambda$ for $n_\lambda$ satisfying $\sum_\lambda n_\lambda = n$, and define the set of such collections of Young diagrams as
\begin{align}
    \young{\bm{m}}{n} \coloneqq \left\{ \bm{\alpha} = (\alpha^\lambda)_{\lambda\in \irrep{G}} \; \middle| \; \alpha^\lambda\vdash_{m_\lambda} n_\lambda, \sum_\lambda n_\lambda = n \right\}.
\end{align}
We use the convention $n_\lambda=0$ and $\alpha^\lambda=\varnothing$ whenever $m_\lambda=0$.
We define $\bm{\alpha}+e^{\lambda}_j$ as the collection of Young diagrams obtained by adding one box to the $j$-th row of $\alpha^\lambda$, and define the set $\bm{\alpha}+_{\bm{m}} \square$ as
\begin{align}
    \bm{\alpha}+_{\bm{m}} \square \coloneqq \left\{ \bm{\alpha}+e^\lambda_j \; \middle| \; 1\leq j\leq m_\lambda \text{ s.t. } \alpha^\lambda+e_j \text{ is a valid Young diagram} \right\},
\end{align}
and define $\bm{\alpha}+_{\bm{M}} \square$ similarly.
We also define $\Lambda\coloneqq \{\lambda \in \irrep{G} \mid m_\lambda>0\}$.
Using this, we define the estimation matrix $E_n$ by
\begin{align}
    (E_n)_{\bm{\alpha}\bm{\beta}} \coloneqq {1\over d_{\mcI}^2} \sum_{\lambda, \mu\in \Lambda} \sum_{i,j} \delta_{\bm{\alpha}+e^\lambda_i, \bm{\beta}+e^\mu_j} d_\lambda d_\mu f_\lambda(\alpha^\lambda_i-i) f_\mu(\beta^\mu_j-j) \quad \forall \bm{\alpha}, \bm{\beta} \in \young{\bm{m}}{n},
    \label{eq:learning-estimation-matrix}
\end{align}
where $f_\lambda$ is a function defined by
\begin{align}
    f_\lambda(x) \coloneqq \sqrt{x+m_\lambda+1\over x+M_\lambda+1},
\end{align}
and the summation is taken over $i\leq m_\lambda, j\leq m_\mu$ such that $\alpha^\lambda+e_i$ and $\beta^\mu+e_j$ are valid Young diagrams.
Then, similarly to the estimation of an unknown unitary/isometry channel~\cite{chiribella2005optimal,yoshida2025quantum}, the optimal average-case channel fidelity can be expressed as the largest eigenvalue of the estimation matrix $E_n$:
\begin{theorem}[Optimal mean squared Choi trace distance error of estimating covariant isometries]
\label{thm:learning-colored-estimation}
For every finite query number $n$, we have
\begin{align}
\inf_{\sfT} \int_{\W_G} \dd W \mathbb{E}_\sfT \qty[d_\mathrm{tr}(\rho_W, \rho_{\widehat{W}})^2] = 1-\lambda_{\max}(E_n),
\label{eq:learning-finite-estimation}
\end{align}
where $\lambda_{\max}(E_n)$ is the largest eigenvalue of the estimation matrix $E_n$, and the optimum is achieved by a parallel tester.
\end{theorem}

\begin{proof}
    For isometries $W,\widehat{W}:\mcI\to\mcO$, we define their Choi channel fidelity by
    \begin{align}
        F_\mathrm{ch}(W,\widehat{W})\coloneqq F(\rho_W,\rho_{\widehat{W}})={\left|\Tr[W^\dagger\widehat{W}]\right|^2\over d_{\mcI}^2},
        \label{eq:learning-isometry-channel-fidelity}
    \end{align}
    where $F$ is the squared state fidelity in Eq.~\eqref{eq:state-fidelity} and $d_{\mcI}=\dim\mcI$.
    Since the normalized Choi states of $W$ and $\widehat{W}$ are pure, the pure-state trace-distance identity gives $d_\mathrm{tr}(\rho_W,\rho_{\widehat{W}})^2=1-F_\mathrm{ch}(W,\widehat{W})$.
    Thus, minimizing the mean squared Choi trace distance error is equivalent to maximizing the average channel fidelity.
    By using a tester $\sfT = (\sfT_{\widehat{W}})_{\widehat{W}\in \W_G}$ with the Choi operator $T_{\widehat{W}}$, the probability distribution of the estimated multiplicity isometries $\widehat{W}$ given the true multiplicity isometries $W$ can be expressed as
    \begin{align}
        \Prob_{\sfT}(\widehat{W} \mid W) = \Tr[T_{\widehat{W}}^{\top} \dketbra{W}^{\otimes n}].
    \end{align}
    Similarly to Eq.~\eqref{eq:generalized-schur-decomposition}, $W^{\otimes n}$ can be decomposed as
    \begin{align}
        W^{\otimes n} \cong \bigoplus_{\bm{\alpha}} \1_{\mcG_{\bm{\alpha}}} \otimes f_{\bm{\alpha}}^{(\bm{m} \to \bm{M})}(W),
    \end{align}
    where $f_{\bm{\alpha}}^{(\bm{m} \to \bm{M})}(W)$ is defined by
    \begin{align}
        f_{\bm{\alpha}}^{(\bm{m} \to \bm{M})}(W) \coloneqq \bigotimes_{\lambda\in \irrep{G}} f_{\alpha^\lambda}^{(m_\lambda \to M_\lambda)}(W_\lambda): \mcU_{\bm{\alpha}}^{(\bm{m})} \to \mcU_{\bm{\alpha}}^{(\bm{M})},
    \end{align}
    where we use the decomposition of $W_\lambda^{\otimes n_\lambda}$ as given in Eq.~\eqref{eq:prelim-isometry-schur-decomposition}.
    We also define $f_n(U), f_n(V)$ for $U\in \U_G(\mcI)$ and $V\in \U_G(\mcO)$ by
    \begin{align}
        f_n(U) \coloneqq \bigoplus_{\bm{\alpha}} f_{\bm{\alpha}}^{(\bm{m})}(U), \quad f_n(V) \coloneqq \bigoplus_{\bm{\alpha}} f_{\bm{\alpha}}^{(\bm{M})}(V).
    \end{align} 
    Thus, the Choi vector $\dket{W}^{\otimes n}$ is given by
    \begin{align}
        \dket{W}^{\otimes n} = \bigoplus_{\bm{\alpha}} \dket{\1_{\mcG_{\bm{\alpha}}}} \otimes \dket{f_{\bm{\alpha}}^{(\bm{m} \to \bm{M})}(W)},
    \end{align}
    which can be compressed into
    \begin{align}
        \dket{f_n(W)} \coloneqq \bigoplus_{\bm{\alpha}} \dket{f_{\bm{\alpha}}^{(\bm{m} \to \bm{M})}(W)} \in \bigoplus_{\bm{\alpha}\in \young{\bm{m}}{n}} \overline{\mcU_{\bm{\alpha}}^{(\bm{m})}} \otimes \mcU_{\bm{\alpha}}^{(\bm{M})} \subset \bigoplus_{\bm{\alpha}\in \young{\bm{m}}{n}} \overline{\mcU_{\bm{\alpha}}^{(\bm{m})}} \otimes \bigoplus_{\bm{\beta} \in \young{\bm{M}}{n}}\mcU_{\bm{\beta}}^{(\bm{M})},
    \end{align}
    such that
    \begin{align}
        \dket{W}^{\otimes n} = \iota_{G,n} \dket{f_n(W)},
    \end{align}
    where $\iota_{G, n}: \bigoplus_{\bm{\alpha}} \overline{\mcU_{\bm{\alpha}}^{(\bm{m})}} \otimes \mcU_{\bm{\alpha}}^{(\bm{M})} \to (\overline{\mcI} \otimes \mcO)^{\otimes n}$ is the embedding defined by
    \begin{align}
        \iota_{G,n}\coloneqq \bigoplus_{\bm{\alpha}} \dket{\1_{\mcG_{\bm{\alpha}}}} \otimes \1_{\overline{\mcU_{\bm{\alpha}}^{(\bm{m})}}} \otimes \1_{\mcU_{\bm{\alpha}}^{(\bm{M})}}.
    \end{align}
    Defining $\widetilde{S}_{\widehat{W}}$ by
    \begin{align}
        \widetilde{S}_{\widehat{W}}^{\top} \coloneqq \iota_{G,n}^{\dagger} T_{\widehat{W}}^{\top}\iota_{G,n}  \in \mcL\qty(\bigoplus_{\bm{\alpha}\in \young{\bm{m}}{n}} \overline{\mcU_{\bm{\alpha}}^{(\bm{m})}} \otimes \mcU_{\bm{\alpha}}^{(\bm{M})}),
    \end{align}
    we have
    \begin{align}
        \Prob_{\sfT}(\widehat{W}\mid W)
        =\Tr[\widetilde{S}_{\widehat{W}}^{\top} \dketbra{f_n(W)}].
    \end{align}

    The compressed Choi operators satisfy $\widetilde{S}_{\widehat{W}}^{\top}\succeq0$ and $\int_{\W_G}\dd\widehat{W}\,\Tr[\widetilde{S}_{\widehat{W}}^{\top}\dketbra{f_n(W)}]=1$ for every $W\in\W_G$.
    These are the only properties used below, and they hold for every valid tester, including sequential or indefinite-causal-order ones~\cite{oreshkov2012quantum}.
    Thus, the optimality of Thm.~\ref{thm:learning-colored-estimation} holds for all testers, including those with indefinite causal order.
    We denote the twirled tester by $\widetilde{\sfT}\coloneqq(\widetilde T_{\widehat{W}})_{\widehat{W}\in\W_G}$.
    Since the Haar measure on $\W_G$ is invariant under the left and right multiplication of $\U_G(\mcI)$ and $\U_G(\mcO)$, we can show that the average-case channel fidelity does not change under the following twirling:
    \begin{align}
        \widetilde{T}_{\widehat{W}}^{\top} \coloneqq \int_{\U_G(\mcI)} \dd U \int_{\U_G(\mcO)} \dd V [\overline{f_n(U)} \otimes f_n(V)] \widetilde{S}_{V^\dagger \widehat{W} U}^{\top} [\overline{f_n(U)} \otimes f_n(V)]^\dagger.
    \end{align}
    This covariance also shows that the average-case channel fidelity coincides with the worst-case channel fidelity.
    We define the operator
    \begin{align}
        \widetilde{T}_\mathrm{total} \coloneqq \int_{\W_G} \dd \widehat{W} \widetilde{T}_{\widehat{W}},
    \end{align}
    satisfying
    \begin{align}
        \label{eq:total-tester-invariance}
        \qty[\widetilde{T}_\mathrm{total}^{\top}, \overline{f_n(U)} \otimes f_n(V)] = 0
        \quad \forall U\in \U_G(\mcI), V\in \U_G(\mcO).
    \end{align}
    Due to Schur's lemma, the invariance~\eqref{eq:total-tester-invariance} implies that $\widetilde{T}_\mathrm{total}$ can be decomposed as
    \begin{align}
        \widetilde{T}_\mathrm{total}^{\top} \cong \bigoplus_{\bm{\alpha}\in \young{\bm{m}}{n}} t_{\bm{\alpha}} \1_{\overline{\mcU_{\bm{\alpha}}^{(\bm{m})}}} \otimes \1_{\mcU_{\bm{\alpha}}^{(\bm{M})}},
    \end{align}
    using $t_{\bm{\alpha}} \geq 0$.
    Defining
    \begin{align}
        \label{eq:double-tilde-total-tester}
        \doublewidetilde{T}_\mathrm{total}^{\top} \cong \bigoplus_{\bm{\alpha}\in \young{\bm{m}}{n}} t_{\bm{\alpha}} \1_{\overline{\mcU_{\bm{\alpha}}^{(\bm{m})}}} \otimes \1_{\mcU_{\bm{\alpha}}^{(\bm{m})}},
    \end{align}
    we have
    \begin{align}
        [\1 \otimes f_n(W)] \doublewidetilde{T}_{\mathrm{total}}^{\top} = \widetilde{T}_\mathrm{total}^{\top} [\1 \otimes f_n(W)]
        \quad \forall W\in \W_G.
    \end{align}
    We then define the quantum state $\ket{\phi_\mathrm{est}}$ and a POVM $\{M_{\widehat{W}} \dd \widehat{W}\}_{\widehat{W}}$ by
    \begin{align}
        \label{eq:estimation-state}
        \ket{\phi_\mathrm{est}} &\coloneqq \sqrt{\doublewidetilde{T}_\mathrm{total}}^\top \bigoplus_{\bm{\alpha}} \dket{\1_{\mcU_{\bm{\alpha}}^{(\bm{m})}}} = \bigoplus_{\bm{\alpha}} \sqrt{t_{\bm{\alpha}}} \dket{\1_{\mcU_{\bm{\alpha}}^{(\bm{m})}}},\\
        M_{\widehat{W}} &\coloneqq (\widetilde{T}_\mathrm{total}^{-1/2} \widetilde{T}_{\widehat{W}} \widetilde{T}_\mathrm{total}^{-1/2})^\top,
    \end{align}
    where $\widetilde{T}_\mathrm{total}^{-1/2}$ denotes the Moore--Penrose inverse on $\supp\widetilde{T}_\mathrm{total}$.
    The state is normalized since
    \begin{align}
        \braket{\phi_\mathrm{est}}{\phi_\mathrm{est}}
        &= \bra{\phi_\mathrm{est}} (\1\otimes f_n(W))^\dagger (\1\otimes f_n(W)) \ket{\phi_\mathrm{est}}\\
        &= \Tr[\widetilde{T}_\mathrm{total}^\top \dketbra{f_n(W)}]\\
        &= 1,
    \end{align}
    while $\{M_{\widehat{W}} \dd \widehat{W}\}_{\widehat{W}}$ defines a POVM since
    \begin{align}
        \int \dd \widehat{W} M_{\widehat{W}}
        = (\widetilde{T}_\mathrm{total}^{-1/2} \widetilde{T}_\mathrm{total} \widetilde{T}_\mathrm{total}^{-1/2})^\top
        = \Pi_{\supp \widetilde{T}_\mathrm{total}}^\top\preceq \1
    \end{align}
    holds, where $\Pi_{\supp \widetilde{T}_\mathrm{total}}$ is the projection onto the support of $\widetilde{T}_\mathrm{total}$.
    We can add a positive operator so that $\{M_{\widehat{W}} \dd \widehat{W}\}_{\widehat{W}}$ becomes a POVM on the whole space $\bigoplus_{\bm{\alpha}\in \young{\bm{m}}{n}} \overline{\mcU_{\bm{\alpha}}^{(\bm{m})}} \otimes \bigoplus_{\bm{\beta} \in \young{\bm{M}}{n}}\mcU_{\bm{\beta}}^{(\bm{M})}$.
    We construct a parallel tester $\doublewidetilde{\sfT}$ using $\ket{\phi_\mathrm{est}}$ and $\{M_{\widehat{W}} \dd \widehat{W}\}_{\widehat{W}}$ by
    \begin{align}
        \Prob_{\doublewidetilde{\sfT}}(\widehat{W} \mid W) \coloneqq \bra{\phi_\mathrm{est}} (\1\otimes f_n(W))^\dagger M_{\widehat{W}} (\1\otimes f_n(W)) \ket{\phi_\mathrm{est}}.
    \end{align}
    Then it reproduces the probabilities of the twirled tester $\widetilde{\sfT}$, and hence achieves the same average-case channel fidelity as the original tester $\sfT$, since
    \begin{align}
        \Prob_{\doublewidetilde{\sfT}}(\widehat{W} \mid W)
        &= \sum_{\bm{\alpha}, \bm{\beta}} \dbra{\1_{\mcU_{\bm{\alpha}}^{(\bm{m})}}}\sqrt{\doublewidetilde{T}_{\mathrm{total}}}^\top [\1\otimes f_n(W)]^\dagger M_{\widehat{W}}[\1\otimes f_n(W)]\sqrt{\doublewidetilde{T}_{\mathrm{total}}}^\top \dket{\1_{\mcU_{\bm{\beta}}^{(\bm{m})}}}\\
        &= \dbra{f_n(W)} \sqrt{\widetilde{T}_{\mathrm{total}}}^{\top} M_{\widehat{W}} \sqrt{\widetilde{T}_{\mathrm{total}}}^{\top} \dket{f_n(W)}\\
        &= \dbra{f_n(W)} \qty(\widetilde{T}_{\mathrm{total}}^{1/2} \widetilde{T}_{\mathrm{total}}^{-1/2} \widetilde{T}_{\widehat{W}} \widetilde{T}_{\mathrm{total}}^{-1/2} \widetilde{T}_{\mathrm{total}}^{1/2})^{\top} \dket{f_n(W)}\\
        &= \dbra{f_n(W)} \widetilde{T}_{\widehat{W}}^{\top} \dket{f_n(W)}\\
        &= \Prob_{\widetilde{\sfT}}(\widehat{W} \mid W)
    \end{align}
    holds.
    Thus, we show that the optimal average-case channel fidelity can be achieved by a parallel protocol.

    Finally, we optimize the average-case channel fidelity over parallel protocols to obtain the expression~\eqref{eq:learning-finite-estimation}.
    We rewrite Eq.~\eqref{eq:estimation-state} as
    \begin{align}
        \ket{\phi_\mathrm{est}} = \bigoplus_{\bm{\alpha}\in \young{\bm{m}}{n}} {v_{\bm{\alpha}} \over \sqrt{D_{\bm{\alpha}}^{(\bm{m})}}} \dket{\1_{\mcU_{\bm{\alpha}}^{(\bm{m})}}},
    \end{align}
    where $D_{\bm{\alpha}}^{(\bm{m})} \coloneqq \dim \mcU_{\bm{\alpha}}^{(\bm{m})}$ and $v_{\bm{\alpha}}\coloneqq \sqrt{t_{\bm{\alpha}}D_{\bm{\alpha}}^{(\bm{m})}}\geq 0$.
    By applying $f_n(W)$ to $\ket{\phi_\mathrm{est}}$, we have
    \begin{align}
        \ket{\phi_W} \coloneqq (\1\otimes f_n(W)) \ket{\phi_\mathrm{est}} = \bigoplus_{\bm{\alpha}\in \young{\bm{m}}{n}} {v_{\bm{\alpha}} \over \sqrt{D_{\bm{\alpha}}^{(\bm{m})}}} \dket{f_{\bm{\alpha}}^{(\bm{m} \to \bm{M})}(W)}.
    \end{align}
    Then, the average-case channel fidelity can be expressed as
    \begin{align}
        \int_{\W_G} \dd W \mathbb{E}[F_\mathrm{ch}(W, \widehat{W})]
        &= {1\over d_\mcI^2} \int_{\W_G} \dd W \int_{\W_G} \dd \widehat{W} \Tr[\ketbra{\phi_W} \otimes \dketbra{\bm{W}} \cdot M_{\widehat{W}} \otimes \dketbra{\widehat{\bm{W}}}],
    \end{align}
    where $\bm{W}$ and $\widehat{\bm{W}}$ are defined by $\bm{W} \coloneqq \bigoplus_\lambda \sqrt{d_\lambda} W_\lambda$ and $\widehat{\bm{W}} \coloneqq \bigoplus_\lambda \sqrt{d_\lambda} \widehat{W}_\lambda$.
    Defining the operator
    \begin{align}
        \Pi\coloneqq \int \dd W \ketbra{\phi_{W}} \otimes \dketbra{\bm{W}},
    \end{align}
    we have
    \begin{align}
        \int_{\W_G} \dd W \mathbb{E}[F_\mathrm{ch}(W, \widehat{W})] = {1\over d_{\mcI}^2} \int_{\W_G} \dd \widehat{W} \Tr[\Pi \cdot M_{\widehat{W}} \otimes \dketbra{\widehat{\bm{W}}}].
    \end{align}
    Defining
    \begin{align}
        R_n(U, V)\coloneqq
        \bigoplus_{\bm\alpha\in\young{\bm{m}}{n}}
        \overline{f_{\bm\alpha}(U)}\otimes f_{\bm\alpha}(V),
        \qquad
        R_1(U,V)\coloneqq
        \bigoplus_{\lambda\in\irrep{G}}\overline U_\lambda\otimes V_\lambda \quad \forall U\in \U_G(\mcI), V\in \U_G(\mcO),
    \end{align}
    the operator $\Pi$ satisfies the invariance
    \begin{align}
        \qty[\Pi,R_n(U,V)\otimes R_1(U,V)] = 0
        \quad \forall U\in \U_G(\mcI), V\in \U_G(\mcO).
    \end{align}
    Thus, we can assume that the POVM $\{M_{\widehat{W}} \dd \widehat{W}\}_{\widehat{W}}$ is covariant:
    \begin{align}
        M_{\widehat{W}} = R_n(U, V)
        M_{V^\dagger \widehat{W} U}
        R_n(U,V)^\dagger
        \quad \forall U\in \U_G(\mcI), V\in \U_G(\mcO).
    \end{align}
    Then, for an arbitrary fixed $W_0\in \W_G$, we have
    \begin{align}
        \int_{\W_G} \dd W \mathbb{E}[F_\mathrm{ch}(W, \widehat{W})] = {1\over d_{\mcI}^2} \Tr[\Pi\cdot M_{W_0} \otimes \dketbra{\bm{W}_0}].
    \end{align}
    We decompose $M_{W_0}$ as
    \begin{align}
        M_{W_0} = \sum_{i=1}^{r} \dketbra{\eta^i}, \quad
        \dket{\eta^i} = \bigoplus_{\bm{\alpha}\in \young{\bm{m}}{n}} \sqrt{D_{\bm{\alpha}}^{(\bm{M})}} \dket{\eta_{\bm{\alpha}}^i},
    \end{align}
    using linear operators $\eta_{\bm{\alpha}}^i: \mcU_{\bm{\alpha}}^{(\bm{m})} \to \mcU_{\bm{\alpha}}^{(\bm{M})}$ satisfying the POVM normalization
    \begin{align}
        \1
        &= \int_{\W_G}\dd\widehat{W} M_{\widehat{W}}\\
        &=\int_{\U_G(\mcO)}\dd V M_{V W_0}\\
        &=\int_{\U_G(\mcO)}\dd V
        R_n(\1,V)M_{W_0}R_n(\1,V)^\dagger\\
        &=\bigoplus_{\bm\alpha\in\young{\bm{m}}{n}}
        \left[\left(\sum_{i=1}^r
        \eta_{\bm\alpha}^{i\dagger}\eta_{\bm\alpha}^i\right)^\top
        \otimes\1_{\mcU_{\bm\alpha}^{(\bm{M})}}\right],
    \end{align}
    where the second equality uses that $V W_0$ is distributed according to the product invariant measure on $\W_G$.
    Equivalently, we have
    \begin{align}
        \sum_{i=1}^{r} \eta_{\bm{\alpha}}^{i\dagger} \eta_{\bm{\alpha}}^i = \1_{\mcU_{\bm{\alpha}}^{(\bm{m})}} \quad \forall \bm{\alpha}\in \young{\bm{m}}{n}.
    \end{align}
    Defining the operator $\eta_{\bm{\mu} \to \bm{\nu}}^{i, \bm{\alpha}}: \mcU_{\bm{\mu}}^{(\bm{m})} \to \mcU_{\bm{\nu}}^{(\bm{M})}$ for $\bm{\mu}\in \bm{\alpha}+_{\bm{m}}\square, \bm{\nu}\in \bm{\alpha}+_{\bm{M}}\square$ by
    \begin{align}
       \eta_{\bm{\alpha}}^i \otimes \bigoplus_\lambda W_{0, \lambda} \cong \left[\eta_{\bm{\mu} \to \bm{\nu}}^{i, \bm{\alpha}}\right]_{\substack{\bm{\nu}\in\bm{\alpha}+_{\bm{M}}\square\\\bm{\mu}\in\bm{\alpha}+_{\bm{m}}\square}}, \qquad \eta_{\bm{\mu} \to \bm{\nu}}^{i, \bm{\alpha}}: \mcU_{\bm{\mu}}^{(\bm{m})} \to \mcU_{\bm{\nu}}^{(\bm{M})},
    \end{align}
    we have
    \begin{align}
        \bigoplus_{\bm{\mu} \in \bm{\alpha} +_{\bm{m}} \square} \1_{\mcU_{\bm{\mu}}^{(\bm{m})}}
        &= \1_{\mcU_{\bm{\alpha}}^{(\bm{m})}} \otimes \bigoplus_\lambda \1_{m_\lambda}\\
        &= \sum_{i=1}^{r} \eta_{\bm{\alpha}}^{i\dagger} \eta_{\bm{\alpha}}^i \otimes \bigoplus_{\lambda} W_{0, \lambda}^\dagger W_{0, \lambda}\\
        &= \left[\sum_{\bm{\nu}\in \bm{\alpha}+_{\bm{M}}\square} \sum_{i=1}^{r} \eta_{\bm{\mu} \to \bm{\nu}}^{i, \bm{\alpha}\dagger} \eta_{\bm{\mu}' \to \bm{\nu}}^{i, \bm{\alpha}}\right]_{\bm{\mu},\bm{\mu}'\in\bm{\alpha}+_{\bm{m}}\square},
    \end{align}
    i.e.,
    \begin{align}
        \sum_{\bm{\nu}\in \bm{\alpha}+_{\bm{M}}\square} \sum_{i=1}^{r} \eta_{\bm{\mu} \to \bm{\nu}}^{i, \bm{\alpha}\dagger} \eta_{\bm{\mu}' \to \bm{\nu}}^{i, \bm{\alpha}} = \delta_{\bm{\mu},\bm{\mu}'} \1_{\mcU_{\bm{\mu}}^{(\bm{m})}}.
    \end{align}
    Taking $\bm{\mu}'=\bm{\mu}$ and retaining only the positive summands with $\bm{\nu}=\bm{\mu}$ gives
    \begin{align}
        \sum_{i=1}^{r} \eta_{\bm{\mu} \to \bm{\mu}}^{i, \bm{\alpha}\dagger} \eta_{\bm{\mu} \to \bm{\mu}}^{i, \bm{\alpha}} \preceq \1_{\mcU_{\bm{\mu}}^{(\bm{m})}}.
        \label{eq:learning-diagonal-block-contraction}
    \end{align}
    In the operator sums below, vectors in individual sectors are embedded in the corresponding direct-sum space.
    Thus, we can bound the average-case channel fidelity as
    \begin{align}
    &\int_{\W_G} \dd W \mathbb{E}[F_\mathrm{ch}(W, \widehat{W})] \notag\\
    &= {1\over d_{\mcI}^{2}}
    \Tr\!\left[
    M_{W_0}\otimes
    \dketbra{\bm{W}_0}\cdot\Pi
    \right]
    \\
    &= {1\over d_{\mcI}^{2}}
    \Tr\!\left[
    \sum_{\bm{\alpha},\bm{\beta}}
    \sqrt{D_{\bm{\alpha}}^{(\bm{M})}D_{\bm{\beta}}^{(\bm{M})}}
    \sum_{i=1}^{r}
    \dketbra{\eta_{\bm{\alpha}}^i\otimes\bm{W}_0}{\eta_{\bm{\beta}}^i\otimes\bm{W}_0}
    \cdot\Pi
    \right]
    \\
    &= {1\over d_{\mcI}^{2}}
    \Tr\!\left[
    \sum_{\bm{\alpha},\bm{\beta}}
    \sqrt{D_{\bm{\alpha}}^{(\bm{M})}D_{\bm{\beta}}^{(\bm{M})}}
    \sum_{i=1}^{r}
    \sum_{\substack{
    \bm{\mu}=\bm{\alpha}+e_j^\lambda,\,
    \bm{\nu}=\bm{\alpha}+e_{j'}^\lambda\\
    \bm{\mu}'=\bm{\beta}+e_\ell^\mu,\,
    \bm{\nu}'=\bm{\beta}+e_{\ell'}^\mu}}
    \sqrt{d_\lambda d_\mu}\,
    \dket{\eta_{\bm{\mu}\to\bm{\nu}}^{i,\bm{\alpha}}}
    \dbra{\eta_{\bm{\mu}'\to\bm{\nu}'}^{i,\bm{\beta}}}
    \cdot\Pi
    \right]
    \\
    &= {1\over d_{\mcI}^{2}}
    \sum_{\bm{\alpha},\bm{\beta}}
    \overline{v_{\bm{\alpha}}}v_{\bm{\beta}}
    \sum_{\bm{\mu}\in(\bm{\alpha}+_{\bm{m}}\square)\cap(\bm{\beta}+_{\bm{m}}\square)}
    d_\lambda d_\mu
    \sqrt{{D_{\bm{\alpha}}^{(\bm{M})}D_{\bm{\beta}}^{(\bm{M})}
    \over D_{\bm{\alpha}}^{(\bm{m})}D_{\bm{\beta}}^{(\bm{m})}}}
    \sum_{i=1}^{r}
    {\Tr\!\left[
    \eta_{\bm{\mu}\to\bm{\mu}}^{i,\bm{\alpha}}
    \eta_{\bm{\mu}\to\bm{\mu}}^{i,\bm{\beta}\dagger}
    \right]
    \over D_{\bm{\mu}}^{(\bm{M})}}
    \\
    &\leq {1\over d_{\mcI}^{2}}
    \sum_{\bm{\alpha},\bm{\beta}}
    \overline{v_{\bm{\alpha}}}v_{\bm{\beta}}
    \!\sum_{\bm{\mu}\in(\bm{\alpha}+_{\bm{m}}\square)\cap(\bm{\beta}+_{\bm{m}}\square)}
    d_\lambda d_\mu
    \sqrt{{D_{\bm{\alpha}}^{(\bm{M})}D_{\bm{\beta}}^{(\bm{M})}
    \over D_{\bm{\alpha}}^{(\bm{m})}D_{\bm{\beta}}^{(\bm{m})}}}
    {\sqrt{\sum_{i=1}^{r}\Tr\!\left[
    \eta_{\bm{\mu}\to\bm{\mu}}^{i,\bm{\alpha}}
    \eta_{\bm{\mu}\to\bm{\mu}}^{i,\bm{\alpha}\dagger}
    \right]}
    \sqrt{\sum_{i=1}^{r}\Tr\!\left[
    \eta_{\bm{\mu}\to\bm{\mu}}^{i,\bm{\beta}}
    \eta_{\bm{\mu}\to\bm{\mu}}^{i,\bm{\beta}\dagger}
    \right]}
    \over D_{\bm{\mu}}^{(\bm{M})}}
    \\
    &\leq {1\over d_{\mcI}^{2}}
    \sum_{\bm{\alpha},\bm{\beta}}
    \overline{v_{\bm{\alpha}}}v_{\bm{\beta}}
    \sum_{\bm{\mu}\in(\bm{\alpha}+_{\bm{m}}\square)\cap(\bm{\beta}+_{\bm{m}}\square)}
    d_\lambda d_\mu
    \sqrt{{D_{\bm{\alpha}}^{(\bm{M})}D_{\bm{\beta}}^{(\bm{M})}
    \over D_{\bm{\alpha}}^{(\bm{m})}D_{\bm{\beta}}^{(\bm{m})}}}
    {D_{\bm{\mu}}^{(\bm{m})}\over D_{\bm{\mu}}^{(\bm{M})}},
    \end{align}
    where in each inner sum, the labels $\lambda,\mu,j,\ell$ are determined by $\bm{\mu}=\bm{\alpha}+e_j^\lambda=\bm{\beta}+e_\ell^\mu$.
    We use $\operatorname{Re}z\leq|z|$ and the Cauchy--Schwarz inequality in the fifth line, and Eq.~\eqref{eq:learning-diagonal-block-contraction} in the last line.
    The bound is attained by the rank-one covariant seed obtained by taking $r=1$ and $\eta_{\bm\alpha}=f_{\bm\alpha}(W_0)$, namely
    \begin{align}
        M_{\widehat{W}}= \sum_{\bm{\alpha}, \bm{\beta}\in\young{\bm{m}}{n}} \sqrt{D_{\bm{\alpha}}^{(\bm{M})}D_{\bm{\beta}}^{(\bm{M})}} \dketbra{f_{\bm{\alpha}}^{(\bm{m} \to \bm{M})}(\widehat{W})}{f_{\bm{\beta}}^{\bm{m} \to \bm{M}}(\widehat{W})}.
    \end{align}
    In this case, the fidelity is given by
    \begin{align}
        \bm{v}^\dagger E_n \bm{v},
    \end{align}
    by defining the matrix $E_n$ by
    \begin{align}
        (E_n)_{\bm{\alpha}\bm{\beta}} \coloneqq {1\over d_{\mcI}^{2}}
        \sum_{\bm{\mu}\in(\bm{\alpha}+_{\bm{m}}\square)\cap(\bm{\beta}+_{\bm{m}}\square)}
        d_\lambda d_\mu
        \sqrt{{D_{\bm{\alpha}}^{(\bm{M})}D_{\bm{\beta}}^{(\bm{M})}
        \over D_{\bm{\alpha}}^{(\bm{m})}D_{\bm{\beta}}^{(\bm{m})}}}
        {D_{\bm{\mu}}^{(\bm{m})}\over D_{\bm{\mu}}^{(\bm{M})}}.
    \end{align}
    By using the Weyl-dimension formula, we can rewrite this as in Eq.~\eqref{eq:learning-estimation-matrix}.
    The optimal fidelity is then given by
    \begin{align}
        \max_{\|\bm{v}\|_2=1} \bm{v}^\dagger E_n \bm{v} = \lambda_\mathrm{max}(E_n),
    \end{align}
    which concludes the proof.
\end{proof}

The largest eigenvalue of the estimation matrix $E_n$ can be evaluated up to the leading order $\Theta(n^{-1})$ using arguments similar to those in Ref.~\cite{yoshida2025quantum}, when the parameter $C_G$ defined below satisfies $C_G>0$.

\begin{theorem}[Largest eigenvalue of estimation matrix]
\label{thm:learning-leading-estimation}
Fix a compact group $G$ and unitary representations $U_1,U_2$.
The largest eigenvalue of the estimation matrix $E_n$ is given by
\begin{align}
    \label{eq:learning-leading-estimation}
    \lambda_\mathrm{max}(E_n) = 1-{C_G\over n} + O(n^{-4/3}),
\end{align}
where $C_G$ is defined by
\begin{align}
    \label{eq:def-C_G}
    C_G \coloneqq {1\over d_\mcI}\!\left(\sum_{\lambda\in\Lambda}m_\lambda\sqrt{d_\lambda (M_\lambda-m_\lambda)}\right)^2.
\end{align}
\end{theorem}

\begin{proof}
    We first show Eq.~\eqref{eq:learning-leading-estimation} by showing the upper and lower bounds of the largest eigenvalue of the estimation matrix $E_n$.
    For permutation-covariant channels, the asymptotic scaling of $C_{\mfS_k}$ is shown in Prop.~\ref{prop:symmetric-group-parameters-asymptotics} in Appendix~\ref{sec:uniform-fixed-height-young-estimates}.

    \paragraph{Upper bound.}
    Since $(E_n)_{\bm{\alpha}\bm{\beta}} \geq 0$ holds for all $\bm{\alpha}, \bm{\beta}$, due to the Perron--Frobenius theorem, we have
    \begin{align}
        \lambda_\mathrm{max}(E_n) &\leq \max_{\bm{\alpha}} \sum_{\bm{\beta}} (E_n)_{\bm{\alpha}\bm{\beta}}.
    \end{align}
    The right-hand side is bounded as follows:
    \begin{align}
        \max_{\bm{\alpha}} \sum_{\bm{\beta}} (E_n)_{\bm{\alpha}\bm{\beta}}
        &\leq {1\over d_{\mcI}^2} \max_{\bm{\alpha}} \sum_{\bm{\beta}} \sum_{\lambda, \mu\in \Lambda} \sum_{i,j} \delta_{\bm{\alpha}+e^\lambda_i, \bm{\beta}+e^\mu_j} d_\lambda d_\mu f_\lambda(\alpha^\lambda_i-i) f_\mu(\beta^\mu_j-j)\\
        &\leq {1\over d_{\mcI}^2} \max_{\bm{\alpha}} \sum_{\lambda, \mu\in \Lambda} \sum_{i,j} d_\lambda d_\mu f_\lambda(\alpha^\lambda_i-i) f_\mu(\alpha^\mu_j-j)\\
        &\leq \qty[{1\over d_\mcI} \max_{\bm{\alpha}} \sum_{\lambda\in \Lambda} \sum_{i} d_\lambda f_\lambda(\alpha^\lambda_i-i)]^2\\
        &\leq \qty[{1\over d_\mcI} \max_{\bm{n}} \sum_{\lambda\in \Lambda} d_\lambda m_\lambda f_\lambda\qty({n_\lambda\over m_\lambda}-{m_\lambda+1\over 2})]^2\\
        &\leq \qty[{1\over d_\mcI} \max_{\bm{n}} \sum_{\lambda\in \Lambda} d_\lambda m_\lambda \qty[1-{M_\lambda-m_\lambda \over 2\qty({n_\lambda\over m_\lambda}-{m_\lambda+1\over 2}+M_\lambda+1)}]]^2\\
        &= \qty[1-{1\over 2d_\mcI} \min_{\bm{n}} \sum_{\lambda\in \Lambda} {d_\lambda m_\lambda^2 (M_\lambda-m_\lambda) \over n_\lambda-m_\lambda({m_\lambda+1\over 2}-M_\lambda-1)}]^2,
    \end{align}
    where we use the monotonicity of $f_\lambda$ in the second line and Jensen's inequality for the concave function $f_\lambda$ and ${1\over m_\lambda} \sum_i (\alpha_i^\lambda-i) = {n_\lambda\over m_\lambda}-{m_\lambda+1\over 2}$ in the fourth line, and $\sqrt{1-x}\leq 1-{x\over 2}$ for all $x\leq 1$ in the fifth line.
    Due to the Cauchy--Schwarz inequality, the summation in the last equality is bounded as
    \begin{align}
        \sum_{\lambda\in \Lambda} {d_\lambda m_\lambda^2 (M_\lambda-m_\lambda) \over n_\lambda-m_\lambda({m_\lambda+1\over 2}-M_\lambda-1)}
        &\geq {\qty(\sum_{\lambda\in\Lambda} m_\lambda \sqrt{d_\lambda (M_\lambda-m_\lambda)})^2 \over \sum_{\lambda\in \Lambda} \qty[n_\lambda-m_\lambda({m_\lambda+1\over 2}-M_\lambda-1)]}\\
        &= d_\mcI {C_G \over n+O(1)}.
    \end{align}
    Thus, we have
    \begin{align}
        \lambda_\mathrm{max}(E_n)\leq \qty[1-{C_G\over 2(n+O(1))}]^2 = 1-{C_G\over n}+O(n^{-2}).
    \end{align}

    \paragraph{Lower bound.}
    We show a lower bound by constructing a normalized vector $\bm{v} = (v_{\bm{\alpha}})_{\bm{\alpha}\in \young{\bm{m}}{n}}$ such that $\bm{v}^\dagger E_n \bm{v} \geq 1-{C_G\over n}+O(n^{-4/3})$ by extending the construction in Refs.~\cite{yang2020optimal,yoshida2025quantum}.
    Throughout the lower-bound construction, $n$ is assumed to be sufficiently large.
    We define
    \begin{align}
        s_G\coloneqq\sum_{\lambda\in\Lambda}
        m_\lambda\sqrt{d_\lambda(M_\lambda-m_\lambda)}.
    \end{align}
    If $s_G>0$, we choose constants $c_\lambda>0$ for $\lambda\in\Lambda$ satisfying $m_\lambda = M_\lambda$, and set, for $\lambda\in\Lambda$,
    \begin{align}
        \delta_n\coloneqq n^{-1/3}\sum_{\lambda\in\Lambda: m_\lambda = M_\lambda}c_\lambda,\qquad
        q_\lambda\coloneqq
        \begin{cases}
            (1-\delta_n){m_\lambda \sqrt{d_\lambda(M_\lambda-m_\lambda)} \over s_G} & (m_\lambda<M_\lambda)\\
            c_\lambda n^{-1/3} & (m_\lambda = M_\lambda)
        \end{cases},
        \label{eq:learning-nondegenerate-occupation-weights}
    \end{align}
    so that $\sum_{\lambda\in\Lambda}q_\lambda=1$ and $q_\lambda>0$ holds.
    If $s_G=0$, choose instead any fixed strictly positive probability vector $(q_\lambda)_{\lambda\in\Lambda}$.
    We restrict the support of $\bm{v}$ to collections $\bm\alpha$ satisfying $|\alpha^\lambda| = q_\lambda n + O(n^{2/3})$ and $\alpha_i^\lambda = {q_\lambda n \over m_\lambda} + O(n^{2/3})$ for all $\lambda\in\Lambda$ and $i\in[m_\lambda]$.
    To this end, we fix $\lambda_*\in \Lambda$ and nonnegative integers $A_{\lambda, i}$ for $\lambda\in \Lambda$ and $i\in [m_\lambda]$ such that
    \begin{align}
        A_{\lambda, i} &= {q_\lambda n \over m_\lambda} + O(n^{2/3}) \quad \forall \lambda\in\Lambda, i\in [m_\lambda],\\
        A_{\lambda,i} &\geq A_{\lambda, i+1} + \lfloor n^{2/3}\rfloor \quad \forall \lambda\in\Lambda, i\in [m_\lambda-1],\\
        A_{\lambda_*, m_{\lambda_*}} &\geq \sum_{\lambda\in\Lambda} m_\lambda \lfloor n^{2/3}\rfloor,\\
        \sum_{\lambda\in \Lambda} \sum_{i=1}^{m_\lambda} A_{\lambda, i} &= n.
    \end{align}
    Then, for any $\widetilde{\bm{\alpha}} = (\widetilde{\alpha}^\lambda)_{\lambda\in \Lambda}$ with $\widetilde{\alpha}^\lambda \in \{0, \ldots, \lfloor n^{2/3}\rfloor-1\}^{m_\lambda}$ for $\lambda\neq \lambda_*$ and $\widetilde{\alpha}^{\lambda_*} \in \{0, \ldots, \lfloor n^{2/3}\rfloor-1\}^{m_{\lambda_*}-1}$, we define $\bm{\alpha} = (\alpha^\lambda)_{\lambda\in\Lambda}$ by
    \begin{align}
        \alpha^\lambda_i &= A_{\lambda, i} + \widetilde{\alpha}^\lambda_i \quad (i\in [m_\lambda] \text{ for } \lambda \neq \lambda_* \text{ and } i\in [m_{\lambda_*}-1] \text{ for } \lambda = \lambda_*),\\
        \alpha^{\lambda_*}_{m_{\lambda_*}} &= A_{\lambda_*, m_{\lambda_*}} - \sum_{i=1}^{m_{\lambda_*}-1} \widetilde{\alpha}^{\lambda_*}_i - \sum_{\substack{\lambda\in \Lambda\setminus \{\lambda_*\} \\ i\in [m_\lambda]}} \widetilde{\alpha}^\lambda_i.
    \end{align}
    Then, $\bm{\alpha}\in \young{\bm{m}}{n}$ holds for any $\widetilde{\bm{\alpha}}$.
    We define $v_{\bm{\alpha}}$ on this support by
    \begin{align}
        v_{\bm{\alpha}} &\coloneqq \prod_{\lambda\in \Lambda} v_{\widetilde{\alpha}^\lambda}^\lambda\\
        v_{\widetilde{\alpha}^\lambda}^{\lambda}&\coloneqq
        \begin{cases}
            \prod_{i=1}^{m_\lambda} \sqrt{g_{\widetilde{\alpha}^\lambda_i}} & (\lambda \neq \lambda_*)\\
            \prod_{i=1}^{m_{\lambda_*}-1} \sqrt{g_{\widetilde{\alpha}^{\lambda_*}_i}} & (\lambda = \lambda_*)
        \end{cases},\\
        g_k &\coloneqq {2\over \lfloor n^{2/3}\rfloor} \sin^2\qty({\pi(2k+1)\over 2\lfloor n^{2/3}\rfloor}),
    \end{align}
    Set $v_{\bm{\alpha}}$ and each $v^\lambda_{\widetilde{\alpha}^\lambda}$ to zero outside their respective supports.
    Since the coordinates $\widetilde{\alpha}^\lambda$ range independently over the cubes above, we have
    \begin{gather}
        \sum_{\bm{\alpha}}v_{\bm{\alpha}}^2
        =\prod_{\lambda\in\Lambda}\sum_{\widetilde{\alpha}^\lambda}
        \left(v^\lambda_{\widetilde{\alpha}^\lambda}\right)^2=1.
    \end{gather}
    To evaluate $\bm{v}^\dagger E_n \bm{v}$, we first bound $E_n$ on the support of $\bm{v}$ as
    \begin{align}
        (E_n)_{\bm{\alpha}\bm{\beta}}
        &= {1\over d_{\mcI}^2} \sum_{\lambda, \mu\in \Lambda} \sum_{i,j} \delta_{\bm{\alpha}+e^\lambda_i, \bm{\beta}+e^\mu_j} d_\lambda d_\mu f_\lambda(\alpha^\lambda_i-i) f_\mu(\beta^\mu_j-j)\\
        &\geq {1\over d_{\mcI}^2} \sum_{\lambda, \mu\in \Lambda} \sum_{i,j} \delta_{\bm{\alpha}+e^\lambda_i, \bm{\beta}+e^\mu_j} d_\lambda d_\mu \underline{f}_\lambda \underline{f}_\mu,
    \end{align}
    where $\underline{f}_\lambda$ is defined by
    \begin{align}
        \underline{f}_\lambda
        &\coloneqq \min_{\bm{\alpha}: v_{\bm{\alpha}} \neq 0} \min_{i\in [m_\lambda]} f_\lambda(\alpha^\lambda_i-i)\\
        &\geq f_\lambda\qty({q_\lambda n \over m_\lambda} + O(n^{2/3}))\\
        &\geq 1-{m_\lambda(M_\lambda-m_\lambda) \over 2q_\lambda n} + O(n^{-4/3}).
    \end{align}
    Then, we have
    \begin{align}
        \bm{v}^\dagger E_n \bm{v}
        &\geq {1\over d_{\mcI}^2} \sum_{\lambda, \mu\in \Lambda} d_\lambda d_\mu \underline{f}_\lambda \underline{f}_\mu \sum_{\bm{\alpha}, \bm{\beta}} \sum_{i,j} v_{\bm{\alpha}} v_{\bm{\beta}} \delta_{\bm{\alpha}+e_i^\lambda, \bm{\beta}+e_j^\mu}.
    \end{align}
    In shifts of $\widetilde{\alpha}^\lambda$ below, $e_i$ denotes the $i$th coordinate unit vector, except that $e_{m_{\lambda_*}}=0$ in the $\lambda_*$ factor because its last row is dependent.
    Here, we have
    \begin{align}
        \sum_{\bm{\alpha}, \bm{\beta}} \sum_{i,j} v_{\bm{\alpha}} v_{\bm{\beta}} \delta_{\bm{\alpha}+e_i^\lambda, \bm{\beta}+e_j^\mu}
        &= \sum_{\bm{\alpha}} \sum_{i,j} v_{\bm{\alpha}} v_{\bm{\alpha}+e_i^\lambda-e_j^\mu}\\
        &=\begin{cases}
           \sum_{\widetilde{\alpha}^\lambda} \sum_{i,j} v_{\widetilde{\alpha}^\lambda}^{\lambda} v_{\widetilde{\alpha}^\lambda+e_i-e_j}^{\lambda} & (\lambda = \mu)\\
           \sum_{\widetilde{\alpha}^\lambda, \widetilde{\alpha}^\mu} \sum_{i,j} v_{\widetilde{\alpha}^\lambda}^{\lambda} v_{\widetilde{\alpha}^\lambda+e_i}^{\lambda} v_{\widetilde{\alpha}^\mu}^{\mu} v_{\widetilde{\alpha}^\mu-e_j}^{\mu} & (\lambda \neq \mu).
        \end{cases}
    \end{align}
    We can evaluate the above summation as
    \begin{align}
        \sum_{\widetilde{\alpha}^\lambda} v_{\widetilde{\alpha}^\lambda}^{\lambda} v^{\lambda}_{\widetilde{\alpha}^\lambda+e_i-e_j}
        &=
        \begin{cases}
        1 & (i=j)\\
        \qty(\sum_{k=0}^{\lfloor n^{2/3}\rfloor-2} \sqrt{g_k g_{k+1}})^2 & (i\neq j,\ \lambda\neq \lambda_* \text{ or } i,j\neq m_{\lambda_*})\\
        \sum_{k=0}^{\lfloor n^{2/3}\rfloor-2} \sqrt{g_k g_{k+1}} & (i\neq j,\ \lambda = \lambda_*,\ i=m_{\lambda_*} \text{ or } j=m_{\lambda_*}).
        \end{cases}
    \end{align}
    Since
    \begin{align}
        1-\sum_{k=0}^{\lfloor n^{2/3}\rfloor-2} \sqrt{g_k g_{k+1}}
        &= 1-{2\over \lfloor n^{2/3}\rfloor} \sum_{k=0}^{\lfloor n^{2/3}\rfloor-2} \sin\qty({\pi(2k+1)\over 2\lfloor n^{2/3}\rfloor}) \sin\qty({\pi(2k+3)\over 2\lfloor n^{2/3}\rfloor})\\
        &= 1-{1\over \lfloor n^{2/3}\rfloor} \sum_{k=0}^{\lfloor n^{2/3}\rfloor-2} \qty[\cos\qty({\pi\over \lfloor n^{2/3}\rfloor}) - \cos\qty({\pi(2k+2)\over \lfloor n^{2/3}\rfloor})]\\
        &= \qty(1-{1\over \lfloor n^{2/3}\rfloor})\qty(1-\cos\qty({\pi\over \lfloor n^{2/3}\rfloor}))\\
        &= O(n^{-4/3})
    \end{align}
    holds, we have
    \begin{align}
        \sum_{\bm{\alpha}, \bm{\beta}} \sum_{i,j} v_{\bm{\alpha}} v_{\bm{\beta}} \delta_{\bm{\alpha}+e_i^\lambda, \bm{\beta}+e_j^\mu}
        =m_\lambda m_\mu-O(n^{-4/3}).
    \end{align}
    Using these bounds, we have
    \begin{align}
        \bm{v}^\dagger E_n \bm{v}
        &\geq {1\over d_{\mcI}^2} \sum_{\lambda, \mu\in \Lambda} d_\lambda d_\mu \underline{f}_\lambda \underline{f}_\mu \left(m_\lambda m_\mu-O(n^{-4/3})\right)\\
        &= {1\over d_{\mcI}^2} \qty(\sum_{\lambda\in \Lambda} d_\lambda m_\lambda \underline{f}_\lambda)^2 + O(n^{-4/3})\\
        &= 1-{1\over d_\mcI}\sum_{\lambda\in\Lambda}{d_\lambda m_\lambda^2(M_\lambda-m_\lambda) \over q_\lambda n} + O(n^{-4/3}).
    \end{align}
    When $s_G>0$, Eq.~\eqref{eq:learning-nondegenerate-occupation-weights} provides
    \begin{align}
        \sum_{\lambda\in\Lambda}
        {d_\lambda m_\lambda^2(M_\lambda-m_\lambda)\over q_\lambda}
        ={s_G^2\over1-\delta_n}
        =s_G^2+O(n^{-1/3}).
        \label{eq:learning-occupation-weight-optimization}
    \end{align}
    When $s_G=0$, the same sum is identically zero since $M_\lambda=m_\lambda$ holds for every $\lambda\in\Lambda$.
    Therefore Eq.~\eqref{eq:learning-occupation-weight-optimization} yields
    \begin{align}
        \bm{v}^\dagger E_n \bm{v}
        &\geq 1-{s_G^2\over d_\mcI n} + O(n^{-4/3})\\
        &= 1-{C_G\over n} + O(n^{-4/3}),
    \end{align}
    which concludes the proof.
\end{proof}

\section{Lower bound for covariant-state tomography}
\label{app:state-tomography-lower}

We prove the minimax lower bound in Theorem~\ref{thm:state-learning-optimal}, allowing arbitrary collective measurements on the input copies.
All sector sums in this appendix range over the isotypic components for which $m_\lambda>0$.
Throughout this appendix $L$ denotes the number of irreducible sectors with $m_\lambda>0$, and $D_{\mathrm{st}}\coloneqq\sum_\lambda m_\lambda^2$.
In the representation basis, a covariant density operator and its compressed multiplicity state have the forms
\begin{equation}
\rho=\bigoplus_\lambda\frac{\1_{\mcR_\lambda}}{d_\lambda}\otimes\rho_\lambda,
\qquad \sigma=\bigoplus_\lambda\rho_\lambda,
\qquad \rho_\lambda\succeq0,\quad\sum_\lambda\Tr\rho_\lambda=1.
\end{equation}
Measuring the sector and discarding its irrep register maps $\rho$ to $\sigma$, and appending the maximally mixed irrep state on each sector gives the inverse map.
Both maps are quantum channels independent of the unknown state, and they preserve trace distances within the model.
Moreover, applying the compression channel $\mathcal C$ to any estimate $\widehat\rho$ gives an estimate $\widehat\sigma=\mathcal C(\widehat\rho)$ satisfying $d_{\mathrm{tr}}(\widehat\sigma,\sigma)\leq d_{\mathrm{tr}}(\widehat\rho,\rho)$ by contractivity.
It therefore suffices to establish the lower bound for the direct-sum states $\sigma$.
We adapt the Holevo--Fano packing method for unrestricted state tomography~\cite[Sec.~VI]{haah2017sample} to the direct-sum model.
We first show that if a sufficiently large family of states is well separated in trace distance while remaining close in relative entropy, then any learning protocol that succeeds with constant probability must use many copies. 
We then construct such a family by varying the quantum state within the multiplicity spaces.
We also obtain a complementary lower bound by restricting to states that vary only in the probabilities assigned to the irreducible sectors and invoking the classical minimax lower bound for distribution estimation.
Together, these two arguments account for all independent degrees of freedom of the covariant state family.
Finally, we obtain the dependence on the failure probability by reducing state learning to binary hypothesis testing between two nearby states.

\paragraph{Obtaining sample lower bounds from packings.}\label{par:packing-lowerbound}
Let $\mcZ$ be a finite index set, and let $\{\sigma_z:z\in\mcZ\}$ be a collection of states with pairwise trace distance at least $3\varepsilon$.
Suppose there exists a density operator $\tau$ such that
\begin{equation}
D_{\mathrm{rel}}(\sigma_z\Vert\tau)
\coloneqq\Tr[\sigma_z(\log\sigma_z-\log\tau)]
\leq c\varepsilon^2
\qquad\text{for every }z\in\mcZ.
\label{eq:state-lower-relative-entropy}
\end{equation}
Choose an index $Z$ uniformly at random from $\mcZ$, prepare $n$ copies of the corresponding state $\sigma_Z$, and perform an arbitrary collective measurement on them.
Let the outcome of this measurement be the random variable $Y$.
Define the average $n$-copy state by
\begin{equation}
\bar\rho\coloneqq\frac{1}{|\mcZ|}\sum_{z\in\mcZ}\sigma_z^{\otimes n}.
\end{equation}
The Holevo bound, the relative-entropy decomposition with respect to $\tau^{\otimes n}$, and additivity of relative entropy give
\begin{align}
I(Z:Y)
&\leq \frac{1}{|\mcZ|}\sum_{z\in\mcZ}D_{\mathrm{rel}}(\sigma_z^{\otimes n}\Vert\bar\rho) \\
&= \frac{1}{|\mcZ|}\sum_{z\in\mcZ}D_{\mathrm{rel}}(\sigma_z^{\otimes n}\Vert\tau^{\otimes n})-D_{\mathrm{rel}}(\bar\rho\Vert\tau^{\otimes n}) \\
&\leq \frac{1}{|\mcZ|}\sum_{z\in\mcZ}D_{\mathrm{rel}}(\sigma_z^{\otimes n}\Vert\tau^{\otimes n}) \\
&= \frac{n}{|\mcZ|}\sum_{z\in\mcZ}D_{\mathrm{rel}}(\sigma_z\Vert\tau) \\
&\leq nc\varepsilon^2.
\label{eq:state-lower-holevo}
\end{align}
As the states $\{\sigma_z:z\in\mcZ\}$ are separated by at least $3\varepsilon$, any estimate $\widehat\sigma$ satisfying $d_{\mathrm{tr}}(\widehat\sigma,\sigma_Z)\leq\varepsilon$ determines the index $Z$ uniquely.
Suppose we have an algorithm producing such an estimate with probability at least $2/3$. Then, by Fano's inequality,
\begin{equation}
nc\varepsilon^2\geq I(Z:Y)\geq\log |\mcZ|-h_2(1/3)-\frac13\log(|\mcZ|-1)\geq\frac1{20}\log |\mcZ|,
\qquad |\mcZ|\geq2,
\label{eq:state-lower-fano}
\end{equation}
where $h_2$ denotes binary entropy with natural logarithms.
The last inequality always holds for $|\mcZ|\ge 2$.
Thus, to prove a lower bound on the sample complexity, it suffices to construct a family $\{\sigma_z:z\in\mcZ\}$ satisfying the separation and relative-entropy conditions above, since any state-learning protocol that succeeds with probability at least $2/3$ must then use
\begin{equation}
n=\Omega\!\left(\frac{\log |\mcZ|}{\varepsilon^2}\right)
\end{equation}
samples.
\paragraph{Quantum degrees of freedom in the multiplicity spaces.}
For each sector with $m_\lambda\geq2$, put $s_\lambda\coloneqq\lfloor m_\lambda/2\rfloor$ and restrict its multiplicity state to a fixed $2s_\lambda$-dimensional subspace.
Let
\begin{equation}
\Lambda_*\coloneqq\{\lambda:m_\lambda\geq2\},
\qquad
R_*\coloneqq\sum_{\lambda\in\Lambda_*}s_\lambda^2.
\end{equation}
In this paragraph we focus on the case in which $R_*>0$.
\begin{lemma}[Haar-unitary packing]
\label{lem:state-lower-haar-packing}
There exists a universal constant $c_{\mathrm H}>0$ and a set
\begin{equation}
\mcC\subseteq\bigtimes_{\lambda\in\Lambda_*}\mathrm U(s_\lambda)
\end{equation}
such that $\log|\mcC|\geq c_{\mathrm H}R_*$ and, for all distinct tuples $U=(U_\lambda)_{\lambda\in\Lambda_*}$ and $V=(V_\lambda)_{\lambda\in\Lambda_*}$ in $\mcC$,
\begin{equation}
\sum_{\lambda\in\Lambda_*}s_\lambda\|U_\lambda-V_\lambda\|_{\rm HS}^2\geq R_*.
\label{eq:state-lower-haar-packing}
\end{equation}
\end{lemma}
\begin{proof}
For a Haar unitary $W\in\mathrm U(s)$, Schur--Weyl duality gives
\begin{equation}
\mathbb E|\Tr W|^{2k}
=
\dim\operatorname{End}((\CC^s)^{\otimes k})^{\mathrm U(s)}
=
\sum_{\mu\vdash k,\,\ell(\mu)\leq s}(f^\mu)^2
\leq k!,
\end{equation}
where $f^\mu$ is the dimension of the irreducible $\mfS_k$-representation indexed by $\mu$ and we used $\sum_{\mu\vdash k}(f^\mu)^2=k!$.
Since $W$ is Haar-random, the random variable $\operatorname{Re}\Tr W$ is symmetric about zero and hence has vanishing odd moments. We therefore obtain
\begin{equation}\label{eq: gen fun Haar}
\mathbb E\,e^{t\operatorname{Re}\Tr W}
=
\sum_{k\geq0}\frac{t^{2k}}{(2k)!}\mathbb E(\operatorname{Re}\Tr W)^{2k}
\leq
\sum_{k\geq0}\frac{t^{2k}}{(2k)!}\mathbb E|\Tr W|^{2k}
\leq
\sum_{k\geq0}\frac{k!t^{2k}}{(2k)!}
\leq e^{t^2}.
\end{equation}
For two independent Haar tuples $U=(U_\lambda)_{\lambda\in\Lambda_*}$ and $V=(V_\lambda)_{\lambda\in\Lambda_*}$, define
\begin{equation}
S\coloneqq\sum_{\lambda\in\Lambda_*} s_\lambda\operatorname{Re}\Tr(U_\lambda^\dagger V_\lambda),
\qquad
Y\coloneqq\sum_{\lambda\in\Lambda_*} s_\lambda\|U_\lambda-V_\lambda\|_{\rm HS}^2.
\end{equation}
Since
\begin{equation}
\|U_\lambda-V_\lambda\|_{\rm HS}^2
=
2s_\lambda-2\operatorname{Re}\Tr(U_\lambda^\dagger V_\lambda),
\end{equation}
we have
\begin{equation}
Y=2R_*-2S.
\end{equation}
The matrices $U_\lambda^\dagger V_\lambda$ are independent Haar unitaries, so Eq.~\eqref{eq: gen fun Haar} gives
\begin{equation}
\mathbb E e^{tS}
=
\prod_{\lambda\in\Lambda_*}
\mathbb E e^{t s_\lambda\operatorname{Re}\Tr(U_\lambda^\dagger V_\lambda)}
\leq
\exp\!\left(t^2\sum_{\lambda\in\Lambda_*}s_\lambda^2\right)
=
e^{t^2R_*}.
\end{equation}
Therefore, by Chernoff's bound,
\begin{align}
\Prob\{Y<R_*\}
&=
\Prob\{S>R_*/2\}\\
&\leq
\inf_{t>0}\exp\!\left(-\frac{tR_*}{2}+t^2R_*\right)\\
&\leq e^{-R_*/16},
\end{align}
where we have taken $t=1/4$ in the last inequality.

If $R_*\geq128\log2$, draw
\begin{equation}
M\coloneqq\left\lfloor e^{R_*/64}\right\rfloor
\end{equation}
independent Haar tuples.
The probability that at least one pair violates $Y\geq R_*$ is at most
\begin{equation}
\binom{M}{2}e^{-R_*/16}<1,
\end{equation}
so there exists a set of $M$ tuples satisfying the required pairwise separation.
Moreover,
\begin{equation}
\log M\geq \frac{R_*}{128}.
\end{equation}
If $R_*<128\log2$, the two tuples defined by $U_\lambda=\1$ and $V_\lambda=-\1$ satisfy
\begin{equation}
\sum_{\lambda\in\Lambda_*}s_\lambda\|U_\lambda-V_\lambda\|_{\rm HS}^2
=
4R_*\geq R_*,
\end{equation}
and since $\log2>R_*/128$, the same constant $c_{\mathrm H}=1/128$ applies.
\end{proof}

Let $p_\lambda\coloneqq\frac{s_\lambda^2}{R_*}$. For $U\in\mcC$, define
\begin{equation}
H_{U,\lambda}\coloneqq
\begin{pmatrix}
0&U_\lambda\\
U_\lambda^\dagger&0
\end{pmatrix},
\qquad
\sigma_{U,\lambda}\coloneqq\frac{p_\lambda}{2s_\lambda}(\1_{2s_\lambda}+\delta H_{U,\lambda}),
\qquad0<\delta\leq\frac12.
\end{equation}
Since $H_{U,\lambda}^2=\1$ and $\Tr H_{U,\lambda}=0$, each block is positive and has trace $p_\lambda$ on the chosen $2s_\lambda$-dimensional subspace of $\mcM_\lambda$.
Extending each block by zero on the orthogonal complement and taking the direct sum over all sectors gives a density operator on the full compressed multiplicity space.
We now write
\begin{align}
d_{\mathrm{tr}}(\sigma_U,\sigma_V)
&=
\frac12\sum_{\lambda\in\Lambda_*}
\|\sigma_{U,\lambda}-\sigma_{V,\lambda}\|_1\\
&=
\frac{\delta}{4}
\sum_{\lambda\in\Lambda_*}
\frac{p_\lambda}{s_\lambda}
\|H_{U,\lambda}-H_{V,\lambda}\|_1\\
&=
\frac{\delta}{2}
\sum_{\lambda\in\Lambda_*}
\frac{p_\lambda}{s_\lambda}
\|U_\lambda-V_\lambda\|_1\\
&=
\frac{\delta}{2R_*}
\sum_{\lambda\in\Lambda_*}
s_\lambda\|U_\lambda-V_\lambda\|_1\\
&\geq
\frac{\delta}{4R_*}
\sum_{\lambda\in\Lambda_*}
s_\lambda\|U_\lambda-V_\lambda\|_{\rm HS}^2\\
&\geq
\frac{\delta}{4}.
\label{eq:state-lower-quantum-separation}
\end{align}
The first inequality follows from $\|X\|_{\rm HS}^2\leq\|X\|_\infty\|X\|_1$ and $\|U_\lambda-V_\lambda\|_\infty\leq2$, and the last uses Eq.~\eqref{eq:state-lower-haar-packing}.

To apply the argument in Paragraph~\ref{par:packing-lowerbound}, it remains to show that all states in the family are close in relative entropy to a single fixed state.
To this end, define $\tau$ by the blocks $\tau_\lambda\coloneqq p_\lambda\1_{2s_\lambda}/(2s_\lambda)$ on the chosen $2s_\lambda$-dimensional subspaces, extended by zero on their orthogonal complements.
Since $\sigma_{U,\lambda}=\tau_\lambda(\1+\delta H_{U,\lambda})$ and $\tau_\lambda$ commutes with $H_{U,\lambda}$, we have
\begin{align}
D_{\mathrm{rel}}(\sigma_U\Vert\tau)
&=
\sum_{\lambda\in\Lambda_*}
\frac{p_\lambda}{2s_\lambda}
\Tr\!\left[
(\1+\delta H_{U,\lambda})
\log(\1+\delta H_{U,\lambda})
\right]\\
&\leq
\sum_{\lambda\in\Lambda_*}
\frac{p_\lambda}{2s_\lambda}
\left(
\delta\Tr H_{U,\lambda}
+
\delta^2\Tr H_{U,\lambda}^2
\right)\\
&=
\delta^2
\sum_{\lambda\in\Lambda_*}
\frac{p_\lambda}{2s_\lambda}
\Tr H_{U,\lambda}^2\\
&=
\delta^2
\sum_{\lambda\in\Lambda_*}p_\lambda\\
&=
\delta^2.
\end{align}
Here we used $\log(1+x)\leq x$, $\Tr H_{U,\lambda}=0$, and $H_{U,\lambda}^2=\1_{2s_\lambda}$.
Choosing $\delta\coloneqq12\varepsilon$ gives pairwise trace distance at least $3\varepsilon$ by Eq.~\eqref{eq:state-lower-quantum-separation}, while the bound above gives $D_{\mathrm{rel}}(\sigma_U\Vert\tau)\leq144\varepsilon^2$.
Therefore, for $0<\varepsilon\leq1/24$, Paragraph~\ref{par:packing-lowerbound} and $\log|\mcC|\geq c_{\mathrm H}R_*$ imply
\begin{equation}
n=\Omega\!\left(\frac{R_*}{\varepsilon^2}\right).
\end{equation}

\paragraph{Classical degrees of freedom in the sector probabilities.}
The preceding construction captures the matrix degrees of freedom when $R_*>0$.
We now obtain a separate lower bound from the classical degrees of freedom given by determining a distribution over the set of isotypic components of the $G$-representation; this argument also covers the multiplicity-free case $R_*=0$, in which every isotypic component has $m_\lambda=1$.
Assume $L\geq2$, label the isotypic components by $1,\ldots,L$, fix one unit vector in each multiplicity space and consider the subfamily obtained by varying only the probabilities $p=(p_1,\ldots,p_L)$ of the resulting orthogonal sector states.
For two states in this subfamily,
\begin{equation}
d_{\mathrm{tr}}(\sigma_p,\sigma_q)
=
\frac12\|p-q\|_1,
\end{equation}
so learning these states in trace distance is equivalent to learning an unknown distribution on $\{1,2,\dots,L\}$ in total variation distance.
This also holds for quantum estimates not in the subfamily:
measuring the sector of $\widehat\sigma$ produces a distribution
$\widehat p$ with
$d_{\mathrm{TV}}(\widehat p,p)
\leq d_{\mathrm{tr}}(\widehat\sigma,\sigma_p)$
by contractivity.
Moreover, since all states in this subfamily commute in a known basis, any collective measurement on $n$ copies is a classical post-processing of $n$ independent samples from the corresponding distribution.
The standard Assouad lower-bound argument mentioned after Theorem~1 of Ref.~\cite{canonne2020learning} yields universal constants $c_{\mathrm{cl}},\varepsilon_{\mathrm{cl}}>0$ such that, for every $L\geq2$ and $0<\varepsilon\leq\varepsilon_{\mathrm{cl}}$, learning an arbitrary distribution on $L$ outcomes with success probability at least $2/3$ requires
\begin{equation}
n\geq c_{\mathrm{cl}}\frac{L-1}{\varepsilon^2}.
\end{equation}

\paragraph{Combining the two lower bounds}
The two preceding arguments control complementary parts of the model.
The first varies the state within the multiplicity spaces and yields a lower bound proportional to the number of internal matrix degrees of freedom, while the second varies the probabilities assigned to the irreducible sectors and yields a lower bound proportional to the number of distinct labels.
Together these account for all degrees of freedom of a covariant state.
Indeed,
\begin{align}
D_{\mathrm{st}}-1
&=
(L-1)+\sum_\lambda (m_\lambda^2-1)
\\
&\leq
(L-1)+8R_*
\\
&\leq
9\max\{L-1,R_*\},
\end{align}
where in the first inequality we have used $m^2-1 \leq 8 \lfloor m/2 \rfloor ^2$ for all integers $m\ge 1$.
Since any learning procedure that succeeds uniformly over the whole model must in particular succeed on both subfamilies, the two lower bounds combine to give
\begin{equation}
n=\Omega\!\left(\frac{D_{\mathrm{st}}-1}{\varepsilon^2}\right).
\end{equation}

\paragraph{Dependence on the failure probability.}
The $\Omega(\log(1/\eta)/\varepsilon^2)$ term follows from a standard biased-coin testing reduction~\cite[Sec.~1]{canonne2020learning}.
If $D_{\mathrm{st}}\geq2$, the model contains two orthogonal states; let $\sigma_+$ and $\sigma_-$ be the mixtures of these two states with respective weights $(1/2+2\varepsilon,1/2-2\varepsilon)$ and $(1/2-2\varepsilon,1/2+2\varepsilon)$.
They satisfy
\begin{equation}
d_{\mathrm{tr}}(\sigma_+,\sigma_-)=4\varepsilon,
\qquad
D_{\mathrm{rel}}(\sigma_+\Vert\sigma_-)
=
4\varepsilon\log\frac{1+4\varepsilon}{1-4\varepsilon}
\leq64\varepsilon^2
\end{equation}
for $\varepsilon\leq1/16$.
Hence any estimator with trace error at most $\varepsilon$ and failure probability at most $\eta$ yields a binary test between $\sigma_+$ and $\sigma_-$ with both error probabilities at most $\eta$.

By additivity of relative entropy and the bound $D_{\mathrm{rel}}(\sigma_+\Vert\sigma_-)\leq64\varepsilon^2$,
\begin{equation}
D_{\mathrm{rel}}(\sigma_+^{\otimes n}\Vert\sigma_-^{\otimes n})
=
nD_{\mathrm{rel}}(\sigma_+\Vert\sigma_-)
\leq
64n\varepsilon^2.
\end{equation}
Let $a$ be the probability that the binary test outputs $+$ when the input is $\sigma_+^{\otimes n}$, and let $b$ be the probability that it outputs $+$ when the input is $\sigma_-^{\otimes n}$.
Since both error probabilities are at most $\eta$, we have
\begin{equation}
a\geq1-\eta,
\qquad
b\leq\eta.
\end{equation}
Applying the data-processing inequality for relative entropy to this binary test gives
\begin{align}
64n\varepsilon^2
&\geq D_{\mathrm{rel}}(\sigma_+^{\otimes n}\Vert\sigma_-^{\otimes n})
\\
&\geq a\log\frac{a}{b}
+(1-a)\log\frac{1-a}{1-b}
\\
&\geq (1-\eta)\log\frac{1-\eta}{\eta}
+\eta\log\frac{\eta}{1-\eta}
\\
&= (1-2\eta)\log\frac{1-\eta}{\eta}
\\
&\geq \frac15\log(1/\eta),
\qquad 0<\eta\leq\frac13.
\end{align}
Therefore,
\begin{equation}
n=\Omega\!\left(\frac{\log(1/\eta)}{\varepsilon^2}\right).
\end{equation}
Combining this with the previously established bound
\begin{equation}
n=\Omega\!\left(\frac{D_{\mathrm{st}}-1}{\varepsilon^2}\right),
\end{equation}
and using $\max\{x,y\}\geq (x+y)/2$, we conclude that there exist universal constants
$c_0,\varepsilon_0>0$ such that
\begin{equation}
n\geq
c_0\,\frac{D_{\mathrm{st}}-1+\log(1/\eta)}{\varepsilon^2},
\qquad
0<\varepsilon\leq\varepsilon_0,\quad
0<\eta\leq\frac13,\quad
D_{\mathrm{st}}\geq2.
\end{equation}
This proves the claimed minimax lower bound.

\section{Lower bound for covariant-channel tomography}
\label{appendix:lower-bound-channel}
\subsection{Proof of Thm.~\ref{thm:permutation-channel-strong-diamond-lower}: Learning $G$-covariant channels in diamond distance}

We construct a hard family of $G$-covariant channels $\{\Lambda_W\}_{W\in \W_G}$ to show Thm.~\ref{thm:permutation-channel-strong-diamond-lower}, where $\Lambda_W$ is defined by $\Lambda_W(\cdot)\coloneqq\Tr_{\mcE}[W(\cdot)W^\dagger]$.
We define the environment gauge group $\U_G(\mcE)$ by
\begin{align}
    \U_G(\mcE) \coloneqq \{V\in \U(\mcE) \mid [V, U_1(g)\otimes \overline{U_2}(g)] = 0 \quad \forall g\in G\}.
\end{align}
Using the isotypic decomposition of $U_1\otimes \overline{U_2}$ given by
\begin{align}
    U_1(g)\otimes \overline{U_2}(g) \cong \bigoplus_{\lambda\in\irrep{G}} U_\lambda(g)\otimes \1_{\mu_\lambda}
\end{align}
using the multiplicity $\mu_\lambda$, we have
\begin{align}
    \U_G(\mcE) = \{V = \bigoplus_\lambda \1_{d_\lambda}\otimes V_\lambda \mid V_\lambda\in\U(\mu_\lambda)\} \cong \prod_\lambda \U(\mu_\lambda).
\end{align}
Since $V\otimes \1_{\mcO}$ with $V\in \U_G(\mcE)$ satisfies $V\otimes \1_{\mcO} \in \U(\mcE\otimes \mcO) \cap \Comm(U_1\otimes \overline{U}_2\otimes U_2) \cong \prod_\lambda \U(M_\lambda)$, we can define a homomorphism
\begin{align}
    \phi: \U_G(\mcE)\ni V \mapsto (\phi_\lambda(V))_{\lambda} \in \prod_{\lambda\in\irrep{G}} \U(M_\lambda)
\end{align}
such that
\begin{align}
    (V\otimes \1_{\mcO})W \cong \bigoplus_\lambda \1_{d_\lambda}\otimes \phi_\lambda(V) W_\lambda.
\end{align}
This homomorphism defines the action of $\U_G(\mcE)$ on $\prod_\lambda \W_{M_\lambda, m_\lambda}$ given by $V\cdot (W_\lambda)_\lambda \coloneqq (\phi_\lambda(V)W_\lambda)_\lambda$, and we define the quotient space
\begin{align}
     \left(\prod_{\lambda\in\irrep{G}}\W_{M_\lambda,m_\lambda}\right)\big/\phi(\U_G(\mcE)),
\end{align}
whose elements are equivalence classes of the form $[(W_\lambda)_{\lambda}] \coloneqq \{(\phi_\lambda(V)W_\lambda)_{\lambda} \mid V\in \U_G(\mcE)\}$.
We also define the equivalence relation in $\W_G$ by $W\sim W'$ if and only if there exists $V\in \U_G(\mcE)$ such that $W' = (V\otimes \1_\mcO) W$, and we denote the equivalence class by $[W] \coloneqq \{(V\otimes \1_\mcO) W \mid V\in \U_G(\mcE)\}$.
Then, we have the bijection
\begin{align}
    \W_G/\sim \cong \left(\prod_{\lambda\in\irrep{G}}\W_{M_\lambda,m_\lambda}\right)\big/\phi(\U_G(\mcE)).
\end{align}
We introduce the corresponding distance by
\begin{align}
    d_\mathrm{Stin}([(W_\lambda)_{\lambda}], [(W'_\lambda)_{\lambda}])
    &\coloneqq \inf_{V\in \U_G(\mcE)} \max_\lambda \left\| W_\lambda -\phi_\lambda(V) W'_\lambda \right\|_\infty,\\
    d_\mathrm{Stin}([W], [W'])&\coloneqq \inf_{V\in \U_G(\mcE)} \left\| W-(V\otimes \1_\mcO) W' \right\|_\infty.
\end{align}
Then, we have $d_\mathrm{Stin}([(W_\lambda)_{\lambda}], [(W'_\lambda)_{\lambda}]) = d_\mathrm{Stin}([W], [W'])$ for $W, W'\in \W_G$ given by $W = \bigoplus_\lambda \1_{d_\lambda} \otimes  W_\lambda$ and $W' = \bigoplus_\lambda \1_{d_\lambda} \otimes W'_\lambda$.
Then, we have the following generalization of the Kretschmann--Schlingemann--Werner theorem~\cite{kretschmann2008information,singh2026short} for $G$-covariant channels.

\begin{lemma}[Kretschmann--Schlingemann--Werner theorem for $G$-covariant channels]
    \label{lem:ksw-covariant}
    For any $G$-covariant channels $\Lambda_1,\Lambda_2\in\CovChan_G(\mcI, \mcO)$ and their $G$-covariant Stinespring isometries $W_1\in\Dil_G(\Lambda_1)$ and $W_2\in\Dil_G(\Lambda_2)$ with the common environment $\mcE$,
    \begin{align}
        d_\mathrm{Stin}([W_1],[W_2])^2
        \leq4d_\diamond(\Lambda_1,\Lambda_2).
    \end{align}
\end{lemma}
\begin{proof}
    We follow the proof strategy of Ref.~\cite{singh2026short}.
    We define the operational root channel fidelity by
    \begin{align}
        f_{\mathrm{op}}(\Lambda_1,\Lambda_2)
        \coloneqq\min_{\substack{\ket{\psi}\in \overline{\mcI}\otimes\mcI\\\|\psi\|_2=1}}
        \left\|
        \sqrt{(\1_{\mcL(\overline{\mcI})}\otimes\Lambda_1)(\ketbra{\psi})}
        \sqrt{(\1_{\mcL(\overline{\mcI})}\otimes\Lambda_2)(\ketbra{\psi})}
        \right\|_1,
    \end{align}
    which is related to the diamond distance by the Fuchs--van de Graaf inequalities~\cite{fuchs2002cryptographic} as
    \begin{align}
        \label{eq:fuchs-van-de-graaf-diamond}
        1-f_{\mathrm{op}}(\Lambda_1,\Lambda_2)
        \leq d_\diamond(\Lambda_1,\Lambda_2)
        \leq\sqrt{1-f_{\mathrm{op}}(\Lambda_1,\Lambda_2)^2}.
    \end{align}
    Defining 
    \begin{align}
        \Gamma(C)\coloneqq\lambda_{\min}(W_1^\dagger(C\otimes\1_{\mcO})W_2
        +W_2^\dagger(C^\dagger\otimes\1_{\mcO})W_1)
    \end{align}
    for $C\in \mcL(\mcE)$, we can express the operational root channel fidelity as~\cite{vomende2023progress}
    \begin{align}
        f_{\mathrm{op}}(\Lambda_1,\Lambda_2)
        ={1\over2}\max_{\substack{C\in\mcL(\mcE)\\\|C\|_\infty\leq1}}\Gamma(C).
    \end{align}
    Since the twirling of $C$ given by
    \begin{align}
        C\mapsto \int_G \dd g [U_1(g)\otimes \overline{U_2}(g)]^\dagger C [U_1(g)\otimes \overline{U_2}(g)]
    \end{align}
    does not decrease $\Gamma(C)$, we can restrict the optimization to $C\in\Comm(U_1\otimes \overline{U}_2)$, i.e.,
    \begin{align}
        f_{\mathrm{op}}(\Lambda_1,\Lambda_2)
        ={1\over2}\max_{\substack{C\in\Comm(U_1\otimes \overline{U}_2)\\\|C\|_\infty\leq1}}\Gamma(C).
    \end{align}
    We then show the following inequality, which generalizes \cite[Lem.~3]{singh2026short}:
    \begin{align}
        \label{eq:singh-lem3-generalization}
        2\max_{\substack{C\in\Comm(U_1\otimes \overline{U}_2)\\\|C\|_\infty\leq1}}\Gamma(C)
        \leq2+\max_{V\in \U_G(\mcE)}\Gamma(V)
    \end{align}
    as follows.
    The polar decomposition of $C\in\Comm(U_1\otimes \overline{U}_2)$ is given by $C=VP$, where $V\in \U_G(\mcE)$ and $P\in\Comm(U_1\otimes \overline{U}_2)$ satisfies $0\preceq P\preceq\1_{\mcE}$.
    Defining
    \begin{align}
        X_C\coloneqq W_1-(C\otimes\1_{\mcO})W_2,
        \qquad
        Y_C\coloneqq((V-C)\otimes\1_{\mcO})W_2,
    \end{align}
    we have $W_1-(V\otimes\1_{\mcO})W_2=X_C-Y_C$ and
    \begin{align}
        Y_C^\dagger Y_C
        =W_2^\dagger((\1_{\mcE}-P)^2\otimes\1_{\mcO})W_2
        \preceq W_2^\dagger((\1_{\mcE}-P^2)\otimes\1_{\mcO})W_2.
    \end{align}
    Since $C^\dagger C=P^2$, we also have
    \begin{align}
        X_C^\dagger X_C
        =\1_{\mcI}+W_2^\dagger(P^2\otimes\1_{\mcO})W_2-(W_1^\dagger(C\otimes\1_{\mcO})W_2
        +W_2^\dagger(C^\dagger\otimes\1_{\mcO})W_1).
    \end{align}
    Since~\cite[Lem.~2]{singh2026short}
    \begin{align}
        (X_C-Y_C)^\dagger(X_C-Y_C)
        \preceq2(X_C^\dagger X_C+Y_C^\dagger Y_C)
    \end{align}
    holds, we obtain
    \begin{align}
        \bigl(W_1-(V\otimes\1_{\mcO})W_2\bigr)^\dagger
        \bigl(W_1-(V\otimes\1_{\mcO})W_2\bigr)
        \preceq2\bigl(2\1_{\mcI}-(W_1^\dagger(C\otimes\1_{\mcO})W_2
        +W_2^\dagger(C^\dagger\otimes\1_{\mcO})W_1)\bigr).
    \end{align}
    By taking the largest eigenvalues, we have
    \begin{align}
        2-\Gamma(V)\leq2\bigl(2-\Gamma(C)\bigr),
    \end{align}
    which establishes Eq.~\eqref{eq:singh-lem3-generalization}.
    Thus, we have
    \begin{align}
        f_{\mathrm{op}}(\Lambda_1,\Lambda_2)
        \leq {1\over2}+{1\over4}\max_{V\in \U_G(\mcE)}\Gamma(V)
        =1-{1\over4}\min_{V\in \U_G(\mcE)}
        \left\|W_1-(V\otimes\1_{\mcO})W_2\right\|_\infty^2
        =1-{1\over4}d_\mathrm{Stin}([W_1],[W_2])^2,
    \end{align}
    and Eq.~\eqref{eq:fuchs-van-de-graaf-diamond} yields
    \begin{align}
        d_\diamond(\Lambda_1,\Lambda_2)
        \geq1-f_{\mathrm{op}}(\Lambda_1,\Lambda_2)
        \geq {1\over4}d_\mathrm{Stin}([W_1],[W_2])^2,
    \end{align}
    which concludes the proof.
\end{proof}

\begin{proof}[Proof of Thm.~\ref{thm:permutation-channel-strong-diamond-lower}]
We follow the proof strategy of Ref.~\cite{mele2025optimal}.
Suppose $N(\mcX,r)$ and $M(\mcX,r)$ denote the minimum size of an $r$-net and the maximum size of an $r$-packing, respectively, of a metric space $\mcX$, i.e.,
\begin{align}
    N(\mcX,r)&\coloneqq \min\{\abs{\mcN} \mid \forall x\in \mcX, \exists y\in \mcN \text{ such that } d(x,y)\leq r\},\\
    \label{eq:def-packing-number}
    M(\mcX,r)&\coloneqq \max\{\abs{\mcM} \mid \forall x,y\in \mcM, x\neq y \implies d(x,y)>r\},
\end{align}
where $d$ is the metric of $\mcX$. These packing and covering numbers satisfy
\begin{align}
    M(\mcX,2r)\leq N(\mcX,r)\leq M(\mcX,r).
\end{align}
A maximal $r$-packing of $\W_G/\sim$, together with an $r$-net of $\U_G(\mcE)$, induces a $2r$-net of $\W_G$.
Thus, as in Ref.~\cite[Lem.~IV.8]{mele2025optimal}, we have
\begin{align}
    M(\W_G/\sim,r)\geq\frac{N(\W_G,2r)}{N(\U_G(\mcE),r)}.
\end{align}
As shown in Ref.~\cite[Lem.~IV.7]{mele2025optimal} (shown using Refs.~\cite{szarek1997metric,aubrun2017alice}), we have the following bounds on $N(\mcX, r)$ for the Stiefel manifold $\W_{s,t}$ and the unitary group $\U(s)$ for $0<r\leq 1$:
\begin{align}
    N(\W_{s,t},r)&\geq(c/r)^{2st-t^2},\qquad N(\U(s),r)\leq(C/r)^{s^2},
\end{align}
using universal constants $c, C>0$ satisfying $C\geq c$.
Since the product of $4r$-packings in $\W_{M_\lambda, m_\lambda}$ for $\lambda\in \irrep{G}$ induces a $4r$-packing in $\W_G$, we have
\begin{align}
    N(\W_G,2r)\geq M(\W_G, 4r) \geq \prod_\lambda M(\W_{M_\lambda, m_\lambda}, 4r) \geq \prod_\lambda N(\W_{M_\lambda, m_\lambda}, 4r) \geq \qty(c \over 4r)^{\sum_\lambda(2m_\lambda M_\lambda-m_\lambda^2)}.
\end{align}
Since the product of $r$-nets in $\U(\mu_\lambda)$ for $\lambda\in\irrep{G}$ induces an $r$-net in $\U_G(\mcE)$, we have
\begin{align}
    N(\U_G(\mcE),r)\leq \prod_\lambda N(\U(\mu_\lambda),r) \leq \qty(C \over r)^{\sum_\lambda \mu_\lambda^2}.
\end{align}
The exponents are evaluated as follows:
\begin{align}
    \sum_{\lambda} (2m_\lambda M_\lambda -m_\lambda^2)
    &= D_G + D''_G,\\
    \sum_\lambda \mu_\lambda^2
    &= \sum_\lambda \abs{\langle \chi_{U_1\otimes \overline{U}_2}, \chi_\lambda\rangle_{G}}^2\\
    &= \langle \chi_{U_1\otimes \overline{U}_2}, \chi_{U_1\otimes \overline{U}_2}\rangle_{G}\\
    &= \langle \chi_{U_1}, \chi_{U_1\otimes \overline{U}_2 \otimes U_2}\rangle_{G}\\
    &= \sum_\lambda m_\lambda M_\lambda\\
    &= D_G,
\end{align}
where we use the orthogonality of the characters.
Thus, for $0<r\leq {1/4}$, we have
\begin{align}
    M(\W_G/\sim,r)\geq\frac{N(\W_G,2r)}{N(\U_G(\mcE),r)}\geq\qty(c\over 4r)^{\sum_\lambda(2m_\lambda M_\lambda-m_\lambda^2)}\qty(C \over r)^{-\sum_\lambda \mu_\lambda^2} \geq \qty(c_\kappa \over r)^{D''_G},
\end{align}
where $c_\kappa$ is defined by
\begin{align}
    c_\kappa\coloneqq \qty(c\over 4)^{1+1/\kappa} C^{-1/\kappa},
\end{align}
and we use $D_G/D''_G\leq 1/\kappa$ and $c/(4C) \leq 1$.
We choose the accuracy threshold in Thm.~\ref{thm:permutation-channel-strong-diamond-lower} as follows:
\begin{align}
    \varepsilon_0(\kappa)\coloneqq\min\left\{\frac{1}{144},\frac{c_\kappa^2}{8(3e)^{2/\kappa}}\right\}>0.
\end{align}
For $0<\varepsilon_\diamond \leq\varepsilon_0(\kappa)$, we take $r\coloneqq \sqrt{8\varepsilon_\diamond}\leq 1/4$.
Due to Lem.~\ref{lem:ksw-covariant}, an $r$-packing of $\W_G/\sim$ induces an $r^2/4=2\varepsilon_\diamond$-packing of $\CovChan_G(\mcI, \mcO)$.
An estimate of an element in the $2\varepsilon_\diamond$-packing with error at most $\varepsilon_\diamond$ identifies its packing element by nearest-neighbor decoding.
By invoking the discrimination bound in Ref.~\cite[Prop.~IV.4]{mele2025optimal}\footnote{
By Eq.~\eqref{eq:decomposition_W}, the covariant Stinespring isometries lie in a complex linear space of dimension $D_G=\sum_\lambda m_\lambda M_\lambda$.
Thus, the dimension argument in the proof of the cited proposition applies with $D=D_G$.
}, if we have a discrimination strategy with success probability at least $2/3$ on every packing element, then the number of queries $n$ must satisfy
\begin{align}
    {2\over 3} \leq {1\over M(\W_G/\sim,r)} \binom{n+D_G-1}{D_G-1} \leq \qty({\sqrt{8\varepsilon_\diamond}\over c_\kappa})^{D''_G} \qty(e(n+D_G-1)\over D_G-1)^{D_G-1},
\end{align}
i.e.,
\begin{align}
    n
    &\geq(D_G-1)\qty[\frac{(2/3)^{1/(D_G-1)}}{e}\qty(\frac{c_\kappa}{\sqrt{8\varepsilon_\diamond}})^{D''_G/(D_G-1)}-1]\\
    &\geq {D_G \over 6e} \qty(c_\kappa \over \sqrt{8})^\kappa \varepsilon_\diamond^{-\kappa/2},
\end{align}
where we use $D_G-1\geq D_G/2$, $D''_G/(D_G-1)\geq D''_G/D_G\geq \kappa$, $(2/3)^{1/(D_G-1)}\geq 2/3$, and $\left(\frac{c_\kappa}{\sqrt{8\varepsilon_\diamond}}\right)^\kappa\geq 3e$ due to $D_G\geq 2$ and the range of $\varepsilon_\diamond$.
\end{proof}

\subsection{Proof of Thm.~\ref{thm:permutation-channel-trace-matching-lower}: Learning permutation-covariant channels in Choi trace distance}

We construct a hard family of permutation-covariant channels to show Thm.~\ref{thm:permutation-channel-trace-matching-lower}, which satisfies the following entropic conditions:

\begin{lemma}[Conditions of hard family of quantum channels in Choi trace distance]
\label{lem:adaptive-geometric-renyi-packing}
Suppose $\{\Lambda_z\}_{z\in \mathsf{Z}}$ is a finite set of channels with common input and output spaces and $\abs{\mathsf{Z}}\geq2$, and let $\Lambda_\star$ be a reference channel on the same spaces.
We define the Choi operators $J_z\coloneqq J_{\Lambda_z}$ and $J_\star\coloneqq J_{\Lambda_\star}$, and assume that $J_\star\succ 0$.
For $\varepsilon_\mathrm{tr},\xi>0$, suppose this family of channels satisfies
\begin{align}
    \label{eq:trace-norm-bound}
    \min_{z\neq z'}d_\mathrm{tr}(\rho_z,\rho_{z'})&\geq 3\varepsilon_\mathrm{tr},\\
    \widehat{D}_2(\Lambda_z\|\Lambda_\star) &\leq \xi,
\end{align}
where $\rho_z$ is the normalized Choi operator of $\Lambda_z$ and $\widehat{D}_2$ is the geometric R\'enyi divergence of order 2~\cite{fang2021geometric} given by
\begin{align}
    \widehat{D}_2(\Lambda_z\|\Lambda_\star)\coloneqq \log\left\|\Tr_{\rm out}(J_zJ_\star^{-1}J_z)\right\|_\infty.
\end{align}
If a protocol given $n$ queries to the unknown $\Lambda_z$ outputs a channel $\widehat\Lambda$ satisfying
\begin{align}
    \Prob_z\qty[d_\mathrm{tr}(\rho_z,\rho_{\widehat{\Lambda}})\leq \varepsilon_\mathrm{tr}]\geq 1-\eta,
\end{align}
for $0\leq \eta\leq 1/2$, where $\rho_{\widehat{\Lambda}}$ is the normalized Choi operator of $\widehat{\Lambda}$, then we have
\begin{align}
    n\xi \geq \log \abs{\mathsf{Z}}-h_2(\eta)-\eta \log(\abs{\mathsf{Z}}-1),
\end{align}
where $h_2(\eta)\coloneqq -\eta \log \eta-(1-\eta)\log(1-\eta)$ is the binary entropy function.
\end{lemma}

\begin{proof}
    Suppose $z$ obeys the uniform distribution over $\mathsf{Z}$, $\rho_z^{(n)}$ and $\rho_\star^{(n)}$ are the final states produced by applying $\Lambda_z$ and $\Lambda_\star$ to the same $n$-slot quantum comb, respectively, and let $\widehat{\Lambda}$ be the estimate of $\Lambda_z$ obtained by measuring $\rho_z^{(n)}$.
    We define the estimate $\widehat{z}$ of $z$ by $\widehat{z}\in \argmin_{z'\in \mathsf{Z}} d_\mathrm{tr}(\rho_{\widehat{\Lambda}}, \rho_{z'})$.
    Then, by the assumptions of the protocol, we have
    \begin{align}
        \Prob_z\qty[\widehat{z}=z] \geq 1-\eta.
    \end{align}
    The chain rule for the geometric R\'enyi divergence and its comparison with the relative entropy~\cite[Thm.~3]{fang2021geometric} provide
    \begin{align}
        D_\mathrm{rel}(\rho_z^{(n)}\|\rho_\star^{(n)}) \leq \widehat{D}_2(\rho_z^{(n)}\|\rho_\star^{(n)}) \leq n \widehat{D}_2(\Lambda_z\|\Lambda_\star) \leq n\xi,
    \end{align}
    where $D_\mathrm{rel}$ is the quantum relative entropy.
    Then, the Holevo bound~\cite{holevo1973bounds,schumacher2002relative} provides
    \begin{align}
        I(z:\widehat{z}) \leq {1\over \abs{\mathsf{Z}}} \sum_z D_\mathrm{rel}(\rho_z^{(n)}\|\rho_\star^{(n)}) \leq n\xi,
    \end{align}
    where $I(z:\widehat{z})$ is the mutual information between $z$ and $\widehat{z}$.
    On the other hand, Fano's inequality~\cite{fano1961transmission} provides
    \begin{align}
        I(z:\widehat{z}) \geq \log \abs{\mathsf{Z}}-h_2(\eta)-\eta \log(\abs{\mathsf{Z}}-1),
    \end{align}
    which concludes the proof.
\end{proof}

To construct a hard family of permutation-covariant channels, we take a subset $\mathsf{Y}_{\mathrm{hard}} \subset \{\mu\vdash_d k\}$ of the input Young diagrams and consider the Choi space~\eqref{eq:choi-space-decomposition} where the input Young diagram belongs to $\mathsf{Y}_{\mathrm{hard}}$:
\begin{align}
    \bigoplus_{\mu\in \mathsf{Y}_{\mathrm{hard}}} \overline{\mcS_\mu} \otimes \CC^{m_\mu^{(d)}} \otimes (\CC^d)^{\otimes k}
    &\cong \bigoplus_{\nu\vdash k} \mcS_\nu \otimes \CC^{c_\nu},
\end{align}
where $c_\nu$ is defined by
\begin{align}
    c_{\nu} \coloneqq \sum_{\mu\in \mathsf{Y}_{\mathrm{hard}}} w_{\mu\nu}.
\end{align}
Then, we have the decomposition of the multiplicity space:
\begin{align}
    \label{eq:c_nu-decomposition}
    \CC^{c_\nu} \cong \bigoplus_{\mu\in \mathsf{Y}_{\mathrm{hard}}} \CC^{w_{\mu\nu}}.
\end{align}
We also define the dimension $r_{\mathsf{Y}_{\mathrm{hard}}}$ of the input Young sector and the dimension of the set of permutation-invariant operators supported on the input Young sector:
\begin{align}
    r_{\mathsf{Y}_{\mathrm{hard}}} \coloneqq \sum_{\mu\in \mathsf{Y}_{\mathrm{hard}}} d_\mu m_\mu^{(d)}, \quad D_{\mathsf{Y}_{\mathrm{hard}}} \coloneqq \sum_{\nu\vdash k} c_\nu^2.
\end{align}
Then, we can take $\mathsf{Y}_{\mathrm{hard}}$ satisfying the following properties:

\begin{lemma}[Input Young sectors for hard families]
\label{lem:permutation-central-window-recoupling}
For each fixed $d\geq2$, there exists $k_1(d)$ such that, for every $k\geq k_1(d)$, one can choose a set $\mathsf{Y}_{\mathrm{hard}}\subset\left\{\mu\vdash_d k\right\}$ with the following properties
\begin{align}
    \label{eq:r-C}
    {r_{\mathsf{Y}_{\mathrm{hard}}}\over d^k} = \Theta_d(1), \qquad {d_\mu m_\mu^{(d)}\over r_{\mathsf{Y}_{\mathrm{hard}}}}=\Theta_d(k^{-(d-1)/2}) \quad \forall \mu\in\mathsf{Y}_\mathrm{hard},\qquad
    D_{\mathsf{Y}_{\mathrm{hard}}} = \Theta_d(k^{d^4-d^2}),\\
    \label{eq:recoupling-uniformity}
    {\sum_{\mu\in\mathsf{Y}_{\mathrm{hard}}} \sum_{\nu\vdash k} w_{\mu\nu}^2 \over D_{\mathsf{Y}_{\mathrm{hard}}}} = O_d(k^{-(d-1)/2}),\qquad
    {\max_{\mu\in \mathsf{Y}_{\mathrm{hard}}} \sum_{\nu\vdash k} c_\nu w_{\mu\nu} \over D_{\mathsf{Y}_{\mathrm{hard}}}} = O_d(k^{-(d-1)/2}).
\end{align}
\end{lemma}

\begin{proof}
Defining
\begin{align}
\mathsf{X}_d\coloneqq
\left\{\bm x\in\mathbb R^d \; \middle|\;
\sum_{i=1}^d x_i=0,
\ \sum_{i=1}^d(x_i-x_i^\star)^2\leq{1\over8}\right\},
\end{align}
where $x_i^\star\coloneqq(d+1-2i)/2$, we define the set of Young diagrams $\mathsf{Y}_{\mathrm{hard}}$ by
\begin{align}
\label{eq:hard-young-set}
\mathsf{Y}_{\mathrm{hard}}\coloneqq
\left\{\mu\vdash_d k \;\middle|\; \mu_i = {k\over d} + \sqrt{k} x_i, \bm{x}\in \mathsf{X}_d \right\}.
\end{align}
The dimension estimates proved in Appendix~\ref{sec:uniform-fixed-height-young-estimates} provide
\begin{align}
    \label{eq:d-mu-mu-d}
    d_\mu = \Theta_d\left(d^k k^{-(d^2+d-2)/4}\right),
    \qquad
    m_\mu^{(d)} = \Theta_d(k^{d(d-1)/4}),
    \qquad
    m_\mu^{(d^3)} = \Theta_d\!\left(k^{d^4-d^2+d(d-1)/4}\right),
\end{align}
and we also have
\begin{align}
    \abs{\mathsf{Y}_{\mathrm{hard}}} = \Theta_d(k^{(d-1)/2}).
\end{align}
Thus, we have
\begin{align}
    r_{\mathsf{Y}_{\mathrm{hard}}} = \sum_{\mu\in \mathsf{Y}_{\mathrm{hard}}} d_\mu m_\mu^{(d)} = \Theta_d(d^k k^{-(d^2+d-2)/4})\cdot\Theta_d(k^{d(d-1)/4})\cdot \Theta_d(k^{(d-1)/2}) = \Theta_d(d^k).
\end{align}
This proves the first two bounds in Eq.~\eqref{eq:r-C}.
Since
\begin{align}
    \sum_{\nu\vdash k} c_\nu m_\nu^{(d^2)}
    &= \sum_{\mu\in\mathsf{Y}_{\mathrm{hard}}} m_\mu^{(d)} \sum_{\alpha, \nu\vdash k} m_\alpha^{(d)} g_{\mu\alpha\nu} m_\nu^{(d^2)}\\
    &= \sum_{\mu\in \mathsf{Y}_{\mathrm{hard}}} m_\mu^{(d)} m_\mu^{(d^3)}\\
    &= \Theta_d(k^{d^4-(d^2+1)/2})
\end{align}
holds, by applying the Cauchy--Schwarz inequality to $D_{\mathsf{Y}_{\mathrm{hard}}} = \sum_{\nu\vdash k} c_\nu^2$, we have
\begin{align}
    D_{\mathsf{Y}_{\mathrm{hard}}} \geq {\left(\sum_{\nu\vdash k} c_\nu m_\nu^{(d^2)}\right)^2 \over \sum_{\nu\vdash k} \left(m_\nu^{(d^2)}\right)^2} = {\left(\sum_{\nu\vdash k} c_\nu m_\nu^{(d^2)}\right)^2 \over \binom{k+d^4-1}{d^4-1}} = \Omega_d(k^{d^4-d^2}).
\end{align}
    Since $\mathsf{X}_d$ is bounded, there is a constant $\Delta_d\geq1$ such that $\mathsf{Y}_{\mathrm{hard}}\subset\young{d}{k}(\Delta_d)$ for all $k$, where $\young{d}{k}(\Delta)$ is defined in Eq.~\eqref{eq:young-window}.
    The multiplicities $c_\nu$ are nonnegative, so Eq.~\eqref{eq:prefactor-free-compressed-dimension} applied to $\Delta = \Delta_d$ shows
    \begin{align}
        D_{\mathsf{Y}_{\mathrm{hard}}}
        \leq D_{\Delta_d}
        =O_d(k^{d^4-d^2}),
    \end{align}
    which proves the last bound in Eq.~\eqref{eq:r-C}.

To prove Eq.~\eqref{eq:recoupling-uniformity}, we obtain an upper bound on
\begin{align}
    \sum_{\nu\vdash k} w_{\mu \nu} w_{\lambda \nu}
    &= m_\mu^{(d)} m_\lambda^{(d)} \sum_{\alpha, \beta \vdash_d k} m_\alpha^{(d)} m_\beta^{(d)} \sum_{\nu\vdash k} g_{\mu\alpha\nu} g_{\lambda\beta\nu}\\
    &=m_\mu^{(d)} m_\lambda^{(d)} \sum_{\alpha, \beta \vdash_d k} m_\alpha^{(d)} m_\beta^{(d)} \sum_{\nu\vdash k} \langle \chi_{\mu} \chi_{\alpha}, \chi_\nu \rangle_{\mfS_k} \langle \chi_\nu, \chi_{\lambda} \chi_{\beta} \rangle_{\mfS_k}\\
    &= m_\mu^{(d)} m_\lambda^{(d)} \sum_{\alpha, \beta \vdash_d k} m_\alpha^{(d)} m_\beta^{(d)} \langle \chi_{\mu} \chi_{\alpha}, \chi_{\lambda} \chi_{\beta} \rangle_{\mfS_k}\\
    &= m_\mu^{(d)} m_\lambda^{(d)} \sum_{\alpha, \beta \vdash_d k} m_\alpha^{(d)} m_\beta^{(d)} \langle \chi_{\mu} \chi_{\lambda}, \chi_{\alpha} \chi_{\beta} \rangle_{\mfS_k}\\
    &=m_\mu^{(d)} m_\lambda^{(d)} \sum_{\alpha, \beta \vdash_d k} m_\alpha^{(d)} m_\beta^{(d)} \sum_{\nu\vdash k} \langle \chi_{\mu} \chi_{\lambda}, \chi_\nu \rangle_{\mfS_k} \langle \chi_\nu, \chi_{\alpha} \chi_{\beta} \rangle_{\mfS_k}\\
    &= m_\mu^{(d)} m_\lambda^{(d)} \sum_{\alpha, \beta \vdash_d k} m_\alpha^{(d)} m_\beta^{(d)} \sum_{\nu\vdash k} g_{\mu\lambda\nu} g_{\alpha\beta\nu},
\end{align}
where we use the orthogonality of characters in the third and fifth lines.
To this end, we apply the Cauchy identity (see, e.g., Ref.~\cite{macdonald1995symmetric}):
\begin{align}
    \prod_{a,b}(1-t u_a v_b)^{-1} = \sum_{k=0}^{\infty} t^k \sum_{\nu\vdash k} s_\nu(\bm{u}) s_\nu(\bm{v}),
\end{align}
where $\bm{u} = (u_a)_a$ and $\bm{v} = (v_b)_b$ are two sets of variables, and $s_\nu$ is the Schur polynomial.
By applying the Cauchy identity to $\bm{u} = \bm{x} \otimes \bm{y} \coloneqq (x_i y_j)_{i, j=1}^{d}$ with $\bm{x} = (x_i)_{i=1}^{d}$, $\bm{y} = (y_j)_{j=1}^{d}$ and $\bm{v} = 1^{d^2}$, we have
\begin{align}
    \prod_{i,j=1}^{d} (1-t x_i y_j)^{-d^2} = \sum_{k=0}^{\infty} t^k \sum_{\nu\vdash k} s_\nu(\bm{x}\otimes \bm{y}) s_\nu(1^{d^2}).
\end{align}
The Schur polynomial $s_\nu(\bm{x}\otimes \bm{y})$ can be decomposed as
\begin{align}
    s_\nu(\bm{x}\otimes \bm{y}) = \sum_{\mu, \lambda} g_{\mu\lambda\nu} s_\mu(\bm{x}) s_\lambda(\bm{y}),
\end{align}
and we have
\begin{align}
    s_\nu(1^{d^2}) = m_{\nu}^{(d^2)} = \sum_{\alpha, \beta\vdash_d k} m_\alpha^{(d)} m_\beta^{(d)} g_{\alpha\beta\nu}.
\end{align}
Thus, we have
\begin{align}
    \prod_{i,j=1}^{d} (1-t x_i y_j)^{-d^2}
    &= \sum_{k=0}^{\infty} t^k \sum_{\mu, \lambda \vdash k} s_\mu(\bm{x}) s_\lambda(\bm{y}) \sum_{\alpha, \beta\vdash_d k} m_\alpha^{(d)} m_\beta^{(d)} \sum_{\nu\vdash k} g_{\mu\alpha\nu} g_{\lambda\beta\nu}\\
    &= \sum_{k=0}^{\infty} t^k \sum_{\mu, \lambda \vdash_d k} {\sum_{\nu\vdash k} w_{\mu \nu} w_{\lambda \nu} \over m_\mu^{(d)} m_\lambda^{(d)}}s_\mu(\bm{x}) s_\lambda(\bm{y}).
    \label{eq:permutation-central-recoupling-generating-function-2}
\end{align}
Suppose $U, V\in \U(d)$ have eigenvalues $\bm{x} = (x_i)_{i=1}^{d}$ and $\bm{y} = (y_j)_{j=1}^{d}$, respectively.
Then, Eq.~\eqref{eq:permutation-central-recoupling-generating-function-2} yields
\begin{align}
    \det(\1_d\otimes \1_d-t U\otimes V)^{-d^2} = \sum_{k=0}^{\infty} t^k \sum_{\mu, \lambda \vdash_d k} {\sum_{\nu\vdash k} w_{\mu \nu} w_{\lambda \nu} \over m_\mu^{(d)} m_\lambda^{(d)}} \chi_\mu(U) \chi_\lambda(V),
\end{align}
by using the characters $\chi_\mu, \chi_\lambda$ of $\U(d)$.
The orthogonality of characters then provides
\begin{align}
    t^k {\sum_{\nu\vdash k} w_{\mu \nu} w_{\lambda \nu} \over m_\mu^{(d)} m_\lambda^{(d)}}
    &= \int_{\U(d)\times \U(d)} \dd U \dd V \overline{\chi_\mu(U) \chi_\lambda(V)} \det(\1_d\otimes \1_d-t U\otimes V)^{-d^2}\\
    &\leq m_\mu^{(d)} m_\lambda^{(d)} I_d(t),
\end{align}
where we use $\abs{\chi_\mu(U)} \leq m_\mu^{(d)}$ since $\chi_\mu(U)$ is the trace of an $m_\mu^{(d)}$-dimensional unitary matrix and we define $I_d(t)$ by
\begin{align}
    I_d(t)\coloneqq \int_{\U(d)\times \U(d)} \dd U \dd V \abs{\det(\1_d\otimes \1_d-t U\otimes V)}^{-d^2},
\end{align}
i.e.,
\begin{align}
    \label{eq:sum-w-mu-nu-bound}
    \sum_{\nu\vdash k} w_{\mu \nu} w_{\lambda \nu}\leq t^{-k} (m_\mu^{(d)} m_\lambda^{(d)})^2 I_d(t).
\end{align}

To obtain a desired bound, we show that the divergence of $I_d(t)$ as $t\nearrow 1$ is at most of order $(1-t)^{-(d^2-1)^2}$.
The Weyl integration formula provides (see, e.g., Ref.~\cite{meckes2019random})
\begin{align}
    \label{eq:Id-weyl-integration}
    I_d(t) = {1\over (d!)^2} \int_{[0,2\pi]^{2d}} {\prod_i \dd \theta_i \dd \phi_i \over (2\pi)^{2d}} \abs{\Delta(e^{\mathrm{i}\theta_1}, \ldots, e^{\mathrm{i}\theta_d}) \Delta(e^{\mathrm{i}\phi_1}, \ldots, e^{\mathrm{i}\phi_d})}^2 \prod_{i,j=1}^{d} \abs{1-t e^{\mathrm{i}(\theta_i + \phi_j)}}^{-d^2},
\end{align}
where $\Delta(z_1, \ldots, z_d) \coloneqq \prod_{i<j} (z_i - z_j)$ is the Vandermonde determinant of $\mathrm{diag}(z_1, \ldots, z_d)$.
The Cauchy determinant identity provides (see, e.g., Ref.~\cite{macdonald1995symmetric})
\begin{align}
    \det[(1-tu_iv_j)^{-1}]_{i,j=1}^{d}={t^{d(d-1)/2}\Delta(u_1, \ldots, u_d)\Delta(v_1, \ldots, v_d) \over \prod_{i,j=1}^{d}(1-tu_iv_j)}.
\end{align}
The Cauchy--Schwarz inequality applied to the expansion of the determinant
\begin{align}
    \det[(1-tu_iv_j)^{-1}]_{i,j=1}^{d} = \sum_{\sigma\in \mfS_d} \sgn(\sigma) \prod_{i=1}^{d} (1-tu_iv_{\sigma(i)})^{-1},
\end{align}
where $\sgn{\sigma}$ is the sign of the permutation $\sigma$, provides
\begin{align}
    \abs{\det[(1-tu_iv_j)^{-1}]_{i,j=1}^{d}}^2 \leq d! \sum_{\sigma\in \mfS_d} \prod_{i=1}^{d} \abs{1-tu_iv_{\sigma(i)}}^{-2}.
\end{align}
Thus, we have
\begin{align}
    \abs{\Delta(u_1, \ldots, u_d)\Delta(v_1, \ldots, v_d)}^2 \leq d! t^{-d(d-1)} \sum_{\sigma\in \mfS_d} \prod_{i=1}^{d} \abs{1-tu_iv_{\sigma(i)}}^{-2} \prod_{i,j=1}^{d} \abs{1-tu_iv_j}^{2}.
\end{align}
From Eq.~\eqref{eq:Id-weyl-integration}, we have
\begin{align}
    I_d(t)
    &\leq {t^{-d(d-1)} \over d!} \sum_{\sigma\in \mfS_d} \int_{[0,2\pi]^{2d}} {\prod_i \dd \theta_i \dd \phi_i \over (2\pi)^{2d}} \prod_{i,j=1}^{d} \abs{1-t e^{\mathrm{i}(\theta_i + \phi_j)}}^{-a_{ij}^{\sigma}}\\
    &= t^{-d(d-1)} \int_{[0,2\pi]^{2d}} {\prod_i \dd \theta_i \dd \phi_i \over (2\pi)^{2d}} \prod_{i,j=1}^{d} \abs{1-t e^{\mathrm{i}(\theta_i + \phi_j)}}^{-a_{ij}}
\end{align}
where $a_{ij}^{\sigma}$ and $a_{ij}$ are defined by
\begin{align}
    a_{ij}^{\sigma} \coloneqq d^2-2+2\delta_{\sigma(i), j}, \qquad a_{ij} \coloneqq d^2-2+2\delta_{ij},
\end{align}
and we use the symmetry of the integrand to obtain the last line.
Thus, for $1/2\leq t<1$, we have
\begin{align}
    I_d(t) = O_d\qty(\int_{[0,2\pi]^{2d}} \prod_i \dd \theta_i \dd \phi_i \prod_{i,j=1}^{d} \abs{1-t e^{\mathrm{i}(\theta_i + \phi_j)}}^{-a_{ij}}).
    \label{eq:permutation-central-recoupling-direct-integral}
\end{align}
We bound the integral in Eq.~\eqref{eq:permutation-central-recoupling-direct-integral}.
Since $\abs{1-t e^{\mathrm{i} \theta}}^2 = (1-t)^2 + 2t(1-\cos \theta) \geq (1-t)^2$ for all $\theta$, we have
\begin{align}
    &\int_{[0,2\pi]^{2d}} \prod_i \dd \theta_i \dd \phi_i \prod_{i,j=1}^{d} \abs{1-t e^{\mathrm{i}(\theta_i + \phi_j)}}^{-a_{ij}}\notag\\
    &\leq (1-t)^{-\sum_{i\neq d, j\neq d} a_{ij}} \int_{[0,2\pi]^{2d}} \prod_i \dd \theta_i \dd \phi_i \prod_{i=d \lor j=d} \abs{1-t e^{\mathrm{i}(\theta_i + \phi_j)}}^{-a_{ij}}.
\end{align}
We introduce new variables $\widetilde{\theta}_i\coloneqq \theta_i+\phi_d$ for $1\leq i\leq d$, $\widetilde{\phi}_j\coloneqq \theta_d+\phi_j$ for $1\leq j\leq d-1$, and $\widetilde{\phi}_d \coloneqq \phi_d$ as the last coordinate.
The Jacobian of this change of variables is 1, and we obtain
\begin{align}
    \int_{[0,2\pi]^{2d}} \prod_i \dd \theta_i \dd \phi_i
    \prod_{i=d \lor j=d}\abs{1-t e^{\mathrm{i}(\theta_i+\phi_j)}}^{-a_{ij}}
    =2\pi\prod_{i=d\lor j=d}
    \int_{-\pi}^{\pi}\dd\theta\,
    \abs{1-t e^{\mathrm{i}\theta}}^{-a_{ij}},
\end{align}
where we use the periodicity of the integrand to change the interval $[0,2\pi]$ to $[-\pi, \pi]$.
For $\theta\in[-\pi,\pi]$ and $t\geq1/2$, the inequality $1-\cos\theta\geq2\theta^2/\pi^2$ provides
\begin{align}
    \abs{1-t e^{\mathrm{i}\theta}}^2
    =(1-t)^2+2t(1-\cos\theta)
    \geq(1-t)^2+{2\theta^2\over\pi^2}.
\end{align}
Thus, using $x=\theta/(1-t)$, we have
\begin{align}
    \int_{-\pi}^{\pi}\dd\theta\,
    \abs{1-t e^{\mathrm{i}\theta}}^{-a_{ij}}
    \leq(1-t)^{1-a_{ij}}
    \int_{\mathbb{R}}\dd x\,
    \left(1+{2x^2\over\pi^2}\right)^{-a_{ij}/2}.
\end{align}
The last integral converges since $a_{ij}>1$ holds for all $i,j$, and provides a constant factor depending only on $d$.
Thus, we have
\begin{align}
    I_d(t)=O_d\qty((1-t)^{2d-1-\sum_{ij}a_{ij}})
    =O_d\qty((1-t)^{-(d^2-1)^2}).
\end{align}
Taking $t=1-k^{-1}$ for $k\geq2$ in Eq.~\eqref{eq:sum-w-mu-nu-bound} and using Eq.~\eqref{eq:d-mu-mu-d}, we obtain
\begin{align}
    \sum_{\nu\vdash k}w_{\mu\nu}w_{\lambda\nu}
    =O_d(k^{d^4-d^2-d+1}).
\end{align}
By taking $\lambda = \mu$ and summing over $\mu\in \mathsf{Y}_{\mathrm{hard}}$, we have
\begin{align}
    \sum_{\mu\in\mathsf{Y}_{\mathrm{hard}}} \sum_{\nu\vdash k} w_{\mu \nu}^2 = \abs{\mathsf{Y}_{\mathrm{hard}}} \cdot O_d(k^{d^4-d^2-d+1}) = O_d(k^{d^4-d^2-(d-1)/2}),
\end{align}
which proves the first bound in Eq.~\eqref{eq:recoupling-uniformity}.
Similarly, we have
\begin{align}
    \sum_{\nu\vdash k} c_\nu w_{\mu\nu} = \sum_{\nu\vdash k} \sum_{\lambda\in\mathsf{Y}_{\mathrm{hard}}} w_{\lambda\nu} w_{\mu\nu}  = \abs{\mathsf{Y}_{\mathrm{hard}}} \cdot O_d(k^{d^4-d^2-d+1}) = O_d(k^{d^4-d^2-(d-1)/2}),
\end{align}
which proves the second bound in Eq.~\eqref{eq:recoupling-uniformity}.
\end{proof}

\begin{proof}[Proof of Thm.~\ref{thm:permutation-channel-trace-matching-lower}]
    We first show the proof sketch.
    In this proof, we write $\pi\coloneqq\pi_d$ for simplicity.
    Using the set $\mathsf{Y}_{\mathrm{hard}}$ constructed in Lem.~\ref{lem:permutation-central-window-recoupling}, we construct a hard family of permutation-covariant channels $\Lambda_z$, whose Choi operators are given by
    \begin{align}
        \label{eq:def-J-z}
        J_z = J^{1/2} (\1+\alpha X_z) J^{1/2} + (\1_{\mcI} - \overline{P}_{\mathsf{Y}_\mathrm{hard}}) \otimes {\1_{\mcO} \over d^k}
    \end{align}
    with a parameter $\alpha>0$ and Hermitian operators $J$, $X_z$ for $z\in \mathsf{Z}$ satisfying
    \begin{align}
        \Tr_{\mcO} J = \overline{P}_{\mathsf{Y}_\mathrm{hard}}, \quad
        \Tr_{\mcO} [J^{1/2} X_z J^{1/2}] = 0, \quad X_z^2\preceq \1, \quad J, X_z\in \Comm(\overline{\pi}\otimes \pi),
    \end{align}
    where $\overline{P}_{\mathsf{Y}_\mathrm{hard}}$ is defined by
    \begin{align}
        \overline{P}_{\mathsf{Y}_{\mathrm{hard}}} \coloneqq \sum_{\mu\in\mathsf{Y}_{\mathrm{hard}}} \overline{P}_\mu.
    \end{align}
    Since $X_z^2\preceq \1$ holds, by choosing $\alpha\leq 1$, we have $\1+\alpha X_z\succeq 0$, and thus $J_z\succeq 0$.
    We also have $\Tr_{\mcO} J_z = \1_{\mcI}$, and $[J_z, \overline{\pi}(\sigma)_{\overline{\mcI}} \otimes \pi(\sigma)_{\mcO}] = 0$ for all $\sigma\in\mfS_k$, i.e., $J_z$ is a valid Choi operator of a permutation-covariant channel.
    In particular, we take $X_z, J$ whose supports are constrained to the subspace $\supp [\overline{P}_{\mathsf{Y}_\mathrm{hard}}] \otimes \mcO$.
    We also define $J_\star$ by taking $\alpha=0$ in the definition~\eqref{eq:def-J-z} of $J_z$.
    Then, we have
    \begin{align}
        d_\mathrm{tr}(\rho_z, \rho_{z'})
        &= {\alpha \over 2 d^k} \left\| J^{1/2} (X_z - X_{z'}) J^{1/2} \right\|_1
    \end{align}
    and
    \begin{align}
        \widehat{D}_2(\Lambda_z\|\Lambda_\star)
        &= \log \left\| \Tr_{\mcO}(J_z J_\star^{-1} J_z) \right\|_\infty\\
        &= \log \left\|\Tr_{\mcO}[J]+2\alpha \Tr_{\mcO}[J^{1/2} X_z J^{1/2}] + \alpha^2 \Tr_{\mcO}[J^{1/2} X_z^2 J^{1/2}] + (\1-\overline{P}_{\mathsf{Y}_\mathrm{hard}}) \right\|_\infty\\
        &= \log \|\1 + \alpha^2 \Tr_{\mcO}[J^{1/2} X_z^2 J^{1/2}]\|_\infty\\
        &\leq \log (1+\alpha^2 \|\Tr_{\mcO}[J]\|_\infty)\\
        &= \log(1+\alpha^2 \|\overline{P}_{\mathsf{Y}_\mathrm{hard}}\|_\infty)\\
        &= \log(1+\alpha^2)\\
        &\leq \alpha^2,
    \end{align}
    where we use the triangle inequality and $X_z^2\preceq \1$ in the fourth line, and the inequality $\log(1+t)\leq t$ for $t\geq 0$ in the last line.
    We construct the operators $J$, $X_z$ such that
    \begin{align}
        \label{eq:trace-norm-bound-Xz}
        {1\over d^k} \min_{z\neq z'} \left\|J^{1/2} (X_z-X_{z'}) J^{1/2}\right\|_1 &= \Omega_d(1),\\
        \label{eq:packing-size-bound}
        \log \abs{\mathsf{Z}} &= \Omega_d(k^{d^4-d^2}).
    \end{align}
    Then, by taking $\alpha = \Theta_d(\varepsilon_\mathrm{tr})$ so that Eq.~\eqref{eq:trace-norm-bound} holds and invoking Lem.~\ref{lem:adaptive-geometric-renyi-packing} with $\xi = \alpha^2$, $\eta = 1/3$, we obtain
    \begin{align}
        n = \Omega_d\qty(\log \abs{\mathsf{Z}} \over \alpha^2) = \Omega_d\qty({k^{d^4-d^2} \over \varepsilon_\mathrm{tr}^2}).
    \end{align}

    We first construct an operator $J$ on the Hilbert space~\eqref{eq:choi-space-decomposition}.
    We take the Young projector $\overline{P}_\mu$ on $\overline{\mcI}$ and $P_\nu$ on $\overline{\mcI} \otimes \mcO$, and define
    \begin{align}
        Q_{\mu\nu}\coloneqq (\overline{P}_\mu\otimes \1_{\mcO})(P_\nu)_{\overline{\mcI}\mcO},
    \end{align}
    which defines the projector onto the subspace $\mcS_\nu\otimes\mathbb{C}^{w_{\mu\nu}}$ in the decomposition~\eqref{eq:choi-space-decomposition}, with $\Tr Q_{\mu\nu}=d_\nu w_{\mu\nu}$.
    Since $P_\nu$ on $\overline{\mcI} \otimes \mcO$ commutes with $\overline{\pi}(\sigma) \otimes \pi(\sigma)$ for all $\sigma\in\mfS_k$, we have
    \begin{align}
        [\Tr_{\mcO} Q_{\mu\nu}, \overline{\pi}(\sigma)] = 0 \quad \forall \sigma\in\mfS_k,
    \end{align}
    i.e.,
    \begin{align}
        \label{eq:commutation-trace-Q-mu-nu-1}
        [\Tr_{\mcO} Q_{\mu\nu}, \overline{g}_{\mu}(\sigma) \otimes \1_{m_\mu^{(d)}}] = 0 \quad \forall \sigma\in\mfS_k.
    \end{align}
    Since $P_\nu$ commutes with any unitary operator on the multiplicity space $\CC^{w_{\mu\nu}}$, it commutes with a unitary operator on the multiplicity space $\CC^{m_\mu^{(d)}}$ in the decomposition~\eqref{eq:choi-space-decomposition-2} as well.
    Thus, we have
    \begin{align}
        \label{eq:commutation-trace-Q-mu-nu-2}
        [\Tr_{\mcO} Q_{\mu\nu}, \1_{\mcS_\mu} \otimes U] = 0 \quad \forall U\in \U(m_\mu^{(d)}).
    \end{align}
    By applying Schur's lemma to Eqs.~\eqref{eq:commutation-trace-Q-mu-nu-1} and \eqref{eq:commutation-trace-Q-mu-nu-2}, we have
    \begin{align}
        \Tr_{\mcO} Q_{\mu\nu}\propto \overline{P}_\mu,
    \end{align}
    and by comparing the traces of both sides, we obtain
    \begin{align}
        \label{eq:trace-Q-mu-nu}
        \Tr_{\mcO} Q_{\mu\nu}
        ={d_\nu w_{\mu\nu}\over d_\mu m_\mu^{(d)}}\overline{P}_\mu.
    \end{align}
    We define $J$ by
    \begin{align}
        J\coloneqq
        \sum_{\substack{\mu\in\mathsf{Y}_{\mathrm{hard}}\\\nu\vdash k}}
        {d_\mu m_\mu^{(d)}c_\nu\over
        d_\nu s_\mu}\,Q_{\mu\nu},
    \end{align}
    where $s_\mu$ is defined by
    \begin{align}
        s_\mu\coloneqq\sum_{\kappa\vdash k}c_\kappa w_{\mu\kappa}>0 \quad \forall \mu\in \mathsf{Y}_{\mathrm{hard}}.
    \end{align}
    Then, Eq.~\eqref{eq:trace-Q-mu-nu} yields
    \begin{align}
        \Tr_\mcO J
        =\overline{P}_{\mathsf{Y}_{\mathrm{hard}}}.
    \end{align}
    By decomposing $J$ as $J\cong\sum_{\nu\vdash k}\1_{\mcS_\nu}\otimes J_\nu$ in the decomposition~\eqref{eq:choi-space-decomposition}, we have
    \begin{align}
        \label{eq:J-nu-bound}
        {1\over d^k} J_\nu
        =\bigoplus_{\mu\in\mathsf{Y}_{\mathrm{hard}}}
        {d_\mu m_\mu^{(d)}c_\nu\over d^k
        d_\nu s_\mu}\,
        \1_{\mathbb{C}^{w_{\mu\nu}}}
        \succeq {c_\nu \over d_\nu}
        \Omega_d\!\left(k^{-(d^4-d^2)}\right)
        \1_{\mathbb{C}^{c_\nu}},
    \end{align}
    where we use Lem.~\ref{lem:permutation-central-window-recoupling}.

    We then take a set of Hermitian operators $X_z$ on the Hilbert space~\eqref{eq:choi-space-decomposition} such that $X_z\cong\sum_{\nu\vdash k}\1_{\mcS_\nu}\otimes X_{z,\nu}$ in the decomposition~\eqref{eq:choi-space-decomposition}, and $X_{z,\nu}$ is a Hermitian operator on $\CC^{c_\nu}$.
    For each $\nu\vdash k$, we divide the set $\{\mu\in\mathsf{Y}_\mathrm{hard} \mid w_{\mu\nu}>0\}$ into two disjoint sets $U_\nu$ and $V_\nu$ and define a bipartition of $\CC^{c_\nu}$ by
    \begin{align}
        \label{eq:bipartition-U-V}
        \CC^{c_{\nu}} \simeq \mcU_\nu \oplus \mcV_\nu, \qquad \mcU_\nu\coloneqq \bigoplus_{\mu\in U_\nu} \CC^{w_{\mu\nu}}, \qquad \mcV_\nu\coloneqq \bigoplus_{\mu\in V_\nu} \CC^{w_{\mu\nu}}.
    \end{align}
    Then, using linear operators $Z_\nu: \mcU_\nu \to \mcV_\nu$, we define $X_{z,\nu}$ by
    \begin{align}
        X_{z,\nu} \coloneqq \begin{pmatrix}
            0 & Z_\nu^\dagger\\
            Z_\nu & 0
        \end{pmatrix}_{\mcU_\nu \oplus \mcV_\nu}.
    \end{align}
    When forming products with $J_\nu$, we regard $Z_\nu$ as an operator on $\mcU_\nu\oplus\mcV_\nu$ by extending it by zero on $\mcV_\nu$, and embed the $\nu$-block into the full Choi space using the decomposition~\eqref{eq:choi-space-decomposition}.
    Let $T_\nu\coloneqq\Tr_{\mcO}(\1_{\mcS_\nu}\otimes J_\nu^{1/2}Z_\nu J_\nu^{1/2})$.
    Similarly to Eq.~\eqref{eq:commutation-trace-Q-mu-nu-1}, we have
    \begin{align}
        \overline{\pi}(\sigma)^\dagger T_\nu\overline{\pi}(\sigma)=T_\nu \quad \forall \sigma\in\mfS_k.
    \end{align}
    Since $J_\nu$ preserves each $\mu$-block, $T_\nu$ maps the input Young sectors in $U_\nu$ into those in $V_\nu$.
    The sets $U_\nu$ and $V_\nu$ are disjoint, so Schur's lemma gives
    \begin{align}
        T_\nu=0,
    \end{align}
    i.e.,
    \begin{align}
        \Tr_{\mcO}[J^{1/2} X_z J^{1/2}] = 0.
    \end{align}
    By taking $Z_\nu$ such that $\|Z_\nu\|_\infty\leq 1$, we have $X_z^2\preceq \1$.
    We construct bipartitions $\mcU_\nu, \mcV_\nu$ for each $\nu\vdash k$ and a set of matrices
    \begin{align}
        \mathsf{Z} \subset \left\{(Z_\nu)_{\nu\vdash k} \in \bigoplus_{\nu\vdash k} \mcL(\mcU_\nu \to \mcV_\nu) \; \middle| \; \|Z_\nu\|_\infty \leq 1 \right\}
    \end{align}
    such that Eqs.~\eqref{eq:trace-norm-bound-Xz} and \eqref{eq:packing-size-bound} hold to complete the proof.

    We take two disjoint sets $U_\nu$ and $V_\nu$ such that
    \begin{align}
        \label{eq:permutation-fixed-fidelity-cross-dimension}
        \sum_\nu u_\nu v_\nu = \Theta_d(k^{d^4-d^2}), \qquad u_\nu \coloneqq \dim \mcU_\nu = \sum_{\mu\in U_\nu} w_{\mu\nu}, \qquad v_\nu \coloneqq \dim \mcV_\nu = \sum_{\mu\in V_\nu} w_{\mu\nu},
    \end{align}
    which defines the bipartition~\eqref{eq:bipartition-U-V}.
    To show the existence of such a bipartition, we introduce independent random variables $\zeta_{\mu\nu}\in\{U,V\}$ with $\Prob[\zeta_{\mu\nu}=U]=\Prob[\zeta_{\mu\nu}=V]=1/2$ and define
    \begin{align}
        U_\nu \coloneqq \{\mu\in \mathsf{Y}_{\mathrm{hard}}\mid \zeta_{\mu\nu}=U\}, \qquad V_\nu \coloneqq \{\mu\in \mathsf{Y}_{\mathrm{hard}}\mid \zeta_{\mu\nu}=V\}.
    \end{align}
    Then, we have
    \begin{align}
        \mathbb{E}\sum_\nu u_\nu v_\nu
        &= {1\over 4} \sum_\nu \sum_{\mu\neq \lambda} w_{\mu\nu} w_{\lambda\nu} \\
        &= {1\over 4} \sum_\nu \qty[\qty(\sum_\mu w_{\mu\nu})^2 - \sum_\mu w_{\mu\nu}^2] \\
        &= {1\over 4} \qty[D_{\mathsf{Y}_{\mathrm{hard}}} - \sum_{\mu\in \mathsf{Y}_{\mathrm{hard}}}\sum_{\nu\vdash k} w_{\mu\nu}^2] \\
        &= \Omega_d(k^{d^4-d^2}),
    \end{align}
    where we use Lem.~\ref{lem:permutation-central-window-recoupling} in the last line.
    Thus, there exists a realization of the random variables $\zeta_{\mu\nu}$ such that
    \begin{align}
        \sum_\nu u_\nu v_\nu = \Omega_d(k^{d^4-d^2})
    \end{align}
    holds.
    The converse upper bound holds since
    \begin{align}
        \sum_\nu u_\nu v_\nu \leq \sum_\nu \qty(\sum_\mu w_{\mu\nu})^2 = D_{\mathsf{Y}_{\mathrm{hard}}} = O_d(k^{d^4-d^2}),
    \end{align}
    i.e., Eq.~\eqref{eq:permutation-fixed-fidelity-cross-dimension} holds.

    For $u_\nu v_\nu=0$, set $\mathsf{Z}_{v_\nu,u_\nu}\coloneqq\left\{0\right\}$ and $l_\nu\coloneqq0$.
    For each $\nu$ with $u_\nu v_\nu>0$, we next construct a set $\mathsf{Z}_{v_\nu,u_\nu}$ of matrices $Z:\mcU_\nu\to\mcV_\nu$, represented by $Z\in\mathbb{C}^{v_\nu\times u_\nu}$, such that
    \begin{align}
        \label{eq:condition-for-Z}
        \abs{\mathsf{Z}_{v_\nu,u_\nu}}
        \geq\exp\!\left(\Omega(u_\nu v_\nu)\right),
        \qquad \|Z\|_\infty\leq1,
        \qquad \|Z-Z'\|_1
        \geq\Omega\!\left(\min\left\{u_\nu,v_\nu\right\}\right)
        \quad\forall Z\neq Z'\in\mathsf{Z}_{v_\nu,u_\nu}.
    \end{align}
    First suppose $u_\nu\leq v_\nu$, fix $W_0\in\W_{v_\nu,u_\nu}$, and draw $V$ from Haar measure on $\U(v_\nu)$, so that $VW_0$ is Haar distributed on $\W_{v_\nu,u_\nu}$.
    Defining $F:\U(v_\nu)\to\mathbb{R}$ by
    \begin{align}
        F(V)\coloneqq\Re\Tr(W_0W_0^\dagger V),
    \end{align}
    we obtain from the Cauchy--Schwarz inequality
    \begin{align}
        \abs{F(V)-F(V')}
        \leq\|W_0W_0^\dagger\|_2\|V-V'\|_2
        =\sqrt{u_\nu}\|V-V'\|_2,
    \end{align}
    i.e., $F$ is $\sqrt{u_\nu}$-Lipschitz with respect to the Frobenius norm.
    On the other hand, Haar invariance gives $\mathbb{E}F(V)=0$ and
    \begin{align}
        \|VW_0-W_0\|_2^2=2u_\nu-2F(V).
    \end{align}
    The concentration inequality for Haar measure~\cite[Cor.~17]{meckes2013spectral} therefore yields
    \begin{align}
    \Prob\left[\|VW_0-W_0\|_2\leq\sqrt{u_\nu}\right]=\Prob\left[F(V)\geq {u_\nu\over2}\right]
    \leq\exp\!\left(-{u_\nu v_\nu\over48}\right).
    \label{eq:permutation-fixed-fidelity-stiefel-small-ball}
    \end{align}
    This bound holds  for any $W_0 \in \W_{v_\nu, u_\nu}$.
    Let $M(\W_{v_\nu, u_\nu}, \sqrt{u_\nu})$ be the maximum size of a $\sqrt{u_\nu}$-packing of $\W_{v_\nu, u_\nu}$ in the Frobenius norm, and let $\mcM_{v_\nu,u_\nu}$ be the corresponding maximum packing [see also Eq.~\eqref{eq:def-packing-number}].
    Then, by maximality, the $\sqrt{u_\nu}$-balls centered at the points of $\mcM_{v_\nu,u_\nu}$ cover $\W_{v_\nu, u_\nu}$, so the probability bound~\eqref{eq:permutation-fixed-fidelity-stiefel-small-ball} implies
    \begin{align}
        \abs{\mcM_{v_\nu,u_\nu}} = M(\W_{v_\nu, u_\nu}, \sqrt{u_\nu}) \geq \exp\left({u_\nu v_\nu\over48}\right).
    \end{align}
    We take $\mathsf{Z}_{v_\nu, u_\nu}\coloneqq \mcM_{v_\nu, u_\nu}$.
    Then, since $Z, Z'\in \mathsf{Z}_{v_\nu, u_\nu}$ are isometries, we have
    \begin{align}
        \|Z\|_\infty = 1, \qquad \|Z-Z'\|_\infty \leq \|Z\|_\infty + \|Z'\|_\infty = 2.
    \end{align}
    By definition, we have $\|Z-Z'\|_2\geq \sqrt{u_\nu}$.
    Thus, we have
    \begin{align}
        \|Z-Z'\|_1
        \geq {\|Z-Z'\|_2^2\over\|Z-Z'\|_\infty}
        \geq {u_\nu\over2},
    \end{align}
    which proves Eq.~\eqref{eq:condition-for-Z} when $u_\nu\leq v_\nu$.
    When $u_\nu>v_\nu$, we take $\mathsf{Z}_{v_\nu,u_\nu}\coloneqq\left\{Z^\dagger\;\middle|\;Z\in\mcM_{u_\nu,v_\nu}\right\}$ to obtain the same bound.

    Finally, we define a subset $\mathsf{Z}\subset \{\bigoplus_{\nu\vdash k} Z_\nu \mid Z_\nu \in \mathsf{Z}_{v_\nu, u_\nu}\}$ such that Eqs.~\eqref{eq:trace-norm-bound-Xz} and \eqref{eq:packing-size-bound} hold.
    The left-hand side of Eq.~\eqref{eq:trace-norm-bound-Xz} is lower bounded by
    \begin{align}
        {1\over d^k} \min_{z\neq z'} \left\|J^{1/2} (X_z-X_{z'}) J^{1/2}\right\|_1
        &= {1\over d^k} \min_{z\neq z'} \sum_{\nu:Z_\nu\neq Z_\nu'} d_\nu \left\|J_\nu^{1/2} (X_{z,\nu}-X_{z',\nu}) J_\nu^{1/2}\right\|_1\\
        &\geq \min_{z\neq z'} \sum_{\nu:Z_\nu\neq Z_\nu'} c_\nu \Omega_d(k^{-(d^4-d^2)}) \left\|X_{z,\nu}-X_{z',\nu}\right\|_1\\
        &=\Omega_d(k^{-(d^4-d^2)}) \min_{z\neq z'} \sum_{\nu:Z_\nu\neq Z_\nu'} c_\nu \|Z_\nu - Z'_\nu\|_1\\
        &= \Omega_d(k^{-(d^4-d^2)}) \min_{z\neq z'} \sum_{\nu:Z_\nu\neq Z_\nu'} c_\nu \Omega(\min\{u_\nu, v_\nu\})\\
        &= \Omega_d(k^{-(d^4-d^2)}) \min_{z\neq z'} \sum_{\nu:Z_\nu\neq Z_\nu'} \Omega(u_\nu  v_\nu),
        \label{eq:trace-norm-bound-Xz-lower-bound}
    \end{align}
    Here the first line uses $\|\1_{\mcS_\nu}\otimes H\|_1=d_\nu\|H\|_1$.
    The second line follows from Eq.~\eqref{eq:J-nu-bound}, applied on $\CC^{c_\nu}$, and the inequality
    \begin{align*}
        \|A^{1/2}HA^{1/2}\|_1\geq\lambda_{\min}(A)\|H\|_1 \quad (A\succ0).
    \end{align*}
    The factor $d_\nu$ then cancels the denominator in Eq.~\eqref{eq:J-nu-bound}.
    We also use $\|X_{z,\nu}-X_{z',\nu}\|_1=2\|Z_\nu-Z'_\nu\|_1$, Eq.~\eqref{eq:condition-for-Z}, and $c_\nu=u_\nu+v_\nu$, which implies $c_\nu\min\{u_\nu,v_\nu\}\geq u_\nu v_\nu$.
    Thus, we need to construct a set $\mathsf{Z}$ such that the last line is lower bounded by $\Omega_d(1)$.
    To this end, we take a subset $\mathsf{Z}$ such that $\sum_{\nu:Z_\nu\neq Z'_\nu}u_\nu v_\nu=\Omega_d(k^{d^4-d^2})$ uniformly over all distinct $Z,Z'\in\mathsf{Z}$.
    We show the existence of such a subset $\mathsf{Z}$ by using the Gilbert--Varshamov bound~\cite{gilbert1952comparison,varshamov1957estimate} as follows.
    We take a subset $\mathsf{Z}'_{v_\nu, u_\nu} \subset \mathsf{Z}_{v_\nu, u_\nu}$ such that $\abs{\mathsf{Z}'_{v_\nu, u_\nu}} = 2^{l_\nu}$, with $l_\nu = \Theta(u_\nu v_\nu)$ for $u_\nu v_\nu>0$ and $l_\nu=0$ otherwise, which is possible since Eq.~\eqref{eq:condition-for-Z} holds.
    Then, we label the elements of $\mathsf{Z}'_{v_\nu, u_\nu}$ by $z_\nu \in \{0,1\}^{l_\nu}$, and concatenate the labels to obtain $z = (z_\nu)_\nu \in \{0,1\}^L$ with 
    \begin{align}
        L \coloneqq \sum_\nu l_\nu = \Theta\left(\sum_\nu u_\nu v_\nu\right) = \Theta_d(k^{d^4-d^2}).
    \end{align}
    Then, the Gilbert--Varshamov bound shows the existence of a codebook $\mathsf{Z} \subset \{0,1\}^L$ such that
    \begin{align}
        d_{\rm Ham}(z, z') \geq \lfloor \delta L \rfloor \quad (z\neq z'), \qquad \log\abs{\mathsf{Z}} \geq L(\log 2 - h_2(\delta))
    \end{align}
    for $0<\delta<1/2$, where $d_\mathrm{Ham}(z, z')\coloneqq \sum_i \abs{z_i-z'_i}$ is the Hamming distance and $h_2(\delta) \coloneqq -\delta\log\delta-(1-\delta)\log(1-\delta)$ is the binary entropy function.
    The Hamming distance lower bound implies
    \begin{align}
        \delta L-1 \leq d_\mathrm{Ham}(z, z')\leq \sum_{\nu: Z_\nu\neq Z'_\nu} l_\nu \leq O\left(\sum_{\nu:Z_\nu\neq Z'_\nu} u_\nu v_\nu\right),
    \end{align}
    and by taking $\delta = 1/3$, we have
    \begin{align}
        \sum_{\nu:Z_\nu\neq Z'_\nu} u_\nu v_\nu = \Omega_d(k^{d^4-d^2}) \quad \forall z\neq z'\in\mathsf{Z}, \qquad \log \abs{\mathsf{Z}} = \Omega_d(k^{d^4-d^2}).
    \end{align}
    Then, Eq.~\eqref{eq:trace-norm-bound-Xz-lower-bound} implies Eq.~\eqref{eq:trace-norm-bound-Xz}, and we have Eq.~\eqref{eq:packing-size-bound} as well.
    By following the discussion in the beginning of the proof, we obtain the desired lower bound on the query complexity.
\end{proof}

\section{Asymptotic scaling of parameters related to Young diagrams}
\label{sec:uniform-fixed-height-young-estimates}

This appendix proves the asymptotic scaling of the parameters $C_G, C'_G, D_G, D'_G, D''_G, K_G$ given in Eqs.~\eqref{eq:parameters-for-diamond-norm}, \eqref{eq:parameters-for-choi-trace-norm}, \eqref{eq:parameters-for-diamond-lower-bound}, and \eqref{eq:def-C_G} for the symmetric group $G=\mfS_k$, which are used for the lower and upper bounds for learning permutation-covariant channels.

\begin{proposition}
\label{prop:symmetric-group-parameters-asymptotics}
Fix $d\geq2$ and let $k\to\infty$.
For $G=\mfS_k$, the parameters $C_G, C'_G, D_G, D'_G, D''_G, K_G$ given in Eqs.~\eqref{eq:parameters-for-diamond-norm}, \eqref{eq:parameters-for-choi-trace-norm}, \eqref{eq:parameters-for-diamond-lower-bound}, and \eqref{eq:def-C_G} with $m_\lambda=m_\lambda^{(d)}$, $M_\lambda=m_\lambda^{(d^3)}$, $d_\mcI=d^k$, and $L=|\young{d}{k}|$ are given by
\begin{align}
    K_{\mfS_k}=\Theta_d\!\left(k^{(d-1)(d+2)/2}\right), \qquad C_{\mfS_k}, C'_{\mfS_k} = \Theta_d\left(k^{d^4-(d^2+1)/2}\right), \qquad D_{\mfS_k}, D'_{\mfS_k}, D''_{\mfS_k} = \Theta_d\left(k^{d^4-1}\right).
\end{align}
\end{proposition}

\begin{proof}
    We first prove the estimate for $K_{\mfS_k}$.
    For every $\lambda\vdash_d k$, the Weyl dimension formula~\eqref{eq:weyl-dimension} gives $m_\lambda^{(d)}=O_d(k^{d(d-1)/2})$, since each of its $d(d-1)/2$ factors is $O_d(k)$.
    Together with $L\leq\binom{k+d-1}{d-1}=O_d(k^{d-1})$, this yields $K_{\mfS_k}=O_d(k^{(d-1)(d+2)/2})$.
    For the matching lower bound, we define
    \begin{align}
        \alpha_i\coloneqq\frac{2(d-i+1)}{d(d+1)}\quad(i=1,\ldots,d),\qquad h\coloneqq\frac{1}{2d^2(d+1)},
    \end{align}
    and choose the integers $\lambda_i$ independently in $[k(\alpha_i-h),k(\alpha_i+h)]$ for $i=1,\ldots,d-1$, and set $\lambda_d=k-\sum_{i=1}^{d-1}\lambda_i$.
    Since $\sum_i\alpha_i=1$ and $\alpha_i-\alpha_{i+1}=2/[d(d+1)]$ hold, every such choice satisfies
    \begin{align}
        \lambda_d\geq\frac{3k}{2d(d+1)},\qquad \lambda_i-\lambda_{i+1}\geq\frac{3k}{2d(d+1)}\quad(1\leq i<d).
    \end{align}
    Thus all these choices define $\Theta_d(k^{d-1})$ distinct partitions $\lambda\vdash_d k$.
    The Weyl dimension formula~\eqref{eq:weyl-dimension} gives $m_\lambda^{(d)}=\Omega_d(k^{d(d-1)/2})$ for these partitions.
    Thus, we have
    \begin{align}
        K_{\mfS_k}=\Omega_d\!\left(k^{d-1}k^{d(d-1)/2}\right)=\Omega_d\!\left(k^{(d-1)(d+2)/2}\right).
    \end{align}

    By using the identities $\sum_{\lambda\vdash k} m_\lambda^{(a)} m_\lambda^{(b)} = \binom{k+ab-1}{ab-1}$ for any $a, b\in \ZZ_{>0}$ and $L = \abs{\young{d}{k}} \leq \binom{k+d-1}{d-1} = O_d(k^{d-1})$, we have
    \begin{align}
        D_{\mfS_k} &= \sum_{\lambda\vdash_d k} m_\lambda^{(d)} m_\lambda^{(d^3)} = \binom{k+d^4-1}{d^4-1} = \Theta_d(k^{d^4-1}),\\
        D'_{\mfS_k} &= D_{\mfS_k} + K_{\mfS_k}\log(4L) = \Theta_d(k^{d^4-1}),\\
        D''_{\mfS_k} &= \sum_{\lambda\vdash_d k} m_\lambda^{(d)} [m_\lambda^{(d^3)} -m_\lambda^{(d)}] = D_{\mfS_k} - \binom{k+d^2-1}{d^2-1} = \Theta_d(k^{d^4-1}).
    \end{align}
    Since $C_{\mfS_k} \leq C'_{\mfS_k}$ holds by definition, it suffices to show an upper bound on \begin{align}
        C'_{\mfS_k} = \frac{1}{d^k}\left(\sum_{\lambda\vdash_d k} m_\lambda^{(d)} \sqrt{d_\lambda m_\lambda^{(d^3)}}\right)^2
    \end{align} and a lower bound on
    \begin{align}
        C_{\mfS_k} = \frac{1}{d^k}\left(\sum_{\lambda\vdash_d k} m_\lambda^{(d)} \sqrt{d_\lambda \qty[m_\lambda^{(d^3)} - m_\lambda^{(d)}]}\right)^2
    \end{align}
    of the same order.

    To show an upper bound on $C'_{\mfS_k}$, we introduce $x_i\coloneqq \frac{\lambda_i-k/d}{\sqrt{k}}$ and use the following upper bounds:
    \begin{align}
        \label{eq:d-lambda-upper-bound}
        d_\lambda&= O_d\qty(d^k k^{-(d-1)(d+2)/4}(1+\|\bm x\|_2)^{d(d+1)/2}e^{-c\|\bm x\|_2^2}),\\
        \label{eq:m-lambda-d-upper-bound}
        m_\lambda^{(d)}&= O_d\qty(k^{d(d-1)/4}(1+\|\bm x\|_2)^{d(d-1)/2}),\\
        \label{eq:m-lambda-d3-upper-bound}
        m_\lambda^{(d^3)}&= O_d\qty(k^{d^4-d^2+d(d-1)/4}(1+\|\bm x\|_2)^{d(d-1)/2}),
    \end{align}
    where $c$ is a positive constant depending only on $d$. We prove these bounds as follows.
    The Frobenius formula~\eqref{eq:hook-length} provides
    \begin{align}
        \label{eq:hook-length-modified}
        d_\lambda&=\frac{k!}{\Lambda!}\frac{\Lambda!}{\prod_{i=1}^d\widetilde{\lambda}_i!}\prod_{1\leq i<j\leq d}(\widetilde{\lambda}_i-\widetilde{\lambda}_j),
    \end{align}
    where $\widetilde{\lambda}_i\coloneqq \lambda_i +d-i$ and $\Lambda\coloneqq \sum_i \widetilde{\lambda}_i = k+d(d-1)/2$.
    First, we have
    \begin{align}
        {k!\over \Lambda!} \leq k^{-d(d-1)/2}.
    \end{align}
    Defining $\bm{p} = (p_i)_i$ with $p_i \coloneqq {\widetilde{\lambda}_i \over \Lambda}$ and using the Stirling approximation\footnote{We use $\sqrt{n+1}$ instead of $\sqrt{n}$ to use it even for the case of $n=0$.} $n! = \Theta\qty(\sqrt{n+1} (n/e)^n)$, we have
    \begin{align}
        \label{eq:stirling-combinatorics}
        {\Lambda! \over \prod_{i=1}^{d} \widetilde{\lambda}_i!} = \Theta_d\qty(e^{\Lambda H(\bm{p})} (\Lambda+1)^{-(d-1)/2}\prod_{i=1}^d \sqrt{\Lambda+1 \over \widetilde{\lambda}_i+1}),
    \end{align}
    where $H(\bm{p}) \coloneqq -\sum_{i=1}^d p_i \log p_i$ is the entropy of $\bm{p}$.
    By applying Pinsker's inequality~\cite{pinsker1964information} to the relative entropy between $\bm{p}$ and the uniform distribution $\bm{u}_d = (1/d, \ldots, 1/d)$, we have
    \begin{align}
        H(\bm{p}) -\log d \leq - {1\over 2}\|\bm{p}-\bm{u}_d\|_2^2,
    \end{align}
    and the right-hand side is further bounded by
    \begin{align}
        - {1\over 2}\|\bm{p}-\bm{u}_d\|_2^2
        &= -{1\over 2\Lambda^2} \sum_{i=1}^{d} \qty(\sqrt{k}x_i+{d+1\over 2}-i)^2\\
        &\leq -{1\over 2\Lambda^2} \qty[{1\over 2} k \|x\|_2^2-\sum_{i=1}^{d} \qty({d+1\over 2}-i)^2],
    \end{align}
    where we use $(u+v)^2 \geq {1\over 2} u^2-v^2$ for $u, v\in \RR$.
    Thus, we have
    \begin{align}
        e^{\Lambda H(\bm{p})} = O_d\qty(d^k e^{-{k \|x\|^2 \over 4\Lambda}}) = O_d\qty(d^k e^{-c\|x\|_2^2}),
    \end{align}
    where $c\coloneqq {1\over 2d(d-1)+4}$ and we use $\Lambda = k+{d(d-1) \over 2}\leq \qty(1+{d(d-1)\over 2})k$ for $k\geq1$.
    We also have
    \begin{align}
        \sqrt{\Lambda+1 \over \widetilde{\lambda}_i+1}
        &\leq
        \begin{cases}
            \sqrt{2d(\Lambda+1) \over k} & (\widetilde{\lambda}_i+1 \geq k/(2d))\\
            \sqrt{\Lambda+1} & (\widetilde{\lambda}_i+1 < k/(2d))
        \end{cases}\\
        &=
        \begin{cases}
            O_d(1) & (\widetilde{\lambda}_i+1 \geq k/(2d))\\
            O_d(1+\abs{x_i}) & (\widetilde{\lambda}_i+1 < k/(2d))
        \end{cases}\\
        &= O_d(1+\|\bm{x}\|_2),
    \end{align}
    where we use $x_i<-\sqrt{k}/(2d)$ for the second case.
    Thus, we have
    \begin{align}
        \prod_{i=1}^{d} \sqrt{\Lambda+1 \over \widetilde{\lambda}_i+1} = O_d\qty((1+\|\bm{x}\|_2)^d),
    \end{align}
    i.e.,
    \begin{align}
        {\Lambda! \over \prod_{i=1}^{d} \widetilde{\lambda}_i!} = O_d\qty(d^k k^{-(d-1)/2}(1+\|\bm{x}\|_2)^d e^{-c\|\bm{x}\|_2^2}).
    \end{align}
    Since $\widetilde{\lambda}_i-\widetilde{\lambda}_j = \sqrt{k}(x_i-x_j) + j-i = O_d(\sqrt{k}(1+\|\bm{x}\|_2))$ holds, we have
    \begin{align}
        \label{eq:lambda-difference-upper-bound}
        \prod_{1\leq i<j\leq d} (\widetilde{\lambda}_i-\widetilde{\lambda}_j) = O_d\qty(k^{d(d-1)/4}(1+\|\bm{x}\|_2)^{d(d-1)/2}).
    \end{align}
    Thus, combining these upper bounds, we obtain Eq.~\eqref{eq:d-lambda-upper-bound}.
    The Weyl dimension formula~\eqref{eq:weyl-dimension}
    with Eq.~\eqref{eq:lambda-difference-upper-bound} provides Eq.~\eqref{eq:m-lambda-d-upper-bound}.
    Finally, the Weyl dimension formula shows
    \begin{align}
        \label{eq:weyl-dimension-d3}
        {m_\lambda^{(d^3)} \over m_\lambda^{(d)}} = \prod_{i=1}^{d} \prod_{j=d+1}^{d^3} {\lambda_i+j-i \over j-i} \leq (k+d^3-1)^{d(d^3-d)} = O_d\qty(k^{d^4-d^2}),
    \end{align}
    which provides Eq.~\eqref{eq:m-lambda-d3-upper-bound}.
    Using the upper bounds in Eqs.~\eqref{eq:d-lambda-upper-bound}, \eqref{eq:m-lambda-d-upper-bound}, and \eqref{eq:m-lambda-d3-upper-bound}, we have
    \begin{align}
        \sum_{\lambda\vdash_d k} m_\lambda^{(d)} \sqrt{d_\lambda m_\lambda^{(d^3)}} = O\qty(d^{k/2} k^{(2d^4-d^2-2d+1)/4}) \sum_{\lambda\vdash_d k} (1+\|\bm x\|_2)^{d^2-d/2} e^{-c\|\bm x\|_2^2/2}.
    \end{align}
    For simplicity, we assume that $k$ is a multiple of $d$.
    Then, the summation can be upper bounded by
    \begin{align}
        \sum_{\lambda\vdash_d k} (1+\|\bm x\|_2)^{d^2-d/2} e^{-c\|\bm x\|_2^2/2}
        &\leq \sum_{\bm{x} \in k^{-1/2}A_{d-1}} (1+\|\bm x\|_2)^{d^2-d/2} e^{-c\|\bm x\|_2^2/2} \label{eq:x-range}\\
        &= O_d\qty(k^{(d-1)/2} \int_{\mathbb{R}^{d-1}} \dd \bm{x} (1+\|\bm{x}\|_2)^{d^2-d/2} e^{-c\|\bm{x}\|_2^2/2})\\
        &= O_d\qty(k^{(d-1)/2}),
    \end{align}
    where $A_{d-1}$ is the set of lattice points defined by
    \begin{align}
        A_{d-1} \coloneqq \left\{\bm z\in\mathbb Z^d\;\middle|\;\sum_i z_i=0\right\}.
    \end{align}
    Thus, we obtain
    \begin{align}
        C'_{\mfS_k} = O_d\qty(k^{d^4-(d^2+1)/2}).
    \end{align}
    We can show the same bound for arbitrary $k$ by shifting $\bm{x}$ by $O(k^{-1/2})$ in Eq.~\eqref{eq:x-range}, which does not change the order of the summation.

    We then prove a converse lower bound on $C_{\mfS_k}$ using the set $\mathsf{Y}_{\mathrm{hard}}$ defined in Eq.~\eqref{eq:hard-young-set}.
    We write $\lambda_i = k/d + \sqrt{k} x_i$ for $\lambda\in \mathsf{Y}_{\mathrm{hard}}$ with $\bm x \in \mathsf{X}_d$, and show the following lower bounds for $\lambda\in \mathsf{Y}_{\mathrm{hard}}$:
    \begin{align}
        \label{eq:d-lambda-lower-bound}
        d_\lambda &= \Omega_d\qty(d^k k^{-(d^2+d-2)/4}),\\
        \label{eq:m-lambda-d-lower-bound}
        m_\lambda^{(d)} &= \Omega_d\qty(k^{d(d-1)/4}),\\
        \label{eq:m-lambda-d3-lower-bound}
        m_\lambda^{(d^3)} &= \Omega_d\qty(k^{d^4-d^2+d(d-1)/4}).
    \end{align}
    We first show a lower bound on $d_\lambda$ using the Frobenius formula~\eqref{eq:hook-length-modified}.
    We have ${k!\over \Lambda!} = \Theta_d(k^{-d(d-1)/2})$.
    Since $p_i = \widetilde{\lambda}_i/\Lambda = 1/d+O_d(k^{-1/2})$, we have $H(\bm p) = \log d + O_d(k^{-1})$.
    Using the Stirling-based estimate shown in Eq.~\eqref{eq:stirling-combinatorics}, we have
    \begin{align}
        {\Lambda! \over \prod_{i=1}^{d} \widetilde{\lambda}_i!} = \Theta_d\qty(d^k k^{-(d-1)/2}).
    \end{align}
    Since
    \begin{align}
        \abs{(x_i-x_i^\star)-(x_j-x_j^\star)}^2 \leq 2[(x_i-x_i^\star)^2+(x_j-x_j^\star)^2] \leq 2\sum_{i=1}^{d}(x_i-x_i^\star)^2 \leq 1/4
    \end{align}
    holds for any $\bm{x}\in \mathsf{X}_d$, we have
    \begin{align}
        x_i - x_j = (x_i-x_i^\star)-(x_j-x_j^\star) + (j-i) \geq j-i-1/2 \geq 1/2,
    \end{align}
    for any $i<j$.
    Thus, we have $\widetilde{\lambda}_i-\widetilde{\lambda}_j \geq \sqrt{k}/2$ for any $i<j$, which shows
    \begin{align}
        \label{eq:lambda-difference-lower-bound}
        \prod_{1\leq i<j\leq d} (\widetilde{\lambda}_i-\widetilde{\lambda}_j) = \Theta_d(k^{d(d-1)/4}).
    \end{align}
    Combining these bounds, we obtain Eq.~\eqref{eq:d-lambda-lower-bound}.
    The Weyl dimension formula~\eqref{eq:weyl-dimension} with Eq.~\eqref{eq:lambda-difference-lower-bound} provides Eq.~\eqref{eq:m-lambda-d-lower-bound}.
    Finally, the Weyl dimension formula~\eqref{eq:weyl-dimension-d3} with $\lambda_i = \Theta_d(k)$ provides Eq.~\eqref{eq:m-lambda-d3-lower-bound}.
    Since $\abs{\mathsf{Y}_\mathrm{hard}} = \Theta_d(k^{(d-1)/2})$ holds, we have
    \begin{align}
        C_{\mfS_k} &\geq {1\over d^k} \qty(\sum_{\lambda\in \mathsf{Y}_\mathrm{hard}} m_\lambda^{(d)} \sqrt{d_\lambda [m_\lambda^{(d^3)} - m_\lambda^{(d)}]})^2\\
        &= {1\over d^k} \qty[\Omega_d(k^{(d-1)/2}) \cdot \Omega_d(k^{d(d-1)/4}) \cdot \sqrt{\Omega_d(d^k k^{-(d^2+d-2)/4}) \cdot \Omega_d(k^{d^4-d^2+d(d-1)/4})}]^2\\
        &= \Omega_d(k^{d^4-(d^2+1)/2}).
    \end{align}
\end{proof}

\end{document}